\documentclass[12pt]{article}
\usepackage[utf8]{inputenc}
\usepackage[nosort]{cite}
\usepackage[usenames, dvipsnames]{xcolor}
\usepackage{graphicx}
\usepackage{multicol}
\usepackage{amsfonts}
\usepackage{amssymb}
\usepackage{amsmath}
\usepackage{heck}
\usepackage{afterpage}
\usepackage{setspace}
\usepackage{verbatim}
\usepackage{color}
\usepackage{longtable}
\usepackage{float}
\usepackage{subcaption}
\usepackage{epsfig}
\usepackage{enumerate}
\usepackage{epstopdf}
\usepackage[enableskew, vcentermath]{youngtab}
\usepackage{adjustbox}
\usepackage{multirow}
\usepackage{tikz}
\usepackage[margin=1in]{geometry}
\usepackage{titletoc}
\usepackage{hyperref}
\usepackage[percent]{overpic}
\usepackage{gensymb}
\usepackage{mathtools}%
\usepackage{tikz-cd}
\usepackage{quiver}
\usepackage{booktabs}
\providecommand{\U}[1]{\protect\rule{.1in}{.1in}}

\pdfoutput=1
\newsavebox{\mysavebox}

\hypersetup{colorlinks,citecolor=black,filecolor=black,linkcolor=black,urlcolor=black}
\usetikzlibrary{decorations.markings}

\numberwithin{equation}{section}

\newcommand{\R}{\mathbb R}

\newcommand{\RP}{\mathbb{RP}}
\newcommand{\CP}{\mathbb{CP}}

\usetikzlibrary{chains}
\allowdisplaybreaks

\newcommand{\ba}{\begin{eqnarray}}
\newcommand{\ea}{\end{eqnarray}}
\newcommand{\cA}{\mathcal{A}}

\newcommand{\cE}{\mathcal{E}}

\newcommand{\be}{\begin{equation}}
\newcommand{\ee}{\end{equation}}

\tikzstyle{startstop} = [rectangle, rounded corners, minimum width=3cm, minimum height=1cm,text centered, draw=black, fill=blue!10]
\tikzstyle{startstop} = [rectangle, rounded corners, minimum width=3cm, minimum height=1cm,text centered, draw=black, fill=blue!10]
\tikzstyle{io} = [trapezium, trapezium left angle=70, trapezium right angle=110, minimum width=3cm, minimum height=1cm, text centered, draw=black, fill=blue!30]
\tikzstyle{process} = [rectangle, minimum width=3cm, minimum height=1cm, text centered, draw=black, fill=orange!30]
\tikzstyle{decision} = [diamond, minimum width=3cm, minimum height=1cm, text centered, draw=black, fill=green!30]
\tikzstyle{arrow} = [thick,->,>=stealth]
\tikzset{->-/.style={decoration={
  markings,
  mark=at position #1 with {\arrow[scale=2.4]{>}}},postaction={decorate}}}
\makeatletter \@addtoreset{equation}{section} \makeatother

\usepackage{amsthm}
\usepackage[capitalize, noabbrev]{cleveref}
\newcommand{\Spin}{\mathrm{Spin}}

\newcommand{\SSO}{\mathrm{SO}}
\newcommand{\SL}{\mathrm{SL}}
\newcommand{\Mp}{\mathrm{Mp}}
\newcommand{\GL}{\mathrm{GL}}

\newcommand{\F}{\mathbb F}
\newcommand{\Z}{\mathbb Z}

\newcommand{\MTSpin}{\mathit{MTSpin}}
\newcommand{\inj}{\hookrightarrow}

\newcommand{\pt}{\mathrm{pt}}
\newcommand{\term}{\emph}
\newcommand{\BP}{\mathit{BP}}
\newcommand{\ko}{\mathit{ko}}

\newcommand{\Sq}{\mathrm{Sq}}
\newcommand{\abs}[1]{\left\lvert #1 \right\rvert}
\newcommand{\ang}[1]{\left\langle #1 \right\rangle}
\newcommand{\set}[1]{\left\{ #1 \right\}}
\newcommand{\bl}{\text{--}}

\newcommand{\C}{\mathbb C}

\newcommand{\Ext}{\mathrm{Ext}}
\newcommand{\Hom}{\mathrm{Hom}}

\newcommand{\paren}[1]{\left( #1\right)}

\newcommand{\id}{\mathrm{id}}
\newcommand{\ExZ}{\mathbb E}

\usepackage{xspace}

\usepackage{xparse}
\DeclareDocumentCommand{\shortexact}{s O{} O{} mmmm}{
\IfBooleanTF{#1}{
\begin{tikzcd}[ampersand replacement=\&]
	1 \& {#4}
	\&  {#5}
	\& {#6}
	\& 1#7
	\arrow[from=1-1, to=1-2]
	\arrow["#2", from=1-2, to=1-3]
	\arrow["#3", from=1-3, to=1-4]
	\arrow[from=1-4, to=1-5]
\end{tikzcd}
}{ % no star
\begin{tikzcd}[ampersand replacement=\&]
	0 \& {#4}
	\&  {#5}
	\& {#6}
	\& 0#7
	\arrow[from=1-1, to=1-2]
	\arrow["#2", from=1-2, to=1-3]
	\arrow["#3", from=1-3, to=1-4]
	\arrow[from=1-4, to=1-5]
\end{tikzcd}
}}

\usepackage[newparttoc]{titlesec}
\usepackage{titletoc}
\titleclass{\subpart}{straight}[\part]

\newcounter{subpart}

\renewcommand{\thesubpart}{\Roman{part}.\Alph{subpart}}
\newcommand{\subpartname}{Subpart}

\titleformat{\subpart}{\normalfont\Large\bfseries}%
{\subpartname~\thesubpart}{1pc}{\Large\bfseries}

\titlespacing{\subpart}{0pt}{0pt}{0pt}

\titlecontents{subpart}[0pt]{\addvspace{1pc}\normalfont\bfseries}%
{\thecontentslabel\enspace \enspace\large}%
{\normalfont\large\bfseries}{\hspace{2em plus 1fill}\large\contentspage}

\usepackage{subcaption}
\usepackage{spectralsequences}
\usepackage{adamsmacros}
\newcommand{\AdamsTower}[1]{\DoUntilOutOfBounds{
	\class[#1](\lastx, \lasty+1)
	\structline[#1]
}}
\newcommand{\CetaExt}[1]{
    \foreach \y in {0, ..., 10} {
        \class(2*\y+#1, \y)\AdamsTower{}
    }
}
\newcommand{\TwoCEtaExt}[3]{
\begin{scope}[#1]
    \begin{scope} % TODO: macro
        \CetaExt{#2}
        \classoptions["\kappa_{#3}"](#2, 0, -1)
        \classoptions["\lambda_{#3}"](2+#2, 1, -1)
        \classoptions["\mu_{#3}"](4+#2, 2, -1)
        \classoptions["\xi_{#3}"](6+#2, 3, -1)
        \classoptions["w\kappa_{#3}"](8+#2, 4, -1)
    \end{scope}
    \begin{scope}[fill]
        \CetaExt{#2+1}
        \classoptions["\theta\kappa_{#3}"](1+#2, 0, -1)
        \classoptions["\theta\lambda_{#3}"](3+#2, 1, -1)
        \classoptions["\theta\mu_{#3}"](5+#2, 2, -1)
        \classoptions["\theta\xi_{#3}"](7+#2, 3, -1)
    \end{scope}
\end{scope}
}

\newtheorem{thm}[equation]{Theorem}
\newtheorem{lem}[equation]{Lemma}
\newtheorem{prop}[equation]{Proposition}
\newtheorem{cor}[equation]{Corollary}
\theoremstyle{definition}
\newtheorem{defn}[equation]{Definition}
\newtheorem{exm}[equation]{Example}
\theoremstyle{remark}
\newtheorem{rem}[equation]{Remark}

\crefname{thm}{Theorem}{Theorems}
\crefname{lem}{Lemma}{Lemmas}
\crefname{prop}{Proposition}{Propositions}
\crefname{defn}{Definition}{Definitions}
\crefname{exm}{Example}{Examples}
\crefname{rem}{Remark}{Remarks}
\crefname{cor}{Corollary}{Corollaries}

\title{Cobordism Utopia: \\ U-Dualities, Bordisms, and the Swampland}
\author{Noah Braeger}
\author{Arun Debray}
\author{Markus Dierigl}
\author{Jonathan J.\ Heckman}
\author{Miguel Montero}

\date{May 2025}

\begin{document}

\preprint{CERN-TH-2025-103 \\ IFT-025-51}

\institution{PENN}{\centerline{$^1$Department of Physics and Astronomy, University of Pennsylvania, Philadelphia, PA 19104, USA}}
\institution{KENTUCKY}{\centerline{$^2$Department of Mathematics, University of Kentucky, 719 Patterson Office Tower, Lexington, KY
40506-0027}}
\institution{MUNICH}{\centerline{$^3$Theoretical Physics Department, CERN, 1211 Geneva 23, Switzerland}}
\institution{PENNMATH}{\centerline{$^4$Department of Mathematics, University of Pennsylvania, Philadelphia, PA 19104, USA}}
\institution{MADRID}{\centerline{$^5$Instituto de Física Teórica IFT-UAM/CSIC, C/ Nicolás Cabrera 13-15, 28049 Madrid, Spain}}

\authors{Noah Braeger\worksat{\PENN}\footnote{e-mail: {\tt braeger@sas.upenn.edu}},
 Arun Debray\worksat{\KENTUCKY}\footnote{e-mail: {\tt a.debray@uky.edu}},
Markus Dierigl\worksat{\MUNICH}\footnote{e-mail: {\tt markus.dierigl@cern.ch}},\\[4mm]
Jonathan J.\ Heckman\worksat{\PENN,\PENNMATH}\footnote{e-mail: {\tt jheckman@sas.upenn.edu}}, and
Miguel Montero\worksat{\MADRID}\footnote{e-mail: {\tt miguel.montero@uam.es}}}

\longabstract{\noindent The U-dualities of maximally supersymmetric supergravity theories lead to celebrated non-perturbative constraints on the structure of quantum gravity. They can also lead to the presence of global symmetries since manifolds equipped with non-trivial duality bundles can carry topological charges captured by non-trivial elements of bordism groups. The recently proposed Swampland Cobordism Conjecture thus predicts the existence of new singular objects absent in the low-energy supergravity theory, which break these global symmetries.

We investigate this expectation in two directions, involving the different choices of U-duality groups $G_U$, as well as $k$, the dimension of the closed manifold carrying the topological charge. First, we compute for all supergravity theories in dimension $3 \leq D \leq 11$ the bordism groups $\Omega_1^{\text{Spin}}(BG_U)$. Second, we treat in detail the case of $D = 8$, computing all relevant bordism groups $\Omega_k^{\text{Spin}}(BG_U)$ for $1 \leq k \leq 7$. In all cases, we identify corresponding string, M-, or F-theory backgrounds which implement the required U-duality defects. In particular, we find that in some cases there is no purely geometric background available which implements the required symmetry-breaking defect. This includes non-geometric twists as well as non-geometric strings and instantons.

This computation involves several novel computations of the bordism groups for $G_U = \SL(2,\mathbb{Z}) \times \SL(3,\mathbb{Z})$, which localizes at primes $p=2,3$. Whereas an amalgamated product structure greatly simplifies the calculation of purely $\SL(2,\mathbb{Z})$ bundles, this does not extend to $\SL(3,\mathbb{Z})$. Rather, we leverage the appearance of product / ring structures induced from cyclic subgroups of $G_U$ which naturally act on the relevant bordism groups.}

\maketitle

\newpage

\thispagestyle{empty}
\begin{figure}[H]
\vspace*{-2cm}
\makebox[\linewidth]{\includegraphics[width = 1.15 \linewidth]{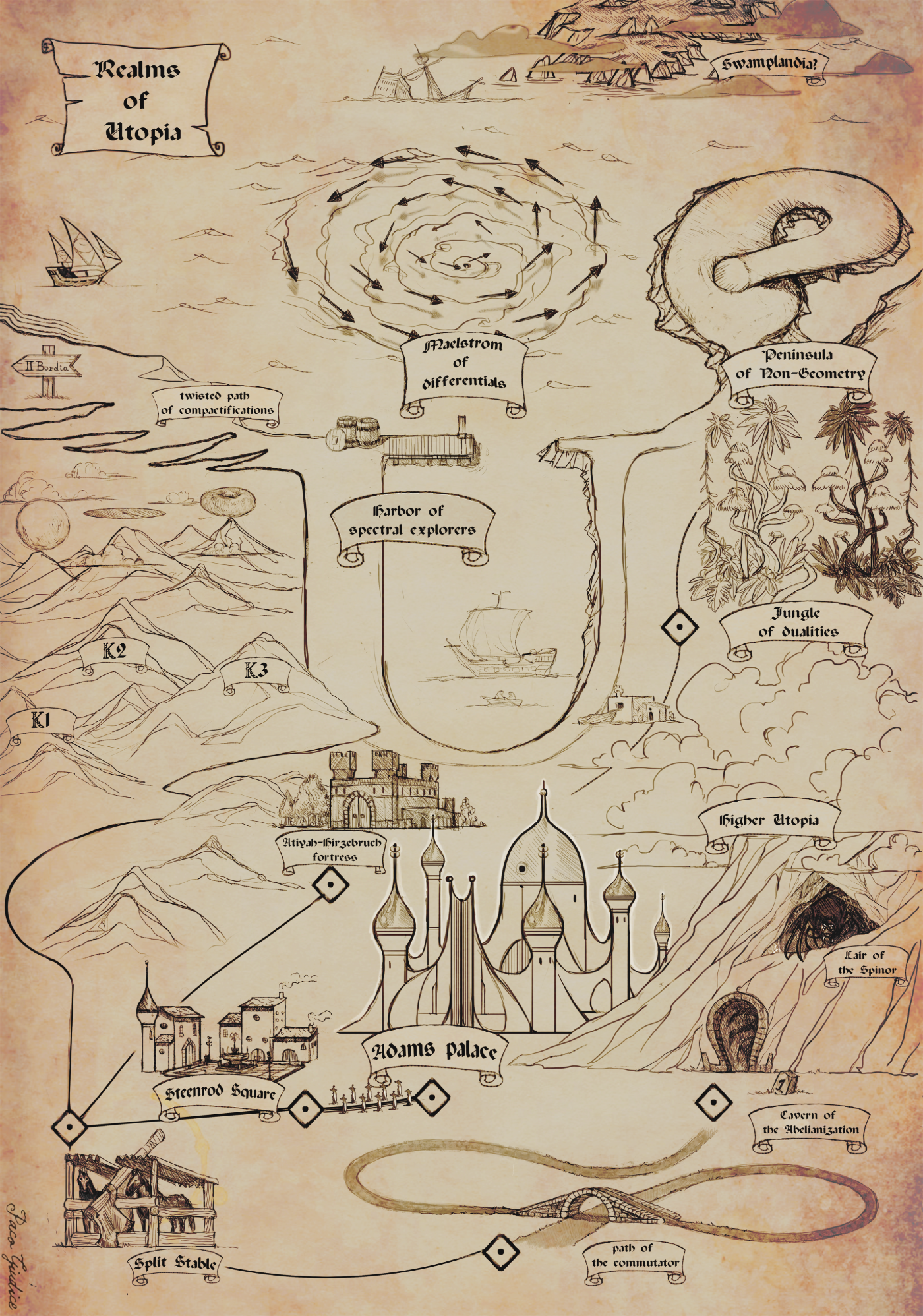}}
\end{figure}

\tableofcontents

\newpage

\section{Introduction and Conclusions}
\label{sec:intro}

There is a general expectation that quantum gravity has no global symmetries.\footnote{See, e.g., \cite{Banks:1988yz, Susskind:1995da} as well as more recent progress such as \cite{Banks:2010zn, Harlow:2018tng, Hsin:2020mfa, Harlow:2020bee, Daus:2020vtf, Bah:2022uyz, Cvetic:2023pgm, Heckman:2024oot, Bah:2024ucp}.} Much of the evidence for this comes from interrogating the fate of global symmetries in specific spacetime backgrounds involving black holes, wormholes, AdS/CFT, and specific top down string constructions. On the other hand, it is natural to expect that the very nature of having a quantum mechanically fluctuating spacetime topology ought to play a role in ruling out global symmetries. Presumably, then, such putative symmetries are broken or gauged.

Along these lines, it is natural to consider field profiles in the effective field theory which can be interpreted as an extended operator.\footnote{Here we use the terminology of quantum field theory (QFT) to refer to these heavy non-dynamical objects as extended operators.} One way to build these is with the spacetime geometry itself. Introducing a manifold $X_k$ of dimension $k$, we can in principle also decorate $X_k$ with additional smooth structures such as non-trivial field profiles / gauge bundles over $X_k$. Gluing this into the rest of the spacetime results in a geometry extended over $k+1$ dimensions with boundary given by our $X_k$ (equipped with suitable asymptotic field profiles). In particular, this can be interpreted as an extended object which fills $D-(k+1)$ spacetime dimensions. If we cannot deform away this extended operator in the low energy effective field theory then it carries a charge under a global symmetry. The Swampland Cobordism Conjecture of \cite{McNamara:2019rup} asserts that in a full theory of quantum gravity, all such $X_k$'s are in fact bordant, i.e., they can be viewed as the boundary of a bulk $(k+1)$-dimensional space. This is often stated as the condition that in quantum gravity, the bordism group is trivial:
\begin{equation}\label{eq:OMEGAQG}
\Omega_{k}^{\mathrm{QG}} = 0
\end{equation}
It is important to note that while there is no precise formulation of $\Omega_{k}^{\mathrm{QG}}$, this conjecture has many direct practical consequences. Indeed, what makes this conjecture so powerful is that this in turn requires supplementing the low energy effective field theory by additional dynamical degrees of freedom which would be \textit{singular} in the original low energy effective field theory.\footnote{If they were smooth field configurations in the original low energy effective field theory they could be deformed away.}

Said differently, the appearance of an extended operator which cannot be smoothed away implies it carries a charge under a generalized global symmetry (in the sense of \cite{Gaiotto:2014kfa}).\footnote{For a bordism group $\Omega_k$, this amounts to a $p$-form symmetry $\mathrm{Hom}(\Omega_k, \mathrm{U}(1))$ with $p = D - (k+1)$.} From a bottom up point of view, one can attempt to either gauge or break this symmetry. For higher-form symmetries, breaking involves supplementing the spectrum by additional degrees of freedom which can terminate at the end of the extended operator.\footnote{Recall that in the framework of \cite{Gaiotto:2014kfa} one considers topological symmetry operators which link with the extended operator in question. This linking is destroyed if the extended operator can now terminate.} Alternatively, one might contemplate gauging the generalized symmetry. Even if this can be carried out (assuming the gauging is not obstructed by anomalies) it turns out that this results in another, magnetic dual generalized global symmetry, and the spectrum of lightest objects presumably also shuffles around. One is thus still faced with the issue of how to get rid of this new global symmetry. So, enriching the spectrum by additional degrees of freedom is inevitable in either case.

Assuming the Swampland Cobordism Conjecture is true, this suggests a natural strategy for testing the consistency of a given low energy effective field theory, and in particular whether the spectrum of dynamical objects in the theory is complete. First, given a theory with a known class of admissible bundle configurations, we can construct mathematical equivalence relations between smooth manifolds equipped with these structures. These are the bordism groups and we schematically label these as $\Omega_{k}^{\mathcal{G}}$ (we shall be more precise later). A priori, there is no reason for these bordism groups to vanish and as such, the appearance of a non-trivial element of $\Omega_{k}^{\mathcal{G}}$ predicts the existence of a dynamical object of codimension $k+1$ in the effective theory, namely it fills $D - (k+1)$ dimensions. Tests of the Swampland Cobordism Conjecture in a variety of different contexts have now been performed, e.g., in \cite{McNamara:2019rup, Montero:2020icj, Dierigl:2020lai, McNamara:2021cuo, Blumenhagen:2021nmi, Buratti:2021yia, Debray:2021vob, Andriot:2022mri, Dierigl:2022reg, Blumenhagen:2022bvh, Velazquez:2022eco, Angius:2022aeq, Blumenhagen:2022mqw, Angius:2022mgh, Blumenhagen:2023abk, Heckman, Dierigl:2023jdp, Kaidi:2023tqo, Huertas:2023syg, Angius:2023uqk, Kaidi:2024cbx, Angius:2024pqk, Fukuda:2024pvu}. In many cases, this sheds new light on previously known string backgrounds, and in some cases, predicts the existence of new objects \cite{Dierigl:2022reg, Heckman, Kaidi:2023tqo, Kaidi:2024cbx, Fukuda:2024pvu}. One of the overarching themes in many of these examples is the interplay between a defect predicted by the Cobordism Conjecture and the realization of an explicit object / background in a UV completion of the effective field theory.

On the other hand, a given effective field theory can in principle have many different UV completions. A classic example is maximally supersymmetric supergravity in $D$ spacetime dimensions; it arises both from placing M-theory on $T^\ell$ as well as type IIB on $T^{\ell - 1}$ with $\ell = 11 - D$. These different characterizations are of course connected through dualities, providing a single coherent picture of
string / M- / F-theory and the spectrum of objects under various dualities. Indeed, the supergravity theory itself enjoys an important U-duality symmetry which acts on the solitonic objects of the spectrum \cite{Hull:1994ys, Witten:1995ex}. A celebrated feature of U-dualities is that they involve strong / weak dualities as well as T-dualities. As such, the full set of symmetries do not have an obvious geometric characterization in terms of a single higher-dimensional parent theory. Historically such dualities played a pivotal role in many aspects of the duality revolution of the mid 1990's.

Our aim in this paper will be to construct the bordism groups and generators associated with such U-dualities. More precisely, we shall be interested in the mathematical structures known as spin-bordism groups for such duality bundles, i.e., $\Omega^{\mathrm{Spin}}_{k}(BG_U)$, where $G_U$ denotes the U-duality group of a supergravity theory, $BG_U$ denotes its classifying space, and the superscript ``$\mathrm{Spin}$'' indicates that we require a spin structure for all of our manifolds. We shall refer to this as exploring the world of ``Utopia,'' namely we view it as a natural continuation of investigations already performed in the study of ``IIBordia,'' i.e., the bordism groups associated with type IIB dualities \cite{Dierigl:2020lai, Debray:2021vob, Dierigl:2022reg, Heckman, Dierigl:2023jdp}.

To set the stage, let us briefly recall some recent explorations of IIBordia. As is well-known, type IIB string theory and its non-perturbative lift to F-theory enjoys a duality symmetry $\mathrm{SL}(2,\mathbb{Z})$. Taking into account reflections on the F-theory torus associated with worldsheet parity and left-moving fermion parity, this extends to the $\mathrm{Pin}^{+}$ cover of $\mathrm{GL(2,\mathbb{Z})}$.\footnote{See \cite{Pantev:2016nze} for the $\mathrm{Spin}$-lift and \cite{Tachikawa:2018njr} for the $\mathrm{Pin}^{+}$ lift.} Moreover, the interplay between the $\mathrm{Spin} / \mathrm{Pin}^{+}$ structures means that the relevant bordism groups involve a non-trivial twisting with the accompanying duality bundles. A very helpful feature of these duality groups is that they admit an amalgamated product structure $G_\mathrm{IIB} = H_1 \ast_{K} H_2$, where each of the constituents involve either cyclic groups or dihedral groups. Treating the bordism groups as a generalized (co)homology theory, there is a corresponding Mayer-Vietoris long exact sequence which breaks up the contributions into far simpler constituent building blocks. These in turn can be constrained / computed using the Adams spectral sequence \cite{Ada58}. In physical terms, these simpler building blocks correspond to quite simple string backgrounds, such as 7-branes with a constant axio-dilaton profile. As found in \cite{Dierigl:2022reg, Heckman}, this leads to the discovery of essentially one new non-supersymmetric object, the reflection 7-brane (R7-brane), as well as a number of supersymmetric and non-supersymmetric string backgrounds. Geometrically, the topology of a corresponding torus fiber decorating the ten-dimensional spacetime produces a beautiful geometric characterization of all of these duality defects.

But from the perspective of U-dualities, $\mathrm{SL}(2,\mathbb{Z})$ is merely a special case of a more general structure. This strongly motivates the exploration of Utopia with its broader class of non-perturbative and stringy dualities.

It is here that we encounter a number of surprises compared with the relatively ``tame'' world of IIBordia. First of all, the analogous $\mathrm{Spin}$ and $\mathrm{Pin}$ lifts of the duality group action are (as yet) unknown. As such, in this work we shall primarily concentrate on just spin-bordism groups associated with duality bundles. At a technical level, some of the key simplifications used in the exploration of IIBordia no longer work. In particular, whereas the IIBordia duality groups admit an amalgamated product structure, this is unavailable for the U-dualities of the maximally supersymmetric $D \leq 8$ supergravity theories. This in turn means that the identification of ``simple physical building blocks'' is also somewhat obscure. The flip side of this is that some of the resulting bordism groups and generators are not fully geometrizable in terms of compactifications descending from a higher-dimensional theory. As such, the bordism group generators we find in these cases directly probe the non-geometric regime of quantum gravity!

The core of our analysis will center on the computation of the bordism groups associated with U-dualities of maximally supersymmetric supergravity theories in $D$ spacetime dimensions:
\begin{equation}
\Omega_{k}^{\mathrm{Spin}}(BG_U^{D}),
\end{equation}
where we work in the range of $3 \leq D \leq 10$ and we restrict to $1 \leq k \leq D-1$, i.e., we focus on codimension-two objects up to instanton-like objects.\footnote{It is also of interest to consider $k = 0$ which tells us about domain walls as well as $k = D$ which tells us about discrete theta-angles. The case of $k = D+1$ is also relevant in the study of duality anomalies of Dai-Freed type \cite{Dai:1994kq}. For $k = 0$, the answer turns out to be either trivial or $\mathbb{Z}$, so we neglect it in what follows. For $k \geq D$, we encounter some technical complications and so defer this to future work.} So, we can scan over both $D$, the total spacetime dimension of the effective field theory, as well as $k$, the dimension of bounding spacetime manifolds equipped with duality bundles.

In this paper we specialize along two different routes, depending on whether we scan over $D$ or $k$. Along these lines, we determine for all $D$, the values of $\Omega_{1}^{\mathrm{Spin}}(BG_U^{D})$, i.e., the spectrum of codimension-two defects. An application of the Atiyah-Hirzebruch spectral sequence reveals that these are controlled by the Abelianization of $G_{U}^{D}$, and as such, can be directly accessed by purely group theoretic techniques. In particular, we find that for $D \leq 7$, the Abelianization of $G^{D}_U$ is trivial, namely the group is ``perfect.'' Any contribution to the bordism group in these cases thus originates from $\Omega_{1}^{\mathrm{Spin}}(\mathrm{pt}) \simeq \mathbb{Z}_2$. One can also use some general techniques to find analogous results for any $D$ and $k = 2$ \cite{BraegerThesis}. Beyond this, however, requires specific details on the structure of the bordism groups, something we defer to future work.

We also determine $\Omega_{k}^{\mathrm{Spin}}(BG_U^{8\mathrm{d}})$ for all $1 \leq k \leq 7$ i.e., the spectrum of defects in 8d maximally supersymmetric supergravity. While there is no amalagamated product structure available for all of:
\begin{equation}
G^{8\mathrm{d}}_U = \mathrm{SL}(2,\mathbb{Z}) \times \mathrm{SL}(3,\mathbb{Z}),
\end{equation}
there is some additional structure available in this case which makes the computation tractable compared with the $D \leq 7$ case (which we defer to future work). The first key idea is (much as in IIBordia) to localize the computation of bordism groups one prime at a time. This is expected to work since the torsional pieces decompose into summands of the form $\mathbb{Z}_{p^{\ell}}$ for suitable prime powers $p^{\ell}$. In particular, we find that in the case at hand, localization occurs at primes $p = 2,3$. In this case, the relevant pieces of the bordism groups often simplify to simpler bundle structures involving the classifying spaces of cyclic groups, dihedral and symmetric groups, and their smash products.\footnote{We give the precise definition of smash products later.} In principle, this boils the computation down to more tractable questions in group cohomology which can in some cases be located in the literature. To proceed further, however, we also leverage the fact that these bordism groups come with a module structure, admitting a group action by suitable cyclic groups. In physical terms, this suggests an organization of some of our bordism defects into symmetry multiplets, but we defer a full treatment to future work.

The end result of these computations is that at least for $D = 8$ maximally supersymmetric supergravity, we have a full list of bordism groups, generators, as well as candidate string / M- / F-theory backgrounds which implement these defects. In particular, a few striking lessons emerge from this general picture. The first statement is that for dualities involving just the $\mathrm{SL}(3,\mathbb{Z})$ factor, there is a geometric realization available involving M-theory backgrounds with a suitable $T^3$ fibration. Likewise, dualities involving just the $\mathrm{SL}(2,\mathbb{Z})$ factor can be interpreted in IIB vacua as coming from fibration of a spacetime 2-torus. On the other hand, there are also backgrounds where \textit{all} of the duality bundle structure group participates. In these cases it appears necessary to appeal to more non-geometric ingredients such as exceptional field theory (a non-perturbative generalization of double field theory) to give a target space ``geometry'' interpretation.\footnote{See e.g., references \cite{Hull:2004in, Dabholkar:2005ve, Hull:2006va, Hull:2009mi, Berman:2011jh, Coimbra:2011ky, Hohm:2013jma, Hohm:2013pua} for recent work in this direction.} As one might suspect, identifying the mathematical computation of bordism groups and generators in these cases is also more challenging, but also more rewarding since it provides access to the non-geometric corners of quantum gravity.

The resulting bordism groups for $D = 8$ supergravity with maximal supersymmetry are given in Table \ref{tab:8dbordgroups}. We see that they are far from trivial, suggesting the existence of many conserved topological charges which have to be taken care of.
\begin{table}
\centering
\renewcommand{\arraystretch}{1.2}
    \begin{tabular}{ c  c }
    \toprule
    $k$ & $\Omega^{\text{Spin}}_k \big( B \text{SL}(2,\mathbb{Z}) \times B\text{SL}(3,\mathbb{Z})\big)$ \\ \midrule
    $1$ & $\mathbb{Z}_2  \oplus \mathbb{Z}_3 \oplus \mathbb{Z}_4$ \\
    $2$ &  $\mathbb{Z}_2^{\oplus 4}$ \\
    $3$ & $\mathbb{Z}_3^{\oplus 3} \oplus \mathbb{Z}_2^{\oplus 3} \oplus \mathbb{Z}_8^{\oplus 3}$ \\
    $4$ & $\mathbb{Z} \oplus \mathbb{Z}_3^{\oplus 2} \oplus \mathbb{Z}_2^{\oplus 3} \oplus \mathbb{Z}_4^{\oplus 2}$ \\
    $5$ & $\mathbb{Z}_3^{\oplus 2} \oplus \mathbb{Z}_9 \oplus \mathbb{Z}_2^{\oplus 5} \oplus \mathbb{Z}_4$ \\
    $6$ & $\mathbb{Z}_3^{\oplus 2} \oplus \mathbb{Z}_2^{\oplus 3} \oplus \mathbb{Z}_4^{\oplus 2}$ \\
    $7$ & $\mathbb{Z}_3^{\oplus 2} \oplus \mathbb{Z}_9^{\oplus 3} \oplus \mathbb{Z}_2^{\oplus 6} \oplus \mathbb{Z}_8^{\oplus 2} \oplus \mathbb{Z}_{16}^{\oplus 2} \oplus \mathbb{Z}_{32}$ \\ \bottomrule
    \end{tabular}
\renewcommand{\arraystretch}{1.0}
\caption{Bordism groups for 8d supergravity theories with 32 supercharges (the Spin bordism $\Omega^{\text{Spin}}_k (\text{pt})$ contributes a $\mathbb{Z}_2$ for $k=1\,,2$ and a $\mathbb{Z}$ summand for $k =4$).}
\label{tab:8dbordgroups}
\end{table}

Part of this work is based on the recently completed PhD thesis of N. Braeger \cite{BraegerThesis}.

The rest of this paper is organized as follows: In the remainder of this section we provide a summary of the bordism generators and symmetry-breaking defects, as well as a collection of general lessons and surprises. We also present some potential areas of future investigation. In Section~\ref{sec:Udual} we review U-duality groups of maximal supergravity theories as well as its full geometrization within the context of exceptional field theory. We then focus on the particular U-duality group SL$(2,\mathbb{Z}) \times \text{SL}(3,\mathbb{Z})$ for 8d supergravity and its interpretation in torus compactifications of M-theory and type IIB in Section~\ref{sec:Udual8d}. Part \ref{part:physics} contains a detailed description of the singular fibers in the various UV descriptions in Section~\ref{sec:monodrom}, a general analysis of codimension-two defects in $3\leq D \leq 7$ in Section~\ref{sec:gencodim2}, as well as the in-depth investigation of the bordism generators as well as a physics interpretation of their associated defects for $D=8$ in Section~\ref{sec:defects}.

The derivation of the bordism groups is given in Part~\ref{part:math}. We begin by giving a general treatment of the bordism groups $\Omega_{k}^{\mathrm{Spin}}(BG_U^{D})$ for maximally supersymmetric supergravity theories in dimensions $3 \leq D \leq 10$. Specializing further to 8d supergravity, we begin by laying out our general plan of attack, namely we show how the calculation localizes at primes $2$ and $3$, and in part \ref{Part:p=3} we carry out the relevant bordism group calculations at prime $3$ and in part \ref{Part:p=2} we perform the same analysis at prime $2$. Some additional technical details are deferred to the Appendices; in Appendix \ref{app:cryssingfiber} we analyze the shape of singular torus fibrations in some prominent examples, and in Appendix \ref{App:Knit} we review the knit product which enters in the structure of the U-duality groups.

\subsection{Summary of the bordism groups and generators}
\label{subsec:summarydef}

\begin{table}
\centering
\renewcommand{\arraystretch}{1.3}
\resizebox{\textwidth}{!}{
\begin{tabular}{| c | c | c | c |}
    \hline
    $k$ & $\Omega^{\text{Spin}}_k \big( B \text{SL}(2,\mathbb{Z}) \times B \text{SL}(3,\mathbb{Z})\big)$ & Generators & Defect \\ \hline \hline   
    $1$ & $\mathbb{Z}_2$ & $S^1_+$ & spin defect \\
    & $\mathbb{Z}_3$ & $S^1_{\gamma_3}$ & $\mathcal{N}=(2,0)$ 6d SCFTs \\ 
    & $\mathbb{Z}_4$ & $S^1_{\gamma_4}$ & $\mathcal{N}=(2,0)$ 6d SCFTs \\ \hline 
    $2$ & $\mathbb{Z}_2$ & $S^1_+ \times S^1_+$ & spin defect on $S^1$ \\
    & $\mathbb{Z}_2$ & $S^1_+ \times S^1_{\gamma_4}$ & codimension-two defect on $S^1_+$ \\ 
    & $\mathbb{Z}_2 \oplus \mathbb{Z}_2$ & $S^1_{M_1^{(i)}} \times S^1_{M_2^{(i)}}$ & twisted compactification of additional \\ 
    &  &  & codimension-two objects (hosting SCFTs) \\ \hline
    $3$ & $\mathbb{Z}_3$ & $L^3_{3,\gamma_3}$ & type IIB on singular local Calabi-Yau \\
    & $\mathbb{Z}_3$ & $L^3_{3,\Gamma_3^{(1)}}$ & non-Higgsable cluster on $T^2$ \\ 
    & $\mathbb{Z}_3$ & $L^3_{3,\Gamma_3^{(2)}}$ & twisted compactification of 5d SCFTs \\
    & $\mathbb{Z}_2$ & $S^1_+ \times S^1_+ \times S^1_{\gamma_4}$ & codimension-two defect on $S^1_+ \times S^1_+$ \\ 
    & $\mathbb{Z}_2 \oplus \mathbb{Z}_2$ & $S^1_{\gamma_4} \times S^1_{M_1^{(i)}} \times S^1_{M_2^{(i)}}$ & (non-geometrically) twisted compactification \\
    & & & of codimension-two defect \\
    & $\mathbb{Z}_8$ & $L^3_{4,\gamma_4}$ & type IIB on singular geometry \\
    & $\mathbb{Z}_8$ & $L^3_{4,\Gamma_4^{(1)}}$ & F-theory on singular geometry \\
    & $\mathbb{Z}_8$ & $L^3_{4,\Gamma_4^{(2)}}$ & M-theory on singular geometry \\ \hline
    $4$ & $\mathbb{Z}$ & K3 & codimension-five spin defect \\ 
    & $\mathbb{Z}_3 \oplus \mathbb{Z}_3$ & $S^1_{\gamma_3} \times L^3_{3,\Gamma_3^{(i)}}$ & (non-geom.) twisted compactification of defects \\
    & $\mathbb{Z}_4 \oplus \mathbb{Z}_4$ & $S^1_{\gamma_4} \times L^3_{4,\Gamma_4^{(i)}}$ & (non-geom.) twisted compactification of defects \\
    & $\mathbb{Z}_2$ & $W_4$ & (topolog.) twisted compactification of defect \\
    & $\mathbb{Z}_2 \oplus \mathbb{Z}_2$ & $A \,, A'$ & (non-geom.) twisted compactification of defects \\ \hline
\end{tabular}}
\renewcommand{\arraystretch}{1.0}
    \caption{Bordism groups, their generators, and defects for $1 \leq k \leq 4$.}
    \label{tab:gen1to4}
\end{table}

In this section we summarize the individual generators of the bordism groups in Table \ref{tab:8dbordgroups} as well as their associated defects. Whenever possible, we describe the the configurations in a duality frame in which a geometric description is possible. The groups and generators for dimension $1 \leq k \leq 4$ are provided in Table \ref{tab:gen1to4} and the generators in dimension $5 \leq k \leq 7$ are given in Table \ref{tab:gen5to7}. We also include a brief description of the properties of the necessary defects and their relation to other (known) string and M-theory backgrounds.

In more detail, the various generators are given by:
\begin{itemize}
    \item{$S^1_+$: A circle, $S^1$, with periodic boundary conditions for fermions.}
    \item{$S^1_M$: A circle with non trivial duality bundle specified by the monodromy $M \in \text{SL}(2,\mathbb{Z}) \times \text{SL}(3,\mathbb{Z})$. For a detailed description of the monodromies see Section \ref{sec:monodrom}.}
    \item{$L^n_{k,M}$: An $n$-dimensional lens space obtained via the quotient $S^n/\mathbb{Z}_k$ with the monodromy $M$ when traversing the torsion 1-cycle. In case the lens space admits more than one spin structure we distinguish the two choices by $L^n_k$ and $\widetilde{L}^n_k$.}
    \item{$\mathbb{RP}^n_{M}$: Real projective space $S^n/\mathbb{Z}_2$ with the monodromy $M$ around the torsion 1-cycle.}
    \item{K3: The K3 manifold.}
    \item{$E$: The Enriques 4-manifold. It appears as a base for the generators of the form $(L^3_4 \times \text{K3})/\mathbb{Z}_2$ above, where the $\mathbb{Z}_2$ acts as complex conjugation on $\mathbb{C}^2$ which we use to describe the lens space and an involution on K3. The resulting geometry is a fibration of $L^3_4$ over the Enriques surface $E$.}
    \item{$W_4$: There are two equivalent ways to construct $W_4$; either one starts with $L^3_4 \times S^1$ and mods out a $\mathbb{Z}_2$ acting as complex conjugation on the lens space ($L^3_4 \sim \partial(\mathbb{C}^2 /\mathbb{Z}_4)$) and the antipodal map on $S^1$, or one starts with $S^3 \times S^1$ and mods out a $D_8$, whose $\mathbb{Z}_4$ subgroup acts to generate the lens space $S^3/\mathbb{Z}_4$ and the reflection $\mathbb{Z}_2$ as described above. The duality bundle is defined via the embedding $D_8 \hookrightarrow S_4 \hookrightarrow \text{SL}(3,\mathbb{Z})$ together with the fibration $D_8 \hookrightarrow (S^3 \times S^1) \rightarrow W_4$.}
    \item{$A$ and $A'$: The 4-manifold $\mathbb{RP}^3 \times S^1$ with duality bundle specified by $M_2^{(i)}$ around the circle factor and $\big(\gamma_4^2,M_1^{(i)}\big)$ when traversing the torsion 1-cycle of the $\mathbb{RP}^3$ factor. We choose, $i=1,2$ for $A$ and $A'$, respectively.} 
    \item{$Q^5_4$: This space is given, as in \cite{Heckman}, by the total space of the lens space bundle $L^3_4$ over $\mathbb{CP}^1$. The fibration is generated by regarding the covering $S^3$ of the lens space as the sphere of constant radius in the sum of complex line bundles $\mathcal{L}_1 \oplus \mathcal{L}_2$ over $\mathbb{CP}^1$, where in our case we always have
    \begin{equation}
        \mathcal{L}_1 \oplus \mathcal{L}_2 = H^{\pm2} \oplus \underline{\mathbb{C}} \,,
    \end{equation}
    with hyperplane bundle $H \cong \mathcal{O}(1)$ and trivial line bundle $\underline{\mathbb{C}} \cong \mathcal{O}(0)$. As opposed to the lens space $L^5_4$ these spaces admit a spin structure.}
    \item{$W_6$: The generator $W_6$ is geometrically the product $\mathbb{RP}^3 \times \mathbb{RP}^3$ with duality bundle specified by $(\gamma_4^2, M_2^{(1)})$ on the torsion 1-cycle of the first $\mathbb{RP}^3$ and $\widetilde{R} \in \text{SL}(3,\mathbb{Z})$ on the torsion 1-cycle of the second $\mathbb{RP}^3$ factor.}
    \item{$(T^4 \times \mathbb{RP}^3)_{\gamma_4, M_1^{(i)},M_2^{(i)}}$: The duality bundle on this manifold can be understood from different embeddings of $\mathbb{Z}_4 \times S_4$ into the U-duality group, using the indicated monodromies. The monodromies around each circle composing the $T^4$ factor is given by $\big(\gamma_4 \,, M_1^{(i)} \,, M_2^{(i)} \,, M_2^{(i)}\big)$. The monodromy on the torsion 1-cycle of $\mathbb{RP}^3$ is given by the duality element $\big(\gamma_4^2, M_1^{(i)}\big)$. We see that this manifold contains $A$/$A'$ as a factor.}
\end{itemize}
\begin{table}
\centering
\renewcommand{\arraystretch}{1.3}
\resizebox{\textwidth}{!}{
\begin{tabular}{| c | c | c | c |}
    \hline
    $k$ & $\Omega^{\text{Spin}}_k \big( B \text{SL}(2,\mathbb{Z}) \times B \text{SL}(3,\mathbb{Z})\big)$ & Generators & Defect \\ \hline \hline   
    $5$ & $\mathbb{Z}_9$ & $L^5_{3,\gamma_3}$ & type IIB on singular local Calabi-Yau (CY)\\ 
    & $\mathbb{Z}_3$ & $L^5_{3,(\gamma_3,\Gamma_3^{(1)})}$ & composite S-string (see \cite{Heckman}) \\
    & $\mathbb{Z}_3$ & $L^5_{3,(\gamma_3,\Gamma_3^{(2)})}$ & non-geometric string \\
    & $\mathbb{Z}_2 \oplus \mathbb{Z}_2$ & $S^1_{\gamma_4} \times A \,, S^1_{\gamma_4} \times A'$ & (non-geom.) twisted compactification of defects \\
    & $\mathbb{Z}_2$ & $S^1_{\gamma_4} \times W_4$ & (non-geometrically and topological) \\
    & & & twisted compactification of defect \\
    & $\mathbb{Z}_4$ & $Q^5_{4,\gamma_4}$ & (topolog.) twisted compactification of defect \\
    & $\mathbb{Z}_2$ & $Q^5_{4,(\gamma_4,\Gamma_4^{(1)})}$ & composite S-fold compactification with \\
    & & & topological twist \\ 
    & $\mathbb{Z}_2$ & $Q^5_{(\gamma_4,\Gamma_4^{(2)})}$ & non-geometric 3-branes compactified \\
    & & & with topological twist \\ \hline
    $6$ & $\mathbb{Z}_3 \oplus \mathbb{Z}_3$ & $L^3_{3,\gamma_3} \times L^3_{3, \Gamma_3^{(i)}}$ & (non-geom.) twisted compactification of defects \\
    & $\mathbb{Z}_4 \oplus \mathbb{Z}_4$ & $L^3_{4,\gamma_4} \times L^3_{4,\Gamma_4^{(i)}}$ & (non-geom.) twisted compactification of defects \\
    & $\mathbb{Z}_2 \oplus \mathbb{Z}_2$ & $\mathbb{RP}^3_{M_1^{(i)}} \times \mathbb{RP}^3_{M_2^{(i)}}$ & (non-geom.) twisted compactification of defects \\
    & $\mathbb{Z}_2$ & $W_6$ & (non-geom.) twisted compactification of \\
    & & & additional codimension-four S-fold \\ \hline
    $7$ & $\mathbb{Z}_9$ & $L^7_{3,\gamma_3}$ & type IIB on singular local CY \\
    & $\mathbb{Z}_9$ & $L^7_{3,\Gamma_3^{(1)}}$ & composite S-string on $T^2$ (see \cite{Heckman}) \\
    & $\mathbb{Z}_9$ & $L^7_{3,\Gamma_3^{(2)}}$ & twisted M-theory compact. on singular local CY \\
    & $\mathbb{Z}_3$ & $L^7_{3,(\gamma_3, \Gamma_3^{(1)})}$ & composite S-instanton (see \cite{Heckman}) \\
    & $\mathbb{Z}_3$ & $L^7_{3,(\gamma_3, \Gamma_3^{(2)})}$ & non-geometric instanton \\
    & $\mathbb{Z}_2 \oplus \mathbb{Z}_2$ & $(L^3_{4,\Gamma_4^{(i)}} \times \text{K3}) /\mathbb{Z}_2$ & (topolog.) twisted compactification of defects \\
    & $\mathbb{Z}_2 \oplus \mathbb{Z}_2$ & $S^1_{\gamma_4} \times \mathbb{RP}^3_{M_1^{(i)}} \times \mathbb{RP}^3_{M_2^{(i)}}$ & (non-geom.) twisted compactification of defects \\
    & $\mathbb{Z}_2$ & $S^1_{\gamma_4} \times W_6$ & (non-geom.) twisted compactification of defect \\
    & $\mathbb{Z}_2 \oplus \mathbb{Z}_{32}$ & $L^7_{4,\gamma_4} \,, \widetilde{L}^7_{4,\gamma_4}$ & type IIB on singular geometry \\
    & $\mathbb{Z}_8 \oplus \mathbb{Z}_8$ & $(T^4 \times \mathbb{RP}^3)_{(\gamma_4,M_1^{(i)},M_2^{(i)})}$ & (non-geom.) twisted compactification of defects \\
    & $\mathbb{Z}_{16} \oplus \mathbb{Z}_{16}$ & $L^7_{4,\Gamma_4^{(i)}}$ & F-/M-theory on singular geometry \\ \hline
\end{tabular}}
\renewcommand{\arraystretch}{1.0}
    \caption{Bordism groups, their generators, and defects for $5 \leq k \leq 7$.}
    \label{tab:gen5to7}
\end{table}

\subsection{Highlights and surprises}
\label{subsec:highlights}

During the (long) investigations leading to this work, we were able to learn some general lessons about the nature of bordism defects and their interpretation. We will highlight these, and some surprises in the following.

\subsubsection{The completeness of string theory}
\label{subsec:compl}

The fact that supergravity theories do not know about all the objects of the UV-complete theory certainly does not come as a surprise. All the D-branes, perfectly well-defined dynamical objects in string theory, appear as non-perturbative singular object in the supergravity description. What is surprising, however, is that our knowledge of string and M-theory is complete enough to find an interpretation of all the symmetry-breaking defects required by the many non-trivial bordism groups in Table \ref{tab:8dbordgroups}. In fact, as opposed to the ten dimensional analysis which required the R7-brane \cite{Dierigl:2022reg, Heckman, Dierigl:2023jdp}, we do not need to introduce any fundamentally new object. 

This does not mean that the backgrounds we encounter are uninteresting. In several situations we are led to studying configurations that host highly interesting worldvolume theories, including:
\begin{itemize}
    \item{Interacting superconformal field theories (with various amounts of supersymmetry)}
    \item{Twisted compactifications (with geometric as well as non-geometric twists)}
    \item{Lower-dimensional T- and U-duality defects}
\end{itemize}
Often these backgrounds go beyond what is studied in the literature, see, e.g., \cite{Kumar:1996zx, Liu:1997mb, Curio:1998bv}, and motivate further investigation of these rather exotic backgrounds. But they are not beyond the ingredients that are known to be part of string theory compactification.

This completeness and parsimony reinforces the status of string theory as the natural candidate of a UV-completion to supergravity theories and might even be used as an argument in favor of string universality (see, e.g., \cite{Montero:2020icj, Kumar:2009us, Adams:2010zy, Kim:2019vuc, Kim:2019ths, Cvetic:2020kuw, Tarazi:2021duw, Hamada:2021bbz, Bedroya:2021fbu, Bedroya:2023tch}), i.e., that all consistent theories of supergravity have a string, M-, or F-theory completion.

\subsubsection{Twisted compactifications}
\label{subsec:twist}

In the analysis below we often encounter bordism generators whose geometry decomposes as a direct product or fibration of two lower-dimensional pieces. These pieces typically appear on their own as generators of the lower-dimensional bordism group. As in \cite{Heckman}, we do not introduce new defects for this type of generators but describe them as the compactification of the already included symmetry-breaking objects. In cases where the bordism group generator is a fibration it will be a twisted compactification, in the sense that the normal coordinates of the defect transform as non-trivial vector bundles (see Figure \ref{fig:twisted_comp}). 
\begin{figure}
\centering
\includegraphics[width = .9 \textwidth]{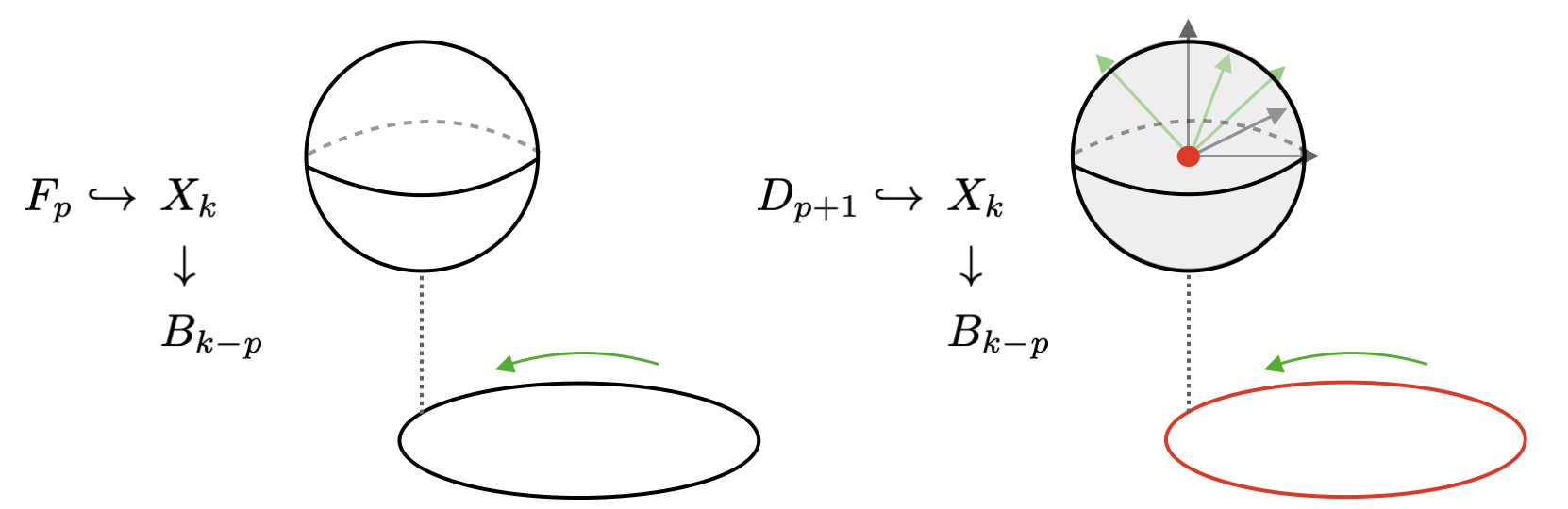}
\caption{Bordism generator in dimension $k$ as fibration of fiber $F_p$ over the base $B_{k-p}$ (left). Associated bordism defect interpreted as the twisted compactification of the defect associated to $F_p$ compactified on $B_{k-p}$ with non-trivial action on normal directions (right).}
\label{fig:twisted_comp}
\end{figure}
However, there also exist other twists of the defect theories that we want to distinguish:
\begin{itemize}
    \item{Topological twists: In these cases the fiber geometry undergoes rotations when moving around the base. As in \cite{Witten:1990bs} this can be understood as mixing the rotations of spacetime with rotations of an internal space that can be interpreted as R-symmetry transformation.}
    \item{Monodromy twists: In certain cases the defects naturally split into individual building blocks, such as SCFT sectors, which are compactified on spaces with a non-trivial monodromy action. These twists already appear at the level of the singular M-theory fibers as will be described in Section \ref{sec:monodrom}.}
    \item{Geometric duality twists: These configurations demand the presence of a non-trivial duality bundle on the base manifold over which the defect is wrapped. Here, the full duality bundle allows for a geometric description in one of the UV completions.}
    \item{Non-geometric duality twists: As above, the base manifold, over which the defect filling the fiber is wrapped, requires the presence of a non-trivial duality bundle. This time, however, there is no duality frame in which both the defect as well as the twist can be lifted to a geometric configuration in string and M-theory.}
\end{itemize}
The first two twists are generally well-defined for wrapped defects, while the last two involve the structure of the duality bundle in a more intricate way. In practice, checking that a wrapped object serves as a viable defect requires embedding the total space in a consistent\footnote{We defer issues of supersymmetry, stability and worldvolume dynamics to future work.} string background.

\subsubsection{Non-geometric defects}
\label{subsec:nongeom}

We already encountered one form of non-geometric backgrounds above, i.e., defects that are compactified on manifolds with duality bundle that does not have a geometric interpretation (at least not in the duality frame where the defect does). Beyond that, we encounter defects that themselves are non-geometric.

The prototypical example of such non-geometric backgrounds are T-folds \cite{Hull:2004in}, which involve transition functions acting as T-duality on the internal space. These transition functions can end on a codimension-two object which can be identified with an exotic brane \cite{deBoer:2010ud, deBoer:2012ma, Lust:2015yia}, depicted in Figure \ref{fig:Tfold}.
\begin{figure}
\centering
\includegraphics[width = .4 \textwidth]{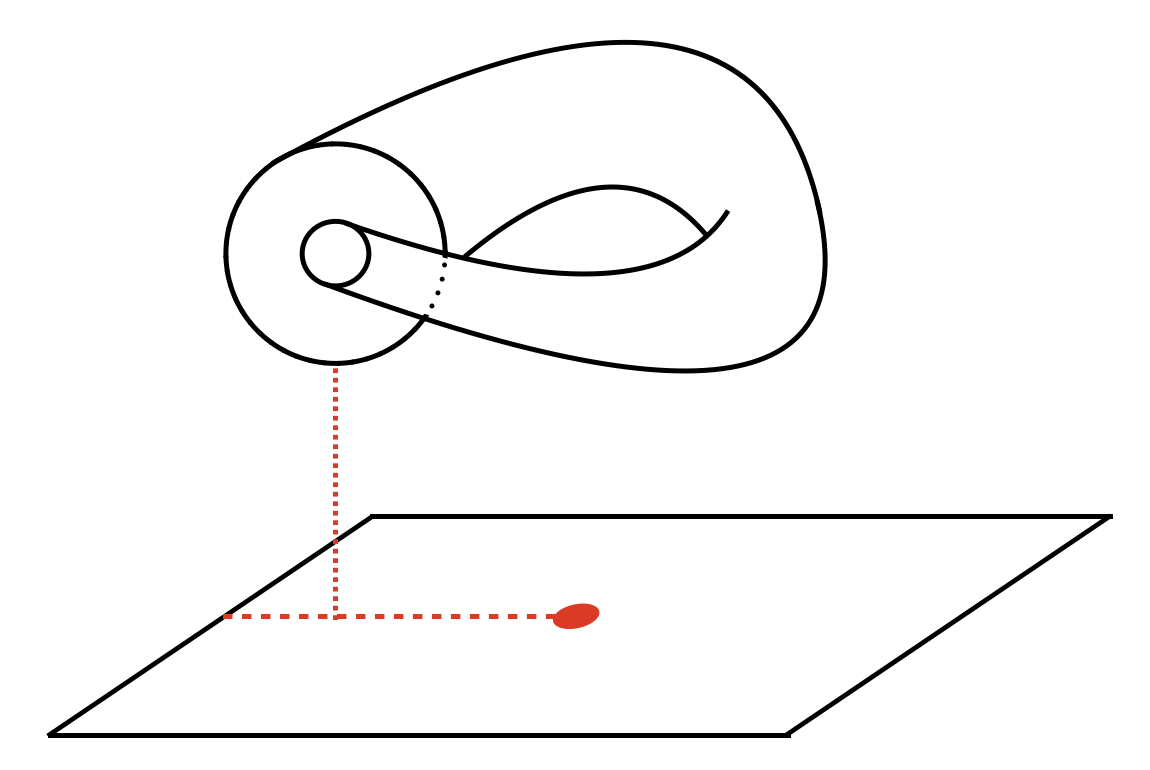}
    \caption{Non-geometric defect (red dot) as endpoint of a transition function that acts as T-duality on an internal circle.}
    \label{fig:Tfold}
\end{figure}
While field theories on such a background would be ill-defined, the exchange of Kaluza-Klein momentum and winding modes allows these configurations in string theory. One can extend the analysis to include more general U-duality transformations leading to whole orbits of exotic brane states \cite{deBoer:2010ud, deBoer:2012ma} in codimension-two.

In our work we encounter non-geometric defects in higher-codimension. At their location the moduli fields degenerate, as for example the torus fiber in F-theory. In general, however, we find that there is no F-theory or M-theory background in which the duality bundle is fully geometrized.

Interestingly, such non-geometric defects seem to appear only in high codimension. In particular, for eight dimensional supergravity theories we find non-geometric defects associated to $\Omega^{\text{Spin}}_k (BG_U^{8\text{d}})$ for $k \geq 5$. One of these defects is of codimension-six, while another appears as a twisted compactification of a codimension-four defect. For lower $k$, all defects have a geometric understanding in at least one duality frame, but might be subject to non-geometric duality twists.

\subsubsection{Dimensionality of the symmetry-breaking defects}
\label{subsec:dimdef}

Having a non-trivial deformation class $X_k$ in $\Omega^{\text{Spin}}_k (BG_U^{D})$ is associated to a conserved topological charge carried by $X_k$. To make sure that this charge is not conserved we include defects that allow the description of $X_k$ as its asymptotic boundary. In this way $X_k$ becomes the boundary of a $(k+1)$-dimensional manifold $Y_{k+1}$ with the defect included. In case the defect is a localized object in $Y_{k+1}$ it is naturally of (real) codimension $(k+1)$. The extra coordinate can be understood as radial direction with $Y_{k+1}$ being conical (see Figure \ref{fig:conicaldef}).
\begin{figure}
\centering
\includegraphics[width = .4 \textwidth]{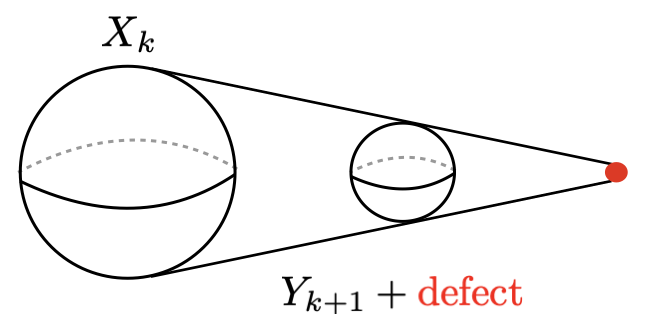}
    \caption{Sketch of symmetry-breaking defect in codimension $(k+1)$ and the conical structure of the bounding space $Y_{k+1}$.}
    \label{fig:conicaldef}
\end{figure}

We have already seen that this relation between the dimension of the bordism generator and the dimensionality of the defect can be modified in cases the generator splits into fiber and base, see Section \ref{subsec:twist}. If the fiber is given by a generator of the bordism group in lower dimension, the fibration can be bounded by wrapping the defect on the base, see Figure \ref{fig:twisted_comp}. 

One surprise in this work is that sometimes we also encounter new defects of lower codimension. This happens if the generator of the bordism group splits as a direct product, but neither of the two factors appeared as generator for bordism groups in lower dimensions. In these situations it is natural to include new (singular) symmetry-breaking objects with codimension smaller than $(k+1)$ and wrap these on the other factor.

In our case this happens for example for $k=2$ with the generator given by a torus with duality bundle associated to two commuting transition functions in the SL$(3,\mathbb{Z})$ factor of $G_U^{8 \text{d}}$, i.e., the generator splits into two circles with SL$(3,\mathbb{Z})$ duality bundle. These are not generators of $\Omega^{\text{Spin}}_1 (BG_U^{8 \text{d}})$ since they can be bounded by smooth 2-manifolds equipped with a duality bundle and do not require any singular defect. This raises the immediate question, why the same configurations cannot be used to bound the generator in $k = 2$. The reason is that the duality bundle of the gravitational configuration is incompatible with the SL$(3,\mathbb{Z})$ transition function on the other circle and therefore does not lead to a well-defined smooth 3-manifold with the torus as its boundary, see Section \ref{subsec:codim3} for more details. Similar behavior also appears in other dimensions.

We want to mention another possibility, which is that despite the fact that the bordism generator splits into two factors the symmetry-breaking object might still be of codimension $(k+1)$, see Figure \ref{fig:dif_codim}. However, in this case it is more complicated to think about the bounding geometry $Y_{k+1}$ as a cone (as depicted in Figure \ref{fig:conicaldef}) and it most likely has a more intricate internal geometry which might involve topology changes, see also \cite{Ruiz:2024gzv}.
\begin{figure}
\centering
\includegraphics[width = .5 \textwidth]{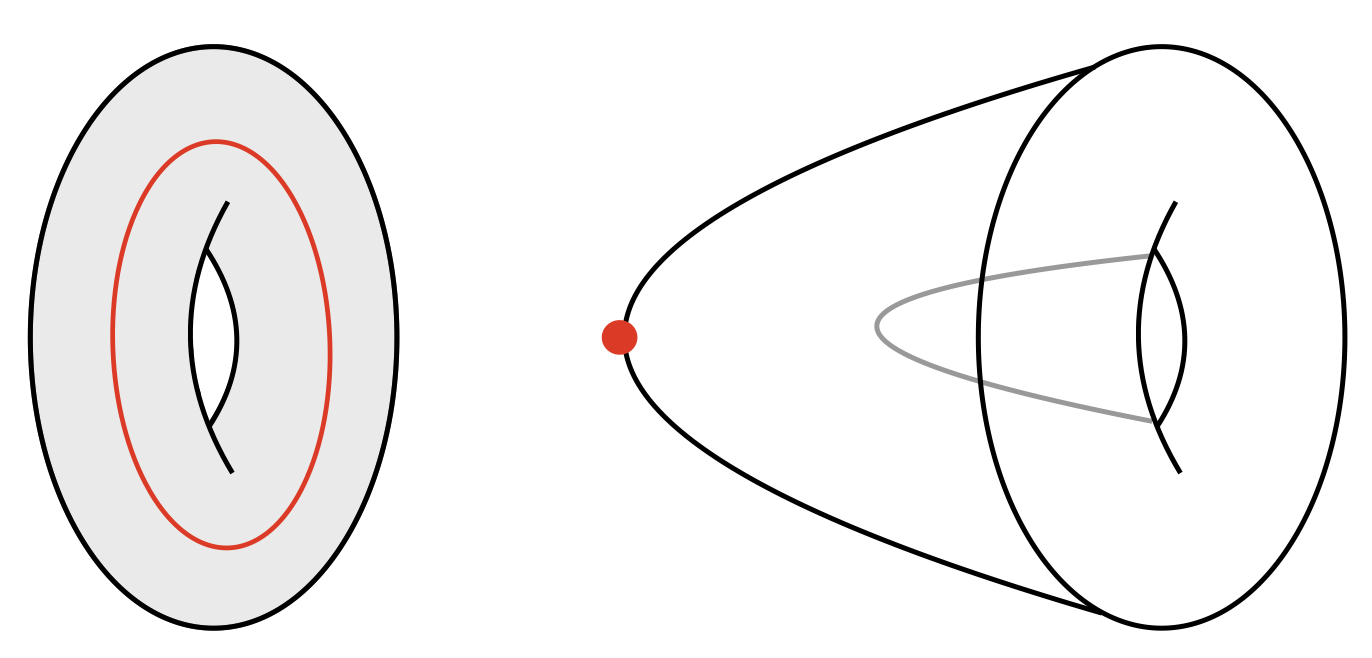}
\caption{Bounding one factor with lower codimension object (left). Bounding the entire product with higher codimension defect (right).}
\label{fig:dif_codim}
\end{figure}

\subsubsection{Deformations in supergravity}
\label{subsec:deform}

The relation between non-trivial bordism classes and symmetry-breaking defects also guarantees that there are no deformations within the low-energy supergravity that trivialize the deformation class. Thus, all deformations that preserve the structure of the low-energy theory, which in our case includes the duality bundle as well as the spin structure of spacetime, are not able to desingularize the symmetry-breaking object. Deformations might however change the specific realization and interpretation of the object. To illustrate this we want to discuss two examples:

The first example is the geometric configuration which appears as one of the generators in $\Omega^{\text{Spin}}_3 (BG_U^{8\text{d}})$
\begin{equation}
(T^2_F \times \mathbb{C}^2)/\mathbb{Z}_3 \times T^2 \,,
\end{equation}
where $T^2_F$ refers to the F-theory fiber torus. 
The resolved geometry is given by F-theory on a $(-3)$ curve decorated by a type IV fiber. This supports the celebrated 6d $\mathfrak{su}(3)$ non-Higgsable cluster of reference \cite{Morrison:2012np}. The quotient singularity obtained by blowdown leads to a 6d SCFT \cite{Witten:1996qb, Heckman:2013pva}.\footnote{See \cite{Heckman:2018jxk, Argyres:2022mnu} for recent reviews of 6D SCFTs.} Observe that while there is a blowup of the $(-3)$ curve available, this curve supports an $\mathfrak{su}(3)$ gauge theory sector, namely a strongly coupled bound state of $[p,q]$ 7-branes. As such, it includes degrees of freedom which are not present in the 8d supergravity theory.
As in the discussion in Section \ref{subsec:dimdef}, the deformation changes the dimensionality of the symmetry breaking defect which instead of codimension-four is now give by a codimension-two object wrapped over a compact two-dimensional submanifold (see Figure~\ref{fig:defdeform}).
\begin{figure}
\centering
\includegraphics[width = .8 \textwidth]{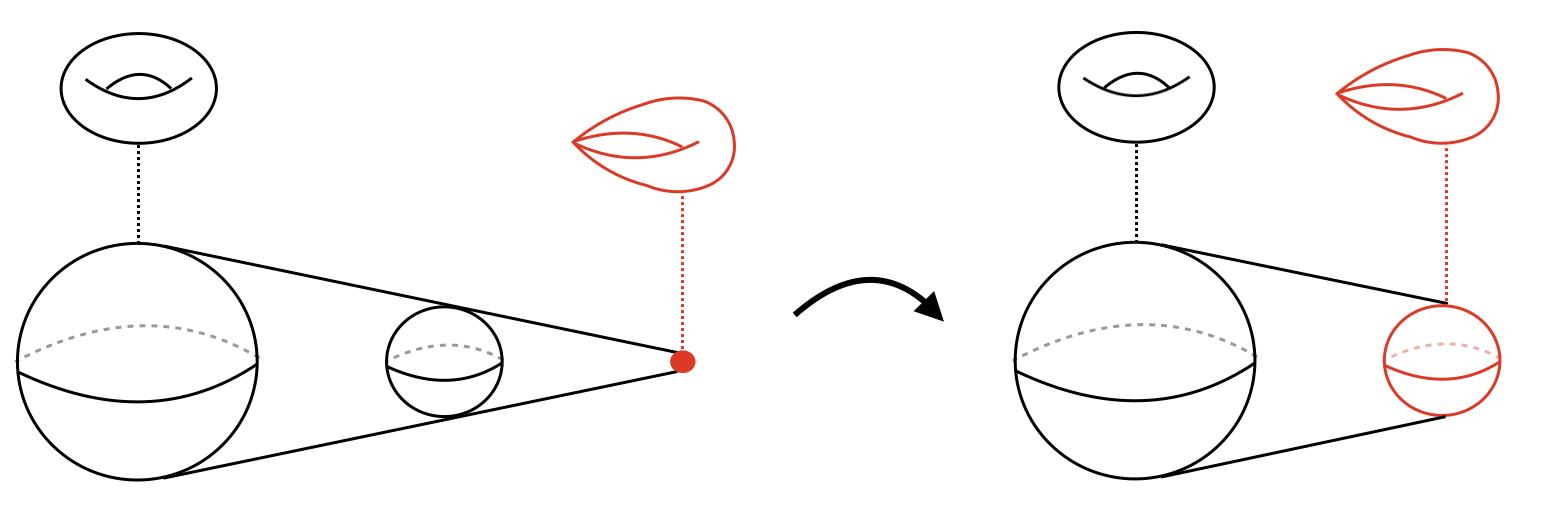}
\caption{Deformation of the bounding manifold in the low-energy theory can change the specific realization of the defect but cannot fully desingularize it.}
\label{fig:defdeform}
\end{figure}

The second example is that of a seven-dimensional lens space $L^7_4$ with an SL$(2,\mathbb{Z})$ bundle, which can be described by type IIB on the background
\begin{equation}
    (T^2 \times \mathbb{C}^4)/\mathbb{Z}_4 \,.
\end{equation}
Since the $\mathbb{Z}_4$ action does not preserve a holomorphic $(5,0)$-form in the type IIB geometry, this does not correspond to a local singular Calabi-Yau geometry and supersymmetry appears to be broken completely, see also the next Section \ref{subsec:liftsusy}. However, even if it was a local Calabi-Yau geometry, the central $\mathbb{C}^4/\mathbb{Z}_4$ singularity does not allow for a deformation that preserves supersymmetry since it is a terminal singularity \cite{a6207765-8e51-37bb-af3c-43a72b254a77}. This, however, does not affect the topological bordism discussion and we can, for example, describe the lens space as the asymptotic boundary of the complex line bundle $\mathcal{O}(-4)$ over $\mathbb{CP}^3$ following the discussion in \cite{Hsieh:2020jpj}. Once more the resolved singularity, the blown-up $\mathbb{CP}^3$, remains a singular locus, over which the $T^2$ degenerates into $T^2/\mathbb{Z}_4$ and the background requires the inclusion of defects.

In the UV description, where there is an internal space, there are more possibilities that include resolutions involving the fiber coordinates. These deformations are not contained in the supergravity framework, but can `trivialize' the bordism defect in the UV lift. Indeed, we will find in various circumstances that symmetry-breaking defects of maximal supergravity are related to string and M-theory on singular manifolds which do allow for a desingularization (that might break supersymmetry). In these cases we do not find analogs of this class of defects in the higher-dimensional theory, e.g., in the classification of \cite{Heckman}.

It is in this sense that the symmetry-breaking defects are necessarily singular in the low-energy description. Also note that any allowed deformation should not change the bordism class of the asymptotic boundary, there cannot be any objects that extend all the way to infinity. It would be a very interesting question to analyze what generalization of bordism groups might allow for such singularities.

In general it becomes a dynamical question which of the bounding manifolds with defect is energetically preferred, e.g., whether the $\mathbb{CP}^3$ in the example above will dynamically contract to zero size. The only situation in which we can be sure to find a stable configuration is whenever the singular backgrounds preserve part of supersymmetry, and the symmetry-breaking defects are BPS. This is why in our analysis of the bordism defects we focus on configurations that are most likely to preserve some supersymmetry.

\subsubsection{Spin-lift of defects and supersymmetry}
\label{subsec:liftsusy}

Unfortunately the question whether the defects we introduce preserve part of the supersymmetry is a difficult one in our analysis. The reason is that we focus on the U-duality group acting on the bosons of the supergravity theory. 

That this cannot be the full answer, can for example be deduced from the description of maximal supergravity from M-theory on $T^k$ for which an SL$(k,\mathbb{Z})$ subgroup of the full U-duality is given by the large diffeomorphisms of $T^k$. The $D$-dimensional supercharges lift to the eleven-dimensional supercharges that transforms as a $32$ component Majorana fermion. At certain points in moduli space there will be a $T^2$ square sub-torus for which one can act with the usual $S$ generator of the associated SL$(2,\mathbb{Z})$ subgroup:
\begin{equation}
    S = \begin{pmatrix*}[r] 0 & -1 \\ 1 & 0 \end{pmatrix*} \,.
\end{equation}
At these particular points $S$ can be understood as a rotation by $\frac{\pi}{2}$ in the internal space, which makes $S^4$ a rotation by $2 \pi$. Of course this acts as the identity on the bosons in the eleven-dimensional theory and their $D$-dimensional descendants. On fermions including the supercharges, however, it acts as multiplication by $(-1)$. Moving away from the loci of moduli space where one has a square sub-torus the transformation become more complicated and are in general moduli dependent. Nevertheless, this reasoning shows that the full duality group should include a non-trivial $\mathbb{Z}_2$ extension of the bosonic duality group, a Spin lift, see also \cite{Pantev:2016nze, Tachikawa:2018njr}.

On top of this, M-theory does not require spacetime to be oriented \cite{FH19b} and one can further include non-orientable internal spaces leading to a further extension of the duality group. This M-theory data then suggest a Pin$^+$ lift of the duality group for fermions. While this lift is known for type IIB in ten dimensions \cite{Tachikawa:2018njr}, and was used in \cite{Debray:2021vob, Heckman} in the determination of bordism groups 
\begin{equation}
    \Omega^{\text{Spin-Mp}(2,\mathbb{Z})}_k (\text{pt}) \,, \quad \Omega^{\text{Spin-GL}^+(2,\mathbb{Z})}_k (\text{pt}) \,,
\end{equation}
as opposed to $\Omega^{\text{Spin}}_k (B\text{SL}\big(2,\mathbb{Z})\big)$ we do not know of a full generalization to U-duality groups of interest here, see \cite{Pantev:2016nze} for partial results. This has an important effect on the defects, which we want to illustrate in a simple example.

For this we focus on the Spin lift from SL$(2,\mathbb{Z})$ to the metaplectic group Mp$(2,\mathbb{Z})$, which describes a non-trivial $\mathbb{Z}_2$ extension. Since everything happens at prime $2$ the essential features are captured by
\begin{equation}
    \Omega^{\text{Spin}}_k (B \mathbb{Z}_4) \quad \text{vs.} \quad \Omega^{\text{Spin-}\mathbb{Z}_8}_k (\text{pt}) \,,
\end{equation}
whose bordism groups can be found in \cite{Heckman}. The $\mathbb{Z}_4$ is generated by the $S$ generator with the Spin-$\mathbb{Z}_8$ structure accounting for its non-trivial lift. For $k=3$ one has
\begin{equation}
    \Omega^{\text{Spin}}_3 (B\mathbb{Z}_4) = \mathbb{Z}_2 \oplus \mathbb{Z}_8 \,, \quad \Omega^{\text{Spin-}\mathbb{Z}_8}_3 (\text{pt}) = \mathbb{Z}_2 \,,
\end{equation}
with both the $\mathbb{Z}_8$ summand and the $\mathbb{Z}_2$ for Spin-$\mathbb{Z}_8$ structure being generated by the lens space $L^3_4$. However, the bundles are different. While for the Spin generator one can choose a $\mathbb{Z}_4$ action
\begin{equation}
    \psi \rightarrow e^{2 \pi i/4} \, \psi \,,
\end{equation}
when traversing the torsion 1-cycle of $L^3_4$, the same is not possible for the Spin-$\mathbb{Z}_8$ structure since the geometric action is linked with the `gauge' action and only
\begin{equation}
    \psi \rightarrow - \psi \,,
\end{equation}
is allowed. This difference is reflected in the bordism groups, which can be determined using Dirac $\eta$-invariants. Lifting back to the full duality group we see that the bosonic monodromy around the torsion 1-cycle is given by $S$ for Spin and $S^2$ for Spin-$\mathbb{Z}_8$ structure.\footnote{In the geometry of F-theory, where the $S$ action is implemented on the auxiliary torus fiber, this statement becomes that while the total space for the torus fibration constructed from the $\Omega^{\text{Spin}}_3 (B\mathbb{Z}_4)$ generator is not Spin, the total space of the $\Omega^{\text{Spin-}\mathbb{Z}_8}_3 (\text{pt})$ generator is.} 

In $d=5$ the difference becomes even more drastic since the non-Spin manifold $L^5_4$ is a generator of $\Omega^{\text{Spin-}\mathbb{Z}_8}_5 (\text{pt})$ it cannot be a generator of $\Omega^{\text{Spin}}_5 (B\mathbb{Z}_4)$ which would have to be Spin, and is rather given by a $L^3_4$ bundle over $\mathbb{CP}^1$ denoted as $Q^5_4$ above. 

We see that the Spin lift of the duality group has an important influence on the bordism group generators. While the Spin lift only affects the prime $2$ part of the bordism group, the reflections in general lead to modification for other primes as well. Since the BPS nature of defects associated to non-trivial duality bundles clearly depends on the transformation properties of the fermions in particular the supercharges, we defer the discussion of the supersymmetric properties of the bordism defects to a case-by-case basis until the appropriate lifts of the U-duality groups are determined. 

\subsubsection{Torsion in duality}
\label{subsec:QvsZ}

There is a quite curious feature of our calculations for $\text{SL}(3,\mathbb{Z})$ that turns out to have a deeper physical meaning, which may extend beyond the borders of Utopia. In Section~\ref{subsec:SL3central}, we will see how the calculation for the bordism groups of $B\text{SL}(3,\mathbb{Z})$ proceeds by finding different embeddings of easier groups (like the symmetric group, or cyclic groups) into SL$(3,\mathbb{Z})$. To get all bordism classes for $B$SL$(3,\mathbb{Z})$, it turns out that we need to consider some cyclic groups $\mathbb{Z}_n$ \emph{twice}; in particular, for $n=4$, we will have to consider defects coming from two different embeddings of $\mathbb{Z}_4$ into SL$(3,\mathbb{Z})$. In each of these, the generator of $\mathbb{Z}_4$ is represented by a different matrix, which we label as $\Gamma_{4}^{(i)}$ for $i=1,2$. A similar thing happens with $\mathbb{Z}_3$, and the corresponding pair of inequivalent matrices of order three is denoted by $\Gamma_{3}^{(i)}$. Explicit expressions for these matrices can be found in Section~\ref{subsec:SL3central}.

Here, we comment on the nature of the difference between these two embeddings. More concretely, consider a bordism class represented by a manifold $X$ with $\mathbb{Z}_4$ bundle. The $\mathbb{Z}_4$ can be taken to be generated by either $\Gamma_{4}^{(1)}$ or $\Gamma_{4}^{(2)}$. How are these two related?

The two bordism classes so constructed are not physically equivalent in many cases, and the reason is that the matrices $\Gamma_{4}^{(1)}$ and $\Gamma_{4}^{(2)}$ are \emph{not} conjugate in SL$(3,\mathbb{Z})$. More concretely, the SL$(3,\mathbb{Z})$ duality bundle over $X$ can be characterized in terms of transition functions, which are in turn described by matrices $M_i\in \text{SL}(3,\mathbb{Z})$ specifying the holonomies of the discrete gauge field. This can be equivalently described by the conjugated matrices
\begin{equation}
M_i\,\rightarrow U\, M_i \, U^{-1} \,,
\label{gtr}
\end{equation}
for any $U \in \text{SL}(3,\mathbb{Z})$; this is just a gauge transformation. Conversely, we would expect two sets $\{M_i\}$, $\{M_i'\}$ not related as in \eqref{gtr} to be physically inequivalent, in general. 

The above explains in which way the classes constructed out of $\Gamma_{4}^{(1)}$ are different from those constructed out of $\Gamma_{4}^{(2)}$, but it is more interesting to ask in which sense they are similar. Part of the duality group (or rather, a bosonic approximation to it) in the full quantum gravity theory is SL$(3,\mathbb{Z})$. However, as familiar from similar discussions about SL$(2,\mathbb{Z})$ in the IIB context, this is because the theory has (due to completeness principle \cite{Polchinski:2003bq}) a complete spectrum of extended charged states satisfying the Dirac quantization condition. By contrast, a low-energy observer, having only access to the supergravity low-energy degrees of freedom, will not be able to access this information, at least in generic points of moduli space, since the massive objects are too heavy for them. The duality group still constrains the low-energy effective field theory action, but the selection rules are indistinguishable from those of the groups 
\begin{equation}
\mathrm{SL}(3,\mathbb{R}) \,, \quad \text{or} \quad \mathrm{SL}(3,\mathbb{Q})\,.
\label{sugra}
\end{equation}
Both of these provide valid selection rules at low energies. A Swampland-minded effective field theorist, who knows that at the end of the day there will be extended objects satisfying a Dirac quantization condition, might more reasonably choose SL$(3,\mathbb{Q})$ as the duality group; this more-or-less corresponds to allowing for $\mathbb{Q}$-valued charges, which is closer to the properly integer quantized charges, but this is really a matter of taste.

Irrespectively of the choice in \eqref{sugra}, the key point we wish to emphasize is that the low-energy physics of any given bordism class only depends on the SL$(3,\mathbb{Q})$ class of the associated bundle, not on their SL$(3,\mathbb{Z})$ class. In other words, the low-energy physics will remain unaltered when $U$ in \eqref{gtr} is taken in the wider class of matrices
\begin{equation}
U \in \mathrm{SL}(3,\mathbb{Q}) \,.
\end{equation}
Although this is not a symmetry of the full theory, it is a symmetry of the low-energy physics. As a result, \emph{any two bordism classes where the U-duality bundles are SL$(3,\mathbb{Q})$ conjugate describe the same low-energy physics}. This is interesting precisely because there are  SL$(3,\mathbb{Q})$-conjugate matrices that are not  $SL(3,\mathbb{Z})$-conjugate; the pair $\Gamma_4^{(i)}$ for $i=1,2$ (and also $\Gamma_3^{(i)}$) is precisely one such example. This means that the smooth bordism classes constructed via these two different $\mathbb{Z}_4$ embeddings will yield identical low-energy physics, but will differ at high energies, e.g. in their spectrum of extended states. At energies much below the scale of spacetime topology fluctuations or massive states, the partition function $\mathcal{Z}(X)$ of the theory is a well-defined quantity. If we have two theories differing only in massive states as above, the quotient
\begin{equation}
\mathcal{T}(X)\equiv\frac{\mathcal{Z}_1(X)}{\mathcal{Z}_2(X)} \,,
\end{equation}
will be insensitive to small deformations of the metric or background fields, since these are deformations of massless fields and hence dependence in $\mathcal{Z}_1$ and $\mathcal{Z}_2$ will cancel out. We conclude that $\mathcal{T}(X)$ is topological, and since it satisfies all the standard axioms of gluing and pasting, it defines a TQFT. Thus, very generally, two $\mathbb{Z}$-inequivalent theories that are $\mathbb{Q}$-equivalent differ by topological couplings. One way to think about this is to observe that integrating out the massive states in either theory can produce a TQFT at low energies; since the massive states are different in each theory, they will give rise to different TQFTs, and $\mathcal{T}(X)$ precisely accounts for the difference between these topological sectors.

The above discussion can be used to understand some of the results in \cite{Montero:2022vva} in a new light. In that reference, two discrete $\theta$ angles were constructed for compactifications of IIB on quotients of $T^3$, equipped with different SL$(2,\mathbb{Z})$ duality bundles. These compactifications are mapping tori $T^2\rightarrow S^1$, so they correspond to circle compactifications with non-trivial U-duality holonomy. They correspond precisely to holonomies $\Gamma_4^{(i)}$; the discrete $\theta$ angles described in \cite{Montero:2022vva} for seven-dimensional theories are therefore an example of non-trivial $\mathcal{T}(X)$.

Although there is no non-trivial bordism class associated to these seven-dimensional backgrounds (meaning a smooth geometry can interpolate from one to the other), non-trivial bordism classes involving these different choices exist in dimension three, since the bordism going from one to the other in 7d is now obstructed by the additional geometry of lens spaces. Therefore, the bordism classes in degree three obtained from the different $\Gamma_{3,4}^{(i)}$ have the same low-energy physics, but differ at high energies. The physics in their bordism defects will also be different in general, as inflow from $\mathcal{T}(X)$ can require the presence of chiral degrees of freedom in the defect.

We expect the idea of looking at $\mathbb{Q}$-conjugate elements in a properly quantized duality group may be used to uncover many more examples of discrete $\theta$ angles and similar topological couplings, in string compactifications with 32 supercharges and below.

\subsection{Future Directions}

Having dealt with the main contours of what we do in this work, let us now turn to some natural future directions.

Perhaps the most direct continuation of the present work would be to perform similar computations for the other U-duality groups.
A general complication is that the relevant stable splittings and group cohomology computations do not appear to be in the literature.
That being said, it is plausible that the knit product structure can be leveraged to extract many aspects of the physical problem.

Below $D < 3$, gravity behaves rather differently, and so do the U-duality groups of maximal supergravity, passing instead to the affine symmetries $\mathrm E_{9(9)}(\mathbb{Z}), \mathrm E_{10(10)}(\mathbb{Z}), $ and $\mathrm E_{11(11)}(\mathbb{Z})$ (see e.g., \cite{Julia:1981wc, Nicolai:1987kz}). It would nevertheless be interesting to compute the relevant bordism groups in these cases as well.

As already mentioned, one item of immediate interest would be to understand appropriate $\mathrm{Spin}$ and $\mathrm{Pin}$-lifts of the U-duality groups. We expect these extra symmetries to emerge at special points of the moduli space, where they can act on the fermionic degrees of freedom of our supergravity theory. Computation of the corresponding Spin / Pin-twisted bordism groups would provide an important refinement of the considerations presented here, especially with regards to finding possibly new, non-supersymmetric backgrounds similar to the R7-branes found in IIBordia \cite{Dierigl:2022reg, Heckman}.

Another curiosity we find in our investigations is an apparent nested structure of defects, which is highly suggestive of ``branes within branes,'' similar to what is found for the D-branes of superstring theory \cite{Douglas:1995bn}. This suggests that we may have access to even further topological structures present on these defects, an exciting prospect for future investigations.

On general grounds, the defects we discover involve objects which must be added to a given low energy effective field theory. In particular, this means that they ought to have a suitable non-zero tension / mass, namely we can integrate out their effects below some threshold scale.
Giving a proper notion of tension / mass, especially for orbifold type backgrounds would be quite natural to develop, and it is tempting to speculate that it involves the stress tensor correlators and / or free energy of the corresponding worldvolume theory.

While we have primarily focused on the case of $\Omega^{\mathrm{Spin}}_{k}(BG_{U})$ for $k$ ``small,'' it is natural to ask whether our methods extend to large, and possibly arbitrarily big values of $k$. Aside from its importance in pure mathematics, there are potential applications of this in the study of defects in supercritical string backgrounds, where the effective target space can also be much bigger than the critical dimension.

\newpage

\paragraph*{Acknowledgements} Once more we are indebted to Paco Giudice for his time and creativity to create the map of Utopia.

We thank B.S.~Acharya, D.S.~Berman, J.~Block, G.~Bossard, V.~Chakrabhavi, J.~de Boer, C.~Hull, M.~Cveti\v{c}, C.~Lawrie, D.~Lüst, R.~Minasian, J.~McNamara, F.~Riccioni, E.~Torres, C.~Vafa, and E.~Witten for helpful discussions.  MD thanks the ESI in Vienna, in particular the program ``The Landscape vs. the Swampland,'' for hosting him during part of the time in which this work was completed. MD and JJH thank the organizers of the 2024 conference ``Geometry, Strings and the Swampland,'' held at Schloss Ringberg for hospitality during part of this work. MD, JJH, and MM thank the Harvard Swampland Initiative for hospitality during part of this work. JJH and MM thank the Harvard CMSA for hospitality during part of this work. The work of NB is supported by an NSF Graduate Research Fellowship. The work of JJH is supported by DOE (HEP) Award DE-SC0013528 as well as by BSF grant 2022100. The work of JJH is also supported in part by a University Research Foundation grant at the University of Pennsylvania.  MM is currently supported by the RyC grant RYC2022-037545-I from the AEI and was supported by an Atraccion del Talento Fellowship
2022- T1/TIC-23956 from Comunidad de Madrid in the early stages of this project. The
authors thank the Spanish Research Agency (Agencia Estatal de Investigacion) through the
grants IFT Centro de Excelencia Severo Ochoa CEX2020-001007-S and PID2021-123017NBI00, funded by MCIN/AEI/10.13039/501100011033 and by ERDF A way of making Europe.

\newpage

\section{U-duality}
\label{sec:Udual}

In this section, we briefly review the U-duality of supergravity theories with $32$ real supercharges in ten to three spacetime dimensions. These dualities leave the supergravity action invariant and act non-trivially on the states in the theory. Once charge quantization is included, the duality groups are discrete. For a more detailed description see \cite{Obers:1998fb} (which we mainly base our discussion on), as well as references therein.

\subsection{U-duality groups}
\label{subsec:Ugroups}

Supergravity theories with $32$ real supercharges in $D$ dimensions can be obtained by compactifying eleven-dimensional supergravity, the low-energy limit of M-theory, on a $(11-D)$-dimensional torus. The fields in the supergravity multiplet arise as zero modes of this compactification. This further has an interpretation as type IIA on a $(10-D)$-dimensional torus, where the radius of the remaining circle parametrizes the string coupling in type IIA.

On the level of the supergravity action, the full U-duality is generated by the Lorentz transformations on the M-theory compactification torus as well as the T-dualities of the type IIA description. The resulting group can be denoted as
\begin{equation}
    G_{U, \mathbb{R}}^{D} = \text{SL}(\ell,\mathbb{R}) \bowtie \text{SO}(\ell-1,\ell-1,\mathbb{R}) \,,
\end{equation}
where we introduced $\ell = 11-D$. Here, as in \cite{Obers:1998fb} the symbol $\bowtie$ describes the knit product, with the resulting group being generated by the non-commuting subgroups (see also Appendix \ref{App:Knit}). In the range $\ell \in \{1\,,2 \,, \dots, 8\}$ these are given by the Cremmer-Julia symmetry groups for non-chiral supergravity theories, see \cite{Cremmer:1980gs, Julia:1980gr},
\begin{equation}
    \begin{array}{c  c  c }
    \toprule 
    D & G_{U, \mathbb{R}}^{D} & H^{D}_{U, \mathbb{R}} \\ \midrule
    10 & \mathbb{R}^+ & 1 \\ 
    9 & \text{SL}(2,\mathbb{R}) \times \mathbb{R}^+ & \text{U}(1) \\ 
    8 & \text{SL}(3,\mathbb{R}) \times \text{SL}(2,\mathbb{R}) & \text{SO}(3) \times \text{U}(1) \\
    7 & \text{SL}(5,\mathbb{R}) & \text{SO}(5) \\
    6 & \text{SO}(5,5,\mathbb{R}) & \text{SO}(5) \times \text{SO}(5) \\ 
    5 & \text{E}_{6(6)} & \text{USp}(8) \\ 
    4 & \text{E}_{7(7)} & \text{SU}(8) \\
    3 & \text{E}_{8(8)} & \text{SO}(16) \\ \bottomrule
    \end{array}
\end{equation}
which we reproduce from \cite{Obers:1998fb}.\footnote{$\mathbb{R}^+$ denotes the real half-line.} The $H^{D}_{U, \mathbb{R}}$ describes the R-symmetry group in $D$ dimensions, which is the maximal compact subgroup of $G^{D}_{U,\mathbb{R}}$. The moduli spaces $\mathcal{M}_{D}$ of the associated maximal supergravity theories are parametrized by the coset space
\begin{equation}
    \mathcal{M}_{D} = \frac{G^{D}_{U,\mathbb{R}}}{H^{D}_{U,\mathbb{R}}} \,,
    \label{eq:8dmoduli}
\end{equation}
which further has to be modded out by the discrete duality action below. In the M-theory frame these scalar fields originate from the metric components of $T^{\ell}$ as well as internal components of the M-theory 3-form $C_{MNR}$ and its dual 6-form. For example, for $\ell = 3$ one has a single scalar field from the 3-form and six scalar fields from the metric parametrized by the seven-dimensional coset space
\begin{equation}
    \mathcal{M}_{8\text{d}} = \frac{\text{SL}(3,\mathbb{R}) \times \text{SL}(2,\mathbb{R})}{\text{SO}(3) \times \text{U}(1)} \,,
\end{equation}
and similarly for smaller dimensions.

Once one includes the flux quantization conditions, these continuous groups of the supergravity action are reduced to discrete subgroups that leave the charge lattices invariant obtained from the knit product
\begin{equation}
    G^{D}_U = \text{SL}(\ell,\mathbb{Z}) \bowtie \text{SO}(\ell-1,\ell-1,\mathbb{Z}) \,,
    \label{eq:Ugroup}
\end{equation}
with $\ell = 11-D$. These U-duality groups $G^{D}_U$ are given by:
\begin{equation}
    \begin{array}{c  c }
    \toprule
    D & G^{D}_U \\ \midrule
    10& 1 \\
    9 & \text{SL}(2,\mathbb{Z}) \\ 
    8 & \text{SL}(3,\mathbb{Z}) \times \text{SL}(2,\mathbb{Z}) \\
    7 & \text{SL}(5,\mathbb{Z}) \\ 
    6 & \text{SO}(5,5,\mathbb{Z}) \\
    5 & \text{E}_{6(6)}(\mathbb{Z}) \\
    4 & \text{E}_{7(7)}(\mathbb{Z}) \\
    3 & \text{E}_{8(8)}(\mathbb{Z}) \\ \bottomrule
    \end{array}
\end{equation}
In the following, we will focus on these discrete bosonic U-duality groups.\footnote{See \cite{Debray:2022wcd} for a discussion of anomalies for the continuous group E$_{7(7)}(\mathbb{R})$.} The fields and charges of the supergravity theory will transform under U-duality according to certain representations, which can be deduced from the transformation properties of the central charges in the supersymmetry algebra, as discussed in \cite{Obers:1998fb}.

\subsection{A geometrization of U-duality: Exceptional field theory}
\label{subsec:geomEFT}

Not all duality transformation have a geometric interpretation in terms of the internal $T^\ell$ or $T^{\ell-1}$ of the M-theory or type IIA/B lift, respectively. However, similar to the approach of F-theory \cite{Vafa:1996xn} (see  \cite{Weigand:2018rez} for a review) for SL$(2,\mathbb{Z})$ duality in type IIB, there are approaches to introduce auxiliary internal spaces to geometrize the duality action.

The geometrization of the T-duality subgroup $\text{SO}(\ell - 1, \ell - 1,\mathbb{Z})$ doubles the number of internal dimensions \cite{Siegel:1993th, Hull:2004in, Dabholkar:2005ve, Hull:2006va, Hull:2009mi} and is referred to as double field theory. It can be understood as a geometrization of the winding modes in the underlying string theory. The physical spacetime is then described by an $(\ell-1)$-dimensional subspace of this enhanced internal geometry, which depends on the choice of the duality frame. The additional inclusion of dualities which correspond to the other part in \eqref{eq:Ugroup} requires an extension of this setup known as exceptional field theory (see, e.g., \cite{Berman:2011jh, Coimbra:2011ky, Hohm:2013jma, Hohm:2013pua}).\footnote{See also \cite{Lazaroiu:2021vmb} for a mathematical description of U-duality bundles.}

These $D$-dimensional exceptional field theories present a covariant formulation of eleven-dimensional supergravity under the full U-duality group $G^{D}_U$. Instead of a doubling of the coordinates one includes even further auxiliary directions, for example, for $D=8$ the internal space coordinates are given by the in term of the six-dimensional $(\mathbf{2},\mathbf{3})$ representation of the $\text{SL}(2,\mathbb{Z}) \times \text{SL}(3,\mathbb{Z})$ duality group.\footnote{In lower dimensions the number of auxiliary dimensions is even bigger. For example, for $D=5$ one has 27 coordinates that transform in the 27-dimensional representation of E$_{6(6)}$ only six of which are part of the physical M-theory spacetime.} The physical spacetime in M-theory is formed by a three-dimensional subspace, the choice of which corresponds to the chosen duality frame. The chosen duality frame is also manifest in the solution to the so-called section constraint, which ensures a closure of the algebra of symmetries in exceptional field theory and restricts the dependence of fields in the theory to a subset of coordinates. It can be written as
\begin{equation}
Y^{MN}_{\phantom{MN}PQ} \, \partial_{M} \otimes \partial_{N} = 0 \,,
\end{equation}
which holds when applied to the fields in the theory, and $M,N,P,Q$ label the internal coordinates. The tensor $Y^{MN}_{\phantom{MN}PQ}$ transforms covariantly under the U-duality group. It can be understood as a projection to some particular representation of the duality group.

The bosonic fields of the supergravity theory naturally appear in representations of the U-duality group that can be phrased in terms of the index structure of the internal coordinates. Again, in eight dimensions this involves the duality invariant spacetime metric $g_{\mu \nu}$ and 3-form $\mathcal{C}_{\mu \nu \rho}$, the six vectors $\mathcal{A}_{\mu}$ and three 2-forms $\mathcal{B}_{\mu \nu}$ transforming in the (anti-)fundamental, and the scalar fields which parameterize the moduli space in \eqref{eq:8dmoduli} whose dependence on the internal directions is restricted via the section constraint, as we discuss in more detail below.

\vspace{3mm}

We now focus on the realization of maximal eight-dimensional supergravity, which will be the main subject of our investigation.

\section{U-duality in eight dimensions}
\label{sec:Udual8d}

The (bosonic) U-duality group in eight dimensions is given by
\begin{equation}
G_U^{\text{8d}} \equiv G_U = \text{SL}(2,\mathbb{Z}) \times \text{SL}(3,\mathbb{Z}) \,,
\end{equation}
that acts on the fields in the supergravity theory. The bosons are given by
\begin{equation}
\begin{split}
\text{graviton:}& \quad g_{\mu \nu} \,, \\
\text{vectors:}& \quad A_{\mu}^{\alpha,a} \,, \\
\text{2-forms:}& \quad B^{\alpha}_{\mu \nu} \,, \\
\text{3-form:}& \quad C_{\mu \nu \rho} \,, \\
\text{scalars:}& \quad \varphi^I \,, \enspace I \in \{1 \,, \dots \,, 7 \} \,.
\end{split}
\label{eq:sugrafields}
\end{equation}
The index $\alpha$ takes values in $\{1 \,, 2 \,, 3 \}$, and the index $a$ in $\{1 \,, 2\}$, respectively, while $\mu \in \{ 0 \,, \dots \,, 7\}$ is the Lorentz index of the eight-dimensional theory. All of these fields are contained in the supergravity multiplet, which further contains two gravitinos $\Psi^i_{\mu}$ of opposite chiralities as well as six fermions $\chi^n$. The index structure above demonstrates that the fields transform in certain representations under the U-duality group, where $\alpha \,, \beta$ can be regarded as SL$(3,\mathbb{Z})$ and $a \,, b$ as SL$(2,\mathbb{Z})$ indices, respectively. In particular, one has
\begin{equation}
A_{\mu}^{\alpha, a}: \enspace (\mathbf{3}, \mathbf{2}) \,, \quad B^{\alpha}_{\mu \nu}: \enspace (\mathbf{3}, \mathbf{1}) \,,
\end{equation}
while the scalars transform in a more complicated way, which we will explore further below.

Since configurations related by duality transformations are physically equivalent, the U-duality group appears as a gauge group in supergravity.\footnote{The field backgrounds can break these discrete gauge symmetries spontaneously.} As mentioned above supergravity theories with maximal supersymmetry are closely related to superstring theory and M-theory compactified on tori and hence some subgroups of the full U-duality can be realized as symmetries of the compactification spaces. We will discuss this in detail for M-theory on $T^3$ and type IIB on $T^2$.

\subsection{U-duality from M-theory}
\label{subsec:UM}

Let us start with eleven-dimensional supergravity, the low-energy description of M-theory. The bosonic field content consists of the metric $G_{MN}$ and a (pseudo)three-form potential $C_{MNR}$. The maximal eight-dimensional supergravity theory is obtained via compactification on a 3-torus $T^3$, whose coordinates we label by $\alpha \,,  \beta \,, \gamma \in \{1 \,, 2 \,, 3 \}$. The lower-dimensional fields arise as follows:
\begin{equation}
\begin{split}
\text{graviton}:& \quad G_{\mu \nu} \,, \\
\text{vectors}:& \quad G_{\alpha \mu} \,, C_{\alpha \beta \mu} \,, \\
\text{2-forms}:& \quad C_{\alpha \mu \nu} \,, \\
\text{3-form}:& \quad C_{\mu \nu \rho} \,, \\
\text{scalars}:& \quad G_{\alpha \beta} \,, C_{\alpha \beta \gamma} \,.
\end{split}
\label{eq:Mfields}
\end{equation}
Up to field redefinitions, we find the exact same spectrum as in \eqref{eq:sugrafields}. The SL$(3,\mathbb{Z})$ factor of $G_U$ is realized as large diffeomorphisms of $T^3$ and acts on the internal coordinates labeled by $\alpha \,, \beta \,, \gamma$. This demonstrates that the vector fields transform as two three-dimensional representations under SL$(3,\mathbb{Z})$ and similarly the 2-forms transform in the three-dimensional representation. The scalars split into a singlet under SL$(3,\mathbb{Z})$ originating from the internal components of the 3-form and the symmetric representation associated to the internal metric components that parameterize both the shape and size of the 3-torus.

We can therefore understand pure SL$(3,\mathbb{Z})$ duality bundles geometrically as $T^3$-fibrations in an M-theory framework (see Figure \ref{fig:Mgeom}).
\begin{figure}
\centering
\includegraphics[width = 0.4 \textwidth]{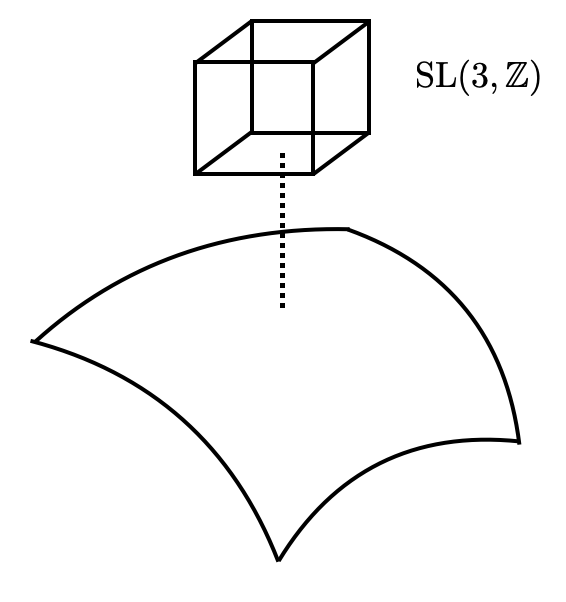}
\caption{M-theory capturing SL$(3,\mathbb{Z}) \subset G_U$ of the U-duality group geometrically.}
\label{fig:Mgeom}
\end{figure}
If one of the circles of $T^3$ is trivially fibered we can reduce to a type IIA framework, i.e., M-theory on a circle (see Figure \ref{fig:IIAgeom}). 
\begin{figure}
\centering
\includegraphics[width = 0.7 \textwidth]{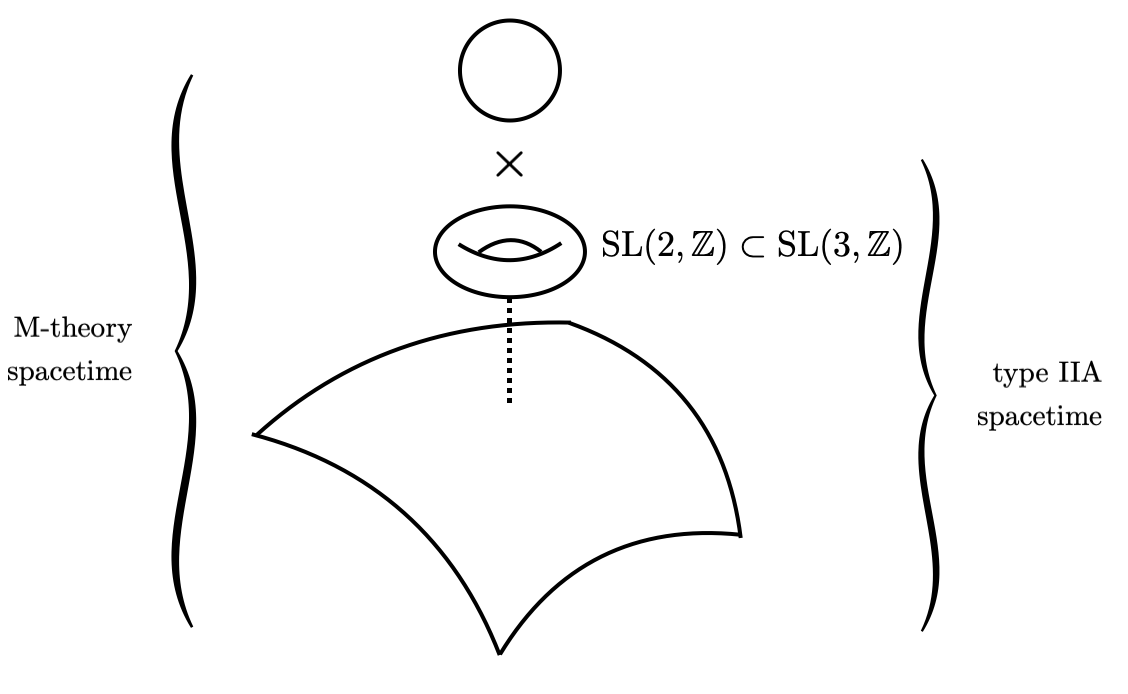}
\caption{Type IIA capturing SL$(2,\mathbb{Z}) \subset \text{SL}(3,\mathbb{Z}) \subset G_U$ of the U-duality group geometrically.}
\label{fig:IIAgeom}
\end{figure}
Type IIA on a 2-torus has a T-duality group 
\begin{equation}
\text{SO}(2,2,\mathbb{Z}) = \big( \text{SL}(2,\mathbb{Z}) \times \text{SL}(2,\mathbb{Z}) \big) / \mathbb{Z}_2 \,,
\end{equation}
the first factor being the remnant of the SL$(3,\mathbb{Z})$ subgroup\footnote{This is not the same SL$(2,\mathbb{Z})$ subgroup as the strong-weak coupling duality SL$(2,\mathbb{Z})_S$ in type IIB.} and the second factor giving rise to the SL$(2,\mathbb{Z})$ factor of $G_U$. This second factor acts on the K\"ahler sector of type IIA, i.e.,
\begin{equation}
\omega_{\text{IIA}} = B + i \, \text{vol} (T^2) \,.
\end{equation}
One can lift that to a complexified volume in M-theory by defining
\begin{equation}
\omega_{\text{M}} = C + i \, \text{vol}(T^3) \,,
\label{eq:MKahler}
\end{equation}
which also transforms under Moebius transformation with respect to the SL$(2,\mathbb{Z})$ factor. Transition functions in the SL$(2,\mathbb{Z})$ factor of $G_U$ thus connect M-theory torus compactifications of different volume and C-field and in general acts non-geometrically.

\subsection{U-duality from type IIB}
\label{subsec:UIIB}

Similarly, starting with type IIB supergravity (i.e., $\mathcal{N} = (2,0) $) in ten dimensions, one obtains maximal supergravity in eight dimensions after compactification on a 2-torus $T^2$. We will denote the 10d bosonic fields by $g_{AB}$ for the metric, $B_{AB}$ and $C_{AB}$ for the RR and NS 2-forms, $C^+_{ABCD}$ for the chiral 4-form with self-dual field strength, as well as $\tau$ for the axio-dilaton, which as usual is given by
\begin{equation}
\tau = C_0 + i e^{- \phi} \,,
\end{equation}
with RR axion $C_0$ and dilaton $\phi$. The bosonic fields in \eqref{eq:sugrafields} arise from the type IIB fields as follows:\footnote{Here we only keep track of the Lorentz structure. The connection to the actual 8d fields \eqref{eq:sugrafields} requires further field redefinitions.}
\begin{equation}
\begin{split}
\text{graviton:}& \quad g_{\mu \nu} \,, \\
\text{vectors:}& \quad g_{a \mu} \,, B_{a \mu} \,, C_{a \mu} \,, \\
\text{2-forms:}& \quad C^+_{a b \mu \nu} \,, B_{\mu \nu} \,, C_{\mu \nu} \,, \\
\text{3-form:}& \quad C^+_{a \mu \nu \rho} \,, \\
\text{scalars:}& \quad \tau \,, g_{ab} \,, B_{ab} \,, C_{ab} \,, 
\end{split}
\label{eq:IIBfields}
\end{equation}
where we label the internal coordinates with $a,b \in \{1 \,, 2\}$. Note that since the 4-form field is chiral (i.e., its field strength is self-dual) it only gives rise to a single 3-form field in eight dimensions. The internal metric components $g_{a b}$ transform in the symmetric representation of SL$(2,\mathbb{Z})$, while the scalars originating from $B_{AB}$ and $C_{AB}$ transform in the antisymmetric, i.e., singlet representation. This fixes some of the transformation properties of the scalar fields.

We see that the internal SL$(2,\mathbb{Z})$ index is given by the coordinates of the $T^2$ compactification torus which acts for example on the two vector fields originating from the same IIB field in ten dimensions. This is to be expected since the SL$(2,\mathbb{Z})$ subgroup of the full duality group $G_U$ is identified with the group of large diffeomorphisms of the compactification torus.
\begin{figure}
\centering
\includegraphics[width = 0.6 \textwidth]{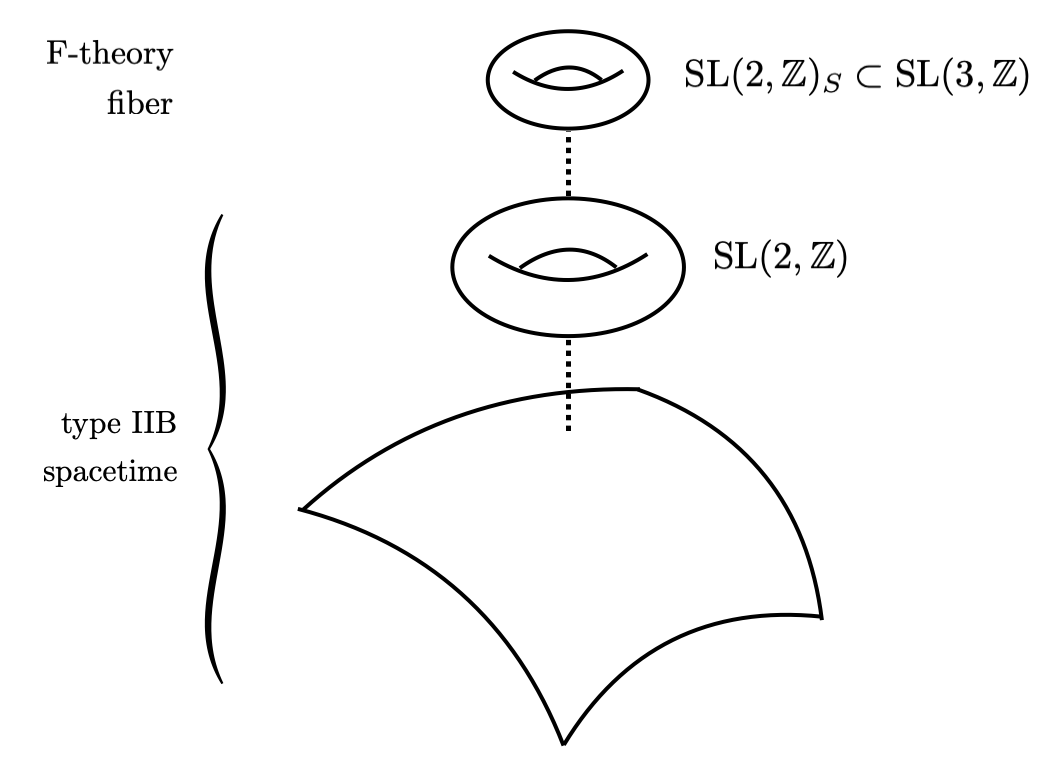}
\caption{F-theory capturing SL$(2,\mathbb{Z}) \times \text{SL} (2,\mathbb{Z})_S \subset G_U$ of the U-duality group geometrically.}
\label{fig:Fgeom}
\end{figure}
U-duality backgrounds that only have a non-trivial SL$(2,\mathbb{Z})$ bundle can therefore be understood geometrically as type IIB geometries described as torus fibrations (see Figure \ref{fig:Fgeom}).

However, one can do even better by going to the non-perturbative description of type IIB provided by F-theory \cite{Vafa:1996xn, Weigand:2018rez}. This geometrizes the strong-weak coupling duality, which appears as an SL$(2,\mathbb{Z})_S$ subgroup of the SL$(3,\mathbb{Z})$ factor in $G_U$. Thus, F-theory allows for a geometric realization of
\begin{equation}
\text{SL}(2,\mathbb{Z}) \times \text{SL}(2,\mathbb{Z})_S \subset G_U = \text{SL}(2,\mathbb{Z}) \times \text{SL}(3,\mathbb{Z}) \,.
\end{equation}
Bundles contained in this subgroup can be understood as F-theory backgrounds with its $T^2_F$ torus fibration encoding the SL$(2,\mathbb{Z})_S$ duality, whose base is given by another $T^2$ torus fibration encoding the SL$(2,\mathbb{Z})$ (see Figure \ref{fig:Fgeom}). Note that only the torus $T^2$ within the F-theory base is part of spacetime and its volume is physical, whereas the volume of the F-theory fiber torus $T^2_F$ is not physical and typically set to zero. 

A non-trivial SL$(2,\mathbb{Z})_S$ bundle indicates the variation of the axio-dilaton in spacetime, typically leading to regions of strong string coupling and therefore going beyond perturbative string backgrounds. This action is captured by the Moebius transformations
\begin{equation}
\text{SL}(2,\mathbb{Z})_S: \quad \tau \mapsto \frac{a \tau + b}{c \tau + d} \,, \quad \begin{pmatrix} a & b \\ c & d \end{pmatrix} \in \text{SL}(2,\mathbb{Z}) \,.
\end{equation}
Moreover, we know that $B_{AB}$ and $C_{AB}$ transform as a doublet under strong weak coupling duality, which shows that their descendants in \eqref{eq:IIBfields} transform accordingly. The full SL$(3,\mathbb{Z})$ includes the K\"ahler sector of type IIB and acts, for example, on the complexified volume of the compactification torus
\begin{equation}
\omega_{\text{IIB}} = B + i \, \text{vol}(T^2) \,,
\end{equation}
which combines the integral of the $B$-field with metric components, e.g., \cite{Basu:2007ck}. In fact this complex scalar is acted upon by Moebius transformations under a different SL$(2,\mathbb{Z})$ subgroup of SL$(3,\mathbb{Z})$. This shows that one can generate transition functions that connect different volumes of the compactification torus and demonstrates the non-geometric nature of these type of backgrounds in the type IIB duality frame. The full SL$(3,\mathbb{Z})$ in the type IIB frame can then recovered by dualities.

\subsection{M-/F-theory duality}
\label{subsec:MF}

Of course these descriptions are connected via M-/F-theory duality which we now describe briefly. 

As we already discussed above we can single out a circle of the M-theory compactification in order to reduce to a well-defined type IIA description at small circle volumes, i.e., small string coupling. The type IIA theory is compactified on the remaining 2-torus. Performing T-duality along one of the circle factors of this IIA compactification leads to type IIB on $T^2$, sending the IIA circle size to zero lets the IIB circle grow and decompactifies to a type IIB configuration on a circle. Indeed, following the transformation behavior of the fields, see, e.g., \cite{Denef:2008wq, Weigand:2018rez}, shows that the 2-torus composed of the M-theory type IIA circle as well as the T-duality circle combine into the F-theory fiber torus. This identifies the SL$(2,\mathbb{Z})_S$ subgroup of SL$(3,\mathbb{Z})$. In the following we will choose it to be given by elements of the form
\begin{equation}
\text{SL}(2,\mathbb{Z})_S \subset \text{SL}(3,\mathbb{Z}): \quad \begin{pmatrix} 1 & 0 & 0 \\ 0 & a & b \\ 0 & c & d \end{pmatrix} \,, \quad \begin{pmatrix} a & b \\ c & d \end{pmatrix} \in \text{SL}(2,\mathbb{Z}) \,,
\end{equation}
which is always possible using conjugations.

The remaining circle direction of the M-theory 3-torus is part of the type IIA compactification torus and is a spectator in the T-duality to type IIB. It therefore corresponds to the type IIB compactification circle in the limit of zero fiber volume. This means that general SL$(3,\mathbb{Z})$ transformations mix spacetime directions in type IIB with the F-theory fiber directions, underlining once more the non-geometric nature of these transformations in the type IIB duality frame. This can be understood as the familiar fact that T-duality acts on the (axio-)dilaton field \cite{polchinski1998string}. Vice versa, the SL$(2,\mathbb{Z})$ factor in $G_U$ acts on the compactification torus $T^2$ in type IIB and therefore is not realized as a geometric SL$(3,\mathbb{Z})$ element in M-theory which only has access to one of the compactification circles.

Of course one can also follow the various supergravity fields and string and M-theory objects, such as branes, under the duality chains, which is conveniently captured by the representations of the eight-dimensional supergravity fields under the U-duality group, see \cite{Obers:1998fb}. In situations where the compactification tori are fibered over the eight-dimensional spacetime, one needs to perform the dualities above fiber-wise, which also allows one to track the effect of U-duality monodromies in the various duality frames.

\subsection{Exceptional field theory}
\label{subsec:exft}

We now want to discuss the geometrization of the 8d U-duality group within exceptional field theory as described in the more general context in Section \ref{subsec:geomEFT}.

The 8d supergravity theory has six vector fields which as shown above transform in the $(\mathbf{3}, \mathbf{2})$ representation of the duality group. This motivates the introduction of six extra internal coordinates $Y^M$, with $M \in \{ 1 \,, 2 \,, \dots \,, 6 \}$, which describe the action of the duality bundle. We can also introduce a two index notation in analogy to the vector fields $Y^{\alpha a}$, with $\alpha \in \{1 \,, 2 \,, 3\}$ and $a \in \{ 1 \,, 2 \}$, which are acted upon by the SL$(3,\mathbb{R})$ and SL$(2,\mathbb{R})$, respectively. Including the eight coordinates of spacetime $x^{\mu}$ one finds a 14-dimensional spacetime parametrized by $(x^{\mu},Y^M)$. To reduce to the physical spacetime one needs to introduce the section constraint, which in eight dimensions takes the form, see, e.g., \cite{Hohm:2015xna},
\begin{equation}
\epsilon^{\alpha \beta \gamma} \epsilon^{ab} \partial_{\alpha a} \otimes \partial_{\beta b} = 0 \,,
\end{equation}
where $\epsilon^{\alpha \beta \gamma}$ is the totally antisymmetric tensor acting on SL$(3,\mathbb{R})$ indices and $\epsilon^{a b}$ its counterpart acting on SL$(2,\mathbb{R})$ indices. We see that this isolates the $(\mathbf{1},\mathbf{3})$ component.

There are several ways to satisfy the section constraint. For example, we can demand that fields only depend on $Y^{\alpha 1}$ singling out the $a = 1$ direction. In this way we find 
\begin{equation}
\epsilon^{\alpha \beta \gamma} \partial_{\alpha 1} \otimes \partial_{\beta 2} = 0 \,,
\end{equation}
since nothing depends on $\partial_{\beta 2}$. This suggests the physical realization of three of the six coordinates, which we expect to correspond to an M-theory framework. Similarly, we can choose that fields only depend on $Y^{1 a}$, finding that
\begin{equation}
\epsilon^{a b} \epsilon^{1 \beta \gamma} \partial_{1 a} \otimes \partial_{\beta b} = 0 \,,
\end{equation}
since $\partial_{\beta b}$ for $\beta \neq 1$, which is enforced by $\epsilon^{1 j k}$, acts trivially. In this frame there are only two physical coordinates and we expect to find the type IIB interpretation. Note that here one has a remaining SL$(2,\mathbb{R})$ action on the coordinates $\beta$ and $\gamma$, which encodes the S-duality and can be thought of as the F-theory torus directions.\footnote{One can also demand the fields to only depend on $Y^{\alpha 1}$ with $i \in \{ 1 \,, 2\}$ which solves the section constraint and reduces to a type IIA framework.}

With this one can reconstruct the field content depending on the solution of the section constraint in the following way. One starts with the exceptional field theory fields $(g_{\mu \nu} \,, \mathcal{M}_{MN} \,, \mathcal{A}_\mu \,, \mathcal{B}_{\mu \nu} \,, \dots)$, describing the spacetime metric, internal metric, gauge fields, and tensor fields, respectively. All of these fields transform in certain U-duality representations. As we have seen from the particular duality frames in eight-dimensional supergravity one has
\begin{equation}
\begin{split}
\mathcal{A}_{\mu} &\sim (\mathbf{2}, \mathbf{3}) \,, \\ 
\mathcal{B}_{\mu \nu} &\sim (\mathbf{1}, \mathbf{3}) \,, \\ 
\mathcal{C}_{\mu \nu \rho} &\sim (\mathbf{1}, \mathbf{1}) \,.
\end{split}
\end{equation}
The choice of solution of the section constraint then breaks the full U-duality group to a subgroup, which is geometrically realized in the particular duality frame. This leads to a decomposition of the full $G_U$ representations that can be matched to the original fields. For example, picking the coordinate dependence to be in $Y^{\alpha 1}$ preserves SL$(3,\mathbb{R})$ and we find
\begin{equation}
\begin{split}
(\mathbf{2},\mathbf{3}) &\rightarrow \mathbf{3} \oplus \mathbf{3} \sim \{ G_{\alpha \mu} \,, C_{\alpha \beta \mu} \} \,, \\
(\mathbf{1}, \mathbf{3}) &\rightarrow \mathbf{3} \sim \{ C_{\alpha \mu \nu} \} \,, \\
(\mathbf{1}, \mathbf{1}) &\rightarrow \mathbf{1} \sim \{ C_{\mu \nu \rho}\} \,.
\end{split}
\end{equation}
As it should, this precisely matches our M-theory discussion \eqref{eq:Mfields} above. Similarly, we can pick the physical coordinates to be $Y^{1 a}$ which leads to
\begin{equation}
\begin{split}
(\mathbf{2},\mathbf{3}) &\rightarrow \mathbf{2} \oplus \mathbf{2} \oplus \mathbf{2} \sim \{ g_{a \mu} \,, B_{a \mu} \,, C_{a \mu} \}\,, \\
(\mathbf{1}, \mathbf{3}) &\rightarrow \mathbf{1} \oplus \mathbf{1} \oplus \mathbf{1} \sim \{ C^{+}_{a b \mu \nu} \,, B_{\mu \nu} \,, C_{\mu \nu} \} \,, \\
(\mathbf{1}, \mathbf{1}) &\rightarrow \mathbf{1} \sim \{ C^{+}_{a \mu \nu \rho} \} \,,
\end{split}
\end{equation}
i.e., the type IIB interpretation \eqref{eq:IIBfields}. Here, one could also resolve the unbroken SL$(2, \mathbb{R})_S$  $\subset$ SL$(3,\mathbb{R})$ which describes the transformation under S-duality, acting for example on the two 2-forms $B_{\mu \nu}$ and $C_{\mu \nu}$.

By using the framework of exceptional field theory we can geometrize all duality bundles. One can then reconstruct the non-geometric backgrounds by different solutions of the section constraint. The same should work for the associated duality defects, which one can formulate in the 14-dimensional theory and then take appropriate subspaces for the M-theory or type IIB/F-theory realization. We hope to come back to a detailed analysis of duality defects in this framework.\footnote{Note that the geometrization within exceptional field theory is different from S-theory \cite{Kumar:1996zx} in which one encodes the $\text{SL}(2,\mathbb{Z}) \times \text{SL}(3,\mathbb{Z})$ bundle in terms of a $T^2 \times T^3$ fibration leading to a 13-dimensional theory.}

\vspace{1cm}

With the geometrical interpretation of the U-duality group in various duality frames at hand we can now move to analyze the bordism defects and their string theory origin.

%%%%%%%%%%%%%%%%%%%%%%%%%%%%%%%%%%%%%%%%%%%%%%
%%%%%%%%%%%%%%%%%%%%%%%%%%%%%%%%%%%%%%%%%%%%%%

\newpage

\part{Supergravity backgrounds and defects}
\label{part:physics}

\section{Monodromies and fiber degenerations}
\label{sec:monodrom}

Decomposing the U-duality bundles using various subgroups of $\text{SL}(2,\mathbb{Z}) \times \text{SL}(3,\mathbb{Z})$ according to its stable splitting (see Section \ref{subsec:stablesplit}), we find a distinguished set of duality transformations which appear as monodromies in the generators of the bordism group. In many cases these monodromies have a geometric interpretation in terms of the $T^2$ compactification of type IIB, the $T_F^2$ of an F-theory fiber, or the $T^3$ in M-theory compactifications, as discussed in the UV realization of the U-duality group in Section~\ref{sec:Udual}. The duality defects that break the associated global symmetry are singular objects in the low-energy effective description. These objects are located where the fiber geometry becomes singular. Since we are interested in relating the singular duality defects of the eight-dimensional supergravity theory with stringy objects, it will be important to have a good understanding of these fiber degenerations.

\subsection{\texorpdfstring{$\SL(2,\mathbb{Z})$}{SL(2,Z)}}
\label{subsec:SL2central}

The two relevant matrices for SL$(2,\mathbb{Z})$ are given by
\begin{equation}
\gamma_3 = (ST)^2 = \begin{pmatrix*}[r] -1 & -1 \\ 1 & 0 \end{pmatrix*} \,, \quad \gamma_4 = S = \begin{pmatrix*}[r] 0 & -1 \\ 1 & 0 \end{pmatrix*} \,,
\end{equation}
where we used the usual convention for the generators of SL$(2,\mathbb{Z})$. These monodromy matrices act on a 2-torus whose shape is parameterized by a single complex parameter $\tau$ in the upper half-plane. One can understand that as
\begin{equation}
T^2 = \mathbb{R}^2 / \Lambda_2 \,,
\end{equation}
with the two-dimensional lattice $\Lambda_2$ spanned by $\langle 1 \,, \tau \rangle$. For general SL$(2,\mathbb{Z})$ transformations this parameter transforms via Moebius transformations
\begin{equation}
\tau \mapsto \frac{a \tau + b}{c \tau + d} \,, \quad \text{with} \enspace \begin{pmatrix} a & b \\ c & d \end{pmatrix} \in \text{SL}(2,\mathbb{Z}) \,.
\end{equation}
For the degenerate fiber on the duality defect the monodromy should leave the parameter $\tau$ invariant. Thus, we find that $\tau$ is fixed to special values
\begin{equation}
\begin{alignedat}{3}
\gamma_3:  \quad \tau &= - \frac{\tau + 1}{\tau} \quad &\Rightarrow \quad \tau &=  e^{2 \pi i/3} \,, \\
\gamma_4:  \quad \tau &= - \frac{1}{\tau} \quad &\Rightarrow \quad \tau &=  i \,.
\end{alignedat}
\end{equation}
For the 2-tori at these special values of $\tau$, $\gamma_k$ acts via rotation by $2 \pi / k$ in the complex plane forming the group $\mathbb{Z}_k$. The degenerate fiber then takes the form of a torus orbifold $T^2 / \mathbb{Z}_k$. For $k = 3$ and $k = 4$ these are depicted in Figure \ref{fig:2dorbi}. The $T^2/\mathbb{Z}_3$ orbifold has three conical singularities of the form $\mathbb{C}/\mathbb{Z}_3$ located at the orbifold fixed points $\{0 \,, \tfrac{1}{3} + \tfrac{2 \tau}{3}, \tfrac{2}{3} + \tfrac{\tau}{3}  \}$. The $T^2/\mathbb{Z}_4$ orbifold has two singularities of the form $\mathbb{C}/\mathbb{Z}_4$ at $\{ \tfrac{1}{2} \,, \tfrac{\tau}{2} \}$ and one of the form $\mathbb{C}/\mathbb{Z}_2$ at the origin.
\begin{figure}
\includegraphics[width = \textwidth]{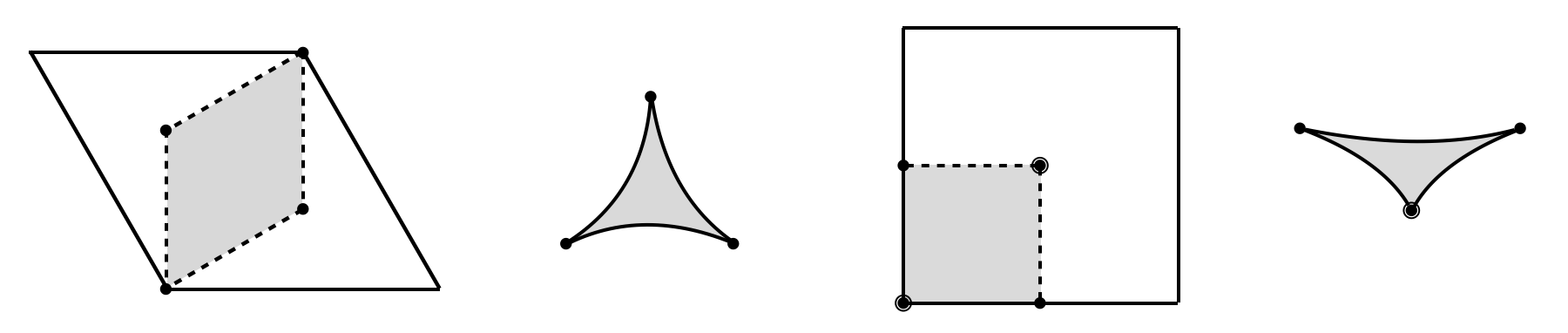}
\caption{The two relevant orbifolds $T^2/\mathbb{Z}_3$ (left), $T^2/\mathbb{Z}_4$ (right), with fundamental domain given by the shaded region, and fixed points marked with black circles (the $\mathbb{C}/\mathbb{Z}_2$ is indicated by a further circle).}
\label{fig:2dorbi}
\end{figure}

\subsection{\texorpdfstring{$\SL(3,\mathbb{Z})$}{SL(3,Z)}}
\label{subsec:SL3central}

For SL$(3,\mathbb{Z})$ we need to consider the two order-three transformations
\begin{equation}
\Gamma_3^{(1)} = \begin{pmatrix*}[r] 1 & 0 & 0 \\ 0 & -1 & -1 \\ 0 & 1 & 0 \end{pmatrix*} \,, \quad \Gamma_3^{(2)} = \begin{pmatrix*}[r] 0 & 1 & 0 \\ 0 & 0 & 1 \\ 1 & 0 & 0 \end{pmatrix*} \,.
\end{equation}
We see that $\Gamma_3^{(1)}$ is block diagonal with $\gamma_3$ in the lower right corner. Therefore, it acts on a two-dimensional sub-torus of $T^3$, which we choose to be the one corresponding to the F-theory fiber $T^2_F$ (see Section \ref{subsec:MF}) and fix its complex structure to $e^{2 \pi i /3}$ as above. Since one of the circles in $T^3$ does not transform at all, the singular fiber is given by
\begin{equation}
T^3 / (\mathbb{Z}_3)_{\Gamma_3^{(1)}} = (T_F^2 / \mathbb{Z}_3) \times S^1 \,.
\end{equation}
Things become more interesting for $\Gamma_3^{(2)}$, which does not leave any individual circle in $T^3$ invariant. Since one has
\begin{equation}
\big( \Gamma^{(2)}_3 \big)^T \Gamma^{(2)}_3 = \mathbf{1} \,,
\end{equation}
it is an element of SO$(3)$ and can be understood as rotation by $2 \pi/3$ in $\mathbb{R}^3$ along the axis $v = (1,1,1)$. Taking $T^3 = \mathbb{R}^3 / \mathbb{Z}^3$ and using translations by lattice vectors we find the action on $T^3$ (as shown in Figure~\ref{fig:T3Z3_rot}) given by a rotation along one of the diagonals.
\begin{figure}
\centering
\includegraphics[width = 0.8 \textwidth]{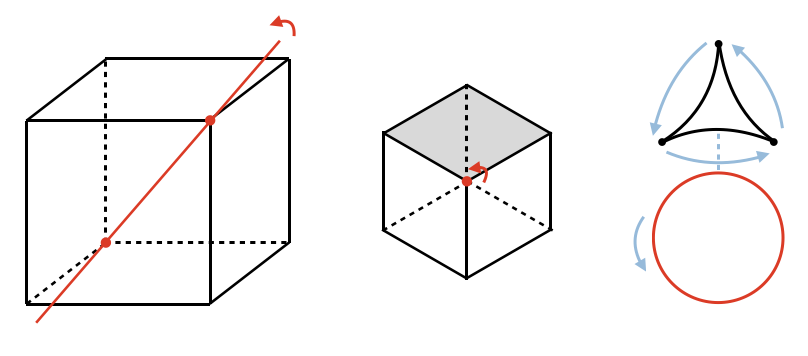}
\caption{The singular fiber $T^3 / (\mathbb{Z}_3)_{\Gamma^{(2)}_3}$ as fibration of $T^2 / \mathbb{Z}_3$ over $S^1$ (right), obtained by modding out a $\mathbb{Z}_3$ rotation indicated in two different perspectives on the left and in the middle.} 
\label{fig:T3Z3_rot}
\end{figure}
From this we expect a non-isolated singularity of the form $(\mathbb{C}/\mathbb{Z}_3) \times S^1$ given by the rotation axis. This is indeed the case and the singular fiber can be described as
\begin{equation}
T^3 / (\mathbb{Z}_3)_{\Gamma_3^{(2)}} = \big( (T^2/\mathbb{Z}_3) \times S^1 \big) / \mathbb{Z}^s_3 \,.
\label{eq:twistedcircleZ3}
\end{equation}
The quotient by $\mathbb{Z}_3^s$ acts as a translation around the base circle and an action on the torus orbifold that exchanges the three fixed points. For a technique to construct this singular fiber we refer to Appendix \ref{app:cryssingfiber}.

Other relevant $\SL(3,\mathbb{Z})$ transformations are of order four and given by
\begin{equation}
\Gamma_4^{(1)} = \begin{pmatrix*}[r] 1 & 0 & 0 \\ 0 & 0 & -1 \\ 0 & 1 & 0 \end{pmatrix*} \,, \quad \Gamma_4^{(2)} = \begin{pmatrix*}[r] 1 & 1 & 1 \\ -1 & 0 & 0 \\ 0 & -1 & 0 \end{pmatrix*} \,.
\label{eq:Gamma4}
\end{equation}
As above $\Gamma_4^{(1)}$ is of block-diagonal form which only acts on a two-dimensional torus, which we choose to be the 2-torus associated to the F-theory fiber under dualities. We find that the singular fiber is simply given by
\begin{equation}
T^3/(\mathbb{Z}_4)_{\Gamma^{(1)}_4} = (T_F^2 / \mathbb{Z}_4) \times S^1 \,.
\end{equation}
For $\Gamma^{(2)}_4$ things are more complicated, especially since it is not an element of SO$(3)$ and we need to find a good basis for the torus lattice $\Lambda_3$. This is done in Appendix \ref{app:cryssingfiber}. Similar to the situation above, we find that the resulting singular fiber is given by a fibration of $T^2/\mathbb{Z}_4$ over $S^1$ (see Figure~\ref{fig:T3Z4_orbi}).
\begin{figure}
\centering
\includegraphics[width = 0.2 \textwidth]{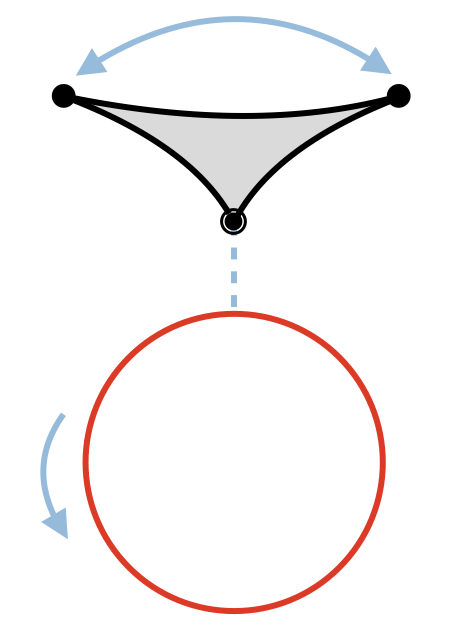}
\caption{The singular fiber $T^3 / (\mathbb{Z}_4)_{\Gamma^{(2)}_4}$ as fibration of $T^2 / \mathbb{Z}_4$ over $S^1$.} 
\label{fig:T3Z4_orbi}
\end{figure}
The fibration structure is given by
\begin{equation}
T^3/(\mathbb{Z}_4)_{\Gamma^{(2)}_4} = \big((T^2 / \mathbb{Z}_4) \times S^1 \big) / \mathbb{Z}_2^s \,,
\end{equation}
where the translational symmetry $\mathbb{Z}_2^s$ shifts halfway around the base circle and acts on the fiber. Going around the base circle this action exchanges the two local singularities of form $\mathbb{C}/\mathbb{Z}_4$, while the local $\mathbb{C}/\mathbb{Z}_2$ singularity is mapped to itself.

Finally, we will also encounter the SL$(3,\mathbb{Z})$ monodromies given by the order-two elements
\begin{equation}\label{eq:M1_matrices}
    M_1^{(1)} = \begin{pmatrix*}[r] -1 & 0 & 0 \\ 0 & -1 & 0 \\ 0 & 0 & 1 \end{pmatrix*} \,, \quad M_2^{(1)} = \begin{pmatrix*}[r] 1 & 0 & 0 \\ 0 & -1 & 0 \\ 0 & 0 & -1 \end{pmatrix*} \,,
\end{equation}
and 
\begin{equation}\label{eq:M2_matrices}
    M_1^{(2)} = \begin{pmatrix*}[r]
        -1&0&0\\
        1&1&1\\
        0&0&-1
    \end{pmatrix*}\,, \quad
    M_2^{(2)} = \begin{pmatrix*}[r]
        0&0&1\\
        -1&-1&-1\\
        1&0&0
    \end{pmatrix*} \,.
 %   0 & 1 & 0 \\ 1 & 0 & 0 \\ -1 & -1 & -1 \end{pmatrix} \,, \quad M_2^{(2)} = \begin{pmatrix} 0 & 0 & 1 \\ -1 & -1 & -1 \\ 1 & 0 & 0 \end{pmatrix} \,.
\end{equation}
It will be enough to discuss the singular fiber of one of each pair of monodromies which we will choose to be $M_2^{(1)}$ and $M_2^{(2)}$. 

For $M_2^{(1)}$ we see the same block diagonal form as for $\Gamma_4^{(1)}$, which acts non-trivially only on the F-theory fiber $T^2_F \subset T^3$. The singular fiber, does not require a special value for the complex structure and takes the product form
\begin{equation}
    T^3/(\mathbb{Z}_2)_{M_2^{(1)}} = (T^2_F / \mathbb{Z}_2) \times S^1 \,,
\end{equation}
with the usual 2-torus orbifold as the F-theory fiber.

For $M_2^{(2)}$ there is no invariant sub-torus and we employ the same techniques as above (see Appendix \ref{app:cryssingfiber}) to obtain the singular fiber geometry:
\begin{equation}
    T^3 / (\mathbb{Z}_2)_{M_1^{(2)}} = \big( (T^2/\mathbb{Z}_2) \times S^1 \big) / \mathbb{Z}_2^s \,,
\end{equation}
with $\mathbb{Z}_2^s$ given by a half-shift around a circle in combination with an exchange of the diagonal fixed points on $T^2/\mathbb{Z}_2$.

\subsection{\texorpdfstring{$\SL(2,\mathbb{Z}) \times\SL(3,\mathbb{Z})$}{SL(2,Z) x SL(3,Z)}}
\label{subsec:nongeomsingfiber}

Finally, we want to describe singular fibers for simultaneous monodromies in both duality subgroups. While pairs of the form $\big(\gamma_3, \Gamma_3^{(1)}\big)$ and $\big( \gamma_4, \Gamma_4^{(1)}\big)$ have a geometric interpretation in F-theory we will focus on the pairs $\big(\gamma_3, \Gamma^{(2)}_3\big)$ and $\big(\gamma_4, \Gamma^{(2)}_4\big)$, which will appear for certain bordism generators below, and do not act geometrically in any duality frame.

For that we first determine the form of the monodromy when acting on the six coordinates $Y^M$ of the exceptional field theory description, which is simply the tensor product of the individual monodromy matrices, for example
\begin{equation}
\big( \gamma_3, \Gamma^{(2)}_3 \big) = \begin{pmatrix*}[r] -1 & -1 \\ 1 & 0 \end{pmatrix*} \otimes \begin{pmatrix} 0 & 1 & 0 \\ 0 & 0 & 1 \\ 1 & 0 & 0  \end{pmatrix} = \begin{pmatrix*}[r] 0 & -1 & 0 & 0 & -1 & 0 \\ 0 & 0 & -1 & 0 & 0 & -1 \\ -1 & 0 & 0 & -1 & 0 & 0 \\ 0 & 1 & 0 & 0 & 0 & 0 \\ 0 & 0 & 1 & 0 & 0 & 0 \\ 1 & 0 & 0 & 0 & 0 & 0 \end{pmatrix*} \,.
\end{equation}
Of course this is also an element of order 3. Taking the six internal coordinates to be periodic, i.e., described by a $T^6$, the duality matrix acts on this $T^6$ as an element of SL$(6,\mathbb{Z})$. We can then proceed to analyze the singular central fiber by using the same techniques as above, which we do explicitly in Appendix \eqref{subapp:singnongeo}. There, we find a 9-fold cover $\widetilde{T}^6$ of the original torus which decomposes as
\begin{equation}
    \widetilde{T}^6 / (\mathbb{Z}_3)_{(\gamma_3,\Gamma_3^{(2)})} = T^2 \times (T^2/\mathbb{Z}_3) \times (T^2/\mathbb{Z}_3) \,,
\end{equation}
with the usual 2-torus orbifolds $T^2/\mathbb{Z}_3$. This geometry has nine orbifold points that locally are described by $\mathbb{C}^2/\mathbb{Z}_3$. These get exchanged by the identification generated by shifts associated to interior points of $\widetilde{T}^6$, discussed in Appendix \ref{subapp:singnongeo}. A single fundamental domain of the singular fiber is described by
\begin{equation}
    T^6 / (\mathbb{Z}_3)_{(\gamma_3,\Gamma_3^{(2)})} = \big( T^2 \times (T^2/\mathbb{Z}_3) \times (T^2/\mathbb{Z}_3) \big) / (\mathbb{Z}_3^s \times \widetilde{\mathbb{Z}}_3^s) \,,
    \label{eq:nongeoorbZ3}
\end{equation}
with $\mathbb{Z}_3^s$ and $\widetilde{\mathbb{Z}}_3^s$ acting as shifts in the first $T^2$ factor and identifying the orbifold fixed points in the other two, which hence are non-trivially fibered. A similar discussion can be performed for the $\mathbb{Z}_4$ case.

The fact that the structure of the singular fiber involves all of the 6-torus coordinates underlines the non-geometric nature of the configuration which also involves other moduli fields. It would be interesting to apply the section constraint to the singular fiber. This would determine certain singular subspaces of the orbifold \eqref{eq:nongeoorbZ3}, that define the singular internal spacetime in a certain duality frame. The non-trivial behavior of the other auxiliary coordinates on the other would capture the monodromies of the remaining moduli fields, involving the volumes of the physical internal space.

For example, imposing the solution to the section constraint such that only the $Y^{1a}$ behavior is non-trivial, one obtains a two-dimensional subspace of \eqref{eq:nongeoorbZ3}, with a complicated singularity structure and monodromies for the other moduli fields, from which one could read off the non-geometrical realization within type IIB (see Figure \ref{fig:orbisection} and Appendix \ref{subapp:singnongeo}).
\begin{figure}
\centering
\includegraphics[width = 0.6 \textwidth]{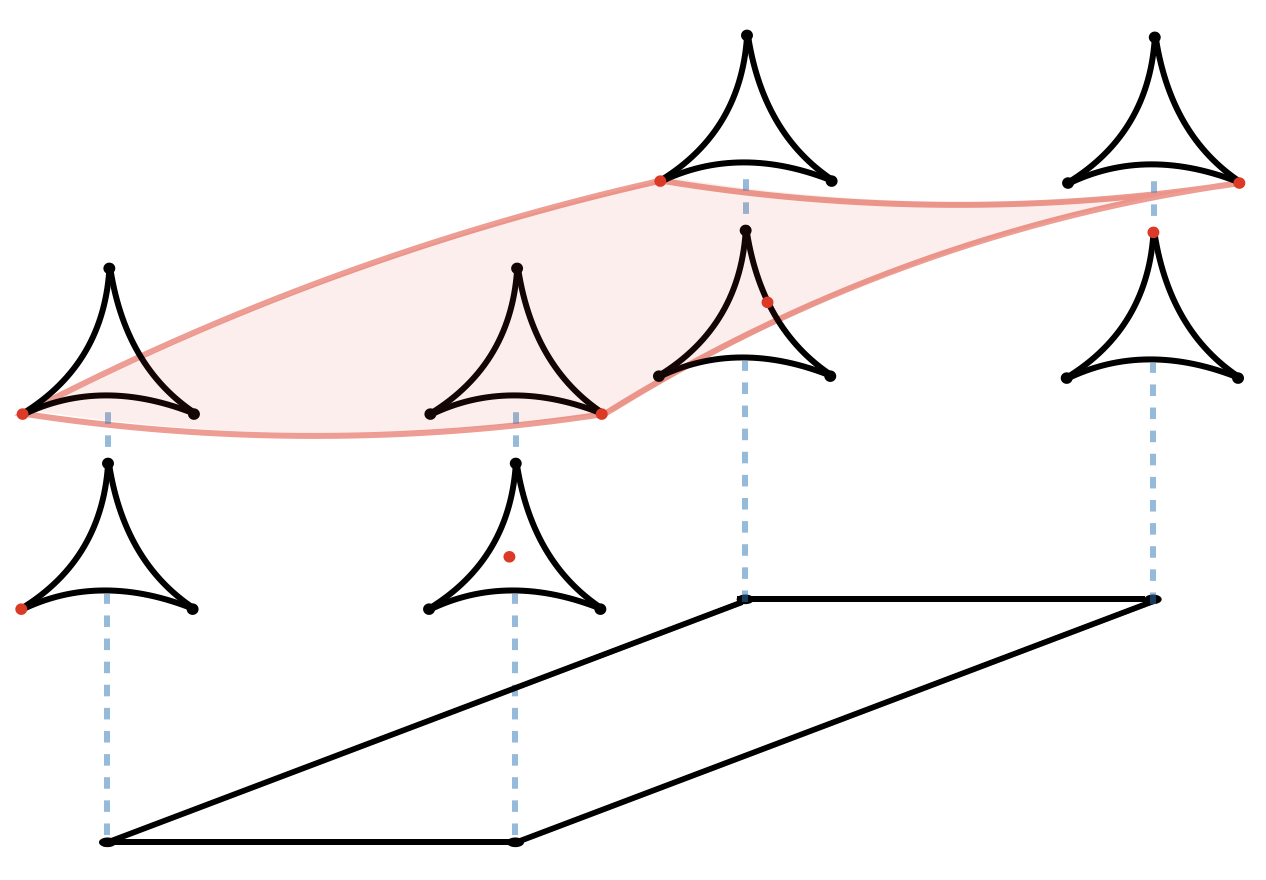}
\caption{Sketch of the solution of the section constraint cutting out a singular subspace of the non-geometric fiber.} 
\label{fig:orbisection}
\end{figure}

With the structure of the singular fibers determined, we can now move to the discussion of the U-duality defects which are located at the point of fiber degeneration.

\section{General U-duality defects in codimension-two}
\label{sec:gencodim2}

Before we discuss the symmetry-breaking defects in maximal supergravity in eight dimensions of any dimensionality we explore the codimension-two defects in spacetime dimension seven and below. This is possible since (as we find in Section \ref{sec:Omega1}), the reduced Spin bordism groups are:
\begin{equation}
\widetilde{\Omega}^{\text{Spin}}_1 (BG^{D}_U) = 0 \,, \quad 3\leq D \leq 7 \,.
\end{equation}
This implies that there are no defects of real codimension-two necessary to break global symmetries induced by non-trivial duality bundles. The reason for this is that the U-duality groups are perfect, which means that their Abelianization $\text{Ab}(G_U^{D})$ is trivial and all elements can be phrased in terms of a product of commutators of the form
\begin{equation}
g = g_1 g_2 g_1^{-1} g_2^{-1} \,,
\end{equation}
for $g \,, g_1 \,, g_2 \in G^{D}_U$. The smooth configuration that bounds the potential generator given by a circle with transition function $g$ is given in analogy to several codimension-two defects in type IIB as discussed in Section 7.1 of \cite{Heckman}, see also the discussion in \cite{McNamara:2021cuo}. The bounding manifold needs to have a non-trivial topology\footnote{For a discussion of topological properties of smooth bounding manifolds utilizing Morse theory see \cite{Ruiz:2024gzv}.} and in the case at hand is a 2-torus with a disk cut out, see Figure~\ref{fig:codim2_bound}.
\begin{figure}
\centering
\includegraphics[width = 0.8 \textwidth]{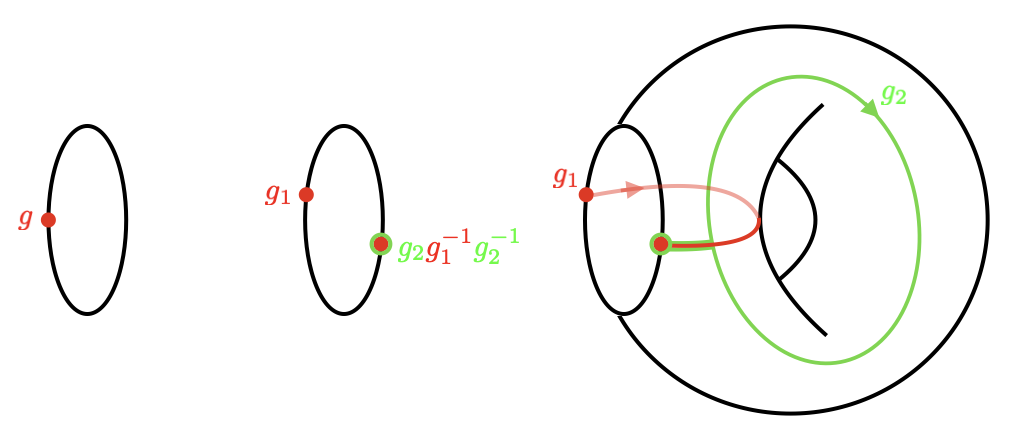}
\caption{Bounding manifold of $S^1$ with transition function given by $g = g_1 g_2 g_1^{-1} g_2^{-1}$.} 
\label{fig:codim2_bound}
\end{figure}

This bounding manifold can in a sense be understood as a `wormhole-like' gravitational soliton \cite{McNamara:2021cuo} which changes the codimension-one topological operator implementing the action of the duality group. These are typically what one would expect from the breaking of a global symmetry in quantum gravity. In particular, the conservation of topological charge violation is already accounted for by allowing for topology changes of the underlying spacetime manifold and does not require extra elementary objects in the theory. The necessity of the non-trivial topology of the gravitational solution further demonstrates the non-perturbative nature of these symmetry-breaking effects (even in gravity).\footnote{A perturbative breaking could, for example, be generated by including terms in the Lagrangian that are suppressed with an inverse power of the Planck mass.} 

Note also that this does not diminish the importance of studying the corresponding codimension-two objects, see, e.g., \cite{deBoer:2010ud, deBoer:2012ma, Lust:2015yia, Bergshoeff:2006jj, Bergshoeff:2011se, Bergshoeff:2013sxa, Achmed-Zade:2018rfc} for a sample of interesting results. In particular, it can very well be that the singular objects might be energetically preferred with respect to the smooth gravitational configurations, for which one has to go beyond the topological properties of the background and instead one has to analyze the dynamics induced by the energy density of the different backgrounds.

Having discussed some general aspects of U-dualities, we now specialize to 8d supergravity with 32 supercharges, where the U-duality group is:
\begin{equation}
G^{8\text{d}}_U \equiv G_U = \text{SL}(2,\mathbb{Z}) \times \text{SL}(3,\mathbb{Z}).
\end{equation}

\section{Duality defects in eight dimensions}
\label{sec:defects}

In this section, we analyze the non-trivial bordism generators in dimension $k \in \{ 1 \,, \dots \,, 7\}$ from a physics viewpoint, and explore what kind of defects we need to introduce in maximal eight-dimensional supergravity in order to describe them as the boundary of a $(k+1)$-dimensional space. The inclusion of the defects as dynamical objects of the theory breaks the associated global symmetry. We then give an interpretation of the necessary configurations within string / F- / M-theory and explore whether we need to include new fundamental objects or allow for more exotic backgrounds. Whenever possible, we explore the duality frame in which the U-duality bundle in $G_U$ has a geometric realization, which simplifies its string theory interpretation. 

\subsection{Codimension-two defects}
\label{subsec:codim2gen}

Besides the generator $S^1_+$ of $\Omega^{\text{Spin}}_1 (\text{pt})$,\footnote{Namely, the circle with periodic boundary conditions for fermions.} which requires the introduction of a codimension-two Spin defect that we will not further explore (see \cite{McNamara:2019rup, Hamada:2025duq}), there are two generators of the reduced bordism group:
\begin{equation}
\widetilde{\Omega}^{\text{Spin}}_1 (B G_U) = \mathbb{Z}_3 \oplus \mathbb{Z}_4 \,.
\end{equation}
They are given by an $S^1$ with the duality bundle specified by the monodromies $\gamma_3$ and $\gamma_4$, i.e., only the SL$(2,\mathbb{Z})$ bundle is non-trivial. This is due to the fact that for $\Omega_1$ only the Abelianization of the gauge group is relevant, and this is trivial for SL$(3,\mathbb{Z})$ since it is a perfect group (as are the higher U-duality groups $G_U^{D}$, see Section \ref{sec:gencodim2}). Note that one can pick any spin structure on the $S^1$ since one can relate the two choices by a disconnected sum with $S^1_+$.

In the type IIB lift this translates into a non-trivial fibration of the compactification torus over the $S^1$. Filling the circle induces a central fiber of the type discussed in Section \ref{subsec:SL2central}, i.e., the fiber torus becomes singular and degenerates to the torus orbifolds $T^2/\mathbb{Z}_3$ and $T^2/\mathbb{Z}_4$ for $\gamma_3$ and $\gamma_4$, respectively. Accordingly, the complex structure of the compactification torus at the singular fiber is fixed to $\tau = e^{2 \pi i/3}$ and $\tau = i$. While this configuration might be reminiscent of the  F-theory geometry with type IV$^*$ and type III$^*$ fiber, (see \cite{Heckman}), here, the torus is part of spacetime and in particular has a finite volume\footnote{In F-theory the fiber volume is unphysical.} The geometry is conveniently described by
\begin{equation}
\big(T^2 \times \mathbb{C} \big) / \mathbb{Z}_k \,, \quad \text{with} \enspace k \in \{3 \,,4 \} \,,
\label{eq:codim2geom}
\end{equation}
which has the chosen generators of the bordism group as boundary (see Figure \ref{fig:codim2}).
\begin{figure}
\centering
\includegraphics[width = 0.35 \textwidth]{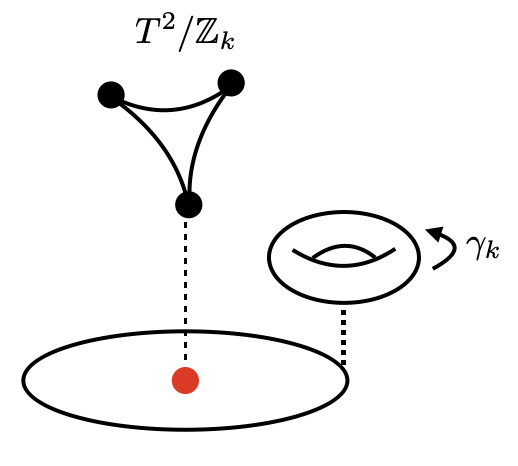}
\caption{Type IIB realization of the U-duality defects in codimension-two ($k \in \{ 3,4\}$).}
\label{fig:codim2}
\end{figure}

Given the singularity structure of the central fiber (see Section \ref{subsec:SL2central}), we find that the geometry \eqref{eq:codim2geom} has several local singularities of the form
\begin{equation}
\mathbb{C}^2 / \mathbb{Z}_n \,,
\end{equation}
where $n$ is inherited from the type of orbifold singularities in the fiber. Since we can describe the total space as a patch of an elliptically-fibered, singular K3, the defects preserve part of the supersymmetry.

\vspace{0.2cm}

The $\mathbb{Z}_3$ case leads to three singularities of the form $\mathbb{C}^2 / \mathbb{Z}_3$, i.e., three A$_2$ singularities. In general an A$_n$ singularity is given by $\mathbb{C}^2 / \mathbb{Z}_{n+1}$ with the action on the local coordinates $(z_1 \,, z_2)$ given by
\begin{equation}
(z_1 \,, z_2) \mapsto (\omega z_1 \,, \omega^{-1} z_2) \,, \quad \omega = e^{2 \pi i \frac{1}{n+1}} \,, \quad \begin{pmatrix} \omega & 0 \\ 0 & \omega^{-1} \end{pmatrix} \in \mathbb{Z}_{n+1} \subset \text{SU}(2) \,.
\end{equation}
Type IIB string theory compactified on such a singular space gives rise to an $\mathcal{N} = (2,0)$ superconformal field theory in six-dimensions \cite{Witten:1995zh, Heckman:2018jxk}. In the present defect they are assembled at the three orbifold fixed points of $T^2/\mathbb{Z}_3$ and lead to a codimension-two defect in the eight-dimensional supergravity.

\vspace{0.2cm}

For $\mathbb{Z}_4$ one finds two singularities of type A$_3$ and one of type A$_1$, giving rise to the associated $\mathcal{N} = (2,0)$ SCFT sectors in the type IIB description. These are located on the orbifold fixed points of $T^2 / \mathbb{Z}_4$. In the 8d supergravity these are located on the same codimension-two defect.

The fact that part of the string theory background is compact, given by $T^2 / \mathbb{Z}_k$, has interesting consequences. In particular it is known that the $\mathcal{N} = (2,0)$ theories possess states charged under 2-form global symmetries corresponding to D3-brane wrapped on relative 2-cycles, see, e.g., \cite{DelZotto:2015isa, Apruzzi:2020zot, Bhardwaj:2020phs, Braun:2021sex}. In our setup the global symmetries of the individual sectors are correlated due to D3-branes stretching between the individual A$_n$ singularities, which should be regarded as excitations on the defect. These states break the 2-form global symmetries at each of the orbifold points to a subgroup, which is $\mathbb{Z}_3$ for the $\gamma_3$ defect and $\mathbb{Z}_2$ for the $\gamma_4$ defect. These remnant symmetries are related to the torsion cycles of the asymptotic geometry $(T^2 \times \mathbb{C}) /\mathbb{Z}_k$, see, e.g., \cite{Cvetic:2021sxm, Cvetic:2022uuu, Baume:2023kkf, Cvetic:2023pgm, Gould:2023wgl}.

We see that the UV lift of the eight-dimensional configuration leads to an interpretation of the symmetry-breaking defect in terms of a well-known string theory background, namely 6d $\mathcal{N} = (2,0)$ SCFT sectors. It further points toward a non-trivial world-volume theory of the defect which, here is related to the string states coming from wrapped D3-branes. This phenomenon of the singular supergravity defects lifting to known string backgrounds seems to be a general feature of the investigation of symmetry-breaking defects which we shall see repeatedly.

\subsection{Codimension-three defects}
\label{subsec:codim3}

The codimension-three defects break the global symmetries associated to 
\begin{equation}
\widetilde{\Omega}^{\text{Spin}}_2 (B G_U) = \mathbb{Z}_2^{\oplus 3} \,,
\end{equation}
where once more we excluded the generator of Spin bordism $S^1_+ \times S^1_+$, which is trivialized by the spin defect in codimension-two, which is unrelated to the existence of a non-trivial duality bundle.

\vspace{0.2cm}

One of the generators can be chosen as a product of $S^1_+$ and the generator in dimension one with $\gamma_4$ monodromy above. Since we already trivialized this, no more defects are necessary for this factor. The resulting string theory background is simply given by the supersymmetry preserving circle compactification of the partially-compactified combination of SCFT sectors discussed in Section \ref{subsec:codim2gen}.

\vspace{0.2cm}

The two remaining generators are given by a 2-torus $T^2$ with non-trivial monodromies around the two 1-cycles. These monodromies sit in two different embeddings of the form (see Section \ref{subsec:stablesplit})
\begin{equation}
\mathbb{Z}_2 \times \mathbb{Z}_2 \rightarrow S_4 \rightarrow \text{SL}(3,\mathbb{Z}) \,.
\label{eq:Membedding}
\end{equation}
We can choose the two $\mathbb{Z}_2$ factors for the different $S_4$ embeddings in the same way, given by
\begin{align}
\begin{split}
M_1^{(i)} =& \big( \Gamma_4^{(i)} \big)^2 R^{(i)}  \big( \Gamma_4^{(i)} \big)^2 R^{(i)} \,, \\
M_2^{(i)} =& \big( \Gamma_4^{(i)} \big)^2 \,,
\end{split}
\label{eq:Monodromies}
\end{align}
with $\Gamma_4^{(i)}$ in \eqref{eq:Gamma4} and
\begin{equation}
R^{(1)} = \begin{pmatrix*}[r] 0 & -1 & 0 \\ -1 & 0 & 0 \\ 0 & 0 & -1 \end{pmatrix*} \,, \quad R^{(2)} = \begin{pmatrix*}[r] -1 & 0 & 0 \\ 0 & -1 & -1 \\ 0 & 0 & 1 \end{pmatrix*} \,.
\end{equation}
Indeed the two monodromy elements correspond to the even permutations of $S_4$, namely
\begin{equation}
M_1^{(i)} = (1\ 2) (3\ 4) \,, \quad M_2^{(i)} = (1\ 3) (2\ 4) \,,
\end{equation}
where we identified the order four element $\Gamma^{(i)}_4$ with the cyclic permutation $(1\ 2\ 3\ 4)$ and the element $R^{(i)}$ with $(1\ 2)$.

\begin{figure}
\centering
\includegraphics[width = 0.25 \textwidth]{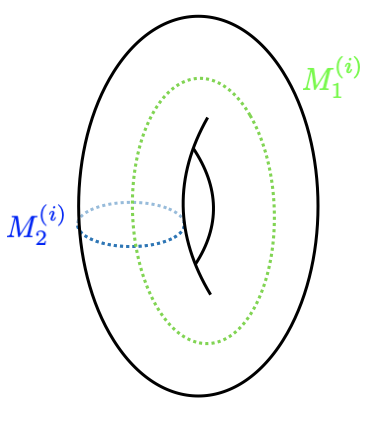}
\caption{Two-dimensional generator given by $T^2$ with monodromies $M_1^{(i)}$ and $M_2^{(i)}$.} 
\label{fig:gendim2}
\end{figure}
Since the topology of the bordism manifold is given by a $T^2$, with two distinct monodromies around the two 1-cycles (see Figure \ref{fig:gendim2}) it is natural to describe it as a boundary by filling one of the circles with the inclusion of a defect. However, this defect is a codimension-two defect rather than codimension-three, and it raises the immediate question why we have not seen them in our discussion of $\Omega^{\text{Spin}}_1 (BG_U)$ (see also Section \ref{subsec:dimdef}). Indeed the monodromies $M^{(i)}$ are part of the commutator subgroup and hence are trivialized by gravitational solitons of the form depicted in Figure \ref{fig:codim2_bound}. So why can we not bound one of the circles in $T^2$ using this smooth configuration? 

The reason the singular defect works but the smooth gravitational soliton described in Figure \ref{fig:codim2_bound} does not is as follows. Filling one of the torus cycles, let us choose the one with transition function $M_1^{(i)}$. One then has an action of $M_2^{(i)}$ implemented when going around the non-trivial 1-cycle of the solid torus. Since $M_1^{(i)}$ and $M_2^{(i)}$ commute, this does not affect the singular codimension-two object implementing the $M_1^{(i)}$ monodromy, see the right-hand side of Figure \ref{fig:boundingdim2}.
\begin{figure}
\centering
\includegraphics[width = 0.6 \textwidth]{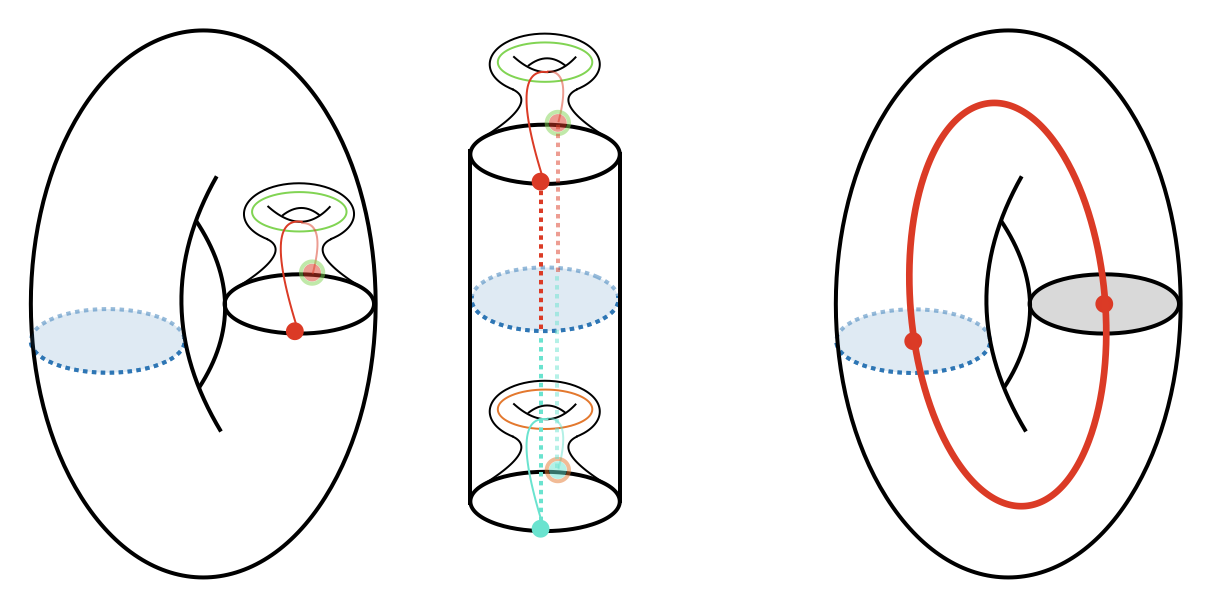}
\caption{Bounding manifold of the generators in dimension-two using a gravitational soliton (left), does not work since the remaining monodromy acts non-trivially on it (middle). Bounding one of the torus cycles with a singular object (right) does work, since the remaining transition function leaves it invariant.} 
\label{fig:boundingdim2}
\end{figure}
For the gravitational soliton, however, one has to express $M_1^{(i)}$ as a (product of) commutators of two SL$(3,\mathbb{Z})$ elements. In order to generate a valid, smooth configuration both these elements have to commute with $M_2^{(i)}$, which we argue cannot be the case.

As a first hint we find that $M_1^{(i)}$ already has the form of an element in the commutator subgroup of $S_4$
\begin{equation}
M_1^{(i)} = g_1 g_2 g_1^{-1} g_2^{-1} \,, \quad \text{with} \enspace g_1 = \big( \Gamma^{(i)}_4 \big)^2 \,, \enspace  g_2 = R^{(i)} \,.
\label{eq:comdecomp}
\end{equation}
While $g_1$ commutes with $M_2^{(i)}$, $g_2$ does not and thus this is not a valid choice for a smooth background of a gravitational soliton. Next, we need to show that this is the case for all possible decompositions of $M_1^{(i)}$ into a product of commutators involving $g_i \in \text{SL}(3,\mathbb{Z})$.

For the particular realizations of monodromies indicated in \eqref{eq:Monodromies} above we find
\begin{equation}
M_1^{(1)} = \begin{pmatrix*}[r] -1 & 0 & 0 \\ 0 & -1 & 0 \\ 0 & 0 & 1 \end{pmatrix*} \,, \quad M_2^{(1)} = \begin{pmatrix*}[r] 1 & 0 & 0 \\ 0 & -1 & 0 \\ 0 & 0 & -1 \end{pmatrix*} \,.
\end{equation}
Commutation with $M_2^{(1)}$, i.e., $g_i \, M_2^{(1)} = M_2^{(1)} \, g_i$, implies that the elements $g_i \in \text{SL}(3,\mathbb{Z})$ used to describe $M_1^{(1)}$ in \eqref{eq:comdecomp} are of the form
\begin{equation}
g_i = \begin{pmatrix} a_i & 0 & 0 \\ 0 & b_i & c_i \\ 0 & d_i & e_i \end{pmatrix} \,.
\label{eq:commutatorform}
\end{equation}
Since $g_i \in \text{SL}(3,\mathbb{Z})$ one has $a_i = \pm 1$, since the determinant is $\text{det}(g_i)= a_i (b_ie_i-d_ic_i) = 1$. The inverse of such a matrix is also of block-diagonal form with the same $a_i$ in the upper left. Consequently, every commutator of such block-diagonal matrices has a $1$ in the upper left
\begin{equation}
    g_1 g_2 g_1^{-1} g_2^{-1} = \begin{pmatrix} a_1^2 a_2^2 & 0 & 0 \\ 0 & * & * \\ 0 & * & * \end{pmatrix} \,,
\end{equation}
and therefore cannot give rise to $M_1^{(1)}$. The same argument can be applied for the product of commutators, all of elements of the form \eqref{eq:commutatorform} in order to commute with $M_2^{(1)}$. For the other embedding of $S_4$ the monodromies read;
\begin{equation}
M_1^{(2)} = \begin{pmatrix*}[r] -1 & 0 & 0 \\ 1 & 1 & 1 \\ 0 & 0 & -1 \end{pmatrix*} \,, \quad M_2^{(2)} = \begin{pmatrix*}[r] 0 & 0 & 1 \\ -1 & -1 & -1 \\ 1 & 0 & 0\end{pmatrix*} \,,
\end{equation}
and commutation with $M_2^{(2)}$ fixes all but four elements in $\tilde{g}_i \in \text{SL}(3,\mathbb{Z})$. This further suggests that the commutator of two such elements does not allow for an expression of $M_1^{(2)}$, since one needs to solve nine equations with eight variables. One can check numerically that this is indeed the case. 

Thus, we see that the singular defect given by additional codimension-two objects cannot be resolved into smooth gravitational configurations, which we discuss below. It would be interesting to explore more generally when such additional defects of lower codimension have to be added for different gauge groups. One path might be to analyze the various characteristic classes of subgroups of the full duality group and their fate under embeddings of the form \eqref{eq:Membedding}.

\vspace{0.2cm}

One might also ask whether there is a genuine codimension-three objects that can trivialize the elements above. For that it would be useful to have a bordism between $T^2$ with the two monodromies and a quotient of $S^2$ which can naturally be identified as enclosing a codimension-three defect (see, e.g., Figure \ref{fig:dif_codim}). This seems hard to achieve since the two non-trivial monodromies require that one takes a free quotient of the sphere with two factors, but the only finite group with free action on $S^2$ is $\mathbb{Z}_2$, via the anti-podal map, which is not sufficient.

Thus, surprisingly our investigation of non-trivial generators of $\Omega^{\text{Spin}}_2(BG_U)$ suggests the existence new objects in codimension-two. Since the monodromy matrices are elements of SL$(3,\mathbb{Z})$ they have a geometric interpretation in the M-theory description. 

For monodromy $M_1^{(1)}$ the M-theory geometry is described by the local elliptically-fibered K3 times a circle, preserving half of supersymmetry,
\begin{equation}
(T_F^2 \times \mathbb{C}) / (\mathbb{Z}_2)_{M_1^{(1)}} \times S^1 \,,
\end{equation}
with the central fiber giving $T^2/\mathbb{Z}_2$ with four local $\mathbb{C}^2/\mathbb{Z}_2$, i.e., A$_1$ singularities. M-theory on such a background gives rise to seven-dimensional super-Yang-Mills theory with gauge algebra $\mathfrak{su}(2)$, which suggests that the defect hosts a six-dimensional gauge theory with gauge algebra $\mathfrak{su}(2)^{\oplus 4}$. Since all of them appear on the same central fiber the generalized symmetries are modified by the presence of M2-branes states stretching from one to another orbifold singularity (see the discussion in \cite{Cvetic:2021sxm, Cvetic:2022uuu, Baume:2023kkf, Gould:2023wgl}, and in particular \cite{Cvetic:2023pgm}, making use of the Mayer-Vietoris sequence). The resulting defect is compactified on the second $S^1$ with monodromy $M_2^{(1)}$.

Alternatively, one can fill the circle with $M_2^{(1)}$ monodromy which also has an F-theory interpretation, where the fiber volume is set to zero. The corresponding defect is the torus compactification of an $\mathfrak{so}(8)$ brane stack (consisting of four D7-branes on top of an O7$^-$ orientifold plane). This type IIB object is then further compactified on a circle with $M_1^{(1)}$ monodromy, which involves the third torus direction and therefore will act as a non-geometric twist. This is potentially dangerous since the SL$(2,\mathbb{Z})_S$ subgroup of SL$(3,\mathbb{Z})$ acts on the brane charges.  For this to be well-defined it is necessary that the monodromies $M_1^{(1)}$ and $M_2^{(1)}$ commute which is indeed the case. This F-theory configuration can be understood as the limit of vanishing fiber volume of the M-theory background above, see Section \ref{subsec:MF}. In this limit the shrinking of the central fiber leads to a gauge enhancement to $\mathfrak{so}(8)$.

For monodromy $M_2^{(2)}$ (an equivalent discussion holds for $M_1^{(2)}$) all three directions are involved and we describe the topology of the singular fiber in Appendix \ref{app:cryssingfiber}. This shows that the central fiber in this case is given by 
\begin{equation}
T^3/(\mathbb{Z}_2)_{M_1^{(2)}} \simeq \big( T^2 / \mathbb{Z}_2 \times S^1 \big) / \mathbb{Z}_2^s \,,
\end{equation} 
where the additional quotient by $\mathbb{Z}_2^s$ encodes the a fibration structure of $T^2/\mathbb{Z}_2$ that, going around the base circle, exchanges two pairs of fixed points. This means that the resulting worldvolume theory in the M-theory lift is given by a twisted circle compactification of the $\mathfrak{su}(2)^{\oplus 4}$ theory, where the twist relates two pairs of $\mathfrak{su}(2)$ sectors, leaving a linear combination of two $\mathfrak{su}(2)$ gauge fields massless. Again, this configuration can be smoothed by resolving the orbifold singularities, but now the volume of some of the resolution curves are related via the $\mathbb{Z}_2^s$ action. This is another instance of the effect discussed in Section \ref{subsec:deform}, namely that a full resolution is possible in M-theory but not in the lower-dimensional supergravity. Indeed in a discussion of defects in M-theory, the backgrounds above do not correspond to non-trivial bordism generators after blowing up the fiber singularities.

\subsection{Codimension-four defects}
\label{subsec:codim4}

The relevant bordism group for codimension-four defects is given by
\begin{equation}
\Omega^{\text{Spin}}_3 (BG_U) = \widetilde{\Omega}^{\text{Spin}}_3 (BG_U) = \mathbb{Z}_3^{\oplus 3} \oplus \mathbb{Z}_2^{\oplus 3} \oplus \mathbb{Z}_8^{\oplus 3} \,.
\end{equation}
The $\mathbb{Z}_3$ and $\mathbb{Z}_8$ summands are generated by lens spaces, which we can interpret as the asymptotic boundary of $\mathbb{C}^2 / \mathbb{Z}_k$, and carry a non-trivial duality bundle, fixed by the monodromy around the torsion 1-cycle of the lens space.

\vspace{0.2cm}

For the $\mathbb{Z}_3$ summands the three different monodromies on $L^3_3 = \partial \big(\mathbb{C}^2 / \mathbb{Z}_3 \big)$ are given by $\gamma_3$, $\Gamma_3^{(1)}$, and $\Gamma_3^{(2)}$. All of these backgrounds admit a geometrical interpretation within string theory or M-theory and so do the necessary symmetry-breaking defects.

For $\gamma_3$ we can use the type IIB description and find that the defect is described by the singular background
\begin{equation}
\big( T^2 \times \mathbb{C}^2 \big) / (\mathbb{Z}_3)_{\gamma_3} \,,
\end{equation}
which is of the form of a local Calabi-Yau 3-fold, and hence preserves one quarter of the supersymmetry. At the location of the defect the $T^2$ degenerates and we find three singularities of the form $\mathbb{C}^3 / \mathbb{Z}_3$. This space allows for a crepant resolution by blowing up a $\mathbb{CP}^2$ at the singular points, see, e.g., \cite{Witten:1996qb,lho2019crepant}. The defect can therefore be understood as a partial compactification of three such sectors in type IIB string theory.\footnote{It might be tempting to try to describe this defect as a compactification of the codimension-two object related to $S^1$ with $\gamma_3$ monodromy. For that one would resolve the singular $\mathbb{C}^2/\mathbb{Z}_3$ to the total space of $\mathcal{O}(-3)$, a complex line bundle over $\mathbb{CP}^1$, and wrap the codimension-two object around the $\mathbb{CP}^1$. However, this curve has self-intersection $(-3)$ and therefore, at least in the presence of supersymmetry, demands a non-trivial axio-dilaton profile dictated by the discussion of non-Higgsable clusters in \cite{Morrison:2012np}. This would change the duality bundle and the generator, and we are therefore guided towards this new genuinely codimension-four object.}

For $\Gamma_3^{(1)}$ the interpretation is also geometrical in non-perturbative type IIB, i.e., F-theory, since the SL$(3,\mathbb{Z})$ monodromy only acts on a 2-torus which can be identified with the F-theory fiber $T^2_F$. This leads to a F-theory background given by
\begin{equation}
\big( T^2_F \times \mathbb{C}^2 \big) / (\mathbb{Z}_3)_{\Gamma_3^{(1)}} \times T^2 \,.
\end{equation}
The first factor now precisely has the geometry of a non-Higgsable cluster,  \cite{Morrison:2012np}. It can be resolved into a curve with self-intersection $(-3)$ over which the fiber develops a type IV singularity, indicating the presence of 7-branes. This is also one of the defects found in the pure type IIB discussion in \cite{Heckman}. Since the additional spacetime $T^2$ is unaffected, the defect is given by the torus compactification of the $(-3)$ non-Higgsable cluster. This also confirms the expectation that at least some of the bordism defects of the parent theory (here: type IIB in ten dimensions) persist after compactification. The resulting 4d theory has a strongly coupled point in the moduli space where a 4d SCFT emerges. For the isolated $(-3)$ non-Higgsable cluster this theory is given by an Argyres-Douglas theory of the form $D_3\big( \text{SU}(2) \big)$, see \cite{DelZotto:2015rca, Carta:2023bqn}.\footnote{We thank C. Lawrie for explaining which 4d SCFT emerges in this case.} Here one might be tempted to use the topological soliton in Section~\ref{sec:gencodim2} to fill the $S^1$ fiber of the lens space description as circle bundle over $\mathbb{CP}^1 \simeq S^2$. However, transition functions of the circle bundle are incompatible with the duality bundle on the smooth gravitational soliton, similar to the discussion in Section \ref{subsec:codim3}. 

The most interesting defect for the last $\mathbb{Z}_3$ summand is given by $L^3_3$ with monodromy  $\Gamma^{(2)}_3$, which requires a discussion within M-theory for a geometric interpretation. Once more one finds that the defect contains singularities of the form $\mathbb{C}^3/\mathbb{Z}_3$, which combine the 8d geometry of the supergravity theory with two of the three fiber coordinates in the M-theory lift. As mentioned above these can be resolved by blowing up a $\mathbb{CP}^2$ at the singular point. This signals the appearance of a non-trivial interacting SCFT in M-theory.\footnote{This is the celebrated $E_0$ 5d SCFT of \cite{Seiberg:1996bd} (see also \cite{Morrison:1996xf, Douglas:1996xp, Intriligator:1997pq}. In the blownup phase, M5-branes wrapped on the $\mathbb{CP}^2$ produce effective strings. Near the singular point, the tension of these effective strings drops to zero, signalling the presence of a conformal fixed point.} As discussed in Section \ref{sec:monodrom}, there locally are three such singularities that get exchanged when traversing the invariant fiber circle of the M-theory $T^3$. We see that the duality defect, with four-dimensional worldvolume, in the M-theory description has interesting dynamics which consists of a twisted circle compactification of three 5d SCFT sectors. From the central fiber geometry in \eqref{eq:twistedcircleZ3}, we find the defect geometry
\begin{equation}
(T^3 \times \mathbb{C}^2) / (\mathbb{Z}_3)_{\Gamma_3^{(2)}} \simeq \big( (T^2 \times \mathbb{C}^2) / \mathbb{Z}_3 \times S^1\big) / \mathbb{Z}_3^s \,,
\end{equation}
where the action of $\mathbb{Z}_3^s$ exchanges the three local SCFTs sectors when going around the $S^1$. In the resolved geometry, this twist exchanges the three $\mathbb{CP}^2$ resolutions of the blow-up, acting accordingly on the string states, which only leaves a linear combination invariant.

\vspace{0.2cm}

The generators of the $\mathbb{Z}_8$ summands are very similar to the ones above with $\gamma_3$ replaced by $\gamma_4$ and $\Gamma_3^{(i)}$ by $\Gamma_4^{(i)}$, respectively.

The first $\mathbb{Z}_8$ is given by type IIB on a spacetime of the form
\begin{equation}
(T^2 \times \mathbb{C}^2) / (\mathbb{Z}_4)_{\gamma_4} \,.
\label{eq:nonSUSYZ4}
\end{equation}
As opposed to the discussion of the $\mathbb{Z}_3$ quotients, however, there is no supersymmetric way of orbifolding when the full $\mathbb{Z}_4$ acts on the torus fiber. This can be seen by considering the transformation of the local $(3,0)$-form composed out of the $\mathbb{C}^2$ coordinates $(z_1,z_2)$ and the $T^2$ coordinate $y$
\begin{equation}
\Omega = dy \wedge dz_1 \wedge dz_2 \rightarrow e^{\frac{2 \pi i}{4} (\pm 1 \pm 1 \pm 1)} dy \wedge dz_1 \wedge dz_2 \neq \Omega \,.
\end{equation}
The $\pm 1$ in the exponent corresponds to the allowed choices for the $\mathbb{Z}_4$ action on the individual coordinates. Thus, it seems that the associated defect breaks supersymmetry. However, we believe this to be an artifact of considering the duality group SL$(2,\mathbb{Z})$ instead of its Spin lift Mp$(2,\mathbb{Z})$. In fact for the metaplectic group the $\mathbb{Z}_4$ element has a geometrical action of
\begin{equation}
y \rightarrow - y \,,
\end{equation}
on the torus direction and there is a perfectly supersymmetric orbifold of the type \eqref{eq:nonSUSYZ4}. This can also be seen in the results of \cite{Heckman}, where the defects were supersymmetric, see also the discussion in Section \ref{subsec:liftsusy}. Again, the same caveats of trying to relate the above backgrounds to a compactification of the $S^1_{\gamma_4}$ defect apply.

Similarly, the second $\mathbb{Z}_8$ summand is given by
\begin{equation}
\big( T^2_F \times \mathbb{C}^2 \big) / (\mathbb{Z}_4)_{\Gamma_4^{(1)}} \times T^2 \,, 
\end{equation}
with the same issues concerning supersymmetry as above. Since, comparing to the analysis for $\mathbb{Z}_3$, we expect this defect to be the $T^2$ compactification of a $(-4)$ non-Higgsable cluster theory this suggests again that the non-supersymmetric nature seems indeed to be an artifact of working with the bosonic duality group. For the supersymmetric version, where $\mathbb{Z}_4$ acts on $T^2_F$ by multiplication with $-1$, there is again a 4d $\mathcal{N} = 2$ SCFT at certain points of the moduli space, given by SU$(2)$ SQCD with four flavors \cite{DelZotto:2015rca, Carta:2023bqn}.

The last $\mathbb{Z}_8$ summand needs an M-theoretic description to be geometric. With the description of the central fiber as given in Appendix \ref{app:cryssingfiber}, we find a geometry of the type
\begin{equation}
\big( (T^2 \times \mathbb{C}^2) / \mathbb{Z}_4 \times S^1 \big) / \mathbb{Z}_2^s \,,
\end{equation}
where the $\mathbb{Z}_2^s$ action exchanges the two $\mathbb{C}^3 / \mathbb{Z}_4$ orbifold points. Again, we find that the quotient is not compatible with supersymmetry, which should be resolved once one takes into account the spin-cover of the duality group. Due to the singularities the associated defect likely carries some SCFT sectors, some of which are exchanged under the $\mathbb{Z}_2^s$ action.

\vspace{0.2cm}

This leaves three $\mathbb{Z}_2$ summands, given by the generators
\begin{equation}
S^1_+ \times S^1_+ \times S^1_{\gamma_4} \,, \quad S^1_{\gamma_4} \times S^1_{M_1^{(1)}} \times S^1_{M_2^{(1)}} \,, \quad S^1_{\gamma_4} \times  S^1_{M_1^{(2)}} \times S^1_{M_2^{(2)}} \,.
\end{equation}
Each of them contains a direct factor of $S^1_{\gamma_4}$ and therefore the associated global symmetry is broken by compactifying the associated codimension-two defect on $S^1_{M_1^{(i)}} \times S^1_{M_2^{(i)}}$. Since the monodromies $M^{(i)}$ involve all coordinates of the $T^3$ these compactifications will in general be non-geometric twists (see Section \ref{subsec:nongeom}), where the moduli fields of the string background will undergo non-trivial monodromies as well. It is also reassuring that the defect we compactify does not contain any type IIB 7-branes which might not be invariant under the monodromies $M^{(i)}$ and therefore might not be allowed in the compactifications. Thus, there are no actually new defects needed for this subset of generators.

\vspace{0.2cm}

We see that in codimension four we obtain genuinely new defects that are associated to singular type IIB geometries but we also find situations in which the lower-codimension defects found above are compactified. Once more this picture will be very common in the higher-codimension cases discussed below. It is also the first time that we are confronted with non-geometry. However, the defects themselves are not non-geometric. Rather, they are compactified on backgrounds with a generally non-geometric action on the type IIB moduli fields.

\subsection{Codimension-five defects}
\label{subsec:codim5}

The $\mathbb{Z}$ summand in the $d = 4$ bordism group is associated to the Spin manifold K3, and we will not discuss it further. Instead we will focus on the reduced bordism group given by
\begin{equation}
\widetilde{\Omega}^{\text{Spin}}_4 (BG_U) = \mathbb{Z}_3^{\oplus 2} \oplus \mathbb{Z}_2^{\oplus 3} \oplus \mathbb{Z}_4^{\oplus 2} \,,
\end{equation}
which relies on the presence of non-trivial duality bundles.

\vspace{0.2cm}

The two $\mathbb{Z}_3$ summands are given by
\begin{equation}
S^1_{\gamma_3} \times L^3_{3,\Gamma_3^{(i)}} \,,
\end{equation}
where the $L^3_3$ carries an SL$(3, \mathbb{Z})$ duality bundle specified by monodromy $\Gamma_3^{(i)}$. We see that these are the direct product of two generators that were already known. Therefore, we can bound it either by compactifying the codimension-two defects on $L^3_3$ or the codimension-four defects on $S^1_{\gamma_3}$. 

As before we can ask whether this further compactification is allowed. Note that the defect associated to $\gamma_3$ is a geometric background in type IIB without the presence of branes. Due to this we expect that the compactification on the three-dimensional lens spaces are possible, with $\Gamma_3^{(1)}$ allowing for an F-theory description with varying axio-dilaton and $\Gamma_3^{(2)}$ leading to a compactification with non-geometric twist, respectively. 

Alternatively, one can compactify the codimension-four defects on $S^1_{\gamma_3}$. For $\Gamma_3^{(1)}$ this is the compactification of the non-Higgsable cluster on a 3-manifold given by the torus fibration over $S^1$, which presumably can be related to twisted compactification of the $D_3\big( \text{SU}(2) \big)$ theory discussed above (at least at scales where the volume of the $T^2$ becomes small). For $\Gamma_3^{(2)}$ it is a circle compactification of the theory discussed above, i.e., a twisted circle reduction of the three copies of 5d SCFTs associated to M-theory on $\mathbb{C}^3/\mathbb{Z}_3$, with a non-geometric twist.\footnote{These twists associated to the $S^1$ compactification are reminiscent of automorphism twists which are often employed to obtain modified SCFTs in lower dimensions, see, e.g., \cite{Tachikawa:2010vg, Zafrir:2016wkk, Bhardwaj:2019fzv}.} It would be interesting to study the effect of this twist on the worldvolume fields at low energy. 

Since the compactifications of the higher-dimensional defects does not seem to lead to inconsistencies, we do not need to include additional objects for these bordism generators.

\vspace{0.2cm}

An analogous story holds for the two $\mathbb{Z}_4$ summands being generated by
\begin{equation}
S^1_{\gamma_4} \times L^3_{4,\Gamma_4^{(i)}} \,,
\end{equation}
with SL$(3,\mathbb{Z})$ bundle on $L^3_4$ determined by $\Gamma_4^{(i)}$. Again we have the choice of compactifying the codimension-two defects on $L^3_4$ or the codimension-four defects on $S^1_{\gamma_4}$. Of course the same caveats concerning supersymmetry as discussed in Section \ref{subsec:codim4} and Section \ref{subsec:liftsusy} apply.

\vspace{0.2cm}

Next, one of the $\mathbb{Z}_2$ summands is generated by $W_4$ which can be described as
\begin{equation}
    W_4 = ( L^3_{4,\Gamma_4^{(1)}} \times S^1)/\mathbb{Z}_2 \,,
\end{equation}
with $\mathbb{Z}_2$ acting as complex conjugation on the $\mathbb{C}^2$ used to describe the $L^3_4$ and the anti-podal map on $S^1$. It therefore can be understood as the lens space fibered over $S^1$. Following the discussion in Section \ref{subsec:twist} one proceeds by filling in the fiber by the inclusion of a codimension-four defect described in \ref{subsec:codim3}, see also the discussion in \cite{Heckman}. The associated defect is naturally understood as a twisted circle compactification of the higher-dimensional defect. Since the complex conjugation flips the sign of two of the normal coordinates of the codimension-four defect it can be understood as a topological twist involving a particular U$(1)$ subgroup of the full R-symmetry group. This twist will likely project out some of the massless localized degrees of freedom, but a detailed discussion goes beyond the scope of our investigation.

\vspace{0.2cm}

The remaining two $\mathbb{Z}_2$ are generated by disconnected sums of $S^1_{\gamma_4} \times L^3_4$ described above and a space, called $A$ in Section \ref{sec:Z4S4}. This $A$ is given by the geometric configuration
\begin{equation}
    S^1 \times \mathbb{RP}^3 \,.
\end{equation}
The duality bundle is specified by the monodromy elements when traversing the two non-trivial 1-cycles, i.e., the $S^1$ and the torsion 1-cycle of $\mathbb{RP}^3$, respectively, 
\begin{equation}
\begin{split}
    \text{monodromy around } S^1:& \quad M_2^{(i)} \in \text{SL}(3,\mathbb{Z}) \,, \\
    \text{monodromy around } \mathbb{RP}^3:& \quad \bigg( \begin{pmatrix*}[r] -1 & 0 \\ 0 & -1 \end{pmatrix*}, M_1^{(i)}\bigg) \in \text{SL}(2,\mathbb{Z}) \times \text{SL}(3,\mathbb{Z}) \,.
\end{split}
\label{eq:genA}
\end{equation}
We will denote the space with $i=2$ by $A'$. Again, we can use the product manifold structure leading to two different interpretations. We can compactify the codimension-two defect that  bounds the $S^1$ factor, which was already introduced in Section \ref{subsec:codim3}, on the $\mathbb{RP}^3$ factor with its non-trivial duality bundle specified by the monodromy indicated in \eqref{eq:genA}. Alternatively, we can define a new codimension-four object bounding $\mathbb{RP}^3$ with its bundle, which did not yet appear from lower-dimensional bordism groups and compactify it on the $S^1$ factor. Since we already argued for the inclusion of a codimension-two defect, the first possibility is more natural and we do not need to include further defects.

\subsection{Codimension-six defects}
\label{subsec:codim6}

The relevant bordism group for codimension-six defect is given by
\begin{equation}
    \widetilde{\Omega}^{\text{Spin}}_5 (BG_U) = \mathbb{Z}_3^{\oplus 2} \oplus \mathbb{Z}_9 \oplus \mathbb{Z}_2^{\oplus2} \oplus \mathbb{Z}_2^{\oplus 2} \oplus \mathbb{Z}_2 \oplus \mathbb{Z}_4 \,,
\end{equation}
where we already paired summands that are related by the different embedding choices.

\vspace{0.2cm}

The $\mathbb{Z}_9$ summand has a geometric interpretation in type IIB string theory and is generated by $L^5_3$ with $\gamma_3$ monodromy. This can be understood as the asymptotic boundary of 
\begin{equation}
    ( T^2 \times \mathbb{C}^3)/ (\mathbb{Z}_3)_{\gamma_3} \,,
\end{equation}
Moreover, by using the $\mathbb{Z}_3$ action of the form $(\omega, \omega^{-1},\omega,\omega^{-1})$ with $\omega = e^{2 \pi i/3}$, we can interpret this as a local patch of a singular elliptically-fibered Calabi-Yau 4-fold,  containing three singularities of the type $\mathbb{C}^4/\mathbb{Z}_3$ at the singular point. The singularities with $\mathbb{Z}_3$ action defined above are Gorenstein, canonical, and terminal, so allow no crepant resolution \cite{a6207765-8e51-37bb-af3c-43a72b254a77}. This lead to a string-like singular object in eight-dimensional supergravity, with dynamics specified by type IIB on $\mathbb{C}^4/\mathbb{Z}_3$. Since it consists of three such singularities one can think of it as a composite object.

\vspace{0.2cm}

The first $\mathbb{Z}_3$ summand is given by the lens space $L^5_3$ with monodromy $\big(\gamma_3,\Gamma_3^{(1)}\big)$ acting as the tensor product $\gamma_3 \otimes \Gamma_3^{(1)}$ and therefore requires a lift to F-theory. It can be interpreted as
\begin{equation}
    (T^2_F \times T^2 \times \mathbb{C}^3) / (\mathbb{Z}_3)_{\gamma_3 \otimes \Gamma_3^{(1)}} \,,
\end{equation}
which again seems to have a supersymmetric realization by choosing the $\mathbb{Z}_3$ action appropriately, e.g., 
\begin{equation}
    \mathbb{Z}_3: \quad (\omega,\omega, \omega,\omega^{-1}, \omega) \,,
\end{equation}
with $\omega$ as above. These local $\mathbb{C}^5 / \mathbb{Z}_3$ singularities cannot be resolved while preserving supersymmetry. These objects can be understood as being composite and consisting of three copies of the S-strings discussed in \cite{Heckman}. The fact that there are three of them, due to the three singularities of the spacetime torus, also gives an intuition why they only generate a $\mathbb{Z}_3$ summand and not a full $\mathbb{Z}_9$ as in \cite{Heckman}.

The second $\mathbb{Z}_3$ summand is the first time that the tensor product of two monodromies appears on the same 1-cycle which cannot be interpreted geometrically in any string or M-theory frame. Thus, this is the first example of a genuinely non-geometric defect, i.e., a defect which cannot be described by the geometry of an internal space of the UV description alone, but involves U-duality twists that involve other moduli. They can be understood as a non-perturbative generalization of T- or U-duality defect (see \cite{deBoer:2010ud, deBoer:2012ma, Lust:2015yia}) in higher codimension. 

Let us analyze the geometry a little bit more in the M-theory frame, where at least the $T^3$ factor has a geometric interpretation. Focusing on the shape moduli of the $T^3$ we find the M-theory background
\begin{equation}
    (T^3 \times \mathbb{C}^3)/(\mathbb{Z}_3)_{\Gamma_3^{(2)}} \simeq \big((T^2 \times \mathbb{C}^3)/\mathbb{Z}_3 \times S^1 \big)/ \mathbb{Z}_3^s \,,
\end{equation}
where the singular orbifold points undergo monodromy when going around the $S^1$ as described in Appendix  \ref{app:cryssingfiber}. However, at the same time one has the $\gamma_3$ monodromy which is non-geometric and instead acts, for example, on the volume of the $T^3$ and the vacuum expectation value of $C_3$ integrated over the torus, see Section \ref{subsec:UM}. This implies that the torus fiber is glued to itself potentially at a different volume and value of $C_3$ (see Figure \ref{fig:Ufect}).
\begin{figure}
\centering
\includegraphics[width = 0.35 \textwidth]{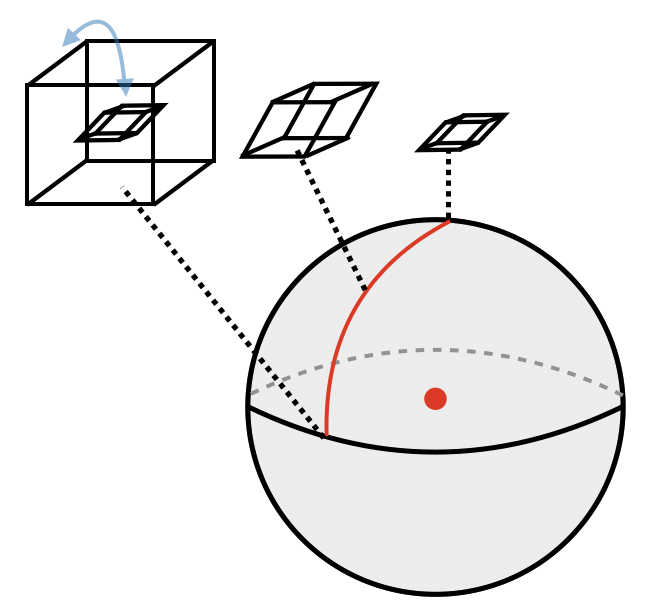}
\caption{Schematic depiction of a non-geometric U-duality defect in M-theory as the central singularity of a generator, where not only the shape of the internal 3-torus varies, but also its volume (and 3-form).} 
\label{fig:Ufect}
\end{figure}
In fact, since we know that $\gamma_3$ acts as Moebius transformations on $\omega_{\text{M}}$ in \eqref{eq:MKahler} we find that
\begin{equation}
    \omega_{\text{M}} \mapsto - \frac{\omega_{\text{M}} +1}{\omega_{\text{M}}} \,,
\end{equation}
which has a fixed point at $\omega_{\text{M}} = e^{2 \pi i/3}$. This suggests that the volume of the central singularity is necessarily of Planck size and the integral of $C_3$ is non-zero. The fact that the geometry involves cycles of the size of the Planck length suggests that quantum gravity corrections become important and modify the classical background described above. Beyond the action on the geometry, there will be further monodromies on the other fields in the eight-dimensional supergravity. The full duality action should also be captured by the singular exceptional geometry fiber, discussed in Section \ref{subsec:nongeomsingfiber}. We will leave a more detailed discussion of these non-geometric defects, their properties, and worldvolume theory for future work. It would also be interesting to study more general flux backgrounds on these non-geometrical string objects and explore what other strings one can find in their duality orbit (see also \cite{Ferrara:1997uz, Ferrara:1997ci, Lu:1997bg, Kleinschmidt:2011vu, Bergshoeff:2012ex}).

\vspace{0.2cm}

Two of the $\mathbb{Z}_2$ summands are generated by $S^1_{\gamma_4} \times A$ and $S^1_{\gamma_4} \times A'$, with $A$ and $A'$, being $S^1 \times \mathbb{RP}^3$ with duality bundle inside the two different embeddings of $\mathbb{Z}_4 \times S_4$ in $G_U$, see Section \ref{subsec:codim5} above. Using the product structure, we can simply compactify the codimension-two defects, associated to $S^1_{\gamma_4}$, on the space $A$. Since this carries duality bundles in both factors of $G_U$ it is a compactification with non-geometric twist, but it does not require the introduction of further elementary objects.

\vspace{0.2cm}

We can proceed along the same lines for another $\mathbb{Z}_2$ summand generated by $S^1_{\gamma_4} \times W_4$, whose defect is give by the codimension-two defect wrapped on $W_4$. With $W_4$ itself a fibration, this leads to a double twisted configuration with both topological and non-geometric duality twists. 

\vspace{0.2cm}

The $\mathbb{Z}_4$ summand is given by a five-dimensional Spin manifold $Q^5_4$, which is described as a $L^3_4$ lens space bundle fibered over $\mathbb{CP}^1$. As discussed in \cite{Heckman}, this can be understood as a compactification of the defect associated to the fiber $L^3_4$ on $\mathbb{CP}^1$ with a topological twist associated to the non-trivial fibration. The duality bundle on $L^3_4$ is provided by $\gamma_4$ and so one has the twisted compactification of the codimension-four defect. Since the codimension-four defect might not lift to the correct defect after the Spin lift of the duality group, the same is true for this configuration.\footnote{Indeed, from the discussion in \cite{Heckman}, one might expect that after the Spin lift the generator is given by $L^5_4$ which itself is not Spin, but it lifts to a Spin manifold when including the type IIB compactification torus, i.e., it is Spin-Mp$(2,\mathbb{Z})$. If this is the case the defect would be type IIB on a local singular Calabi-Yau 4-fold of the form $(T^2 \times \mathbb{C}^3)/\mathbb{Z}_4$ with constant axio-dilaton.}

\vspace{0.2cm}

The remaining two $\mathbb{Z}_2$ summands have the same spacetime geometry $Q^5_4$ given by the lens space $L^3_4$ fibered over $\mathbb{CP}^1$ as above, but the duality bundle over the $L^3_4$ fiber differs and is given by two different embeddings of $\mathbb{Z}_4$ into $G_U$.

The first embedding corresponds to the tensor product $\gamma_4 \otimes \Gamma_4^{(1)}$; it represents an F-theory geometry given by 
\begin{equation}
    (T^2_F \times T^2 \times \mathbb{C}^2)/ (\mathbb{Z}_4)_{\gamma_4 \otimes \Gamma_4^{(1)}} \rightarrow \mathbb{CP}^1 \,,
\end{equation}
non-trivially fibered over $\mathbb{CP}^1$. The symmetry-breaking defect is the defect of the fiber, given by a codimension-four object in the eight-dimensional spacetime wrapped on $\mathbb{CP}^1$ with a topological twist. Typically we would expect that such a defect does not survive the Spin lift to the duality group, because it does not lead to a well-defined spacetime structure for the type IIB lift. Here, however, the F-theory duality bundle is non-trivial as well, which can allow type IIB to make sense on non-spin manifolds. Indeed there is a local $\mathbb{Z}_4$ action on $\mathbb{C}^4$, given by $(\omega, \omega, \omega, \omega)$ with $\omega= e^{2\pi i/4}$ which has the potential to survive the lift, leading to a supersymmetric background described by a local singular Calabi-Yau. The corresponding defects then have an interpretation as composite objects built out of the codimension-six objects, S-folds, of \cite{Heckman}, compactified on $\mathbb{CP}^1$ with twist. The fact that a multiplicity of these objects is needed in the U-duality defect also explains why the associated bordism group summands are reduced. This is another situation in which the codimension of the defect does not match the naive expectation from the dimension of the bordism group generator.

The second embedding is given by the tensor product $\gamma_4 \otimes \Gamma_4^{(2)}$, which is non-geometric. Pending the lift to the fermionic bordism group, this suggests that one should include non-geometric defects already in codimension four, with asymptotic geometry $L_4^3$ and monodromy above. The defect for $Q^5_4$ then corresponds to their twisted compactification on $\mathbb{CP}^1$.

Let us briefly discuss the M-theory defect associated to the filling of the fiber, whose geometry is given by
\begin{equation}
    (T^3 \times \mathbb{C}^2)/(\mathbb{Z}_4)_{\gamma_4 \otimes \Gamma_4^{(2)}} \,.
\end{equation}
Since the action of $\Gamma_4^{(2)}$ is reflected by the variation of geometric moduli one obtains
\begin{equation}
    (T^3 \times \mathbb{C}^2)/(\mathbb{Z}_4)_{\Gamma_4^{(2)}} \simeq \big( (T^2 \times \mathbb{C}^2)/\mathbb{Z}_4 \times S^1\big)/\mathbb{Z}_2^s \,,
\end{equation}
combined with the non-geometric monodromy imposed by $\gamma_4$. As for the non-geometric defect for $\mathbb{Z}_3$ above, this $\gamma_4$ acts on $\omega_{\text{M}}$ via Moebius transformations
\begin{equation}
    \omega_{\text{M}} \mapsto - \frac{1}{\omega_{\text{M}}} \,,
\end{equation}
underlining the non-geometric nature of the background. At the defect itself this further suggests that the $\omega_{\text{M}}$ needs to fixed at the symmetric value $i$, and therefore the 3-form is switched off and the volume is fixed to the Planck scale, suggesting the importance of quantum corrections. We hope to come back to a more detailed study of these non-geometric objects including the correct spin lift of the duality group.

\subsection{Codimension-seven defects}
\label{subsec:codim7}

The defects in codimension-seven associated to 
\begin{equation}
    \widetilde{\Omega}^{\text{Spin}}_6 (BG_U) = \mathbb{Z}_3^{\oplus 2} \oplus \mathbb{Z}_2^{\oplus 2} \oplus \mathbb{Z}_2 \oplus \mathbb{Z}_4^{\oplus 2} \,,
\end{equation}
with one exception, given by the single $\mathbb{Z}_2$ summand, are surprisingly simple. The reason is that all of them are described by product manifolds.

\vspace{0.2cm}

The two $\mathbb{Z}_3$ summands are given by $L^3_3 \times L^3_3$ with $\gamma_3$ monodromy on the first and $\Gamma_3^{(i)}$ monodromy on the second factor. This means that it is simply one of the corresponding codimension-four defect described in Section \ref{subsec:codim4} compactified on the other lens space. No additional defects are necessary.

\vspace{0.2cm}

Very similarly, the two $\mathbb{Z}_4$ summands are given by $L^3_4 \times L^3_4$ with $\gamma_4$ monodromy on the first and $\Gamma_4^{(i)}$ monodromy on the second factor. We can bound these spaces by the compactification of the codimension-four defects on the second lens space factor.

\vspace{0.2cm}

Two of the $\mathbb{Z}_2$ summands are generated by $\mathbb{RP}^3 \times \mathbb{RP}^3$ with monodromy $M_1^{(i)}$ and $M_2^{(i)}$ on the individual factors, respectively. Here, one has two possibilities to introduce symmetry breaking defects. Either one introduces new codimension-four objects, described by the singular geometry M-theory geometry
\begin{equation}
(T^3 \times \mathbb{C}^2 )/(\mathbb{Z}_2)_{M^{(i)}_k} \,,
\end{equation}
with the $\mathbb{Z}_2$ acting as multiplication by $(-1)$ on $\mathbb{C}^2$, which bound one of the $\mathbb{RP}^3$ factors. These would then be compactified (including the non-geometric twist) on the other $\mathbb{RP}^3$ factor. Alternatively, one can resolve the $\mathbb{C}^2 / \mathbb{Z}_2$ to the total space of the complex line bundle $\mathcal{O}(-2)$ over $\mathbb{CP}^1$, in which case one can wrap the codimension-two defects found in Section \ref{subsec:codim3} over the $\mathbb{CP}^1$ and then compactify this configuration on the other $\mathbb{RP}^3$ factor. Since we already needed to include the codimension-two objects the second possibility seems to be the minimal option. 

\vspace{0.2cm}

This leaves a single $\mathbb{Z}_2$ summand with generator $W_6$ described as follows: The spacetime is given by $\mathbb{RP}^3 \times \mathbb{RP}^3$ with duality bundle given by monodromy $(\gamma^2_4, M_2^{(1)})$ around the torsion 1-cycle of the first $\mathbb{RP}^3$ factor and
\begin{equation}
    \widetilde{R} = R^{(1)} \Gamma_4^{(1)} R^{(1)} \big( \Gamma_4^{(1)} \big)^3 R^{(1)} = \begin{pmatrix*}[r] -1 & 0 & 0 \\ 0 & 0 & -1 \\ 0 &  -1 & 0 \end{pmatrix*} \,,
\end{equation}
which corresponds to the permutation $(13)$ in the $S_4$ subgroup generated by $\Gamma_4^{(1)}$ and $R^{(1)}$, around the torsion 1-cycle of the second $\mathbb{RP}^3$. 

The first factor can be described as the boundary of the singular F-theory geometry
\begin{equation}
    (T^2_F \times T^2 \times \mathbb{C}^2)/(\mathbb{Z}_2)_{\gamma_4^2 \otimes M_2^{(1)}} \,,
\end{equation}
which acts as a minus sign on both 2-tori, as well as the $\mathbb{C}^2$. This geometry admits an invariant $(4,0)$-form and thus can be described as a singular, local Calabi-Yau 4-fold with singularities of the type $(T^2_F \times \mathbb{C}^3)/\mathbb{Z}_2$, which signal the presence of gauge degrees of freedom on a collapsed cycle as in various examples above. The monodromy structure further suggests a relation to the defect filling the fiber for $Q^5_{4,(\gamma_4, \Gamma_4^{(1)})}$ in Section~\ref{subsec:codim6}, as well as a composite object built of the S-fold geometries in codimension-six in \cite{Heckman}. This defect then needs to compactified on the second $\mathbb{RP}^3$ with duality bundle induced by monodromy $\widetilde{R}$ inducing a non-geometric twist. This is another instance were we are led to include additional defects in lower codimension due to higher-dimensional generators.\footnote{The same conclusion holds for filling the second $\mathbb{RP}^3$ factor which would correspond to an additional defect in codimension four, which is geometric in the M-theory duality frame.}

\subsection{Codimension-eight defects}
\label{subsec:codim8}

The codimension-eight defects associated to the brodism\footnote{Recall that a brodism is the equivalence relation between dudes/dudettes where one is equivalent to another if they are bromantically connected. This has primarily been developed by brofessional mathematicians and applied to brojective spaces, though it is widely expected that the same methods extend to bro-spectra. To date, the primary method of proof used in the brodism literature is ``just trust me bro.''} group:
\begin{equation}
\Omega^{\text{Spin}}_7 (BG_U) = \mathbb{Z}_3^{\oplus 2} \oplus \mathbb{Z}_9^{\oplus 3} \oplus \mathbb{Z}_2^{\oplus 6} \oplus \mathbb{Z}_8^{\oplus 2} \oplus \mathbb{Z}_{16}^{\oplus 2} \oplus \mathbb{Z}_{32} \,,
\end{equation}
are given by both genuine objects of codimension eight as well as the compactification of the defects introduced above.

\vspace{0.2cm}

The $\mathbb{Z}_9$ factors are given by $L^7_3$ with monodromy $\gamma_3$ or $\Gamma_3^{(i)}$. The first defect is geometric in type IIB and described by the geometry
\begin{equation}
(T^2 \times \mathbb{C}^4) / (\mathbb{Z}_3)_{\gamma_3} \,,
\end{equation}
which can be made into a local singular Calabi-Yau 5-fold by choosing the correct $\mathbb{Z}_3$ action, e.g.,
\begin{equation}
\mathbb{Z}_3: \quad (z_1 \,, z_2 \,, z_3 \,, z_4 \,, y) \rightarrow (\omega z_1 \,, \omega z_2 \,, \omega z_3 \,, \omega z_4 \,, \omega^{-1} y) \,, \quad \omega = e^{2 \pi i / 3} \,.
\end{equation}
This geometry describes three local singularities of the form $\mathbb{C}^5 / \mathbb{Z}_3$ with the action derived from the one above. These singularities are Gorenstein, canonical and terminal \cite{a6207765-8e51-37bb-af3c-43a72b254a77}. They therefore do not possess a crepant resolution. Demanding supersymmetry to be unbroken one therefore obtains a fundamental defect, which is localized in the eight-dimensional spacetime, i.e., an instanton. From the type IIB configuration one can actually see that it is a composite object formed of three connected $\mathbb{C}^5 / \mathbb{Z}_3$ sectors.

Similarly, we can define the defect with monodromy $\Gamma_3^{(1)}$ with the same action as above, but now with the torus interpreted as the F-theory fiber. The full F-theory geometry has the form
\begin{equation}
(T^2_F \times \mathbb{C}^4) / (\mathbb{Z}_3)_{\Gamma_3^{(1)}} \times T^2\,,
\end{equation}
and we find the torus compactification of codimension-eight defect of \cite{Heckman}. This string-like object potentially has an exotic $(8,2)$ supersymmetry in two dimensions which would be retained after further compactification on a (Euclidean) torus.

The last defect of this family is geometric only in the M-theory duality frame and given by the bulk geometry
\begin{equation}
(T^3 \times \mathbb{C}^4) / (\mathbb{Z}_3)_{\Gamma_3^{(2)}} \simeq \big( (T^2 \times \mathbb{C}^4) / \mathbb{Z}_3 \times S^1 \big) / \mathbb{Z}_3^s \,, 
\end{equation}
with the first factor having the same singularity structure as above. Thus, one obtains the twisted circle compactification of three sectors associated to M-theory on $\mathbb{C}^5 / \mathbb{Z}_3$.

\vspace{0.2cm}

The remaining generators at prime $3$ only generate $\mathbb{Z}_3$ summands and are given again by the lens space $L^7_3$ but with monodromy given by the tensor product $\gamma_3 \otimes \Gamma_3^{(i)}$. For $\Gamma_3^{(1)}$ one has the F-theory geometry
\begin{equation}
(T^2_F \times T^2 \times \mathbb{C}^4)/(\mathbb{Z}_3)_{\gamma_3 \otimes \Gamma_3^{(i)}} \,,
\end{equation}
which can be made into a local, singular Calabi-Yau 6-fold. It is associated to an instanton-like object which can also be understood as the composite of nine instanton-like objects discussed in the context of type IIB string theory in \cite{Heckman}, one for each of the orbifold fixed point of $(T_F^2 \times T^2)/\mathbb{Z}_3$. As in type IIB it is tempting to speculate about preserved supersymmetry for this zero-dimensional system. Indeed, the naive count for Calabi-Yau 6-folds would suggest that $\tfrac{1}{32}$ supercharges are conserved, however, since the geometry only is a local patch there can be enhancements.

\vspace{0.2cm}

The last generator at prime $3$ is non-geometric. The geometric action in the M-theory frame is given by the generator of $L^7_3$ with $\Gamma_3^{(2)}$ monodromy above. However, now there is also an action on the moduli fields when traversing the torsion 1-cycle. Thus, it can be called a non-geometric instanton. It would be interesting to investigate the zero modes of this background, include further fluxes, or use the techniques of exceptional field theory to extract more details.

\vspace{0.2cm}

Two of the $\mathbb{Z}_2$ summands are described geometrically by
\begin{equation}
(L^7_{3,\Gamma_4^{(i)}} \times \text{K3}) / \mathbb{Z}_2 \,,
\end{equation}
where $\mathbb{Z}_2$ acts by the Enriques involution on K3 and complex conjugation on $L^3_4$. The resulting space can also be understood as a fiber bundle over the Enriques surface $E$ with fiber $L^3_4$. The duality bundle is given by monodromy $\Gamma_4^{(i)}$ when traversing the torsion 1-cycle of the fiber. Since the $\mathbb{Z}_2$ acts on this monodromy this space has a $D_8 = \mathbb{Z}_4 \rtimes \mathbb{Z}_2$ bundle. The associated defects are obtained by filling the fiber with the codimension-four defect associated to $L^3_4$ with monodromy $\Gamma_4^{(i)}$ and compactifying on $E$. However, this in itself would likely not be consistent, since $E$ is not a Spin manifold. However, the fibration adds a topological twist in the compactification which ensures the consistency of fermions living on the defect with appropriate R-symmetry charge.

\vspace{0.2cm}

Two $\mathbb{Z}_2$ summands are generated by $S^1_{\gamma_4} \times \mathbb{RP}^3 \times \mathbb{RP}^3$ with duality bundle on the $\mathbb{RP}^3$ factors as above, i.e., given by $M_1^{(i)}$ and $M_2^{(i)}$ around the non-trivial 1-cycles. The associated defects are hence the compactification of the $S^1_{\gamma_4}$ defect on the two copies of real projective space. Since the monodromies $M^{(i)}$ involve all three components of SL$(3,\mathbb{Z})$ this compactification contains a non-geometric twist, with some of the type IIB moduli undergoing a non-trivial transformation.

\vspace{0.2cm}

Another $\mathbb{Z}_2$ summand is given by $S^1_{\gamma_4} \times W_6$ and the defect corresponds to the compactification of the $S^1_{\gamma_4}$ defect on the spin 6-manifold with duality bundle described by $W_6$. 

\vspace{0.2cm}

The last $\mathbb{Z}_2$ summand combines with the $\mathbb{Z}_{32}$ summand to generate a copy of $\Omega^{\text{Spin}}_7 (B\mathbb{Z}_4)$ coming from the SL$(2,\mathbb{Z})$ factor of the full duality group. The generators are given by the seven-dimensional lens spaces $L^7_4$ and $\widetilde{L}^7_4$, which differ by the choice of Spin structure and duality bundle determined by monodromy $\gamma_4$ around the torsion 1-cycle. The defects have a geometric interpretation on type IIB given by the geometry 
\begin{equation}
(T^2 \times \mathbb{C}^4) / (\mathbb{Z}_{4})_{\gamma_4} \,,
\end{equation}
with constant axio-dilaton. This geometry is affected by the same spin-lift issues as for $L^3_4$ in Section \ref{subsec:codim4}. We expect this to be resolved by refining with the spin-lift of the duality group, for which the correct generator likely becomes a $L^5_4$ bundle over $\mathbb{CP}^1$, similar to the $Q^5_4$ in dimension five, see also \cite{Heckman}.

\vspace{0.2cm}

The two $\mathbb{Z}_8$ summands are constructed using the geometry
\begin{equation}
T^4 \times \mathbb{RP}^3 \,,
\end{equation}
with non-trivial duality bundle. The monodromies around the four different 1-cycles of the torus are given by $\gamma_4$, $M_1^{(i)}$, $M_2^{(i)}$, and $M_2^{(i)}$, respectively. The monodromy around the torsion cycle of $\mathbb{RP}^3$ is given by the tensor product $\gamma^2_4 \otimes M_1^{(i)}$. Due to the appearance of various circles with monodromies encountered before there are several possibilities for the symmetry-breaking defect. For example, we can fill the first circle of the $T^4$ factor by including the defect associated to $S^1_{\gamma_4}$, which then is compactified on the remaining 6-manifold. Since the monodromies on this 6-manifold involve the full duality group this is a compactification with a non-geometric twist. Equivalently, one could fill the circle with $M_2^{(i)}$ monodromy using the codimension-two defect discussed in Section \ref{subsec:codim3}, compactified on the remaining factors involving the non-geometric twist.

\vspace{0.2cm}

The two $\mathbb{Z}_{16}$ summands are give by $L^7_4$ with monodromy given by $\Gamma_4^{(i)}$. For the F-theory generator with monodromy $\Gamma_4^{(1)}$ one encounters a geometry of the form
\begin{equation}
(T^2_F \times \mathbb{C}^4) /(\mathbb{Z}_4)_{\Gamma_4^{(1)}} \times T^2 \,,
\end{equation}
for which one finds the same complications involving the spin-lift and supersymmetry as discussed above. It is therefore described by a torus compactification of a string-like object associated to F-theory on the singular geometry above, producing a point-like object in the supergravity theory.

The other generator requires a lift to M-theory and corresponds to the singular geometry
\begin{equation}
(T^3 \times \mathbb{C}^4)/(\mathbb{Z}_4)_{\Gamma_4^{(2)}} \simeq \big(  (T^2 \times \mathbb{C}^4)/\mathbb{Z}_4 \times S^1 \big) / \mathbb{Z}^s_2 \,.
\end{equation}
This configuration has the same issues concerning the spin-lift. It contains several orbifold singularities which describe particle-like configurations of M-theory which are compactified on an additional circle with monodromy action on these particles, producing a point-like object in supergravity.

\vspace{0.5cm}

We see that the breaking of global symmetries associated to non-trivial bordism classes for the U-duality groups requires the introduction of many interesting and exotic string / M- / F-theory backgrounds. However, none of the required objects is fundamentally new, as opposed to the R7-brane described in \cite{Dierigl:2022reg, Heckman, Dierigl:2023jdp}.\footnote{see also earlier hints of it in \cite{Distler:2009ri}.} Thus, we see that string / M- / F-theory already come equipped with all the ingredients necessary to break the global symmetry induced by the topological charges associated with U-duality bundles. Moreover, we want to stress once more that, in particular at prime $2$, the fate of some of the objects needs to be reexamined after an understanding of the fermionic version of the U-duality groups is established.

\newpage

\part{Bordism calculations}
\label{part:math}

Having laid out the physical interpretation of the various bordism generators of $\Omega_{k}(B G_U^{8\mathrm{d}})$, we now turn to the computation of these structures. There are various stages of analysis we need to go through. First of all, we need to identify the isomorphism class for the various groups. All of our bordism groups have the general form of a direct sum of the form:\footnote{In this work we refer to cyclic groups as $\mathbb{Z}_k$ as opposed to $\mathbb{Z}/k$ or $\mathbb{Z}/k\mathbb{Z}$ which are also common in the math literature.}
\begin{equation}
\Omega_{k} \simeq \mathbb{Z}^{\oplus m} \oplus \underset{i}{\bigoplus} \mathbb{Z}_{p_i^{\ell_i}},
\end{equation}
and in the cases of interest, the reduced Spin-bordism groups (i.e., those obtained by quotienting out by $\Omega^{\mathrm{Spin}}_{k}(\mathrm{pt})$) are pure torsion. Simply identifying the different torsion factors does not suffice for physics applications, we also need to seek out explicit manifolds $X_k$ equipped with specific duality bundles that generator them. While there are typically
algorithmic approaches to extracting the isomorphism class of the groups, finding explicit generators (as far as we are aware) is somewhat of an art form and involves making motivated mathematical (and physical!) guesses. Indeed, we find that a large number of the generators take the form of generalized lens spaces equipped with a duality bundle, and can in turn be viewed as the boundary of a general quotient of the form $(\mathbb{C}^{m} \times T^{n}) / \Gamma$, where the torus fiber geometrizes the duality bundle structure.

It is also worthwhile contrasting the specific strategy adopted here with the case of IIBordia studied in \cite{Heckman}. In the study of IIBordia, the relevant duality groups admit an amalgamated product structure $G_{\mathrm{IIB}} = H_1 \ast_{K} H_2$, where the individual factors $H_i$ and $K$ are cyclic and / or dihedral groups. Since there is a Mayer-Vietoris long exact sequence for amalgamated products, the whole computation (while still very non-trivial!) boils down to simpler constituent building blocks. Moreover, these individual building blocks have natural physical interpretations in terms of specific F-theory backgrounds.

It turns out that this approach does not really extend to the duality groups of Utopia. The main issue is that all of the $\mathrm{SL}(n,\mathbb{Z})$ groups are ``perfect'' for $n \geq 3$, and in particular, they do not decompose into amalgamated products. As such, identifying natural building blocks is somewhat more challenging.

That being said, we are still able to greatly simplify the relevant calculations. In particular, much as in other bordism calculations, the relevant calculations localize at different prime factors. In particular, the relevant bordism groups localize at primes $p=2,3$. In this reduction, the most complicated group cohomology we need to contend with is that of the symmetric group on four letters $\mathcal{S}_4$ which has thankfully already been computed in the extant literature. As such, many of the relevant bordism group calculations can be boiled down to simpler statements.

The other simplification which greatly facilitates the computation of bordism groups is that of a ring / module structure in which a cyclic group acts on the relevant groups. This has a natural physical interpretation as organizing our bordism defects into symmetry multiplets of an appropriate automorphism symmetry.

Our plan in this part will therefore be to first explain the general strategy, and to then proceed to the computation at prime $p = 3$ and then prime $p = 2$.

\subpart{First calculations, and the plan}

\section{First spin bordism of U-duality groups}
\label{sec:Omega1}

In this section, we determine $\Omega_1^{\Spin}(BG)$ for the U-duality groups in dimensions 3 through 10. As in the rest of this paper, we do not consider how these groups mix with fermion parity. We collect the results of this section in \cref{Table:FirstBordism}.

These calculations amount to
determining the Abelianization of $G$, denoted $\mathrm{Ab}(G)$, for each of the U-duality groups $G$. Indeed, an easy application of the Atiyah-Hirzebruch spectral sequence reveals that
\begin{equation}\label{FirstBordism}
    \Omega^{\Spin}_1(BG) \cong H_1(BG;\Z) \oplus \Omega^{\Spin}_1(\mathrm{pt}) \cong \mathrm{Ab}(G) \oplus \Omega_1^{\Spin}(\mathrm{pt}) \,.
\end{equation}

\begin{lem}[Abelianization of 10d to 7d U-duality groups]\label{ab1}
Consider the U-duality groups of 10, 9, 8, and 7 dimensional supergravity:
\begin{equation}
G^{\mathrm{10d}}_U = 1 \,, \;\; G^{\mathrm{9d}}_U =\SL(2,\Z) \,, \;\; G^{\mathrm{8d}}_U = \SL(3,\Z)\times \SL(2,\Z) \,, \;\; G^{\mathrm{7d}}_U =\SL(5,\Z) \,.
\end{equation}
Then,
\begin{equation}
\mathrm{Ab}(G^{\mathrm{10d}}_U) = 1 \,, \;\; \mathrm{Ab}(G^{\mathrm{9d}}_U)= \Z_{12} \,, \;\; \mathrm{Ab}(G^{\mathrm{8d}}_U) = \Z_{12}, \;\; \mathrm{Ab}(G^{\mathrm{7d}}_U) = 0
\end{equation}
\end{lem}
\begin{proof}
The Abelianizations of $\SL(n,\Z)$ are well-known: 
    \begin{equation}
    \mathrm{Ab} \big( \SL(1,\Z) \big) \cong 1, \;\; \mathrm{Ab} \big( \SL(2,\Z) \big) \cong \Z_{12} \;\; \text{and}\;\; \mathrm{Ab}\big(\SL(n,\Z) \big) \cong 0 \;\; \text{for}\;\;n\geq 3
\end{equation}
The claim then follows from the fact that Abelianization respects direct products.
\end{proof}

The remaining U-duality groups require more machinery. In part, they require an identification of the groups $\SSO(5,5,\Z)$, $\SSO(6,6,\Z)$, and $\SSO(7,7,\Z)$ with a certain class of Chevalley-Demazure group schemes; see \cite[\S 2-7]{Vavilov} for an introduction. 

\begin{lem}
\label{lem:perfectgroups}
    The groups $\SSO(5,5,\Z)$, $\SSO(6,6,\Z)$, and $\SSO(7,7,\Z)$ are perfect. 
\end{lem}
\begin{proof}
    Consider $\SSO(5,5,\Z)$, where $\mathrm O(5,5,\Z)$ consists of the matrices $A$ satisfying 
\begin{equation}
A^T\begin{pmatrix}
    0 & 1_n\\
    1_n & 0
    \end{pmatrix}A =  \begin{pmatrix}
        0 & 1_n \\
        1_n & 0
    \end{pmatrix}
\end{equation}
$\SSO(5,5,\Z)$ is the group of $\Z$-points of a Chevalley-Demazure group scheme of type $D_5$; see \cite[\S 8]{Vavilov}. Let $\mathrm{EO}(5,5,\Z)\leq \SSO(5,5,\Z)$ be the associated elementary Chevalley-Demazure group. Since $\Z$ is Euclidean, $\mathrm{EO}(5,5,\Z) = \SSO(5,5,\Z)$. The claim that $\SSO(5,5,\Z)$ is perfect then follows from the fact that $\mathrm{EO}(5,5,\Z)$ is perfect \cite[Corollary 4.4]{Stein}. The same reasoning can be used to show that $\SSO(6,6,\Z)$ and $\SSO(7,7,\Z)$ are perfect.
\end{proof}

We also make use of the following lemma due to Obers and Pioline. 

\begin{lem}[{Obers and Pioline, \cite{Obers}}]\label{lem:UdualityKnit}
    The D-dimensional U-duality group $E_{\ell(\ell)}$, with $\ell = 11 -D$, can be written as 
    \begin{equation}
    \mathrm{E}_{\ell(\ell)}(\Z)=\SL(\ell,\Z)\bowtie \SSO(\ell-1,\ell-1,\Z)
    \end{equation}
    where $\bowtie$ denotes the knit product.  
\end{lem}
A review of the knit product and why we should expect the U-duality group to be given as above is given in Appendix \ref{App:Knit}.

\cref{lem:UdualityKnit} and \cref{lem:perfectgroups} allow us to determine the remaining Abelianizations. 

\begin{lem}[Abelianization of 6d to 3d U-duality groups]\label{ab2}
    Consider the U-duality groups of maximally supersymmetric 6, 5, 4, and 3 dimensional supergravity:
    \begin{equation}
        G^{\mathrm{6d}}_U = \SSO(5,5,\Z), \quad G^{\mathrm{5d}}_U = \mathrm{E}_{6(6)}(\Z), \quad G^{\mathrm{4d}}_U = \mathrm{E}_{7(7)}(\Z), \quad G^{\mathrm{3d}}_U = \mathrm{E}_{8(8)}(\Z)
    \end{equation}
    Then, 
    \begin{equation}
    \mathrm{Ab} (G^{\mathrm{6d}}_U)\cong 0, \quad \mathrm{Ab} (G^{\mathrm{5d}}_U)\cong 0, \quad \mathrm{Ab} (G^{\mathrm{4d}}_U)\cong 0, \quad \mathrm{Ab} (G^{\mathrm{3d}}_U) \cong 0
    \end{equation}
\end{lem}
\begin{proof}
    The claim follows from \cref{lem:UdualityKnit}, \cref{lem:perfectgroups}, and that the knit product of perfect groups is perfect. 
\end{proof}
We now have everything we need to determine the first spin bordism groups of the various U-duality groups using \eqref{FirstBordism}. We summarize the results in Table \ref{Table:FirstBordism}.

\begin{table}[h!]
\centering
\begin{tabular}{c c c} 
\toprule
\text{Dimension $D$} & 
\text{U-duality group $G_U^D$} & $\Omega_1^{\Spin}(BG_U^{D})$ \\
\midrule
$10$ & 1 & $\Z_2$\\
$9$ & $\SL(2, \Z)$ &  $\Z_2\times\Z_{12}$\\
$8$ & $\SL(2, \Z)\times\SL(3, \Z)$ & $\Z_2\times\Z_{12}$\\
$7$ & $\SL(5, \Z)$ & $\Z_2$\\
$6$ & $\SSO(5, 5, \Z)$ & $\Z_2$\\
$5$ & $\mathrm{E}_{6(6)}(\Z)$ & $\Z_2$\\
$4$ & $\mathrm{E}_{7(7)}(\Z)$ & $\Z_2$\\
$3$ & $\mathrm{E}_{8(8)}(\Z)$ & $\Z_2$\\
\bottomrule
\end{tabular}
\caption{\label{Table:FirstBordism}The first spin bordism group of the U-duality groups for the maximally supersymmetric supergravity theories in various dimensions. These computations are a combination of~\eqref{FirstBordism} and \cref{ab1,ab2}.}
\end{table}

The generators of these first spin bordism groups are straightforward to determine. Recall that for a generalized homology theory $E$ and space $X$,
 \begin{equation}
 E_*(X) \cong E_*(\mathrm{pt}) \oplus \widetilde{E}_*(X)
 \end{equation}
 
Furthermore, recall that $\Omega^{\Spin}_1(\mathrm{pt}) \cong \Z_2$ and is generated by $S_+^1$, where $S_+^1$ is the circle with Spin structure induced from the Lie group framing, i.e., periodic boundary conditions for fermions. Since $\Omega^{\Spin}_1(BG) \cong \Z_2$ for the 7d, 6d, 5d, 4d, and 3d U-duality groups, we conclude that the generator for these groups is $S_+^1$. The first spin bordism groups for the remaining U-duality groups, those for 10d, 9d, and 8d supergravity, have additional generators. Each group contains a factor of $\Z_{12}\cong \Z_{3}\oplus \Z_4$; the generators for these groups were determined in \cite{Heckman}.  Finally, the group corresponding  9d supergravity has an additional $\Z_2$ summand coming from $\widetilde\Omega^{\Spin}_1(B\Z_2)$, which is generated by the circle with either spin structure and a nontrivial $\Z_2$-bundle (see for example~\cite[\S 14.3.2]{Heckman}). 

For some further results on the case of $k=2$ see \cite{BraegerThesis}.

\section{\texorpdfstring{Reducing from $\SL(2, \Z)$ and $\SL(3, \Z)$ to finite groups}{Reducing from SL(2,Z) and SL(3,Z) to finite groups}}
\label{sec:stablesplit}

The most important step in our computations, and the reason we are able to go so much farther for 8d than for lower-dimensional theories, is a tool called a \term{stable splitting}. This is an idea coming from homotopy theory which expresses generalized (co)homology of a space as a direct sum of that of simpler spaces. Stable splittings are known for $B\SL(2, \Z)$, $B\SL(3, \Z)$, and $B\SL(2, \Z)\times B\SL(3, \Z)$; the existence of stable splittings for other U-duality groups would allow an extension of our methods to lower-dimensional supergravity theories.

We begin in \S\ref{ss:stabsplit_what} with an introduction to stable splittings and a few variants. This is a standard concept in homotopy theory; we aimed our exposition at a theoretical physics audience.
Then, in \S\ref{subsec:stablesplit}, we give the $2$- and $3$-local stable splittings of $B\SL(2, \Z)$, $B\SL(3, \Z)$, and $B\SL(2, \Z)\times B\SL(3, \Z)$ that we will need in order to compute their spin bordism groups. These results are due to Minami~\cite{Minami}, building on Soul\'{e}'s calculation~\cite{Soule} of $H^*(B\SL(3, \Z); A)$ for various coefficient rings $A$.

In \S\ref{ss:assembly}, we summarize the implications of these results for the bordism computations we want to make. This is the bridge between the bordism computations later in this paper, which are in terms of easier groups than $\SL(3, \Z)$, and the generators in the physics section of this paper.

As will be clear from this section, these stable splittings are absolutely crucial for our computations. We would be interested in learning of results analogous to Soul\'{e}'s or Minami's~\cite{Soule,Minami} for other U-duality groups.

\subsection{What is a stable splitting?}
\label{ss:stabsplit_what}
Recall that a \term{pointed space} is a space $X$ equipped with a distinguished point $x\in X$. A map of pointed spaces is a continuous map that sends the basepoint to the basepoint. The homotopy theory of unpointed spaces embeds into the homotopy theory of pointed spaces by sending a space $X$ to $X_+\coloneqq X\amalg \set{0}$, with $0$ as the basepoint. Thus the unreduced (generalized) (co)homology of $X$ is the reduced (generalized) (co)homology of $X_+$, and indeed this is often taken as a definition. For $E$ a generalized homology theory, we use $\widetilde E_*(X)$ to denote the reduced $E$-homology of $X$.
\begin{defn}[Spanier~\cite{Spa56}]
Let $f\colon X\to Y$ be a map of pointed spaces. We say $f$ is an \term{S-equivalence} if for every generalized homology theory $E_*$, the map $f_*\colon \widetilde E_*(X)\to \widetilde E_*(Y)$ is an isomorphism.
\end{defn}
One may just as well use generalized cohomology theories. Spanier's original definition looks different, but is equivalent to this one for all spaces with the homotopy type of CW complexes.

The \term{wedge sum} of pointed spaces $X$ and $Y$, denoted $X\vee Y$, is the quotient of $X\amalg Y$ identifying the two basepoints into a single point $b$. $X\vee Y$ is pointed, with basepoint $b$. For any generalized homology theory $E_*$, the inclusion maps $i_1\colon X\to X\vee Y$ and $i_2\colon Y\to X\vee Y$ induce an isomorphism
\begin{equation}
    (i_1,i_2)\colon \widetilde E_*(X)\oplus \widetilde E_*(Y)\overset\cong\longrightarrow \widetilde E_*(X\vee Y);
\end{equation}
indeed, this is one of the Eilenberg-Steenrod axioms of a generalized homology theory. An analogous fact is true for generalized cohomology theories.
\begin{defn}\label{stabsplit_defn}
A \term{stable splitting} of a pointed space $X$ is data of maps $f\colon Y\to X$ and $g\colon Z\to X$ such that $(f\vee g)\colon Y\vee Z\to X$ is an S-equivalence. A stable splitting of an unpointed space $X$ is defined to be a stable splitting of $X_+$.

We define stable splittings of spaces into three or more wedge summands analogously. Moreover, in view of the suspension isomorphism on generalized (co)homology, we allow $f$ and $g$ to only be defined after suspending $X$, $Y$, and $Z$ $k$ times for some $k\ge 0$.
\end{defn}
Thus a stable splitting of $X$ realizes the generalized homology groups of $X$ as a direct sum of those of $Y$ and $Z$. We will be interested in this for $E_* = \Omega^\Spin_*$.
\begin{exm}\label{reduced_unreduced}
For any nonempty space $X$, a choice of basepoint in $X$ induces a stable splitting of $X_+$ into $S^0\vee X$. This realizes the splitting of unreduced (generalized) (co)homology into reduced (generalized) (co)homology and the (generalized) (co)homology of a point:
\begin{equation}
    E_*(X) \cong \widetilde E_*(X)\oplus E_*(\pt),
\end{equation}
since $E_*(\pt)\cong\widetilde E_*(S^0)$.
\end{exm}
\begin{exm}\label{distributivity}
For nonempty spaces $X$ and $Y$ with basepoints $x_0\in X$ and $y_0\in Y$, \cref{reduced_unreduced} can be iterated to produce a stable splitting of $X\times Y$ into the four pieces $S^0\vee X\vee Y\vee (X\wedge Y)$, and analogously with products of more than two nonempty spaces. Here $X\wedge Y$ is the \term{smash product}, defined to be the quotient of $X\times Y$ in which we identify all points of the form $(x, y_0)$ and $(x_0, y)$ with the basepoint $(x_0, y_0)$.
%\textcolor{red}{MD:Should we include definition of smash?}.
\end{exm}
\Cref{distributivity} expresses a kind of distributivity of wedge sums/stable splittings over products.

For most spaces of interest in applications to physics, $H_n(X)$ is a finitely generated Abelian group. Thus we may study $H_*(X)$, as well as many generalized (co)homology theories on $X$, by ``working one prime at a time:'' tensoring with the ring $\Z_{(p)}$ of rational numbers whose denominators are coprime to a chosen prime number $p$. For finitely generated Abelian groups, this has the effect of preserving free summands and $\Z_{p^\ell}$ summands, and throwing out $r$-torsion for primes $r \ne p$. This is a standard technique in homotopy theory; see~\cite[\S 10.2]{Heckman} for more information.

In particular, we may apply this philosophy to stable splittings.
\begin{defn}
A \term{$p$-local S-equivalence} is a map $f\colon X\to Y$ such that for all generalized homology theories $E_*$, the map $f_*\colon E_*(X)\otimes\Z_{(p)}\to E_*(Y)\otimes\Z_{(p)}$ is an isomorphism.

A \term{$p$-local stable splitting} is defined identically to \cref{stabsplit_defn}, but with $p$-local S-equivalence in place of S-equivalence.
\end{defn}
Thus, if $H_*(X)$ is finitely generated in each degree, if we have a $p$-local stable splitting of $X$ for every prime $p$ into pieces whose $E$-homology we can calculate, we can recover the $E$-homology of $X$.
\begin{rem}
The standard definition of stable splitting (as in, for example~\cite{Priddy, MP83}) works with spectra: a stable splitting of $X$ into $Y$ and $Z$ is a stable homotopy equivalence $\Sigma^\infty(Y\vee Z)\overset\simeq\to \Sigma^\infty X$, and likewise with the $p$-localizations of these spectra and $p$-local stable splittings. One may also allow $Y$ and $Z$ to be spectra. We mostly do not need this generalization and so chose to give the simpler definition.
\end{rem}

\subsection{\texorpdfstring{Stably splitting $B\SL(2, \Z)\times B\SL(3, \Z)$}{Stably splitting BSL(2,Z) x BSL(3,Z)}}
\label{subsec:stablesplit}

The following result is well-known; see~\cite[Lemmas 12.2 and 12.7]{Heckman}, for example, for a proof.
\begin{prop}\label{SL2_splitting_at_2}\hfill
\begin{enumerate}
    \item There is a $2$-local S-equivalence $B\Z_4\to B\SL(2, \Z)$ induced from the group homomorphism $\rho_{2,2}\colon \Z_4\hookrightarrow \SL(2, \Z)$ defined by sending $1\in\Z_4$ to the matrix $S = \begin{pmatrix*}[r]
        0 & -1\\1 & 0
    \end{pmatrix*}$.
    \item There is a $3$-local S-equivalence $B\Z_3\to B\SL(2, \Z)$ induced from the group homomorphism $\rho_{2,3}\colon \Z_3\hookrightarrow \SL(2, \Z)$ defined by sending $1\in\Z_3$ to the matrix $(ST)^2 = \begin{pmatrix*}[r]
        -1 & -1\\1 & 0
    \end{pmatrix*}$.
    \item If $p\ge 5$ is a prime number, $B\SL(2, \Z)\to\pt$ is a $p$-local S-equivalence.
\end{enumerate}
\end{prop}
We use $D_{2n}$ to refer to the dihedral group with $2n$ elements, with presentation
\begin{equation}\label{dihedral_presentation}
    D_{2n} = \ang{r,s\mid r^n = s^2 = 1, srs = r^{-1}}.
\end{equation}
\begin{defn}
\label{rho12_defn}
Let $\rho_1,\rho_2\colon D_6\to\SL(3, \Z)$ be the homomorphisms defined on generators by:
\begin{subequations}\label{rho12}
\begin{alignat}{2}
    \rho_1(r) &= \begin{pmatrix*}[r] 1 & 0 & 0 \\ 0 & -1 & -1 \\ 0 & 1 & 0 \end{pmatrix*} \qquad\qquad & \rho_1(s) &= \begin{pmatrix*}[r] -1 & 0 & 0 \\ 0 & \phantom{-}1 & 1 \\ 0 & 0 & -1 \end{pmatrix*}\\
    \rho_2(r) & = \begin{pmatrix*}[r] 0 & 1 & 0 \\ 0 & 0 & 1 \\ 1 & 0 & 0 \end{pmatrix*}  & \rho_2(s) &= \begin{pmatrix*}[r] 0 & 0 & -1 \\ 0 & -1 & 0 \\ -1 & 0 & 0 \end{pmatrix*} .
\end{alignat}
\end{subequations}
\end{defn}
The reader can check that the matrices in~\eqref{rho12} have determinant $1$ and satisfy the relations in~\eqref{dihedral_presentation}, so that they define homomorphisms $D_6\to\SL(3, \Z)$ as claimed.
\begin{prop}[{Brown~\cite[\S 6]{Bro76}, Soule~\cite[Theorem 4(iii)]{Soule}}]
\label{5_boring}
For any prime $p\ge 5$, the map $B\SL(3, \Z)\to\pt$ is a $p$-local S-equivalence.
\end{prop}
\begin{prop}[{Brown~\cite[\S 6]{Bro76}, Soule~\cite[Corollary (i) to Lemma 8]{Soule}, Minami~\cite[\S 0]{Minami}}]
\label{SL3_stable_splitting_at_3}
The maps $B\rho_1,B\rho_2\colon BD_6\to B\SL(3, \Z)$ define a $3$-local stable splitting of $B\SL(3, \Z)$ into $BD_6\vee BD_6$.
\end{prop}
\begin{rem}\label{which_rho3}
Brown and Soul\'{e} expressed their results at the level of cohomology, and did not discuss stable splittings. The stable splitting follows from their theorems and the Whitehead theorem in a standard way; Minami~\cite[\S 0]{Minami} explicitly describes Soul\'{e}'s theorem as a stable splitting. Both Soul\'{e} and Minami describe \cref{SL3_stable_splitting_at_3} using $BS_3$ and $BD_{12}$ instead of $BD_6$ twice, but $S_3\cong D_6$ and the standard inclusion $D_6\hookrightarrow D_{12}$ induces a $3$-local homotopy equivalence on classifying spaces~\cite[\S 14.1]{Heckman}, so the result is the same.

The specific matrices come from those used by Soul\'{e} as follows: in~\cite[Prop.\ 1, p.\ 9]{Soule}, he gives explicit matrices generating a $D_{12}$ subgroup of $\SL(3, \Z)$ which he calls $\underline Q$, and the image of $\rho_1$ is the usual $D_6\subset D_{12}$. For $\rho_2$, Soul\'{e} labels a certain $S_4$ subgroup $\underline O$, and the image of $\rho_2$ is the usual $D_6\cong S_3\subset S_4$. Brown chooses different matrices.
\end{rem}

The $2$-local stable splitting is slightly more complicated; unlike in the previous cases, the pieces of the stable splitting are not all induced by inclusions of subgroups.

Mitchell-Priddy~\cite{MP83} define a spectrum $L(2)$ (called $\overline{\mathit{Sp}}{}^4(S^0)$ in their earlier work~\cite{Priddy}, with no relation to the symplectic group). We will not need to know much about $L(2)$, but see~\cite[\S 3]{Priddy} for a definition and basic properties.
\begin{prop}[{Minami~\cite[Theorem 3.4(iii)]{Minami}}]
\label{minami_splitting}
After $2$-localization, there are maps $\psi_1,\psi_2\colon B\SL(3, \F_2)\to B\SL(3, \Z)$ and $\psi_3\colon L(2)\to B\SL(3, \Z)$ such that\footnote{If $q$ is a prime power, $\F_q$ denotes the finite field with $q$ elements, which is unique up to isomorphism. When $q$ is prime, $\F_q\cong\Z_q$ as rings.}
\begin{equation}
    (\psi_1,\psi_2, \psi_3)\colon B\SL(3, \F_2)\vee B\SL(3, \F_2)\vee L(2)\longrightarrow B\SL(3, \Z)
\end{equation}
is a $2$-local stable splitting of $B\SL(3, \Z)$.
\end{prop}

\Cref{minami_splitting} is complicated by the fact that the maps $\psi_1$ and $\psi_2$ are not to our knowledge induced by group homomorphisms from $\SL(3,\F_2)$ to $\SL(3,\Z)$.
We finish out this subsection by providing a description of these maps that works well for describing U-duality defects.
\begin{defn}
\label{rho3_defn}
Recall that the symmetric group $S_4$ admits a presentation
\begin{equation}
\label{S4_presentation}
S_4 = \ang{c, \tau\mid c^4 = \tau^2 = (\tau c)^3 = 1}.
\end{equation}
where $c = (1\ 2\ 3\ 4)$ and $\tau = (1\ 2)$. Let $\rho_3\colon S_4\to\SL(3, \Z)$ be the homomorphism defined on generators by
\begin{equation}\label{rho_3}
    \rho_3(c) = \begin{pmatrix*}[r]
         0&0&1\\
         0&1&0\\
        -1&0&0
    \end{pmatrix*}\qquad\qquad
    \rho_3(\tau) = \begin{pmatrix*}[r]
        -1 & 0 & 0\\
        0 & 0 & -1\\
        0 & -1 & 0
    \end{pmatrix*}.
\end{equation}
\end{defn}
\eqref{rho_3} is the pair of matrices Soul\'{e}~\cite[Theorem 2]{Soule} labeled $\underline O$;
the reader can check these matrices have determinant $1$ and satisfy the relations in~\eqref{S4_presentation}.

\begin{prop}[{Mitchell-Priddy~\cite[Theorem B]{Priddy}}]
\label{S4_splitting}
Let $q\colon\SL(3, \Z)\to \SL(3, \F_2)$ be the homomorphism reducing the entries of a matrix mod $2$, and let $\sigma\colon S_4\to\Z_2$ be the sign homomorphism. Then there is a stable map $\psi_4\colon BS_4\to L(2)$ such that
\begin{equation}
    (B\sigma, B(q\circ\rho_3), \psi_4)\colon BS_4\longrightarrow B\Z_2\vee B\SL(3, \F_2)\vee L(2)
\end{equation}
is a $2$-local stable splitting.
\end{prop}
Thus for any kind of bordism of $B\SL(3, \F_2)$, every bordism class can be represented by a manifold with $\SL(3, \F_2)$-bundle induced from an $S_4$-bundle via $q\circ\rho_3$.\footnote{In fact, one can pull back further to $D_8$~\cite[Theorems A and B]{Priddy}, which we used to help constrain the search space for generators of bordism groups.}
\begin{prop}\label{why_S4}
There is a homotopy equivalence of stable maps
\begin{equation}
    Bq\simeq ((\psi_2,\psi_3)\circ (Bq\circ\rho_3, \psi_4))\colon BS_4\longrightarrow B\SL(3, \F_2)\vee L(2)\longrightarrow B\SL(3, \Z).
\end{equation}
That is, in the isomorphism $\widetilde\Omega_*^\Spin(B\SL(3, \Z))\cong\widetilde\Omega_*^\Spin(B\SL(3, \F_2))\oplus\widetilde\Omega_*^\Spin(B\SL(3, \F_2))\oplus \Omega_*^\Spin(L(2))$ induced by \cref{minami_splitting}, the second and third summands can be realized by computing $\widetilde\Omega_*^\Spin(BS_4)$, throwing out all classes which are nontrivial when pulled back to a $\Z_2$ reflection subgroup of $S_4$, and then changing from $S_4$-bundles to $\SL(3, \Z)$-bundles via $\rho_3$.
\end{prop}
\begin{proof}
This almost completely follows by combining \cref{which_rho3,S4_splitting}, except that we are also claiming that the factors of $L(2)$ match. Fortunately, Minami~\cite[Corollary 3.3]{Minami} shows that the $L(2)$ stable summand of $B\SL(3, \Z)$ is in fact pulled back to a stable summand of $BS_4$ via $\rho_3$, so we are all set.
\end{proof}
See Dwyer-Wilkerson~\cite[Theorem 4.1]{DW93} for a $2$-adic analogue of this story, which may be expressed directly in terms of a representation $\SL_3(\F_2)\to\GL_3(\hat\Z_2)$.

Finally, we have to dispatch $\psi_1$.
\begin{prop}[{Minami~\cite{Minami}}]\label{psi1_done}
Let $\rho_4\colon S_4\to\SL_3(\Z)$ be the homomorphism defined on generators by
\begin{equation}\label{underline_P}
    \rho_4(c) = \begin{pmatrix*}[r] 0&0&-1 \\ -1&0&-1 \\ 0&1&1 \end{pmatrix*} \qquad\qquad
    \rho_4(\tau) = \begin{pmatrix*}[r] -1&0&0 \\ 0 & 0 &1 \\ 0 & 1 & 0 \end{pmatrix*},
\end{equation}
and let $i\colon S_4\to \SL_3(\F_2)$ be the usual three-dimensional representation. Then $B\rho_4 = \psi_1\circ Bi\colon BS_4\to B\SL_3(\Z)$.
\end{prop}
Thus $\psi_1$ itself is transfer $\Sigma_+^\infty B\SL_3(\F_2)\to \Sigma_+^\infty BS_4$ followed by $i\colon BS_4\to B\SL_3(\Z)$. The matrices in~\eqref{underline_P} are labeled $\underline P$ in~\cite[Theorem 2]{Soule}.
\begin{rem}
\label{better_matrices}
The matrices that appear here, in the definitions of $\rho_3$ and $\rho_4$ above in~\eqref{rho_3} and~\eqref{underline_P}, are not the most convenient choices for the physics applications in the first part of this paper. Indeed, comparing with~\eqref{eq:Gamma4} and \eqref{eq:Monodromies}, the reader will notice that we have replaced $\rho_3(c)$, $\rho_3(\tau)$, $\rho_4(c)$, and $\rho_4(\tau)$ with different matrices: $\Gamma_4^{(1)}$, $R^{(1)}$, $\Gamma_4^{(2)}$, and $R^{(2)}$, respectively. These sets of matrices are conjugate inside $\SL(3, \Z)$: specifically, if
\begin{equation}
    C^{(1)} = \begin{pmatrix*}[r]
        0&1&0\\
        0&0&1\\
        1&0&0
    \end{pmatrix*}\quad\text{and}\quad
    C^{(2)} = \begin{pmatrix*}[r]
        -1&0&0\\
        1&-1&0\\
        0&1&1
    \end{pmatrix*},
\end{equation}
then the reader can verify that $\det(C^{(1)}) = \det(C^{(2)}) = 1$ (so that $C^{(1)},C^{(2)}\in\SL(3, \Z)$) and that
\begin{equation}
    \begin{alignedat}{2}
        C^{(1)}\rho_3(c) (C^{(1)})^{-1} &= \Gamma_4^{(1)}
        \qquad\qquad
        & C^{(1)}\rho_3(\tau) (C^{(1)})^{-1} &= R^{(1)}\\
        C^{(2)}\rho_4(c) (C^{(2)})^{-1} &= \Gamma_4^{(2)}
        \qquad\qquad
        & C^{(2)}\rho_4(\tau) (C^{(2)})^{-1} &= R^{(2)}.
    \end{alignedat}
\end{equation}
If $\phi,\psi\colon G_1\rightrightarrows G_2$ are injective group homomorphisms whose images are conjugate in $G_2$, their induced maps on classifying spaces are homotopic; therefore, for the purpose of realizing stable splittings and describing generators of bordism groups, we may use the matrices $\Gamma_4^{(i)}$ and $R^{(i)}$ in~\eqref{eq:Gamma4} and~\eqref{eq:Monodromies}.
\end{rem}

Analogously to the stable and/or cohomological splittings by Brown~\cite{Bro76}, Soul\'{e}~\cite{Soule}, and Minami~\cite{Minami} that we used above, there is also related work studying the $K$-theory of $B\SL(3, \Z)$ and similar objects, including that of Adem~\cite{Ade92, Ade93, Ade93a}, Tezuka-Yagita~\cite{TY92}, Juan-Pineda~\cite{JP98}, Lück~\cite{Luc07}, Sánchez-García~\cite{SG08}, Joachim-Lück~\cite{JL13}, Bárcenas-Velásquez~\cite{BV14, BV16}, Hughes~\cite{Hug21}, Bárcenas~\cite{Bar24}, and Lück-Patchkoria-Schwede~\cite{LPS24}.
\subsection{Implications for spin bordism}
\label{ss:assembly}
The previous subsection tells us in principle how to reduce the computation of $\Omega_*^\Spin(BG)$ for $G = \SL(2, \Z)$, $\SL(3, \Z)$, or $\SL(2, \Z)\times \SL(3, \Z)$ to the analogous computations for a few finite groups $H_i$. Here, we make these reductions explicit: given generating manifolds for $\Omega_*^\Spin(BH_i)$, how do we obtain the generators of $\Omega_*^\Spin(BG)$?

\begin{defn}
Given a group homomorphism $\rho\colon G\to H$ and a principal $G$-bundle $P\to M$, let $\rho(P)\to M$ denote the principal $H$-bundle given by the reduction of the structure group of $P$ along $\rho$; explicitly, $\rho(P)$ is $P\times_G H\to M$.
\end{defn}
For $\SL(2, \Z)$, this is implicit in~\cite[\S 12.2]{Heckman}; we review it to set the stage.
\begin{prop}
Let $\rho_{2,2}\colon\Z_4\to\SL(2, \Z)$ and $\rho_{2,3}\colon\Z_3\to\SL(2, \Z)$ be the homomorphisms defined in \cref{SL2_splitting_at_2}. If $S_2 = \set{(M_1, P_1), \dotsc, (M_k,P_k)}$ is a linearly independent generating set for $\widetilde\Omega_n^\Spin(B\Z_4)$ and $S_3 = \set{(N_1, Q_1), \dotsc, (N_\ell,Q_\ell)}$ is a linearly independent generating set for $\widetilde\Omega_n^\Spin(B\Z_4)$, then
\begin{equation}
    \set{
        (M_1, \rho_{2,2}(P_1)),\dotsc, (M_k, \rho_{2,2}(P_k)),
        (N_1, \rho_{2,3}(Q_1)), \dotsc, (N_\ell, \rho_{2,3}(Q_\ell))
    }
\end{equation}
is a linearly independent generating set for $\widetilde\Omega_n^\Spin(B\SL(2, \Z))$.
\end{prop}
In other words, take the generators for the spin bordism of $B\Z_4$ and use $\rho_{2,2}$ to get to $B\SL(2, \Z)$, and do the same for $B\Z_3$ and $\rho_{2,3}$.
\begin{proof}
This is a corollary of \cref{SL2_splitting_at_2}, together with the definition of $p$-local stable splitting.
\end{proof}
For $\SL(2, \Z)\times\SL(3, \Z)$, the story is a little more complicated, but just as explicit.

The following definition asks for a set of generators of $\widetilde\Omega_k^\Spin(BS_4)$ to be compatible with \cref{S4_splitting}.
\begin{defn}
\label{splitting_compatible}
Let $S$ be a linearly independent generating set of $\widetilde\Omega_k^\Spin(BS_4)$. We say $S$ is a \term{splitting-compatible generating set} if $S = S_1\amalg S_2\amalg S_3$, where $S_1$ is in the image of the inclusion $\widetilde\Omega_k^\Spin(B\Z_2)\to \widetilde\Omega_k^\Spin(BS_4)$, $S_2$ is in the image of the analogous inclusion for $B\SL(3, \F_2)$, and $S_3$ is in the image of the analogous inclusion for $L(2)$. We will refer to $S_1$, $S_2$, and $S_3$ as the \term{first, second, and third components}, respectively, of $S$.
\end{defn}
We make the analogous definition for $B\Z_4\wedge BS_4$.

Though the components are defined in terms of other things than $BS_4$, we have defined \cref{splitting_compatible} so that all elements of a splitting-compatible generating set are manifolds with principal $S_4$-bundles, not for some other group.
\begin{rem}
In \cref{S4_splitting}, we did not specify the inclusions we used in \cref{splitting_compatible}, only the projection maps in the other direction. This is fine: \cref{S4_splitting} implies these inclusions exist, and that, for any choice of these inclusions compatible with the projections in \cref{S4_splitting}, there is a splitting-compatible set of generators for $\widetilde\Omega_k^\Spin(BS_4)$ for all $k$.
Below, we only need that a generating set is splitting-compatible for \emph{some} choice of inclusions, so this is OK.
\end{rem}
We have been representing elements of $\widetilde\Omega_k^\Spin(BG)$ by the data $(M, P)$, where $M$ is a closed spin $k$-manifold and $P\to M$ is a principal $G$-bundle. Below, we will also represent elements of $\widetilde\Omega_k^\Spin(BG\wedge BH)$ and $\widetilde\Omega_k^\Spin(BG\times BH)$ as the data $(M, P, P')$, where $P\to M$ is a principal $G$-bundle and $P'\to M$ is a principal $H$-bundle.\renewcommand{\descriptionlabel}[1]{\hspace{\labelsep}{#1}}
\begin{prop}
\label{some_assembly_required}
Suppose we are given the following linearly independent generating sets:
\begin{enumerate}
    \item $\set{(M_{1,0}, P_{1,0}),\dotsc, (M_{k_0,0}, P_{k_0,0})}$ for $\widetilde\Omega_k^\Spin(B\Z_4)$.
    \item A splitting-compatible generating set for $\widetilde\Omega_k^\Spin(BS_4)$, with the following components:
    \begin{description}
        \item[First component] $\set{(M_{1,1}, P_{1,1}),\dotsc, (M_{k_1,1}, P_{k_1, 1})}$
        \item[Second component] $\set{(M_{1,2}, P_{1,2}),\dotsc, (M_{k_2,2}, P_{k_2, 2})}$, and
        \item[Third component] $\set{(M_{1,3}, P_{1,3}),\dotsc, (M_{k_3,3}, P_{k_3, 3})}$.
    \end{description}
    \item A splitting-compatible generating set for $\widetilde\Omega_k^\Spin(B\Z_4\wedge BS_4)$, with the following components:
    \begin{description}
        \item[First component] $\set{(M_{1,4}, P_{1,4}, P_{1,4}'),\dotsc, (M_{k_4,4}, P_{k_4, 4}, P_{k_4,4}')}$
        \item[Second component] $\set{(M_{1,5}, P_{1,5}, P_{1,5}'),\dotsc, (M_{k_5,5}, P_{k_5, 5}, P_{k_5, 5}')}$, and
        \item[Third component] $\set{(M_{1,6}, P_{1,6}, P_{1,6}'),\dotsc, (M_{k_6,6}, P_{k_6, 6}, P_{k_6,6}')}$.
    \end{description}
    \item $\set{(N_{1,0}, Q_{1,0}),\dotsc, (N_{\ell_0,0}, Q_{\ell_0,0})}$ for $\widetilde\Omega_k^\Spin(B\Z_3)$,
    \item $\set{(N_{1,1}, Q_{1,1}),\dotsc, (N_{\ell_1,1}, Q_{\ell_1,1})}$ for $\widetilde\Omega_k^\Spin(BD_6)$,
    \item $\set{(N_{1,2}, Q_{1,2}, Q_{1,2}'),\dotsc, (N_{\ell_2,2}, Q_{\ell_2,2}, Q_{\ell_2,2}')}$ for $\widetilde\Omega_k^\Spin(B\Z_3\wedge BD_6)$.
\end{enumerate}
Let $\tau_0$ denote a trivial bundle and $\rho_1$, $\rho_2$, $\rho_3$, and $\rho_4$ be as in \cref{rho12_defn,rho3_defn,psi1_done}. Then, the union of the following sets is a linearly independent generating set of $\widetilde\Omega_k^\Spin(B\SL(2, \Z)\times B\SL(3, \Z))$:
\begin{equation}
\begin{gathered}
    \set{(M_{i,0}, \rho_{2,2}(P_{i,0}), \tau_0)\mid 1\le i\le k_0}\\
    \set{(M_{i,2}, \tau_0, \rho_3(P_{i,2}))\text{\rm{} and } (M_{i,2}, \tau_0, \rho_4(P_{i,2}))\mid 1\le i\le k_2}\\
    \set{(M_{i,3}, \tau_0, \rho_3(P_{i,3}))\mid 1\le i\le k_3}\\
     \set{(M_{i,5}, \rho_{2,2}(P_{i,5}), \rho_3(P_{i,5}'))\text{\rm{} and } (M_{i,5}, \rho_{2,2}(P_{i,5}), \rho_4(P_{i,5}'))\mid 1\le i\le k_5}\\
    \set{(M_{i,6}, \rho_{2,2}(P_{i,6}), \rho_3(P_{i,6}'))\mid 1\le i\le k_6}\\
    \set{(N_{j,0}, \rho_{2,3}(Q_{j,0}), \tau_0)\mid 1\le j\le \ell_0}\\
    \set{(N_{j,1}, \tau_0, \rho_1(Q_{j,1}))\text{\rm{} and } (N_{j,1}, \tau_0, \rho_2(Q_{j,1}))\mid 1\le j\le \ell_1}\\
    \set{(N_{j,2}, \rho_{2,3}(Q_{j,2}), \rho_1(Q_{j,2}'))\text{\rm{} and } (N_{j,2}, \rho_{2,3}(Q_{j,2}), \rho_2(Q_{j,2}'))\mid 1\le j\le \ell_2}.
\end{gathered}
\end{equation}
\end{prop}
\begin{proof}
This proposition is essentially an explicit restatement of the stable splittings in \cref{distributivity,SL2_splitting_at_2,5_boring,SL3_stable_splitting_at_3,why_S4,psi1_done}.
\end{proof}
We use this to get the results in \cref{tab:gen1to4,tab:gen5to7}, the tables of generators of $\Omega_*^\Spin(B\SL(2, \Z)\times B\SL(3, \Z))$ in dimensions $7$ and below, from the lists of generators we found for the spin bordism of $B\Z_3$, $BD_6$, $B\Z_3\wedge BD_6$, $B\Z_4$, $BS_4$, and $B\Z_4\wedge BS_4$. The generators and relations for these spin bordism groups are given in \cref{BP_Z3_alg,prop:Z3gens,thm:spin_D6_module,3_wedge_mod,3_wedge_gens_1,3_wedge_gens_2,ringstr_Z4,Z4_mult_gens,algebra_koPSL,complete_L2_decomp,the_Z4_module_str,detected_A,torus_gen,Q5_gen,found_W6}, as well as \S\ref{Sec:D6gens}, \S\ref{ss:Z4S4gens}, \cref{tab:S4gens}, and~\eqref{L2_Z4}.
Here we also use \cref{better_matrices} to improve the embeddings $\rho_3$ and $\rho_4$ to ones which have easier string-theoretic interpretations.

\section{Our strategy: the K\"{u}nneth map}

At the only primes that matter for our computation, $2$ and $3$, we have similar-looking structure: a finite cyclic group $C$ and a finite non-Abelian group $G$ containing $C$, and we need to determine $\Omega_*^\Spin(BC)$, $\Omega_*^\Spin(BG)$, and $\widetilde\Omega_*^\Spin(BC\wedge BG)$. For $p =3$, $C = \Z_3$ and $G = D_6$ (\ref{Part:p=3}), and for $p = 2$, $C = \Z_4$ and $G = S_4$ (\ref{Part:p=2}).

In the degrees we are interested in, $\Omega_*^\Spin(BC)$ and $\Omega_*^\Spin(BG)$ are in the literature, and $\widetilde \Omega_*^\Spin(BC\wedge BG)$ is not: see~\cite[Example 7.3.2]{Bruner} for $\Omega_*^\Spin(B\Z_3)$, \cite[Example 7.3.3]{Bruner} for $\Omega_*^\Spin(B\Z_4)$, \cite[\S 14.1]{Heckman} for $\Omega_*^\Spin(BD_6)$ at $p = 3$, and Bayen~\cite[Chapter 3]{Bayen} for $\Omega_*^\Spin(BS_4)$ at $p = 2$.\footnote{Bayen states the result for $BD_8$, but also explains how it applies for $BS_4$.} We thus could have directly attacked $\widetilde\Omega_*^\Spin(BC\wedge BG)$, but the computation becomes a bit messy, especially when it comes time to find manifold representatives for a generating set of these bordism groups. We will instead do something different: there is extra structure present in this problem, and we will use it.
\begin{thm}
\label{Thm:Kunneth}
Let $C$ and $G$ be as above and $E_*$ be a multiplicative generalized homology theory.\footnote{Precisely, we want $E_*$ to be the generalized homology theory associated to a homotopy-commutative ring spectrum $E$. That is, the product on $E$-cohomology can be lifted to the point-set level, where it may or may not be commutative, but on cohomology it is commutative.}
\begin{enumerate}
    \item\label{KTpart1} The multiplication map $m\colon C\times C\to C$ induces the structure of a graded-commutative $E_*(\pt)$-algebra on $E_*(BC)$, whose multiplication is called the \term{Pontrjagin product}.
    \item The map $m$ promotes the K\"{u}nneth map $E_*(BC)\otimes E_*(BG)$ into the action map of an $E_*(BC)$-module structure on $E_*(BC\times BG)$.
    \item The map $m$ induces the structure of a spectral sequence of algebras on both the Adams and Atiyah-Hirzebruch spectral sequences computing $E_*(BC)$: each $E_r$-page has the structure of a $\Z^2$-graded-commutative algebra over the respective spectral sequence computing $E_*(\pt)$, and differentials satisfy the Leibniz rule
    \begin{equation}\label{leibniz}
        d_r(xy) = d_r(x)y + (-1)^{\abs x}x d_r(y)
    \end{equation}
    for $x$ and $y$ in the spectral sequence for $E_*(C)$ or $E_*(\pt)$.
    \item Likewise, $m$ refines the K\"{u}nneth map on the $E_r$-page of the Atiyah-Hirzebruch and Adams spectral sequences for $E_*(BC\times BG)$ to the structure of a $\Z^2$-graded module over the respective spectral sequence for $E_*(BC)$, with differentials again satisfying the Leibniz rule~\eqref{leibniz}.
    \item All of this is natural in $E_*$ with respect to morphisms of homotopy-commutative ring spectra.
\end{enumerate}
\end{thm}
\Cref{Thm:Kunneth} is a combination of some well-known theorems on the multiplicative structure of the Atiyah-Hirzebruch and Adams spectral sequences for a ring spectrum; here we use the fact that if $A$ is an abelian group, there is a model of $BA$ which is a topological abelian group, so that if $R$ is a commutative ring spectrum, $R\wedge (BA)_+$ is again a commutative ring spectrum, which is direct generalization of the fact that if $S$ is a commutative ring, $S[A]$ is again a commutative ring, with the multiplication extending the addition on $A$. For the multiplicative structure on these two spectral sequences, see Kochman~\cite[Theorem 3.6.8 and \S 4.2]{Koc96}.

$\Omega_*^\Spin$ is a multiplicative generalized cohomology theory, with multiplication induced from the direct product of manifolds. This implies the $\Omega_*^\Spin$-algebra structure on $\Omega_*^\Spin(BC)$ helps us find generators: if $M$ is a $k$-manifold representing a class $x\in\Omega_k^\Spin(BC)$ and $N$ is an $\ell$-manifold representing $y\in\Omega_\ell^\Spin(BC)$, then $M\times N$ represents $xy\in\Omega_{k+\ell}^\Spin(BC)$. The analogous fact is true if $y\in\Omega_\ell^\Spin(BC\times BG)$. This informs our plan of attack:
\begin{enumerate}
    \item Compute the ring structure on the Adams (if $p = 2$) or Atiyah-Hirzebruch (if $p = 3$) spectral sequence for $\Omega_*^\Spin(BC)$, and all differentials, in low degrees. Then use this to determine the ring structure on $\Omega_*^\Spin(BC)$ itself in low degrees.
    \item Run the same spectral sequence computing $\Omega_*^\Spin(BG)$.
    \item Combine these to determine the module structure on the spectral sequence for $\Omega_*^\Spin(BC\times BG)$ in low degrees, and use this spectral sequence to compute all differentials in low degrees.\footnote{We actually modify this strategy slightly: $\Omega_*^\Spin(BC\times BG)$ and the spectral sequence computing it split into the respective data for $BC$ and for $(BC_+)\wedge BG$. At this point in the calculation we already know $\Omega_*^\Spin(BC)$ in the degrees we need, so we focus on $(BC)_+\wedge BG$. We will elaborate on this in \cref{one_sided_smash}.}
    \item Use this to compute the $\Omega_*^\Spin(BC)$-module structure on $\Omega_*^\Spin(BC\times BG)$ in low degrees.
    \item Determine manifold representatives for the low-degree generators of $\Omega_*^\Spin(BC)$ as a ring and $\Omega_*^\Spin(BC\times BG)$ as an $\Omega_*^\Spin(BC)$-module; then take direct products of these generators to obtain additive generators for these bordism groups (i.e.\ what we need for applications to string theory).
\end{enumerate}

\subpart{\texorpdfstring{The calculation at $p = 3$}{The calculation at p = 3}}
\label{Part:p=3}

Now we compute $\Omega_*^\Spin(X)\otimes\Z_{(3)}$ for $X = B\Z_3$, $BD_6$, and $B\Z_3\wedge BD_6$ in low degrees and determine a list of generating manifolds. This plugs into \cref{some_assembly_required} to determine the $3$-torsion in the spin bordism of $B\SL(2, \Z)\times B\SL(3, \Z)$, which is part of~\cref{tab:gen1to4,tab:gen5to7}.

In \S\ref{s:Z3_D6}, we use the Atiyah-Hirzebruch spectral sequence to compute (the $3$-localizations of) $\Omega_*^\Spin(B\Z_3)$ and $\Omega_*^\Spin(BD_6)$ in the degrees relevant for string theory, then determine manifold representatives for generating sets of these bordism groups. A priori, this is not a new computation: the bordism groups were worked out in~\cite[\S 12.2, \S 14.1]{Heckman}, and the generators are known due to Rosenberg~\cite[Proof of Theorem 2.12]{Ros86}. However, we repeat this computation here because we need the $\Omega_*^\Spin(B\Z_3)$-module structure on $\Omega_*^\Spin(BD_6)$, which is new. We then use this module structure in \S\ref{s:smash_3} to compute the Atiyah-Hirzebruch spectral sequence for $\widetilde\Omega_*^\Spin(B\Z_3\wedge BD_6)$ and determine generating manifolds.

\textbf{Throughout this sub-part, we implicitly localize $\Omega_*^\Spin$ at 3.} This does not affect the results that we care about, which are $3$-torsion, but simplifies arguments: for example, it allows us to ignore the $2$-torsion in $\Omega_i^\Spin$ for $i = 1,2$.

\section{\texorpdfstring{Calculations and generators for $B\Z_3$ and $BD_6$}{Calculations and generators for BZ3 and BD6}}\label{s:Z3_D6}

\subsection{\texorpdfstring{Atiyah-Hirzebruch spectral sequence for $\Omega^{\Spin}_*(B\Z_3)$}{Atiyah-Hirzebruch spectral sequence for Spin bordism of BZ3}}
Recall that $H_*(B\Z_3;\Z)$ consists of a $\Z$ in degree $0$, $\Z_3$ in each odd positive degree, and $0$ in all other degrees. By the universal coefficient theorem, the same is true for $\Z_{(3)}$-homology, except with $\Z$ replaced with $\Z_{(3)}$.
\begin{lem}[Cartan~\cite{Car55}]
\label{trivial_pontrjagin}
The Pontrjagin product on $H_*(B\Z_3;\Z_{(3)})$ is trivial: there is a $\Z_{(3)}$-algebra isomorphism
\begin{equation}
    H_*(B\Z_3;\Z_{(3)})\cong \Z_{(3)}[y_1,y_3,y_5,\dotsc]/(3y_i, y_iy_j\text{ for all } i,j\ge 1).
\end{equation}
\end{lem}
Indeed, since all positive-degree classes have odd degree, the product of any two positive-degree classes has even degree and therefore vanishes.
\begin{cor}\label{Z3_AHSS}\hfill
\begin{enumerate}
    \item\label{tensor_ring} The $E^2$-page of the Atiyah-Hirzebruch spectral sequence for $\Omega^{\Spin}_*(B\Z_3)$ is isomorphic as algebras to
    \begin{equation}
        \Z_{(3)}[p_1,p_2,\dotsc, y_1,y_3, y_5,\dotsc]/(3y_i, y_iy_j\text{ for all } i,j\ge 1),
    \end{equation}
    with $p_i\in E^2_{0,4i}$ and $y_i\in E^2_{i,0}$.
    \item\label{Z3_collapse} This spectral sequence collapses on the $E^2$-page.
\end{enumerate}
\end{cor}
\begin{proof}
Part~\eqref{tensor_ring} is straightforward from the definition of the Atiyah-Hirzebruch spectral sequence, except for the ring structure. This follows from the description of the ring structure in \cref{trivial_pontrjagin} via the map $\mu\colon B\Z_3\times B\Z_3\to B\Z_3$ induced by the group operation, then using the map of Atiyah-Hirzebruch spectral sequences induced by $\mu$.

Part~\eqref{Z3_collapse} amounts to showing all differentials vanish. From the relation $y_iy_j = 0$, one learns that every homogeneous element on the $E^2$-page (with respect to the bigrading) is either a polynomial in the classes $p_i$ or $y_j$ times such a polynomial, and the former case occurs only in $E^2_{0,*}$. All differentials to or from $E^2_{0,*}$ vanish, because they can be transferred along the maps $\pt\to B\Z_3\to\pt$ to the Atiyah-Hirzebruch spectral sequence of a point, where they must vanish. The remaining classes, those of the form $y_j$ times a polynomial in the classes $p_i$, all have odd total degree; since differentials decrease total degree by $1$, all differentials must have domain or codomain a group in even total degree, and the only nonzero even-degree elements are those on the line $p = 0$, for which differentials vanish.
\end{proof}
Nevertheless, there are nontrivial extensions in this spectral sequence.
\begin{thm}\label{BP_Z3_alg}
\hfill
\begin{enumerate}
    \item There is an isomorphism of $\Omega^{\Spin}_*$-algebras
    \begin{equation}\label{BP_ring_Z3}
        \Omega^{\Spin}_*(B\Z_3)\cong \Omega^{\Spin}_*[\ell_1, \ell_3, \ell_5, \ell_7, \dots]/(\ell_i\ell_j, 3\ell_1, 3\ell_3, 3\ell_5 - v_1\ell_1, 3\ell_7 - v_1\ell_3, \dots)
    \end{equation}
    where $\abs{\ell_i} = i$ and all generators and relations not listed are in degrees greater than $7$.
    \item The isomorphism in~\eqref{BP_ring_Z3} may be chosen so that the image of $\ell_i$ in the $E_\infty$-page of the Atiyah-Hirzebruch spectral sequence is $y_i$.
\end{enumerate}
\end{thm}
Beware that the pattern suggested by~\eqref{BP_ring_Z3} does not continue in the way one might guess: $\ell_9$ has order $27$, for example.
\begin{proof} 
First we need to resolve extensions by $3$. To do so, we reduce from $3$-localized spin bordism to Brown-Peterson homology $\BP_*$ in the manner  described in~\cite[\S 10.5]{Heckman}; this determines $3$-local spin bordism as a graded Abelian group but does not determine the ring structure.

The graded Abelian group $\BP_*(B\Z_3)$ was computed by Bahri-Bendersky-Davis-Gilkey~\cite[Theorem 1.2(a)]{BBDG89}; see~\cite[\S 12.2]{Heckman} for an explicit description of these groups in low degrees. In particular, $\widetilde{\BP}_1(B\Z_3)\cong\Z_3$, $\widetilde{\BP}_3(B\Z_3)\cong\Z_3$, $\widetilde{\BP}_5(B\Z_3)\cong\Z_9$, and $\widetilde{\BP}_7(B\Z_3)\cong\Z_9$, and all other reduced $\BP$-homology groups of $\BP$ in degrees $7$ and below vanish.

The proof of the theorem then amounts to the observation that~\eqref{BP_ring_Z3} is the only ring structure compatible with both the ring structure on the $E^\infty$-page of the Atiyah-Hirzebruch spectral sequence coming from \cref{Z3_AHSS} and with the additive structure in the previous paragraph.
\end{proof}

\subsubsection{Generators}
\label{sec:Z3gens}
We now move to the generators of $\Omega^{\Spin}_*(B\Z_3)$. We begin by recalling lens spaces and some of their important properties. 

Given an integer $k>1$ and integers $j_1,\dots, j_n$ relatively prime to $k$, we define the lens space 
\begin{equation}
    L^{2n-1}_k(j_1,\dots, j_n) = S^{2n-1}/\Z_k
\end{equation}
to be the orbit space of the unit sphere $S^{2n-1} \subset \C^n$ with the action of $\Z_k$ generated by the rotation 
\begin{equation}
    (z_1,\dots, z_n)=(e^{2\pi i j_i/k}z_1, \dots, e^{2\pi ij_n/k}z_n).
\end{equation}
The lens spaces which will be of primary interest to us are  
\begin{equation}
    L^{2n-1}_k := L^{2n-1}_k(1,\dots, 1).
\end{equation}
Indeed, many of the generators for the spin bordism groups that we will encounter in this section, as well as those in proceeding sections, will be lens spaces. It is a standard exercise to determine which lens spaces admit spin structures (and that all of them are orientable).
\begin{lem}
The lens space $L^{2n-1}_k$ is orientable for all $k,n$. Furthermore, if $k$ is odd, then $L^{2n-1}_k$ admits a unique spin structure. If $k$ is even, then $L^{2n-1}_k$ admits a spin structure if and only if $n$ is even.
\end{lem}

With the above exposition in place, we are ready to state the generators of $\Omega^{\Spin}_*(B\Z_3)\otimes \Z_{(3)}$, which have already been considered in the literature.  

\begin{prop}[{\!\!\!\cite[\S 12.3]{Heckman}}]
\label{prop:Z3gens}
The isomorphism in \cref{BP_Z3_alg} may be chosen so that the lens space $L_3^{2n-1}$ with its canonical orientation and unique spin structure refining that orientation, together with the principal $\Z_3$-bundle $S^{2n-1}\to L_3^{2n-1}$, represents the class $\ell_{2n-1}\in\Omega^\Spin_{2n-1}(B\Z_3)$ for $n = 1,2,3,4$. The K3 surface with trivial $\Z_3$-bundle represents $v_1\in\Omega^\Spin_4(B\Z_3)$.
\end{prop}
Rosenberg~\cite[Proof of Theorem 2.12]{Ros86} showed that lens spaces generate $\Omega_*^\Spin(B\Z_3)$ as an $\Omega_*^\Spin$-module; these specific lens spaces were worked out in~\cite[\S 12.3]{Heckman} building on Rosenberg's result.

\subsection{\texorpdfstring{Atiyah-Hirzebruch spectral sequence for $\Omega^{\Spin}_*(BD_6)$}{Atiyah-Hirzebruch spectral sequence for Spin bordism of BD6}}

The cohomology of $H^*(BD_{2n};\Z)$ was computed by Handel in \cite{Han93}, and his result together with the universal coefficient theorem tells us $H_*(BD_6;\Z_{(3)})$. To summarize: 
\begin{equation}
\label{eqn:D6Hom}
    H_k(BD_6;\Z_{(3)})\cong \begin{cases} 
      \Z_{(3)}, & k=0 \\
      \Z_3, & k\equiv 3 \mod 4 \\
      0, & \text{otherwise.} 
   \end{cases} 
\end{equation}
We are then enabled to compute $\Omega^{\Spin}_*(BD_6)$ using the Atiyah-Hirzebruch spectral sequence. This  has been previously considered in \cite{Heckman}.

\begin{lem}[\!\!{\cite[Theorem 14.3]{Heckman}}]
\label{lem:AhssD6}
\hfill
\begin{enumerate}
      \item The Atiyah-Hirzebruch  spectral sequence computing $\Omega^{\Spin}_*(BD_6)$ collapses on the $E^2$-page.
    \item In total degree less than 8, the $E^2$-page is generated by the classes $v_1\in E^2_{0,4}$, $z_3\in E^2_{3,0}$, $z_7 \in E^2_{7,0}$, and $p\in E^2_{3,4}$.  
\end{enumerate}
\end{lem}

\begin{thm}
\label{thm:spin_D6_module}
\hfill
\begin{enumerate}
\item There is an isomorphism of $\Omega^{\Spin}_*$-modules 
\begin{equation}
    \Omega^{\Spin}_*(BD_6) \cong \Omega^{\Spin}_*\{k_3,k_7,
    \dots\}/(3k_3, 3k_7-v_1k_3,\dots),
\end{equation}
where $|k_i| = i$ and all generators and relations not listed are in degrees greater than 7. 
\item The isomorphism may be chosen so that the image of $k_i$ in the $E^\infty$-page of the Atiyah-Hirzebruch spectral sequence is $z_i$. 
\end{enumerate}
\end{thm}
\begin{proof}
    The theorem follows from \cref{lem:AhssD6} and the fact that $\widetilde{\Omega}^{\Spin}_3(BD_6)\cong \Z_3$ and $\widetilde{\Omega}^{\Spin}_7(BD_6)\cong \Z_9$, see \cite[Theorem 14.3]{Heckman} for reference. 
\end{proof}

\subsubsection{Generators}
\label{Sec:D6gens}
The generators of $\Omega^{\Spin}_*(BD_6)$ are the same as those for \cref{prop:Z3gens}. Indeed, surjectivity of the map $\Omega^{\Spin}_*(B\Z_3)\rightarrow \Omega^{\Spin}_*(BD_6)$ in the range we are interested in implies that the generators for $\Omega^{\Spin}_k(BD_6)$, $k=3,7$,  can be chosen to be the generators we found for $\Omega^{\Spin}_{k}(B\Z_3) $ in \cref{prop:Z3gens}. 

\section{\texorpdfstring{Calculations and generators for $B\Z_3\times BD_6$}{Calculations and generators for BZ3 x BD6}}\label{s:smash_3}

\subsection{\texorpdfstring{$\Omega^\Spin_*(B\Z_3\times BD_6)$ as a  $\Omega^{\Spin}_*(B\Z_3)$-module}{Spin bordism of BZ3 x BD6 as a (Spin bordism of BZ3) module}}
\label{ss:Z3D3}

In this subsection, we determine the $\Omega^\Spin_*(B\Z_3)$-module structure on $\widetilde{\Omega}^{\Spin}_*(B\Z_3\times BD_6)$ up to degree seven.
\begin{rem}\label{one_sided_smash}
As an $\Omega_*^\Spin(B\Z_3)$-module, $\Omega_*^\Spin(B\Z_3\times BD_6)$ splits as a sum of $\Omega_*^\Spin(B\Z_3)$ and $\Omega_*^\Spin((B\Z_3)_+\wedge BD_6)$; the module structure on Atiyah-Hirzebruch spectral sequences also splits in this way. Therefore in this section we will compute the latter summand and then add back in $\Omega_*^\Spin(B\Z_3)$ at the end.

In the stable splitting of the product that we discussed in \cref{distributivity}, $B\Z_3\times BD_6$ splits into $\pt$, $B\Z_3$, $BD_6$, and $B\Z_3\wedge BD_6$; we keep the latter two pieces.
\end{rem}

We begin with the Atiyah-Hirzebruch spectral sequence computing $\widetilde{\Omega}^\Spin_*((B\Z_3)_+\wedge BD_6)$. Recall the homology of $H_*(B\Z_3;\Z_{(3)})$ and $H_*(BD_6;\Z_{(3)})$ given in \cref{trivial_pontrjagin} and \eqref{eqn:D6Hom}. The $E^2$-page of the Atiyah-Hirzebruch spectral sequence computing $\widetilde{\Omega}^\Spin_*((B\Z_3)_+\wedge BD_6)$ can then be computed using the K\"{u}nneth theorem; see \cref{Z3D6AHss} for reference. For the total degree we are interested in, less than 8, all differentials vanish by degree reasons. Indeed, in the homological Atiyah-Hirzebruch spectral sequence, differentials go up and to the left, and decrease total degree by 1. Furthermore, the Atiyah-Hirzebruch spectral sequence computing $\widetilde{\Omega}^\Spin_*(BD_6)$ collapses on the $E^2$-page. 

\begin{figure}[H]
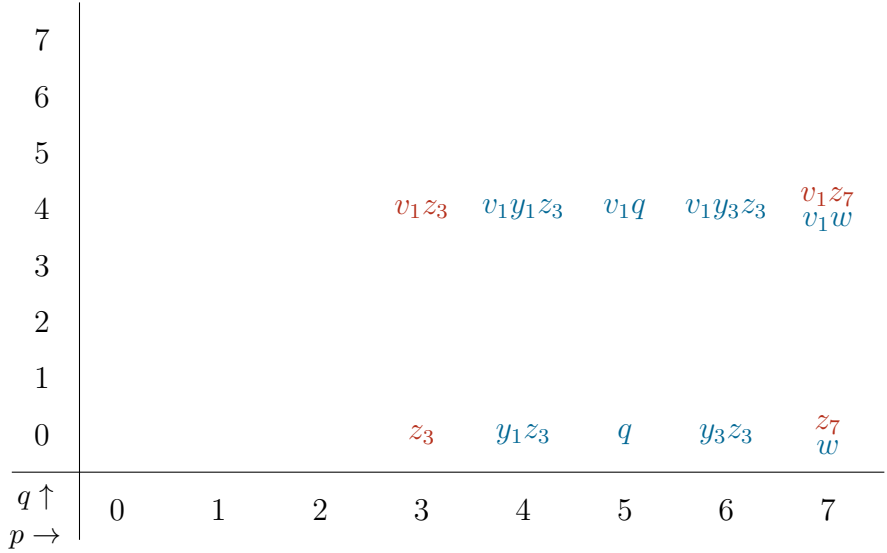

\centering
\begin{sseqdata}[name=Z3D6AHSS, page=2, homological Serre grading, xrange={0}{7}, yrange={0}{7},
classes={draw=none}, xscale=0.75, yscale=0.75,
x label = {$\displaystyle{q\uparrow \atop p\rightarrow}$}, xscale=1.8,
x axis extend end=0.75cm,
x label style = {font = \small, xshift = -32ex, yshift= 3 ex}
]
\begin{scope}[MidnightBlue]
        \class["y_1z_3"](4, 0)
        \class["q"](5, 0)
        \class["y_3z_3"](6, 0)
  %      \class["w"](7, 0)

         \class["v_1y_1z_3"](4, 4)
        \class["v_1q"](5, 4)
        \class["v_1y_3z_3"](6, 4)
  %      \class["v_1w"](7, 4)
\end{scope}
\begin{scope}[BrickRed]
    \class["z_3"](3, 0)
  %  \class["z_7"](7, 0)
    \class["v_1z_3"](3, 4)
 %   \class["v_1z_7"](7, 4)
\end{scope}
\class["{\displaystyle\textcolor{BrickRed}{z_7}\atop\displaystyle\textcolor{MidnightBlue}{w}}"](7, 0)
%\class["\textcolor{BrickRed}{v_1z_7}{,}\textcolor{MidnightBlue}{v_1w}"](7, 4)
\class["{\displaystyle\textcolor{BrickRed}{v_1z_7}\atop\displaystyle\textcolor{MidnightBlue}{v_1w}}"](7, 4)
\end{sseqdata}
\printpage[name=Z3D6AHSS, page=2]
\caption{\label{Z3D6AHss} The $E^2$-page of the Atiyah-Hirzebruch spectral sequence computing $\widetilde\Omega_*^{\Spin}((B\Z_3)_+\wedge BD_6)$. Each named class generates a $\Z_3$.}
\label{AHSS_D6}
\end{figure}

\begin{thm}
\label{Thm:Z3D6E_2}
    In the range $p+q \leq 7$, the Atiyah-Hirzebruch spectral sequence computing $\widetilde{\Omega}^\Spin_*((B\Z_3)_+\wedge BD_6)$ collapses on the $E^2$-page. As a module over $\mathrm{AHSS}(B\Z_3)$,
    \begin{equation}
        \mathrm{AHSS}((B\Z_3)_+\wedge BD_6)\cong \Z_3\{z_3, q, z_7, w, \dots\}/(y_1q = \lambda y_3z_3, \dotsc)
    \end{equation}
    for some $\lambda\in\Z_3$, with all unwritten generators and relations in total degree $8$ and above.
\end{thm}
We will never need to know the precise value of $\lambda$, so do not compute it.
\begin{proof}
A straightforward application of the K\"{u}nneth map (\cref{Thm:Kunneth}) reveals that the classes $y_1 z_3$ and $y_3z_3$ generate $E^2_{4,0}$ and $E^2_{6,0}$, respectively. However, the K\"{u}nneth map misses the generator in degrees five and seven, leaving us with `new' classes $q$ and $w$.     
\end{proof}
We draw this Atiyah-Hirzebruch spectral sequence in \cref{AHSS_D6}.

\cref{BP_Z3_alg,thm:spin_D6_module,Thm:Z3D6E_2} allow us to conclude the $\Omega^\Spin_*(B\Z_3)$-module structure of $\Omega^\Spin_*((B\Z_3)_+ \wedge BD_6)$. 

\begin{thm} \label{3_wedge_mod}
\hfill
\begin{enumerate}
\item 
There is an isomorphism of $\Omega^\Spin_*(B\Z_3)$-modules:
\begin{equation}
    \Omega^{\Spin}_*((B\Z_3)_+ \wedge BD_6)\cong \Omega^{\Spin}_*\{k_3,k_7,n,m\dots\}/%\mathcal{R},
%\end{equation}
%where
%\begin{equation}
%    \mathcal{R} = %(l_il_j, 3l_1, 3l_3, 3l_5-v_1l_1,3l_7-%v_1l_3,
    (3k_3,3k_7-v_1k_3, 3n,3m, \ell_1n - \lambda'\ell_3 k_3, \dots),
\end{equation}
for some $\lambda'\in\Z_3$, and all generators and relations not listed have topological degree greater than 7. 
\item 
The isomorphism can be chosen so that the images of the generators on the $E^\infty$ page of the $E^\infty$-page of the Atiyah-Hirzebruch spectral sequence are: $k_i \mapsto z_i$, $n\mapsto q$ and $m \mapsto w$.
\end{enumerate}
\end{thm}
\subsection{Generators}
The stable splitting in \cref{distributivity}, as well as \cref{prop:Z3gens} and the discussion in \cref{Sec:D6gens}, reveal that the only generators for $\Omega^\Spin_*(B\Z_3\times BD_6)$ we have yet to consider are those for $\Omega^\Spin_*(B\Z_3 \wedge BD_6)$. Fortunately, there are only four such generators and two are in the image of the K\"{u}nneth map. 
\begin{prop}\label{3_wedge_gens_1}
\hfill
\begin{enumerate}
\item
    $L^1_3 \times L^3_3$ with its canonical orientation and unique spin structure, where $L^1_3$ carries the principal $\Z_3$-bundle $S^1 \rightarrow L^1_3$ and $L_3^3$ carries the principal $D_6$-bundle induced from $\Z_3\hookrightarrow D_6$,  represents the class $\ell_1k_3 \in \Omega^\Spin_{4}(B\Z_3 \wedge BD_6)$.
    \item $L^3_3 \times L^3_3$ with its canonical orientation and unique spin structure, where the first $L^3_3$ carries the principal $\Z_3$-bundle $S^3\rightarrow L_3^3$ and the right $L_3^3$ carries the principal $D_6$-bundle induced from $\Z_3 \hookrightarrow D_6$, represents the class $\ell_3k_3 \in \Omega^\Spin_6(B\Z_3 \wedge BD_6)$.  
\end{enumerate}
\end{prop}

The classes $n\in \Omega^\Spin_5(B\Z_3\wedge BD_6)$ and $m \in \Omega^{\Spin}_7(B\Z_3 \wedge BD_6)$ are missed by the K\"{u}nneth map and likewise are not represented as products of generators for $\Omega^\Spin_*(B\Z_3)$ and $\Omega^\Spin_*(BD_6)$.

\begin{prop}\label{3_wedge_gens_2}
Let $\Phi\colon \Z_3 \rightarrow \Z_3 \times \Z_3$ be the diagonal map. $L_3^5$ and $L_3^7$ with their canonical orientation and unique spin structure, together with the $\Z_3 \times D_6$ bundle given by 
    \begin{equation}
    \label{eqn:Z3D6bundle}
        \Z_3 \xrightarrow{\Phi}\Z_3 \times \Z_3 \hookrightarrow \Z_3 \times D_6,
    \end{equation}
    represent the classes $n \in \Omega^\Spin_5(B\Z_3 \times BD_6)$ and $m\in \Omega^\Spin_7(B\Z_3 \times BD_6)$, respectively.
\end{prop}
\begin{proof}
    Botvinnik and Gilkey proved in \cite{Botvinnik} that the bordism classes of $\Omega^{\Spin}_*(B\Z_p \times B\Z_p)$ are represented by classifying maps $L^{2n_1+1}_p\times L^{2n_2+1}_p\rightarrow B(\Z_p)^2$ or by the compositions $L^{2m+1}_p\rightarrow B\Z_p \xrightarrow{B\phi} B(\Z_p)^2$ for the various group homomorphisms $\phi:\Z_p\rightarrow \Z_p\times \Z_p$. The theorem then follows from the fact that the inclusion 
    \begin{equation}
    B\Z_3 \wedge B\Z_3 \hookrightarrow B\Z_3 \wedge BD_6    
    \end{equation}
    admits a section $3$-locally using the transfer $\Sigma^\infty_+ BD_6\rightarrow \Sigma^\infty_+ B\Z_3$, so that $\Omega^\Spin_*(B\Z_3\wedge BD_6)$ is ($3$-locally) a summand of $\Omega^\Spin_*(B\Z_3 \wedge B\Z_3)$.
\end{proof}
Passing these generators through \cref{some_assembly_required}, we obtain the remaining $3$-torsion in \cref{tab:gen1to4,tab:gen5to7}.

\subpart{\texorpdfstring{The calculation at $p = 2$}{The calculation at p = 2}}
\label{Part:p=2}

\section{\texorpdfstring{Generalities for the $p = 2$ computation}{Generalities for the p = 2 computation}}

The computation at $p = 2$ has a similar overall structure to the computation at $p = 3$, but the details are more complicated.
We will first find the ring structure on $\Omega_*^\Spin(B\Z_4)$ in the degrees relevant for us, then use it to determine $\Omega_*^\Spin(B\SL(3, \F_2))$ and $\Omega_*^\Spin(B\Z_4\times B\SL(3, \F_2))$, so that, as we went over in \S\ref{ss:assembly}, we can obtain a complete generating set of $\Omega_*^\Spin(B(\SL(2, \Z)\times\SL(3, \Z))$ in dimensions $7$ and below. The action of $\Omega_*^\Spin(B\Z_4)$ on other bordism groups will play a significant role in systematically discovering many generators.

The first simplification we make is to replace spin bordism with another generalized cohomology theory called \term{connective real $K$-theory}, denoted $\ko$. This is because of a result of Anderson-Brown-Peterson~\cite{ABP67} that implies for any space or connective spectrum $X$, the Atiyah-Bott-Shapiro~\cite{ABS64} map
\begin{equation}
    \widehat A\colon \Omega_k^\Spin(X)\longrightarrow \ko_k(X)
\end{equation}
is an isomorphism for $k\le 7$.\footnote{In higher degrees, there is a more complicated decomposition due to Anderson-Brown-Peterson~\cite{ABP67}.} The theory $\ko$ has better algebraic properties, so we will focus on computing $\ko$-homology, then return to spin bordism as we search for generators.

Another difference between the $p= 3$ and $p = 2$ cases is that because $\Omega_*^\Spin$ and $\ko_*$ both have $2$-torsion, the Atiyah-Hirzebruch spectral sequence is harder to solve compared to a version of the Adams spectral sequence, so we will use the latter.
\begin{thm}\label{its_Adams}
Let $X$ be a space or spectrum homotopy equivalent to a CW complex with with finitely many cells in each dimension. Then there is a spectral sequence
\begin{equation}\label{Adamssig}
    E_2^{s,t} = \Ext_{\cA(1)}^{s,t}(H^*(X;\Z_2), \Z_2) \Longrightarrow \ko_*(X)\otimes\hat\Z_2.
\end{equation}
\end{thm}
Here $\cA(1)$ is a subalgebra of the Steenrod algebra $\cA$ of cohomology operations: $\cA(1)$ is generated by the Steenrod squares $\Sq^1$ and $\Sq^2$. Tensoring with $\hat\Z_2$ is the process of $2$-completion.

We will be brief -- there are plenty of papers in the mathematical physics literature using this technique, and the reader wishing to dig into the definitions and examples of these ingredients is encouraged to read~\cite{Beaudry,Heckman}, both written with physics audiences in mind.
\begin{itemize}
    \item $\cA$ refers to the \term{mod $2$ Steenrod algebra}, the $\Z$-graded $\Z_2$-algebra of all natural transformations $H^*(\bl;\Z_2)\to H^{*+k}(\bl;\Z_2)$ which commute with the suspension isomorphism. This is generated by \term{Steenrod squares} $\Sq^n$ of degree $n$, modulo some relations. Then $\cA(1)$ is the subalgebra generated by $\Sq^1$ and $\Sq^2$.
    \item $\Ext$ is a functor which can be defined in terms of certain equivalence classes of extensions of modules.
    \item $\hat\Z_2$ denotes the $2$-adic numbers. For a finitely generated abelian group, tensoring with $\hat\Z_2$ retains the same information as tensoring with $\Z_{(2)}$, so we will typically be implicit about the appearance of the $2$-adics.
\end{itemize}
\begin{rem}
The Adams spectral sequence was introduced by Adams~\cite{Ada58}, in a considerably more general form than we gave in \cref{its_Adams}; we presented only what we need in this paper. The application to $\ko$-homology by working over $\cA(1)$ first appears in work of Anderson-Brown-Peterson~\cite{ABP69} and and Giambalvo~\cite{Gia73a, Gia73b, Gia76}. The appearance of the Adams spectral sequence in physically motivated computations is a more recent phenomenon: though there were some earlier works such as Rose~\cite{Ros88}, Hill~\cite{Hil08}, and Francis~\cite{Fra11}, the technique began in earnest following its use by Freed-Hopkins~\cite[\S 10]{FH21} in applications to invertible field theories.
\end{rem}
For any space $X$, the $\ko_*(\pt)$-action\footnote{Passing back to spin bordism, this action represents taking the product with a spin manifold with trivial map to $X$.} on $\ko_*(X)$ lifts to an action of the Adams spectral sequence for $\ko_*(\pt)$ on the Adams spectral sequence for $\ko_*(X)$:
\begin{thm}\label{E_converge}
Let $\ExZ\coloneqq \Ext_{\cA(1)}(\Z_2, \Z_2)$, which by \cref{its_Adams} is the $E_2$-page of the Adams spectral sequence computing $\ko_*(\pt)$.
\begin{enumerate}
    \item There is a $\Z^2$-graded commutative $\Z_2$-algebra structure on $\ExZ$ converging to the algebra structure on $\ko_*(\pt)$.
    \item There is a natural $\ExZ$-action on the $E_2$-page of the Adams spectral sequence for $\ko_*(X)$ for any space $X$ which converges to the $\ko_*(\pt)$-action on $\ko_*(X)$. All differentials commute with this $\ExZ$-action.\footnote{Most of this theorem works with $\ko$ replaced with an arbitrary commutative ring spectrum, but in general differentials only obey a Leibniz rule --- the stronger result here is because the Adams spectral sequence for $\ko_*(\pt)$ collapses. This distinction matters if $\ko$ is replaced with $\mathit{tmf}$: see~\cite{BR21}.}
\end{enumerate}
\end{thm}
This is a consequence of the general multiplicative structure of the Adams spectral sequence as discussed in~\cite[\S 3.6]{Koc96}.

Now we need to know what $\ExZ$ is.
\begin{thm}[{Liulevicius~\cite[Theorem 3]{Liulevicius}}]
The $\Z_2$-algebra structure on $\ExZ$ described in \cref{E_converge} is isomorphic to
\begin{equation}
    \Ext_{\cA(1)}(\Z_2,\Z_2)\cong \Z_2[h_0,h_1,v,w]/(h_0h_1, h_1^3, vh_1, h_0^2w-v^2),
\end{equation}
where $\deg(h_0) = (1,1)$, $\deg(h_1) =(1,2)$, $\deg(v) = (3,7)$, and $\deg(w) = (4,12)$.
\end{thm}
The four generators lift to represent classes in the $\ko$-homology and spin bordism of a point.
\begin{itemize}
    \item $h_0$ lifts to $2\in\ko_0(\pt)$ and to the class of $\pt_+\amalg\pt_+$ in $\Omega_0^\Spin$. Thus keeping track of $h_0$-actions in an Adams spectral sequence is very useful --- they provide information about multiplication by $2$ in $\ko$-homology. This is one of the competitive advantages of the Adams spectral sequence over the Atiyah-Hirzebruch spectral sequence; extension problems are usually harder for the latter technique.
    \item $h_1$ lifts to $\eta\in\ko_1(\pt)$ and $S_+^1\in\Omega_1^\Spin$.
    \item $v$ lifts to the K3 surface in $\Omega_4^\Spin\cong\Z$.
    \item $w$ lifts to the Bott class in $\ko_8(\pt)$ and the Bott manifold in $\Omega_8^\Spin$.
\end{itemize}

The bottom row of the Adams spectral sequence is nicely characterized  using mod $2$ cohomology, which we will repeatedly use to find generators. There is a natural isomorphism  
$\varrho\colon \Ext_{\cA(1)}^{0,t}(M, N)\overset\cong\to \Hom_{\cA(1)}(\Sigma^t M, N)$; using this, one can define an ``edge homomorphism''
\begin{equation}\label{edge_hom}
    \Upsilon\colon E_\infty^{0,t}\hookrightarrow E_2^{0,t} = \Ext_{\cA(1)}^{0,t}(H^*(X; \Z_2), \Z_2) \overset\varrho\longrightarrow \Hom_{\cA(1)}(H^{*+t}(X;\Z_2), \Z_2).
\end{equation}
\begin{prop}\label{detector}
Let $c\in H^t(X;\Z_2)$ and $M$ be a closed, $t$-dimensional spin manifold with a map $f\colon M\to X$. Let $m\in E_\infty^{0,t}$ denote the image of $[M, f]\in\Omega_t^\Spin(X)$ in the $E_\infty$-page of the Adams spectral sequence. If $\Upsilon(m)\colon H^t(X;\Z_2)\to\Z_2$ evaluated on $c$ is nonzero, then $\int_M f^*(c)\ne 0$.
\end{prop}
In the case $\Upsilon(m)(c)\ne 0$, we say that $[M, f]$ is \term{detected by the mod $2$ cohomology class $c$}. See~\cite[\S 8.4]{FH19b} and~\cite[\S 3.3]{Debray} for a more detailed discussion of \cref{detector} and Bunke~\cite{Bun17} for an analogue of this result for $s = 1$.

\section{\texorpdfstring{Calculations and generators for $B\Z_4$}{Calculations and generators for BZ4}}

\subsection{\texorpdfstring{Ring structure on the Adams spectral sequence for $\ko_*(B\Z_4)$}{Ring structure on the Adams spectral sequence for ko(BZ4)}}
\label{sec:BZ4}

In this subsection, we describe all classes, products (\cref{Z4_Adams_ring}), and differentials (\cref{ko_Z4_diffs}) in the Adams spectral sequence for $\ko_*(B\Z_{2^n})$. We only need the case $n = 2$ in degrees $8$ and below, but we found that the argument and general structure are cleaner in this generality.

To calculate the ring structure on the Adams spectral sequence computing $\ko_*(B\Z_{2^n})$, we compare it with two simpler Adams spectral sequences computing ordinary homology of $B\Z_{2^n}$. There the ring structures converge to the Pontrjagin product (the product on $\ko_*(B\Z_{2^n})$ induced by the group operation in $\Z_{2^n}$; see \cref{Thm:Kunneth}, part~\ref{KTpart1}) on the homology groups. Thus we can work backwards: the Pontrjagin product is known, so we can work out the ring structure on the $E_2$-pages of the two simpler spectral sequences, and then compare with the spectral sequence we are actually interested in.

For $n\ge -1$, let $\cA(0)$ denote the subalgebra of $\cA(1)$ generated by $\Sq^i$ for $i\le 2^n$; thus $\cA(-1)\cong\Z_2$ and $\cA(0) = \ang{\Sq^1} \cong \Z_2[\Sq^1]/(\Sq^1)^2$. These algebras govern Adams spectral sequences for ordinary cohomology that are analogues of \cref{its_Adams}.
\begin{prop}\label{A0_Adams}
Let $X$ be as in \cref{its_Adams}. Then there are spectral sequences
\begin{subequations}
\begin{align}\label{A0_Adams_sig}
    E_2^{s,t} &= \Ext_{\cA(0)}^{s,t}(H^*(X;\Z_2), \Z_2) \Longrightarrow H_*(X;\Z)\otimes\hat\Z_2\\
    E_2^{s,t} &= \Ext_{\cA(-1)}^{s,t}(H^*(X;\Z_2), \Z_2) \Longrightarrow H_*(X;\Z_2).
    \label{A_minus_one}
\end{align}
\end{subequations}
The inclusions $\cA(-1)\to\cA(0)\to\cA(1)$ induce maps of spectral sequences in the other direction converging to the ($2$-completions of the) usual maps $\ko_*(X)\to H_*(X;\Z)\to H_*(X;\Z_2)$.

In addition, if $X$ is the classifying space of an abelian Lie group, there are ring structures on the pages of these spectral sequences which converge to the Pontrjagin product on the homology of $X$, and the change-of-Ext maps from $\cA(-1)$ to $\cA(0)$ to $\cA(1)$ are ring homomorphisms.
\end{prop}
As $\cA(-1) = \Z_2$ is a field, $\Ext_{\cA(-1)}^{s,*}$ is $\Hom_{\Z_2}$ for $s = 0$ and vanishes for $s > 0$. Therefore the spectral sequence~\eqref{A_minus_one} is trivial: it always collapses without differentials or extension problems, and its $E_2$-page is exactly what it purports to compute. Oddly enough, this makes it useful for us!

The next step over $\cA(-1)$ is to recall the Pontrjagin product on $H_*(B\Z_{2^n};\Z_2)$. While we're here, we also recall a few important facts about the cohomology of $B\Z_{2^n}$. As always in this section, $n >1$.
\begin{thm}[{\!\!\cite[Proposition 4.5.1]{Carlson}}]\label{Thm:CohoZ4}
\label{Z4_mod2_coh}
    \begin{equation}
    \begin{gathered}
        H^*(B\Z_{2^n};\Z_2) \cong \Z_2[x,y]/(x^2), \\
        \abs x = 1,\quad \abs y= 2
    \end{gathered}
    \end{equation}
    The Steenrod squares of the generators are 
    \begin{equation}
        \Sq(x) = x + x^2,\quad \Sq(y) = y + y^2.
    \end{equation}
\end{thm}
\begin{thm}[Cartan~\cite{Car55}]
As a ring with the Pontrjagin product,
\begin{equation}
    H_*(B\Z_{2^n};\Z_2) \cong \Z_2[\overline x, \overline y_i: i\ge 1]/\paren{\overline x^2, \overline y_i \overline y_j = \binom{i+j}{i} \overline y_{i+j}}
\end{equation}
with $\abs{\overline x} = 1$ and $\abs{\overline y_i} = 2i$. The classes $\overline x$, resp.\ $\overline y_i$ are dual to the mod $2$ cohomology classes $x$, resp.\ $y^i$.
\end{thm}
Cartan does not write this ring structure explicitly; see Brown~\cite[Theorem V.6.6]{Bro82} for that.
\begin{cor}\label{minus_one_cor}
The ring structure on the $E_2$-page of the $\cA(-1)$-based Adams spectral sequence~\eqref{A_minus_one} computing $H_*(B\Z_{2^n};\Z_2)$ is
\begin{equation}
    E_2^{*,*} \cong \Z_2[\overline x, \overline y_i: i\ge 1]/\paren{\overline x^2, \overline y_i \overline y_j = \binom{i+j}{i} \overline y_{i+j}}
\end{equation}
with $\overline x\in\Ext^{0,1}$ and $\overline y_i\in\Ext^{0,2i}$.
\end{cor}
This is because, as noted above, the $s = 0$ part of the $E_2$-page and the mod $2$ homology it computes are isomorphic and the spectral sequence collapses.

Now we lift to $\cA(0)$. 
We are not sure who first made the following computation; references include~\cite[\S 2]{Liulevicius} and~\cite[Example
4.5.5]{Beaudry}.
\begin{lem}
\label{Z2_A1_ext}
There is a unique isomorphism of $\Z_2$-algebras
\begin{equation}
        \Ext_{\cA(0)}^{*,*}(\Z_2, \Z_2) \overset\cong\longrightarrow \Z_2[h_0],
\end{equation}
with $h_0\in\Ext^{1,1}$.
\end{lem}
As for $\cA(1)$, $h_0$ detects multiplication by $2$. We will describe the $E_2$-page of~\eqref{A0_Adams_sig}, the Adams spectral sequence over $\cA(0)$, for $X = B\Z_{2^n}$ first as a $\Z_2[h_0]$-module, then as an $\Z_2[h_0]$-algebra.

\begin{lem}\label{A0_E2}
There is an isomorphism of $\Z_2[h_0]$-modules
\begin{equation}
    \Ext_{\cA(0)}(H^*(B\Z_{2^n}; \Z_2), \Z_2) \overset\cong\longrightarrow \Z_2[h_0]\set{\widetilde y^j, \widetilde x\widetilde y^j: j\ge 0}
\end{equation}
where $\widetilde y_j\in\Ext^{0,2j}$ and $\widetilde x\widetilde y^j\in\Ext^{0,2j+1}$. Under the homomorphism
\begin{equation}\label{zero_to_minus_one_redux}
    \Ext_{\cA(0)}^{0,t}(H^*(B\Z_{2^n};\Z_2), \Z_2) = \Hom_{\cA(0)}(H^t(B\Z_{2^n}; \Z_2), \Z_2)\hookrightarrow \Hom_{\Z_2}(H^t(B\Z_{2^n}; \Z_2), \Z_2),
\end{equation}
each class $\widetilde c$ is identified with the dual of $c\in H^t(B\Z_{2^n};\Z_2)$.
\end{lem}
\begin{proof}
\Cref{Z4_mod2_coh} tells us $\Sq^1$ acts trivially on $H^*(B\Z_{2^n};\Z_2)$, and therefore the entire algebra
$\cA(0)$ acts trivially. Therefore $\Ext_{\cA(0)}(H^*(B\Z_{2^n};\Z_2), \Z_2)$ is, as a $\Z_2[h_0]$-module, a direct
sum of copies of $\Ext_{\cA(0)}(\Z_2, \Z_2)$ indexed by a basis of $H^*(B\Z_{2^n};\Z_2)$. This module is generated by
its classes on the line $s = 0$, which are
\begin{equation}
        \Ext^0_{\cA(0)}(H^*(B\Z_{2^n};\Z_2), \Z_2) \overset\cong\longrightarrow \Hom_{\cA(0)}(H^*(B\Z_{2^n};\Z_2), \Z_2)
        \overset\cong\rightarrow \Hom_{\Z_2}(H^*(B\Z_{2^n};\Z_2),\Z_2),
\end{equation}
because $\cA(0)$ acts trivially.
\end{proof}
\begin{cor}\label{A0_module_A_minus_one}
Under the isomorphisms in \cref{minus_one_cor,A0_E2}, the map
\begin{equation}\label{zero_to_minus_one}
    \Ext_{\cA(0)}^{*,*}(H^*(B\Z_{2^n};\Z_2), \Z_2) \longrightarrow
    \Ext_{\cA(-1)}^{*,*}(H^*(B\Z_{2^n};\Z_2), \Z_2)
\end{equation}
induced by the inclusion $\cA(-1)\to\cA(0)$ sends $\widetilde x\mapsto \overline x$, $\widetilde y_j\mapsto\overline y_j$, and $\widetilde x\widetilde y_j\mapsto \overline x\cdot\overline y_j$.
\end{cor}
This is because, under the identification $\Ext_{\cA(-1)} = \Hom_{\Z_2}$, the map~\eqref{zero_to_minus_one} is identified with the map~\eqref{zero_to_minus_one_redux}.
\begin{prop}\label{A0_ring_str}
There is an isomorphism of $\Z_2[h_0]$-algebras
\begin{equation}
    \Ext_{\cA(0)}^{*,*}(H^*(B\Z_{2^n};\Z_2), \Z_2) \overset\cong\longrightarrow
    \Z_2[h_0, \widetilde x, \widetilde y_j: j\ge 1]/\paren{
        \widetilde x^2, \widetilde y_i \widetilde y_j = \binom{i+j}{i}\widetilde y_{i+j}
    }.
\end{equation}
\end{prop}
\begin{proof}
\Cref{A0_Adams} asserts~\eqref{zero_to_minus_one} is a ring homomorphism, and \cref{A0_module_A_minus_one} implies that~\eqref{zero_to_minus_one} is a vector space isomorphism restricted to the line $s = 0$. This uniquely forces the products of all classes on the line $s = 0$, i.e.\ of all $\Z_2[h_0]$-module generators except for $h_0$. The products of $h_0$ with the other generators were already told to us as part of the $\Z_2[h_0]$-module structure in \cref{A0_E2}.
\end{proof}

Now that we have the ring structures over $\cA(-1)$ and $\cA(0)$, we return to $\cA(1)$. First we want to know Ext as an $\ExZ$-module, which means understanding $H^*(B\Z_{2^n};\Z_2)$ as an $\cA(1)$-module.

For an $\cA(1)$-module $M$, $\Sigma^k M$ denotes the same ungraded module but with the gradings of all homogeneous elements increased by $k$; we write $\Sigma M$ for $\Sigma^1 M$. Using the Steenrod squares in \cref{Z4_mod2_coh}, one can make the following calculation.
\begin{prop}
\label{A1_mod_BZ4}
Let $C\eta$ denote the $\cA(1)$-module $\Sigma^{-2}\widetilde H^*(\CP^2;\Z_2)$.\footnote{This module is called the \term{bow} in~\cite[Definition 3.3]{BGHR24}.} Then, there is an isomorphism of $\cA(1)$-modules
\begin{equation}
        H^*(B\Z_{2^n};\Z_2)\overset\cong\longrightarrow \Z_2\oplus \Sigma \Z_2 \oplus \bigoplus_{n\ge 0} \Sigma^{4n+2}
        C\eta \oplus \Sigma^{4n+3}C\eta.
\end{equation}
The $\Sigma\Z_2$ summand is generated by $x$; the $\Sigma^{4n+2}C\eta$ summand is spanned by $y^{2n+1}$ and
$y^{2n+2}$, and the $\Sigma^{4n+3}C\eta$ summand is spanned by $xy^{2n+1}$ and $xy^{2n+2}$.
\end{prop}
For $n = 2$ this result also appears in~\cite[Lemma 3.8]{BGHR24}.
%To set up the Adams spectral sequence, we also need $\Ext_{\cA(1)}(C\eta)$. This is also known.
\begin{prop}[{Bahri-Bendersky~\cite[\S 4]{BB00}}]
\label{ext_Ceta}
There is an isomorphism of $\ExZ$-modules
\begin{equation}
        \Ext_{\cA(1)}(C\eta, \Z_2) \overset\cong\longrightarrow (\ExZ/h_1)\set{\kappa, \lambda, \mu, \xi}/(v\kappa = h_0\mu,
        v\lambda = h_0\xi, v\mu = h_0w\kappa, v\xi = h_0w\lambda),
\end{equation}
where $\abs{\kappa} = (0, 0)$, $\abs{\lambda} = (1, 3)$, $\abs{\mu} = (2, 6)$, and $\abs{\xi} = (3, 9)$.
\end{prop}
We draw this $\ExZ$-module in \cref{Ext_eta_fig}.

\begin{figure}[h!]
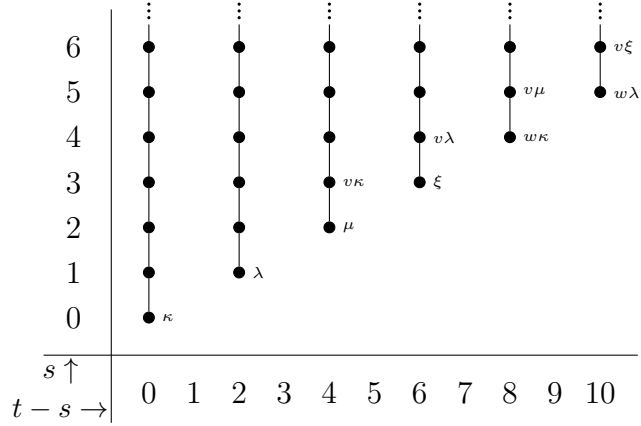

\centering
\begin{sseqdata}[name=ExtCeta, classes=fill, xrange={0}{10}, yrange={0}{6}, scale=0.6, Adams grading, >=stealth,
class labels = {right = 0.05em, font=\tiny},
x label = {$\displaystyle{s\uparrow \atop t-s\rightarrow}$},
x label style = {font = \small, xshift = -23ex, yshift=2.5ex}]
\foreach \y in {0, ..., 5} {
        \class(2*\y, \y)\AdamsTower{}
}
\classoptions["\kappa"](0, 0)
\classoptions["\lambda"](2, 1)
\classoptions["\mu"](4, 2)
\classoptions["\xi"](6, 3)
\classoptions["w\kappa"](8, 4)
\classoptions["w\lambda"](10, 5)
\classoptions["v\kappa"](4, 3);
\classoptions["v\lambda"](6, 4)
\classoptions["v\mu"](8, 5)
\classoptions["v\xi"](10, 6)
\end{sseqdata}
\printpage[name=ExtCeta, page=2]
\caption{The $\ExZ$-module $\Ext_{\cA(1)}(C\eta, \Z_2)$. Vertical lines represent $h_0$-multiplication, so in this
picture $h_0\mu = v\kappa$, $h_0\xi = v\lambda$, etc. This is a picture of \cref{ext_Ceta}.}
\label{Ext_eta_fig}
\end{figure}
Combining \cref{A1_mod_BZ4,ext_Ceta} we obtain a complete description of the $E_2$-page of the Adams spectral sequence computing $\ko_*(B\Z_{2^n})$ as an $\ExZ$-module.
\begin{cor}
\label{Ext_as_module}
There is an isomorphism of $\Z^2$-graded $\ExZ$-modules
\begin{equation}
        \Ext_{\cA(1)}(H^*(B\Z_{2^n};\Z_2), \Z_2) \overset\cong\longrightarrow \ExZ\set{1, \theta, \kappa_i, \kappa_i', \lambda_i,
        \lambda_i', \mu_i, \mu_i', \xi_i, \xi_i': i\ge 1}/\mathcal R_M,
\end{equation}
where $\theta\in\Ext^{0,1}$, $\kappa_i\in\Ext^{0, 4i-2}$, $\kappa_i'\in\Ext^{0, 4i-1}$, $\lambda_i\in\Ext^{1,
4i+1}$, $\lambda_i'\in\Ext^{1, 4i+2}$, $\mu_i\in\Ext^{2, 4i+4}$, $\mu_i'\in\Ext^{2, 4i+5}$, $\xi_i\in\Ext^{3, 4i+7}$, and $\xi_i'\in\Ext^{3, 4i+8}$, and the submodule $\mathcal R_M$ of relations is
\begin{equation}
\begin{aligned}
        \mathcal R_M &= (
                h_1\kappa_i, h_1\kappa_i', h_1\lambda_i, h_1\lambda_i', h_1\mu_i, h_1\mu_i', h_1\xi_i, h_1\xi_i',\\
        &\phantom{= (}
                h_0\mu_i = v\kappa_i, h_0\mu_i' = v\kappa_i', h_0\xi_i = v\lambda_i, h_0\xi_i' = v\lambda_i',\\
        &\phantom{= (}
                h_0w\kappa_i = v\mu_i, h_0w\kappa_i' = v\mu_i',
                h_0w\lambda_i = v\xi_i, h_0w\lambda_i' = v\xi_i'
        ).
\end{aligned}
\end{equation}
\end{cor}
The action of $h_0$ on $\Ext(C\eta)$ is injective, which means that the action of $h_0$ on $\Ext(H^*(B\Z_{2^n};\Z_2))$
is close to injective.
\begin{cor}
\label{injectivity}
$h_0$-action, as a map $E_2^{s,t}\to E_2^{s+1, t+1}$ for the Adams spectral sequence over $\cA(1)$, is injective except in the cases $(s, t) =
(4m+1, 12m+2)$, $(4m+2, 12m+4)$, $(4m+1, 12m+3)$, and $(4m+2, 12m+5)$, $m\ge 0$.
\end{cor}
Now we can compare to $\cA(0)$. First, just for $C\eta$:
\begin{lem}[{\!\!\cite[Appendix D]{Heckman}}]
\label{A1_Ceta_A0}
There is an isomorphism of $\cA(0)$-modules $j\colon C\eta\overset\cong\longleftarrow \Z_2\oplus\Sigma^2 \Z_2$.
The map
\begin{equation}
\begin{aligned}
        f_2\colon \Ext^{s,t}_{\cA(1)}(C\eta) &\longrightarrow \Ext^{s,t}_{\cA(0)}(C\eta)\\
                &\overset{j^*}{\underset{\cong}{\longrightarrow}}
        \Ext_{\cA(0)}^{s,t}(\Z_2) \oplus
        \Ext_{\cA(0)}^{s,t+2}(\Z_2)\\
                &\overset{\eqref{Z2_A1_ext}}{\underset{\cong}{\longrightarrow}} \cong \Z_2[h_0] \oplus
                \Sigma^{0,2}\Z_2[h_0]
\end{aligned}
\end{equation}
sends $\kappa\mapsto (1, 0)$, $\lambda\mapsto (0, h_0)$, $\mu\mapsto 0$, and $\xi\mapsto 0$.
\end{lem}
The notation $\Sigma^{p,q}$ means to increase the $s$-grading by $p$ and the $t$-grading by $q$.
\begin{cor}\label{A1_module_A0}
Under the isomorphisms in \cref{A0_E2,Ext_as_module}, the map
\begin{equation}\label{one_to_zero}
    \Ext_{\cA(1)}^{*,*}(H^*(B\Z_{2^n};\Z_2), \Z_2) \longrightarrow
    \Ext_{\cA(0)}^{*,*}(H^*(B\Z_{2^n};\Z_2), \Z_2)
\end{equation}
induced by the inclusion $\cA(0)\to\cA(1)$ takes on the following values:
\begin{equation}\label{1_to_0_gens}
\begin{alignedat}{2}
    1 &\mapsto 1 \qquad\qquad & \theta &\mapsto \widetilde x\\
    \kappa_i &\mapsto \widetilde y_{2i-1} & \kappa_i' &\mapsto \widetilde x\widetilde y_{2i-1}\\
    \lambda_i &\mapsto h_0\widetilde y_{2i} & \lambda_i' &\mapsto h_0\widetilde x\widetilde y_{2i}\\
    \mu_i &\mapsto 0 & \mu_i' &\mapsto 0\\
    \xi_i &\mapsto 0 & \xi_i' &\mapsto 0.
\end{alignedat}
\end{equation}
\end{cor}
Finally we can compute products in the Adams spectral sequence over $\cA(1)$ by comparing with the ring structure over $\cA(0)$ using that the map in \cref{A1_module_A0} is a ring homomorphism. This time things are a little more exciting because~\eqref{one_to_zero} is not injective.
\begin{prop}\label{some_products}
In $\Ext_{\cA(1)}(H^*(B\Z_{2^n}; \Z_2), \Z_2)$,
\begin{enumerate}
    \item $\theta^2 = 0$,
    \item $\theta\kappa_i = \kappa_i'$,
    \item $\theta\lambda_i = \lambda_i'$,
    \item $\kappa_i\kappa_j = 0$,
    \item $\kappa_i\lambda_j = \binom{2i+2j-1}{2i} h_0\kappa_{i+j}$, and
    \item $\lambda_i\lambda_j = \binom{2i+2j}{2i} h_0\lambda_{i+j}$.
\end{enumerate}
\end{prop}
\begin{proof}
The product $\kappa_i\kappa_j$ vanishes for degree reasons.

The map~\eqref{one_to_zero} is a ring homomorphism, so we can use the ring structure of $\Ext_{\cA(0)}(\dotsb)$ to compute the ring structure of $\Ext_{\cA(1)}(\dotsb)$ \emph{modulo the kernel of~\eqref{one_to_zero}}. Use~\eqref{1_to_0_gens} to compare the remaining products in the theorem statement with the ring structure over $\cA(0)$ as we proved in \cref{A0_ring_str}; one sees that they are compatible, so the products in the theorem statement hold modulo the kernel of~\eqref{one_to_zero}.

To finish the proof and deduce these products without having to quotient, observe that all products in the theorem statement take place in filtration at most $2$, but using~\eqref{1_to_0_gens}, we see that the kernel of the map~\eqref{one_to_zero} consists solely of elements in filtration $3$ and higher. Thus, restricted to the sub-vector space consisting of elements in filtration $2$ and below, \eqref{one_to_zero} is injective, so the products in the theorem statement hold unconditionally.
\end{proof}
\begin{rem}
We can make a slight simplification: $\lambda_i\lambda_j = \binom{i+j}{i} h_0\lambda_{i+j}$ because of the identity $\binom{2(i+j)}{2i}\equiv \binom{i}{j}\bmod 2$ for $i,j\ge 1$.
\end{rem}
\begin{prop}\label{more_products}
In $\Ext_{\cA(1)}(H^*(B\Z_{2^n};\Z_2), \Z_2)$,
\begin{enumerate}
    \item $\theta\mu_i = \mu_i'$ and $\theta\xi_i = \xi_i'$.
    \item $\kappa_i\mu_j = 0$ and $\mu_i\mu_j = 0$.
    \item $\kappa_i\xi_j = \binom{2i+2j-1}{2i} h_0\mu_{i+j}$, $\mu_i\lambda_j = \binom{2i+2j-1}{2i} h_0\mu_{i+j}$, and
    $\mu_i\xi_j = \binom{2i+2j-1}{2i}h_0w\kappa_{i+j}$.
    \item $\lambda_i\xi_j = \binom{i+j}{i} h_0\xi_{i+j}$ and $\xi_i\xi_j = \binom{i+j}{i}h_0 w\lambda_{i+j}$.
\end{enumerate}
\end{prop}
\begin{proof}
Recall from \cref{some_products} that $\theta\kappa_i = \kappa_i'$, so $\theta v\kappa_i = v\kappa_i'$. The relations $v\kappa_i = h_0\mu_i$ and $v\kappa_i' = h_0\mu_i'$ in \cref{Ext_as_module} thus imply $h_0 \theta\mu_i = h_0\mu_i'$. By \cref{injectivity}, multiplication by $h_0$ is an injective map $\Ext^{2,4i+5}\to\Ext^{3,4i+6}$, so $h_0\theta\mu_i = h_0\mu_i'$ implies $\theta\mu_i = \mu_i'$.

The remaining calculations in the theorem statement are completely analogous: take an equation in \cref{some_products}, multiply both sides by $v$ or $v^2$, apply a relation in \cref{Ext_as_module}, then appeal to \cref{injectivity} to cancel factors of $h_0$ off of both sides of the equation and obtain the desired result.
\end{proof}
Combining \cref{Ext_as_module,some_products,more_products}, we get the main theorem in this section.
\begin{thm}
\label{Z4_Adams_ring}
\hfill
\begin{enumerate}
\item
There is an isomorphism of $\Z^2$-graded $\ExZ$-algebras
\begin{subequations}
\begin{equation}\label{Z4_Ext_ring_eqn}
    \Ext_{\cA(1)}(H^*(B\Z_{2^n};\Z_2),\Z_2) \cong
        \ExZ[\theta, \kappa_i, \lambda_i, \mu_i, \xi_i: i\ge 1]/\mathcal R_1
\end{equation}
with generators in the following degrees: $\theta\in\Ext^{0,1}$, $\kappa_i\in\Ext^{0,4i-2}$, $\lambda_i\in\Ext^{1,4i+1}$, $\mu_i\in\Ext^{2,4i+4}$, and $\xi_i\in\Ext^{3,4i+7}$. The ideal $\mathcal R_1$ of relations is
\begin{equation}\label{Z4_rels}
\begin{aligned}
    \mathcal R_1 &= (
        h_1\kappa_i, h_1\lambda_i, h_1\mu_i, h_1\xi_i,
        v\kappa_i = h_0\mu_i, v\lambda_i = h_0\xi_i,
        v\mu_i = h_0w\kappa_i, v\xi_i = h_0w\lambda_i,\\
        &\phantom{= (}
        \theta^2 = 0,
        \kappa_i\kappa_j = 0, \kappa_i\lambda_j = A(i,j) h_0\kappa_{i+j}, \kappa_i\mu_j = 0, \kappa_i\xi_j = A(i,j) h_0\mu_{i+j},\\
                &\phantom{= (}
        \lambda_i\lambda_j = B(i,j) h_0\lambda_{i+j}, \lambda_i\mu_j = A(j,i)h_0\mu_{i+j}, \lambda_i\xi_j = B(i,j)h_0\xi_{i+j},\\
                &\phantom{= (}
        \mu_i\mu_j = 0, \mu_i\xi_j = A(i,j)h_0 w\kappa_{i+j},
        \xi_i\xi_j = B(i,j)h_0w\lambda_{i+j}
        %h_1\kappa_1, h_1\kappa_2, h_1\lambda_1, h_1\lambda_2, h_1\mu_1, h_1\xi_1,
        %v\kappa_1 + \mu_1, v\lambda_1 + \xi_1,
        %\theta^2, \kappa_1^2, \lambda_1^2,
        %\kappa_1\lambda_1 + h_0\lambda_2, \kappa_1\mu_1, \kappa_1\kappa_2, \dots
    ),
\end{aligned}
\end{equation}
where $A(i,j)\coloneqq \binom{2i+2j-1}{2i}\bmod 2$ and $B(i,j)\coloneqq\binom{i+j}{i}\bmod 2$.
\end{subequations}
\item The class $\theta$ is detected by $x\in H^1(B\Z_{2^n};\Z_2)$, $\kappa_i$ is detected by $y^{2i-1}$, and $\theta\kappa_i$ by $xy^{2i-1}$. All other generators listed in~\eqref{Z4_Ext_ring_eqn} are not detected by mod $2$ cohomology.
\end{enumerate}
\end{thm}
For $n = 2$, we draw the $E_2$-page and label these generators in \cref{Z4_E2}, left.

\begin{figure}[h!]
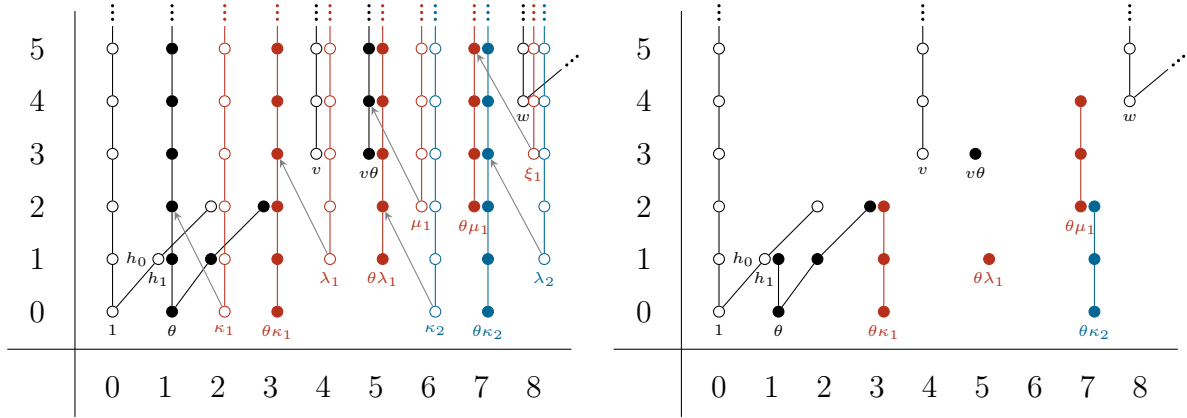

\centering
\begin{sseqdata}[name=Z4ring, xrange={0}{8}, yrange={0}{5}, scale=0.7, Adams grading, >=stealth,
class labels = {below = 0.05em, font=\tiny}]
\begin{scope}
    \class(0, 0)\AdamsTower{}
    \class(1, 1)\structline(0, 0)(1, 1)
    \class(2, 2)\structline
    \class(4, 3)\AdamsTower{}
    \class(8, 4)\AdamsTower{}
    \class(9, 5)\structline(8, 4)(9, 5)
    \class(10, 6)\structline
    \classoptions["1"](0, 0)
    \classoptions[class labels = {right=0.05em}, "h_0"](0, 1)
    \classoptions["h_1"](1, 1)
    \classoptions["v"](4, 3)
    \classoptions["w"](8, 4)
\end{scope}
\begin{scope}[draw=none, fill=none]
    \class(1, 0)
    \class(1, 2)
    \class(2, 0)
    \class(3, 0)
    \class(3, 1)
    \foreach \y in {3, ..., 7} {
        \class(1, \y)
        \class(2, \y)
        \class(3, \y)
    }
    \foreach \x in {4, 5, 8, 9} {
        \class(\x, 1)
        \class(\x, 2)
    }
    \foreach \x in {6, 7, 10} {
        \class(\x, 0)
        \class(\x, 1)
    }
    \foreach \x in {8,9,10} {
        \foreach \y in {1,2,3} {
            \class(\x, \y)
        }
    }
    \class(9, 6)
    \foreach \y in {1, 2, 3, 4, 7} {
        \class(9, \y)
        \class(10, \y)
    }
\end{scope}
\begin{scope}[classes=fill]
    \class(1, 0)\AdamsTower{}
    \class(2, 1)\structline(1, 0, -1)(2, 1)
    \class(3, 2)\structline
    \class(5, 3)\AdamsTower{}
    \class(9, 4)\AdamsTower{}
    \class(10, 5)\structline(9, 4, -1)(10, 5)
    \class(11, 6)\structline
    \classoptions["\theta"](1, 0, -1)
    \classoptions["v\theta"](5, 3, -1)
    \classoptions[class labels = {below right = 0.05em}, "w\theta"](9, 4, -1)
\end{scope}
\TwoCEtaExt{BrickRed}{2}{1}
\TwoCEtaExt{MidnightBlue}{6}{2}
%\TwoCEtaExt{MidnightBlue}{10}{3}
\begin{scope}[gray, >=stealth]
    \d2(2, 0, -1)(1, 2, -1)
    \d2(4, 1, -1)(3, 3, -1)
    \d2(6, 2, 1)(5, 4, 1)
    \d2(6, 0, -1)(5, 2, -1)
    \d2(8, 3, 2)(7, 5, 1)
    \d2(8, 1, -1)(7, 3, -1)
%    \d2(10, 4, 2)(9, 6, 2)
%    \d2(10, 2, 3)(9, 4, 3)
 %   \d2(10, 0, -1)(9, 2, -1)
\end{scope}
\begin{scope}[draw=none, fill=none]
    \class(1, 7)
    \class(3, 7)
    \class(5, 7)\class(5, 7)
    \class(7, 7)\class(7, 7)
    \foreach \y in {1, ..., 5} {
        \d2(2, \y, -1)(1, \y+2, -1)
        \d2(6, \y, -1)(5, \y+2, -1)
    }
    \foreach \y in {2, ..., 5} {
        \d2(4, \y, -1)(3, \y+2, -1)
        \d2(8, \y, -1)(7, \y+2, -1)
    }
    \foreach \y in {3, ..., 5} {
        \d2(6, \y, 1)(5, \y+2, 1)
    }
    \foreach \y in {4, 5} {
        \d2(8, \y, 2)(7, \y+2, 1)
    }
\end{scope}
\end{sseqdata}
\begin{subfigure}[c]{0.48\textwidth}
    \printpage[name=Z4ring, page=2]
\end{subfigure}
\begin{subfigure}[c]{0.48\textwidth}
    \printpage[name=Z4ring, page=3]
\end{subfigure}
\caption{Left: the $E_2$-page of the Adams spectral sequence computing $\ko_*(B\Z_4)$. Generators and some relations labeled. As we discussed in \cref{ko_Z4_diffs}, the May-Milgram theorem establishes $d_2$ differentials between the $h_0$-towers, which we display here. Right: the spectral sequence then collapses at $E_3 = E_\infty$.}
\label{Z4_E2}
\end{figure}

\begin{prop}
\label{ko_Z4_diffs}
The differentials in the Adams spectral sequence for $\ko_*(B\Z_{2^n})$ all vanish except on the $E_n$-page,
where they have the following values for any $i\ge 1$:
\begin{equation}
\begin{alignedat}{3}
        d_n(1) &= 0\qquad\qquad\qquad & 
                %d_n(\kappa_i') &= 0\qquad\qquad &
                d_n(\kappa_i) &= h_0^{n-1}\theta\lambda_{i-1}\ (i\ne 1)\\
         d_n(\theta) &= 0 &
                %d_n(\lambda_i') &= 0 &
                d_n(\lambda_i) &= h_0^{n+1} \theta\kappa_i\\
         d_n(\kappa_1) &= h_0^n \theta &
                %d_n(\mu_1') &= 0 &
                d_n(\mu_i) &= h_0^{n-1}\theta\xi_{i-1}\ (i\ne 1)\\
        d_n(\mu_1) &= h_0^{n-1} v\theta &
 %               d_n(\xi_1') &= 0 &
                d_n(\xi_i) &= h_0^{n+1} \theta\mu_i.
\end{alignedat}
\end{equation}
\end{prop}
For $n = 2$, we draw these differentials in \cref{Z4_E2}, left. The Leibniz rule thus tells us the values of differentials on all classes. For $n = 2$, \cref{ko_Z4_diffs} is proven in a different way by Bárcenas, García-Hernández, and Reinauer~\cite[Theorem 3.15]{BGHR24}.
\begin{proof}
First, $d_s(1) = 0$ for all $s$ by degree reasons. If $d_s(\theta)$ were nonzero for any $s$, then there would be
finitely many $\Z_2$ summands in topological degree $0$ on the $E_\infty$-page, which is incompatible with
$\ko_0(B\Z_{2^n})$ containing a $\Z$ summand (namely $\ko_0(\pt)$). Thus $d_s(\theta) = 0$.

The values of $d_s(\kappa_i)$ and $d_s(\lambda_i)$ follow from the
May-Milgram theorem~\cite{MM81}, specifically from its incarnation in~\cite[Proposition D.13]{Heckman}. For
$d_s(\mu_i)$, use $\ExZ$-linearity of differentials to compute $v d_s(\kappa_i) = d_s(v\kappa_i) = d_s(h_0\mu_i) =
h_0d_s(\mu_i)$; then \cref{injectivity} shows there is a unique class $\sigma\in E_r^{2+n, 4i+3+n}$ such that
$h_0\sigma = d_s(v\kappa_i)$, fixing $d_s(\mu_i) = \sigma$. The same argument works \textit{mutatis mutandis} with $(\xi_i, \lambda_i)$ in place of $(\mu_i, \kappa_i)$.
\end{proof}

\subsection{\texorpdfstring{Lifting to the ring structure on $\ko_*(B\Z_4)$}{Lifting to the ring structure on ko(BZ4)}}

Now we describe the $\ko_*$-algebra structure on $\ko_*(B\Z_4)$ in low degrees.
\begin{thm}\label{ringstr_Z4}\hfill
\begin{enumerate}
\item\label{ring_Z4_part}
There is an isomorphism of $\ko_*$-algebras
\begin{subequations}
\begin{equation}\label{ko_Ring_Z4}
    \ko_*(B\Z_4)\cong\ko_*[\ell_1,\ell_3, q_5, \ell_7, \widetilde \ell_7, \dotsc]/\mathcal R_2,
\end{equation}
with $\abs{\ell_i} = i$, $\abs{q_5} = 5$, and $\lvert\widetilde \ell_7\rvert = 7$, and the ideal $\mathcal R_2$ of relations is
\begin{equation}
\begin{aligned}
    \mathcal R_2 = &(
        4\ell_1, 8\ell_3, 4q_5, 8\widetilde \ell_7, 2\widetilde\ell_7 - 8\ell_7,
        \eta\ell_3, \eta q_5, \eta\ell_7, \eta\widetilde\ell_7,\\
        &\phantom{(}\ell_1^2, \ell_3^2, \ell_1\ell_3, \ell_1 q_5, \ell_1\ell_7, \ell_1\widetilde\ell_7, \ell_3 q_1, v\ell_1 - 2q_5, v\ell_3 - 2\widetilde\ell_7, \dotsc
    ).
\end{aligned}
\end{equation}
\end{subequations}
All generators and relations not listed are in degrees greater than $8$. %\textcolor{red}{MD: Is it obvious of how to find the structure of the ideals, which seem rather complicated.}
\item\label{generator_names}
The isomorphism~\eqref{ko_Ring_Z4} may be chosen so that the images of the generators on the $E_\infty$-page of the Adams spectral sequence are: $\ell_1\mapsto\theta$, $\ell_3\mapsto\theta\kappa_1$, $q_5\mapsto\theta\lambda_1$, $\ell_7\mapsto\theta\kappa_2$, and $\widetilde\ell_7\mapsto\theta\mu_1$.
\end{enumerate}
\end{thm}
\begin{proof}
First compare the theorem statement with the $E_\infty$-page of this spectral sequence, drawn in \cref{Z4_E2}, right: the $E_\infty$-page is generated as an $\ExZ$-algebra by $\theta$, $\theta\kappa_1$, $\theta\lambda_1$, $\theta\kappa_2$, and $\theta\mu_1$. The relation $\theta^2 = 0$ from \cref{Z4_Adams_ring} implies that, unless there are hidden extensions, all products of pairs of these five generators vanish, and the relations $h_1\kappa_i = 0$, $h_1\lambda_1 = 0$, and $h_1\mu_1 = 0$ (also from \cref{Z4_Adams_ring}) imply $\eta$ times each of these classes vanishes unless there is a hidden $\eta$-extension.

Thus if we choose lifts of $\theta$, $\theta\kappa_1$, $\theta\lambda_1$, $\theta\mu_1$, and $\theta\kappa_2$ and label them as in part~\ref{generator_names} of the theorem statement, part~\eqref{ring_Z4_part} is the assertion that, even multiplicatively, there are no extensions on the $E_\infty$-page in degrees $8$ and below apart from extensions by $2$ and the one implied by the relation $v\ell_1 = 2q_5$. This is the form of the theorem we will prove.

Extensions by $2$ are a consequence of the additive structure of $\ko_*(B\Z_4)$, which is calculated by Bruner-Greenlees~\cite[Example 7.3.3]{Bruner}. First, $\ko_1(B\Z_4)\cong\Z_4\oplus\Z_2$ shows $\ell_1$ has order $4$, and $\ko_3(B\Z_4)\cong\Z_8$ implies $\ell_3$ has order $8$. Since $\ko_5(B\Z_4)\cong\Z_4$, but on the $E_\infty$-page we see $\theta\lambda_2$ and $v\theta$ in topological degree $5$, which are not linked by an $h_0$-action, we obtain a hidden extension: $2q_5 = v\ell_1$. Likewise, in degree $7$, Bruner-Greenlees (\textit{ibid.}) shows us that $\ko_*(B\Z_4)\cong\Z_{32}\oplus\Z_2$, which is only possible if $8\ell_7 = 2\widetilde\ell_7$: another hidden extension.

Next, extensions by $\eta$. On $\ell_1$ this follows from the $h_1$-action on the $E_\infty$-page; for the remaining classes we need to argue there are no hidden $\eta$-extensions. Since $\abs{\eta} = 1$ and $2\eta=0$, then if $\ko_{n+1}(B\Z_4)$ is torsion-free then $\eta\cdot\ko_n(B\Z_4) = 0$: $\eta$ times anything is torsion. Thus we deduce $\eta\ell_3 = 0$, $\eta q_5 = 0$, $\eta\ell_7 = 0$, and $\eta\widetilde\ell_7 = 0$.\footnote{By contrast, in the closely related spin-$\Z_8$ bordism, there are many hidden $\eta$-extensions: see~\cite[Footnote 50 and \S 13.4]{Heckman}.} The same argument works to show that $\ell_1$ times $\ell_1$, $\ell_3$, $q_5$, $\ell_7$, or $\widetilde\ell_7$ is $0$, as $\abs{\ell_1} = 1$ and $\ell_1$ is torsion. Replacing $\ko_{n+1}$ with $\ko_{n+3}$, the same argument also applies to products of $\ell_3$ with $\ell_3$, $q_5$, $\ell_7$, and $\widetilde\ell_7$.

For extensions by $v$, we already have $v\ell_1 = 2q_5$, and $v\ell_3 = 2\widetilde\ell_7$ comes from the relation $v\theta\kappa_1 = h_0\theta\mu_1$ on the $E_2$-page described in \cref{Z4_Adams_ring}. All other products of $v$ with generators, and all products of $w$, $q_5$, $\ell_7$, or $\widetilde\ell_7$ with generators that were not already addressed, are in degrees too high to be relevant.
\end{proof}

\subsection{Generators}\label{sec:BZ4gens}
The generators of $\ko_*(B\Z_4)$ have been previously considered in \cite[\S 12.7]{Heckman}. Before we state the result, we recall the notation of the lens space $L^n_k$ introduced in \cref{sec:Z3gens}. Furthermore, we introduce $Q^{2k-1}_n$:

\begin{defn}[{Botvinnik-Gilkey-Stolz~\cite[\S 4]{BGS97}}]
\label{Q5_defn}
Let $V$ denote the complex vector bundle $\mathcal O(2)\oplus \underline\C^k\to\CP^1$ and $Q_n^{2k-1}$ denote the quotient of $S(V)$ by the $\Z_n$-action which acts diagonally by the $n^{\mathrm{th}}$ roots of unity on each fiber.
\end{defn}
Thus $Q_n^{2k-1}$ is a closed manifold equipped with a canonical orientation, namely that arising as the quotient of $S(V)$ by an orientation-preserving $\Z_n$-action. The canonical orientation on $S(V)$ is induced from that on $V$ coming from its complex structure, since $TV|_{S(V)}$ and $TS(V)$ are stably equivalent (see, e.g., \cite[(14.84)]{Heckman}). For all $n$ and $k$, $Q_n^{2k-1}$ is spin~\cite[\S 5]{BGS97}.
\begin{prop}[{\!\cite[\S 12.7]{Heckman}}]
\label{Z4_mult_gens}
The following manifolds represent the generators of $\ko_*(B\Z_4)$ that we identified in \cref{ringstr_Z4}.
\begin{enumerate}
    \item For $n = 1,3$, $L_4^n$ with $\Z_4$-bundle $S^n\to L_4^n$ and either of its two spin structures represents $\ell_n$.
    \item $Q_4^5$ with $\Z_4$-bundle $S(V)\to Q_4^5$ and either of its two spin structures represents $q_5$.
    \item Let $\mathfrak s_\kappa$, $\kappa = 0,2$ denote the two spin structures on $L_4^7$ as classified by $\kappa$ in~\cite[\S C.1]{Heckman}. Then the isomorphism in \cref{ringstr_Z4} can be chosen so that $(L_4^7, \mathfrak s_0)$ represents $\ell_7$ and $7[(L_4^7, \mathfrak s_0)] - 5[(L_4^7, \mathfrak s_2)]$ represents $\widetilde\ell_7$. In all cases we use the $\Z_4$-bundle $S^7\to L_4^7$.
\end{enumerate}
\end{prop}
The fact that lens spaces and lens space bundles of the sort in \cref{Q5_defn} suffice to generate $\ko_{2k-1}(B\Z_4)$ as an abelian group is due to Botvinnik-Gilkey-Stolz~\cite[\S 5]{BGS97}; the specific generators we use were worked out in~\cite[\S 12.3]{Heckman}.

\section{\texorpdfstring{Calculations and generators for $BS_4$ and $B\SL(3, \F_2)$}{Calculations and generators for BS4 and BSL(3, F2)}}
In this section, we compute the ($2$-local) spin bordism groups of $BS_4$ and $B\SL(3, \F_2)$ in degrees $7$ and below, and determine a set of generating manifolds for these bordism groups.\footnote{We do not consider the stable summand $B\Z_2$ of $BS_4$, because our procedure in \cref{some_assembly_required} to product $B\SL(2, \Z)\times B\SL(3, \Z)$ bordism classes discards the spin bordism of that summand.}

These bordism groups have been studied in~\cite{Bayen, Debray, Gri21}, and in fact Bayen~\cite{Bayen} computes all of the bordism groups we need. However, like in previous sections, we repeat the computation because we will make use of the $\Omega_*^\Spin(B\Z_4)$-module structure on $\Omega_*^\Spin(BS_4)$ in \S\ref{2loc_smash}.

The generating manifolds, however, are new. We employ a few different methods to find generators: in \S\ref{Z4_to_S4_gens} we study the map $B\Z_4\to BS_4\to B\SL(3, \F_2)$ on spin bordism, so that we can use the generators of $\widetilde\Omega_*^\Spin(B\Z_4)$ from \cref{Z4_mult_gens}. We then do the same with $\Z_4$ replaced with $\Z_2\times\Z_2$. This leaves two classes unresolved; we wrap them up in \S\ref{sec:X7gen} and \S\ref{sec:W4} by pulling back to $D_8\subset S_4$, then studying relatively simple bundles over previously found generators, similar to the strategies used in~\cite[\S\S14.3.3--14.3.7]{Heckman}, and which go back at least as far as Dold~\cite{Dol56}.

\subsection{\texorpdfstring{Adams spectral sequences for $\ko_*(B\SL(3,\F_2))$ and $\ko_*(L(2))$}{Adams spectral sequences for ko(BSL(3,F2)) and ko(L(2))}}
\label{sec:BPSL}

In this subsection, we describe all classes and differentials in the Adams spectral sequences computing $\ko_*(B\SL(3,\F_2))$ and $\ko_*(L(2))$ in degrees 8 and below. Consider first the case of $\ko_*(B\SL(3,\F_2))$ To begin, we recall previous results from Handel and Mitchell-Priddy.
\begin{thm}[{Handel, \cite[Theorem 5.5]{Han93}}]
    \begin{equation}
    \begin{gathered}
        H^*(BD_8;\Z_2) = \Z_2[x_1,x_2,w]/(x_2^2+x_1x_2),\\
        |x_1|=|x_2| = 1 ,\quad |w| = 2
        \end{gathered}
    \end{equation}
Moreover, the Steenrod algebra action on $H^*(BD_8;\Z_2)$ is:
\begin{equation}
    \begin{gathered}
        \Sq(x_1) = x_1 + x_1^2, \quad \Sq(x_2) = x_2 + x_2^2,\quad \Sq(w) = w + x_1 w + w^2 
    \end{gathered}
\end{equation}
\end{thm}

\begin{thm}[{Mitchell-Priddy, \cite[Theorem 2.6]{Priddy}}] \label{Thm:SteenrodBPSL27} 
    \begin{equation}
    \begin{gathered}
        H^*(B\SL(3,\F_2);\Z_2)= \Z_2[\nu_2,\nu_3,\overline{\nu}_3]/(\overline{\nu}_3^2 +\overline\nu_3\nu_3), \\
        |\nu_2| = 2,\quad |\nu_3| = |\overline{\nu}_3| = 3.
        \end{gathered}
    \end{equation}
The inclusion $j: D_8\rightarrow \SL(3,\F_2)$ induces a monomorphism $j^*:H^*(B\SL(3,\F_2);\Z_2)\rightarrow H^*(BD_8;\Z_2)$ given as follows:
\begin{equation}
    \begin{gathered}
        j^*(\nu_2) = x_1^2+w, \quad j^*(\nu_3) = x_1 w, \quad j^*(\overline{\nu}_3) = x_2 w.
    \end{gathered}
\end{equation}
Moreover, the Steenrod algebra action on $H^*(B\SL(3,\F_2);\Z_2)$ is 
\begin{equation}
    \begin{gathered}
        \Sq(\nu_2) = \nu_2 + \nu_3 + \nu_2^2,\quad \Sq(\nu_3) = \nu_3 + \nu_2\nu_3 + \nu_3^2,\quad \Sq(\overline \nu_3) = \overline\nu_3 + \nu_2\overline \nu_3 + \overline \nu_3^2
    \end{gathered}
\end{equation}
\end{thm}
The Steenrod squares are not given in \textit{loc.\ cit.}\ but can be worked out from the Steenrod squares in $H^*(BD_8;\Z_2)$ and the fact that $j^*$ is injective.

The first step in running the Adams spectral sequence for $\ko_*(B\SL(3, \F_2))$ is to determine the $\cA(1)$-module structure on $H^*(B\SL(3, \F_2);\Z_2)$. For this we must recall a few commonly occurring $\cA(1)$-modules.
\begin{defn}
\label{naming_A1}
For an $\cA(1)$-module $M$, $\Sigma^k M$ refers to the same underlying module with $\Z$-grading increased by $k$.
\begin{enumerate}
    \item $\Z_2$ is an $\cA(1)$-module in which $\Sq^1$ and $\Sq^2$ act trivially.
    \item Let $J\coloneqq\cA(1)/(\Sq^3)$. This module was named the \term{Joker} by Adams.
    \item The \term{upside-down question mark} is the $\cA(1)$-module $Q\coloneqq \cA(1)/(\Sq^1, \Sq^2\Sq^3)$.
    \item There is a unique nonzero $\cA(1)$-module map $\Sigma^{-1}\cA(1)\to\Sigma^{-1}\Z_2$; its kernel is denoted $R_2$, and called the \term{elephant} by Buchanan-McKean~\cite[Figure 1]{BM23}.
\end{enumerate}
\end{defn}

A straightforward calculation using \cref{Thm:SteenrodBPSL27} leads to the following.
\begin{cor}
\label{cor:A(1)decompPSL}
There is an $\cA(1)$-module isomorphism
    \begin{equation}\label{eq:CohoPSL(2,7)}
        \widetilde H^*(B\SL(3,\F_2);\Z_2)\cong \textcolor{BrickRed}{\Sigma^2 J} \oplus \textcolor{Green}{\Sigma^3 Q} \oplus {\Sigma^6 \cA(1)} \oplus \textcolor{MidnightBlue}{\Sigma^7 R_2} \oplus \textcolor{Fuchsia}{\Sigma^8 \Z_2} %\oplus %\textcolor{MidnightBlue}{\Sigma^{10}J}
        \oplus P,
    \end{equation}
    where $P$ is concentrated in degrees 10 and above.
\end{cor}
We summarize \eqref{eq:CohoPSL(2,7)} in \cref{CohoBPSL27}, left.

\begin{figure}[h!]
\centering
\begin{subfigure}[c]{0.42\textwidth}
\begin{tikzpicture}[scale=0.6, every node/.style = {font=\tiny}]
        \node[white] at (-2, 1) {$1$};
	\foreach \y in {2, ..., 12} {
		\node at (-2, \y) {$\y$};
	}
	\begin{scope}[BrickRed]
		\Joker{0}{2}{$\nu_2$}{};
	\end{scope}
	\begin{scope}[Green]
		\SpanishQnMark{2}{3}{$\overline{\nu}_3$}{isosceles triangle};
	\end{scope}
	\begin{scope}
		\Aone{3.5}{6}{$\nu_2^3$}{regular polygon,regular polygon sides=5};
	\end{scope}
	\tikzptB{7.5}{8}{$\nu_2^4$}{Fuchsia, regular polygon,regular polygon sides=3};
	\begin{scope}[MidnightBlue]
		\Rtwo{6}{7}{$\nu_2^2\overline{\nu}_3$}{$\nu_2\overline{\nu}_3^2$}{star};
	\end{scope}
	%\begin{scope}[MidnightBlue]
	%	\Joker{13}{10}{$\nu_2^5$}{rectangle, minimum size=3.5pt};
	%\end{scope}
\end{tikzpicture}
\end{subfigure}
\begin{subfigure}[c]{0.48\textwidth}
\begin{tikzpicture}[scale=0.6, every node/.style = {font=\tiny}]
	\foreach \y in {1, ..., 12} {
		\node at (-2, \y) {$\y$};
	}
        \begin{scope}[gray]
            \foreach \y in {1, 3, ..., 11} {
                \sqone(0, \y);
            }
            \foreach \y in {2, 6, 10} {
                \sqtwoL(0, \y);
            }
            \begin{scope}
                \clip (-1, 2) rectangle (3, 12.5);
                \foreach \y in {3, 7, 11} {
                 \sqtwoR(0, \y);
             }
            \end{scope}
            \emptytikzpt{0}{1}{$a$}{};
            \emptytikzpt{0}{3}{$a^3$}{};
            \emptytikzpt{0}{7}{$a^7$}{};
            \emptytikzpt{0}{11}{$a^{11}$}{};
            \foreach \y in {2, ..., 12} {
                \emptytikzpt{0}{\y}{}{};
            }
        \end{scope}
	\begin{scope}[BrickRed]
		\Joker{2}{2}{$b$}{};
	\end{scope}
	\begin{scope}[Green]
		\SpanishQnMark{4}{3}{$c$}{isosceles triangle};
	\end{scope}
        \begin{scope}[RedOrange]
            \Aonerect{5.5}{4}{$a^2b$}{};
        \end{scope}
	\begin{scope}
		\Aone{8.25}{6}{$b^3$}{regular polygon,regular polygon sides=5};
	\end{scope}
	\tikzptB{12}{8}{$b^4$}{Fuchsia, regular polygon,regular polygon sides=3};
	\begin{scope}[MidnightBlue]
		\Rtwo{10.5}{7}{$b^2c$}{$bc^2$}{star};
	\end{scope}
\end{tikzpicture}
\end{subfigure}
\caption{\label{CohoBPSL27}Left: the $\cA(1)$-module structure on $\widetilde H^*(B\SL(3, \F_2);\Z_2)$ in low degrees, as we calculated in \cref{cor:A(1)decompPSL}. This picture includes all classes in degrees $9$ and below.
Right: the same but for $S_4$ in place of $\SL(3, \F_2)$, realizing the effect on cohomology of the Mitchell-Priddy mod $2$ stable splitting~\cite[Theorem B]{Priddy} of $BS_4$ into $B\Z_2$ (gray/unshaded dots), $L(2)$ (orange/square dots), and $B\SL(3, \F_2)$ (everything else). See \cref{S4_splitting}.
}
\end{figure}
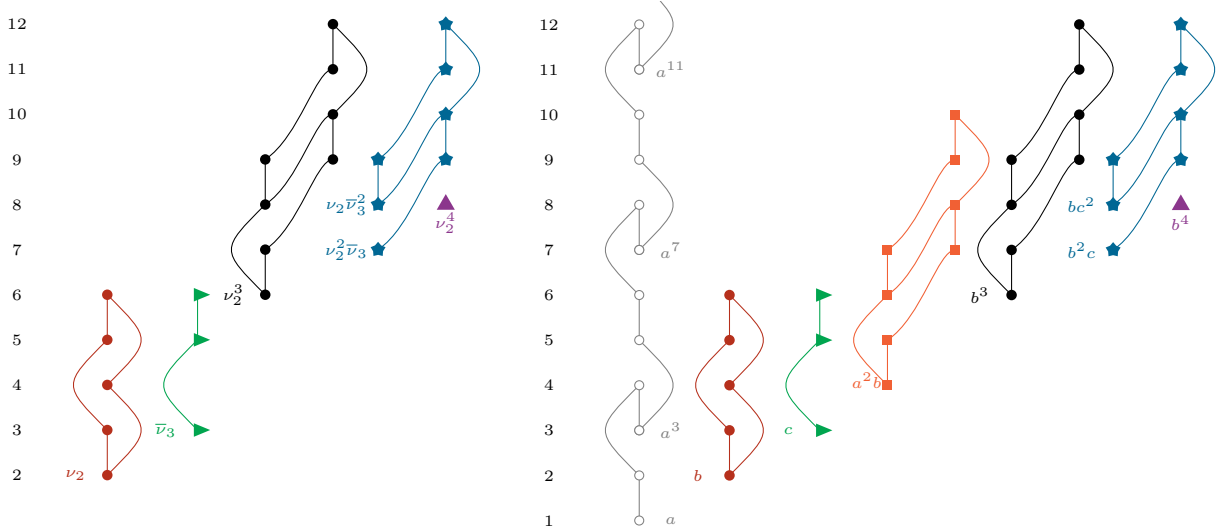

Now we need to compute $\Ext$ of the summands appearing in~\eqref{eq:CohoPSL(2,7)}, except for $P$, which lies in too high of a degree to be relevant for our computations in degrees $7$ and below. Ext commutes with direct sums so we may focus on the modules we named in \cref{naming_A1}. By definition, $\Ext(\Z_2)\cong\ExZ$ as $\ExZ$-modules, and Ext of a rank-one free $\cA(1)$-module is a trivial $\ExZ$-module consisting of a single $\Z_2$ summand concentrated in bidegree $(0, 0)$ with trivial $\ExZ$-action. The other three modules in \cref{naming_A1} have more interesting Ext groups.
\begin{prop}
There are isomorphisms of $\ExZ$-modules as follows.
\begin{subequations}
\begin{align}
    \Ext(J) &\cong\ExZ\set{a, b, c}/(h_0a, h_1a, va, wa + h_1^2c, h_1b, vb + h_0^2c, vc + wb),\\
    \intertext{where $a\in\Ext^{0,0}$, $b\in\Ext^{1, 3}$, and $c\in\Ext^{2,8}$,}
    \Ext(Q) &\cong\ExZ\set{d, e}/(h_1d, vd + h_0^2e, ve + wd),\\
    \intertext{with $d\in\Ext^{0,0}$ and $e\in\Ext^{1,5}$, and}
    \Ext(R_2) &\cong \ExZ\set{g, h, j, k}/(h_1g, h_0h, h_1^2h, vg + h_0k, vh, wg+h_0k, wh+h_1k),
\end{align}
with $g\in\Ext^{0,0}$, $h\in\Ext^{0,1}$, $j\in\Ext^{2,6}$, and $k\in\Ext^{3,11}$.
\end{subequations}
\end{prop}
See~\cite[Figure 29]{Beaudry} for pictures of these $\ExZ$-modules.
\begin{rem}
These $\ExZ$-modules are usually only presented as modules over the subalgebra $\ang{h_0,h_1}\subset\ExZ$; in this form, $\Ext(J)$ and $\Ext(R_2)$ were first described by Adams-Priddy~\cite[\S 3]{AP76}.\footnote{As just a bigraded vector space, $\Ext(J)$ was computed earlier, by Wilson~\cite[Figure 2]{Wil73}.} The complete $\ExZ$-module structure on $\Ext(Q)$ is given by Bruner-Rognes~\cite[Example 2.32]{BR21}, and that of $\Ext(R_2)$ is given by Bruner-Greenlees~\cite[Figure A.4.4]{Bruner}. We do not know of a reference for the complete $\ExZ$-module structure on $\Ext(J)$, but it is straightforward to work out using the long exact sequence of Ext groups associated to a short exact sequence of $\cA(1)$-modules; see Beaudry-Campbell~\cite[\S 4.6]{Beaudry} for examples of this technique.
\end{rem}

Putting these together, we obtain the second page of the Adams spectral sequence computing $\widetilde{\ko}_*(B\SL(3,\F_2))$.
\begin{cor}\label{PSL_E2_gens}
In the range $t-s\le 8$, the $E_2$-page is generated as an $\Ext_{\cA(1)}(\Z_2)$-module by nine classes subject to some relations:
\begin{itemize}
    \item The submodule $\Ext(\textcolor{BrickRed}{\Sigma^2 J})$ has generators $a\in\Ext^{0,2}$, $b\in\Ext^{1,5}$, and $c\in\Ext^{2,10}$ and relations $h_0a = 0$, $h_1a = 0$, $va = 0$, $wa = h_1^2c$, $h_1b = 0$, $vb = h_0^2c$, and $vc = wb$.
    \item The submodule $\Ext(\textcolor{Green}{\Sigma^3Q})$ 
    has generators $d\in\Ext^{0,3}$ and $e\in\Ext^{1,8}$ subject to the relations $h_1d = 0$, $vd = h_0^2e$, and $ve = wd$.
    \item The submodule $\Ext(\Sigma^6\cA(1))$ has generator $f\in\Ext^{0,6}$ with $h_0f = 0$, $h_1f = 0$, $vf = 0$, and $wf = 0$.
    \item The submodule $\Ext(\textcolor{MidnightBlue}{\Sigma^7 R_2})$ has generators $g\in\Ext^{0,7}$ and $h\in\Ext^{0,8}$ with relations $h_1g = 0$, $h_0h = 0$, and $h_1^2h = 0$ (as well as additional generators and relations in topological degree greater than $10$).
    \item The submodule $\Ext(\textcolor{Fuchsia}{\Sigma^8\Z_2})$ has generator $i\in\Ext^{0,8}$ with no relations.
\end{itemize}
The following classes are detected in mod $2$ cohomology: $a$ by $\nu_2$, $d$ by $\overline\nu_3$, $f$ by $\nu_2^3$, $g$ by $\nu_2^2\overline\nu_3$, $h$ by $\nu_2\overline\nu{}_3^2$, and $i$ by $\nu_2^4$. The classes $b$, $c$, and $e$ are not detected by mod $2$ cohomology.
\end{cor}

The convergence, though easily determined using the fact that differentials in the Adams spectral sequence are equivariant for the $\ExZ$-action, was determined by Bayen in \cite{Bayen}. Differentials are determined by their values on the generators, which are $d_2(b) = h_0^3d$, $d_2(c) = h_0^3e$,  and $d_2(i) = h_0^2g$; all other differentials on the nine generators vanish. We summarize in Figure \ref{AdamsPSL(2,7)}. 

\begin{figure}[H]
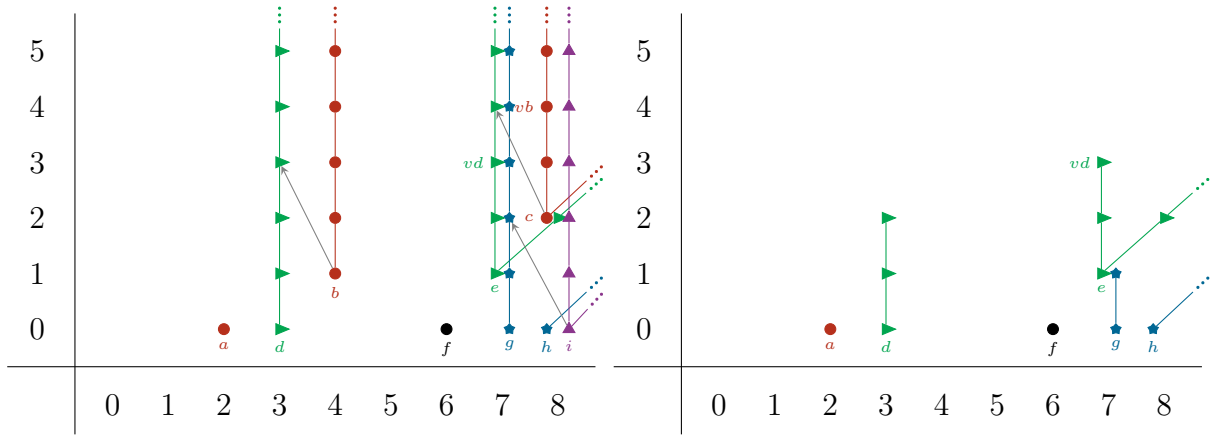

\centering
\begin{subfigure}[c]{0.48\textwidth}
\begin{sseqdata}[name=BPSL2E2, classes=fill, xrange={0}{8}, yrange={0}{5}, 
%class labels = {below = 0.05em, font=\tiny},
scale=0.74, Adams grading, >=stealth,
class labels = {below = 0.05em, font=\tiny}]
% fake classes to make the columns align better
\begin{scope}[draw=none, fill=none]
    \class(7, 0)
    \class(8, 1)\class(8, 1)
\end{scope}
\begin{scope}[BrickRed]
	\class(2, 0)
	\class(4, 1)\AdamsTower{}
	\class(8, 2)\AdamsTower{}
	\class(9, 3)
		\structline(8, 2)(9, 3, -1)
	\class(10, 4)\structline
    \classoptions["a"](2, 0)
    \classoptions["b"](4, 1)
    \classoptions[class labels = {left = 0.05em}, "c"](8, 2)
    \classoptions[class labels = {left = 0.05em}, "vb"](8, 4)
    \classoptions["wa"](10, 4)
\end{scope}
\begin{scope}[Green, isosceles triangle]
	\class(3, 0)\AdamsTower{}
	\class(7, 1)\AdamsTower{}
	\class(8, 2)
		\structline(7, 1)(8, 2,-1)
	\class(9, 3)\structline
 \classoptions["d"](3, 0)
 \classoptions["e"](7, 1, -1)
 \classoptions[class labels = {left = 0.05em}, "vd"](7, 3, -1)
\end{scope}
\begin{scope}
	\class(6, 0)
  \classoptions["f"](6, 0)
\end{scope}
\begin{scope}[MidnightBlue, star]
	\class(7, 0)\AdamsTower{}
	\class(8, 0)
	\class(9, 1)\structline
 \classoptions["g"](7, 0, -1)
 \classoptions["h"](8, 0, -1)
\end{scope}
\begin{scope}[draw=none, fill=none]
    \foreach \y in {0, 3, 4, 5, 6} {
        \class(8, \y)
    }
\end{scope}
\begin{scope}[Fuchsia, regular polygon, regular polygon sides=3, minimum width=1ex]
	\class(8, 0)\AdamsTower{}
	\class(9, 1)
		\structline(8, 0, -1)(9, 1, -1)
	\class(10, 2)\structline
 \classoptions["i"](8, 0, -1)
\end{scope}
\begin{scope}[MidnightBlue, rectangle]
	\class(10, 0)
  \classoptions["j"](10, 0)
\end{scope}
\begin{scope}[draw=none, fill=none]
	\class(3, 7)
	\class(7, 7)\class(7, 7)
\end{scope}
\begin{scope}[gray]
    \d2(4, 1)
    \d2(8, 2, 1)(7, 4, 1)
    \d2(8, 0, 3)(7, 2, 2)
\end{scope}
\begin{scope}[draw=none]
\foreach \y in {1, ..., 5} {
    \d2(8, \y, 3)(7, \y+2, 2)
}
\foreach \y in {2, ...,5} {
    \d2(4, \y)

}
\foreach \y in {3, ..., 5} {
    \d2(8, \y, 1)(7, \y+2, 1)
}
\end{scope}
\d3(10,0)(9,3,2)
\end{sseqdata}
\printpage[name=BPSL2E2, page=2]
\end{subfigure}
%\\\vspace{0.3cm}
\begin{subfigure}[c]{0.48\textwidth}
	\printpage[name=BPSL2E2, page=3]
\end{subfigure}
\caption{\label{AdamsPSL(2,7)}The Adams spectral sequence computing $\widetilde{\ko}_*(B\SL(3,\F_2))$. Left: the $E_2$ page. Right: the $E_4 = E_\infty$ page, up to degree 8.}
\label{the_E2_page}
\label{ThomZ4Einfty}
\end{figure}
We are interested in (the $2$-completion of) $\widetilde{\ko}_*(B\SL(3,\F_2))$ up to degree 7, which amounts to resolving extension questions on the $E_\infty$-page. This has been done by Bayen. 
\begin{prop}[{Bayen~\cite[Theorem 3.6.1]{Bayen}}]\label{algebra_koPSL}
\hfill
\begin{enumerate}
\item
There is an isomorphism of $\Z$-graded $\ko_*$-modules
\begin{subequations}
\begin{equation}\label{PSL_gens}
    \widetilde{\ko}_*(B\SL(3,\F_2)) \otimes \hat\Z_2 \cong (\ko_* \otimes \hat\Z_2)_*\set{\alpha, \delta, \epsilon, \zeta, \dotsc}/\mathcal R_3 \oplus \Sigma^6\Z_2\cdot \phi,
\end{equation}
where $\abs\alpha = 2$, $\abs\delta = 3$, $\abs\phi = 6$, and $\abs\epsilon = \abs\zeta = 7$. The ideal $\mathcal R_3$ of relations is
\begin{equation}\label{PSL_relns}
    \mathcal R_3 = (2\alpha, \eta\alpha, v\alpha, 8\delta, \eta\delta, 8\epsilon, 4\epsilon - v\delta, 4\zeta - 2\epsilon, \dotsc)
\end{equation}
\end{subequations}
and in both~\eqref{PSL_gens} and~\eqref{PSL_relns}, all unlisted classes are in degrees $8$ and above.
\item The above isomorphism can be chosen so that the images of the generators in the $E_\infty$-page of the Adams spectral sequence are: $\alpha\mapsto a$, $\delta\mapsto d$, $\phi\mapsto f$, $\epsilon\mapsto e$, and $\zeta\mapsto g$.
\end{enumerate}
\end{prop}
In particular,
\begin{equation}\label{BayenDeg7}
\begin{alignedat}{2}
    \widetilde{\ko}_2(B\SL(3,\F_2))_2^\wedge &\cong \Z_2 \cdot\textcolor{BrickRed}{\alpha}, \qquad\qquad
    &\widetilde{\ko}_3(B\SL(3,\F_2))_2^\wedge &\cong \Z_8\cdot\textcolor{Green}{\delta},\\
    \widetilde{\ko}_6(B\SL(3,\F_2))_2^\wedge &\cong \Z_2\cdot \phi, &\widetilde{\ko}_7(B\SL(3,\F_2))_2^\wedge &\cong \Z_2\cdot(\textcolor{Green}{\epsilon} - 2\textcolor{MidnightBlue}{\zeta})
        \oplus \Z_{16}\cdot\textcolor{MidnightBlue}{\zeta}.
    \end{alignedat}
\end{equation}

$L(2)$, in fact, is very easy in the degrees we are interested in; for an explicit description we need to use the cohomology of $BS_4$.
\begin{thm}[{Nakaoka~\cite[Theorem 4.1]{Nak62}, Madsen-Milgram~\cite[Example 3.31]{MM79}}] \label{Thm:CohoS4}
    \begin{equation}
    \begin{gathered}
        H^*(BS_4;\Z_2) \cong \Z_2[a,b,c]/(ac),\\
        |a| =1,\quad |b|=2,\quad |c| = 3
    \end{gathered}
    \end{equation}
    The Steenrod squares of the generators are 
    \begin{equation}
        \begin{gathered}
            \Sq(a) = a+a^2,\quad \Sq(b) = b + ab + c+ b^2,\quad \Sq(c) = c + bc + c^2.
        \end{gathered}
    \end{equation}
\end{thm}
The cohomology ring is due to Nakaoka; Madsen-Milgram realized the generators as Stiefel-Whitney classes, so that the Steenrod squares given here follow from the Wu formula.
\begin{rem}
\label{Rem:PSL(2,7)andS4}
    Comparing \cref{Thm:SteenrodBPSL27,Thm:CohoS4} we see immediately that $\nu_2,\nu_3,\overline \nu_3 \in H^*(B\SL(3,\F_2);\Z_2)$ correspond to $b,ab+c,c\in H^*(BS_4;\Z_2)$, respectively. This is also visible in the classes named on the two sides of \cref{CohoBPSL27}.
\end{rem}

\begin{lem}[{Bayen~\cite[\S 3.5.3]{Bayen}}]
\label{A1_triv_L2}
\label{Prop:cohoL2}
$H^*(L(2);\Z_2)$ is a free $\cA(1)$-module on a countably infinite basis consisting of a class in degree $4$ and other classes in degrees greater than $7$.
\end{lem}
\begin{prop}
\label{complete_L2_decomp}
\hfill
\begin{enumerate}
    \item\label{from_Marg} $\Omega_k^\Spin(L(2))$ vanishes for $k\le 7$ except for $k = 4$, where it is isomorphic to $\Z_2$.
    \item\label{coh_S4_A1_free} Recall the map $\psi_4$ from \cref{S4_splitting}. The map
    \begin{equation}\label{detecting_L2_S4}
        \Omega_4^\Spin(BS_4)\overset{\psi_4}{\longrightarrow} \Omega_4^\Spin(L(2)) \overset\cong\longrightarrow \Z_2
    \end{equation}
    is the bordism invariant $\int a^2b$, where $a$ and $b$ are the classes introduced in \cref{Thm:CohoS4}. %\textcolor{red}{MD: Maybe refer to cohomology of $S_4$ below.}
    \item\label{smash_with_Z4L2} $\Omega_k^\Spin(B\Z_4\wedge L(2))$ is $0$ for $0\le k\le 4$ and $\Z_2$ for $5\le k\le 7$. Analogously to~\eqref{detecting_L2_S4}, these summands of $\Omega_*^\Spin(B\Z_4\wedge L(2))$ are detected by $\int xa^2b$, $\int ya^2b$, and $\int xya^2b$ in degrees $5$, $6$, and $7$ respectively.
\end{enumerate}
\end{prop}
\begin{proof}
Part~\eqref{from_Marg} is a direct consequence of Margolis' theorem~\cite{Margolis}: free $\cA(1)$-module summands in cohomology lift to $H\Z_2$ summands after smashing with $\ko$. If $\MTSpin$ denotes the spectrum whose generalized homology theory is spin bordism, then $H^*(\MTSpin;\Z_2)$ is isomorphic to $\cA\otimes_{\cA(1)}M$ for some $\cA(1)$-module $M$~\cite[Theorem B]{Pen83}, so the change-of-rings trick (see for example~\cite[\S 4.5]{Beaudry}) used to apply Margolis' theorem to $\ko$-homology also works for spin bordism, and therefore Margolis' theorem applies in the same way over $\MTSpin$.

For part~\eqref{coh_S4_A1_free} , part~\eqref{from_Marg} and \cref{S4_splitting} imply that $H^*(BS_4;\Z_2)$ contains a $\Sigma^4\cA(1)$ summand corresponding to $L(2)$ as a summand of $\ko\wedge BS_4$, and that this summand can be chosen to be any $\Sigma^4\cA(1)$ summand which is not in the image on cohomology of the maps $\sigma\colon BS_4\to B\Z_2$ or $m\circ q\colon BS_4\to B\SL(3, \F_2)$. Then a quick computation of the $\cA(1)$-module structure on $\widetilde H^*(BS_4;\Z_2)$ in low degrees shows that $a^2b$ generates a free summand satisfying these criteria.

For part~\eqref{smash_with_Z4L2}, begin with the fact that for any $\cA(1)$-module $M$, $\cA(1)\otimes_{\Z_2} M$ is a free $\cA(1)$-module on a basis of the form $\set{\alpha\otimes \beta_i}$, where $\alpha$ is the generator of $\cA(1)$ and $\beta_i$ is a basis of $M$ as a $\Z_2$-vector space. Since $\widetilde H^*(B\Z_4;\Z_2)$ has a vector space basis $\set{x, y, xy, y^2, xy^2, y^3,\dots}$ (\cref{Thm:CohoZ4}), then $\widetilde H^*(B\Z_4\wedge L(2);\Z_2)$ in degrees $7$ and below is a free $\cA(1)$-module on the classes $xa^2b$, $ya^2b$, and $xya^2b$ in degrees $5$, $6$, and $7$ respectively. Margolis' theorem then finishes the proof for us.
\end{proof}

\subsection{Generators}\label{sec:BPSLgens}

\begin{table}[h!]
\centering
\begin{tabular}{c c c c c c}
\toprule
Manifold & Bordism class & Image in $E_\infty$ & Dimension & Char class & Reference\\
\midrule
$T^2$ & $\alpha$ & $a$ & $2$ & $\nu_2$& Prop.\ \ref{f_manifold}
    \\
$L_4^3$ & $\delta$ & $d$ & $3$ & $\overline\nu_3$ & Prop.\ \ref{lens_S4}
    \\
$\RP^3\times\RP^3$ & $\phi$ & $f$ & $6$ & $\nu_2^3$ & Prop.\ \ref{f_manifold}
    \\
$L_4^3\times_{\Z_2} K3$ & $\epsilon$ & $e$ & $7$ & n/a & \S \ref{sec:X7gen}
    \\
$L_4^7$ & $\zeta$ & $g$ & $7$ & $\nu_2^2\overline\nu_3$ & Prop.\ \ref{lens_S4}
    \\
$W_4$ &  &  & $4$ & $a^2b$ & \S\ref{sec:W4}\\
\bottomrule
\end{tabular}
\caption{Generators of $\widetilde\Omega_*^\Spin(B\SL(3,\F_2))$. We also include the generator $W_4$ of $\Omega_4^\Spin(L(2))$, which we need for the spin bordism of $BS_4$.}
\label{tab:S4gens}
\end{table}

We now move to the generators of $\ko_*(B\SL(3,\F_2))$ and $\ko_*(L(2))$ for degree less than 8. A key result by which we will determine these generators is the stable splitting of $BS_4$ given in \cref{S4_splitting}. With this, we can study generators for $\ko_*(B\SL(3,\F_2))$ and $\ko_*(L(2))$ by instead considering the generators of $\ko_*(BS_4)$. 

\subsubsection{\texorpdfstring{Generators coming from $\ko_*(B\Z_4)$ and $\ko_*(B(\Z_2\times \Z_2))$}{Generators coming from ko(BZ4) and ko(B(Z2 x Z2))}}\label{Z4_to_S4_gens}

Many of the generators of $\ko_*(BS_4)$ can be determined by considering subgroups of $S_4$. Indeed, consider the inclusion $ \Z_4 \hookrightarrow S_4$ given by $1\mapsto (1\;2\;3\;4)$. Thinking of $S_4$ as the tetrahedral symmetry group, this map can also be characterized as sending the generator of $\Z_4$ to the tetrahedral symmetry given by the composition of a $90\degree$ turn in the $\R^2_{xy}$-plane with a reflection in $\R_z$, see Figure \ref{fig:Z4inS4}. The inclusion induces a map on $\Z_2$ cohomology. For this, first recall the $\Z_2$ cohomology of $B\Z_4$ from \cref{Thm:CohoZ4}.
\begin{figure}
\centering
\includegraphics[width = 0.8 \textwidth]{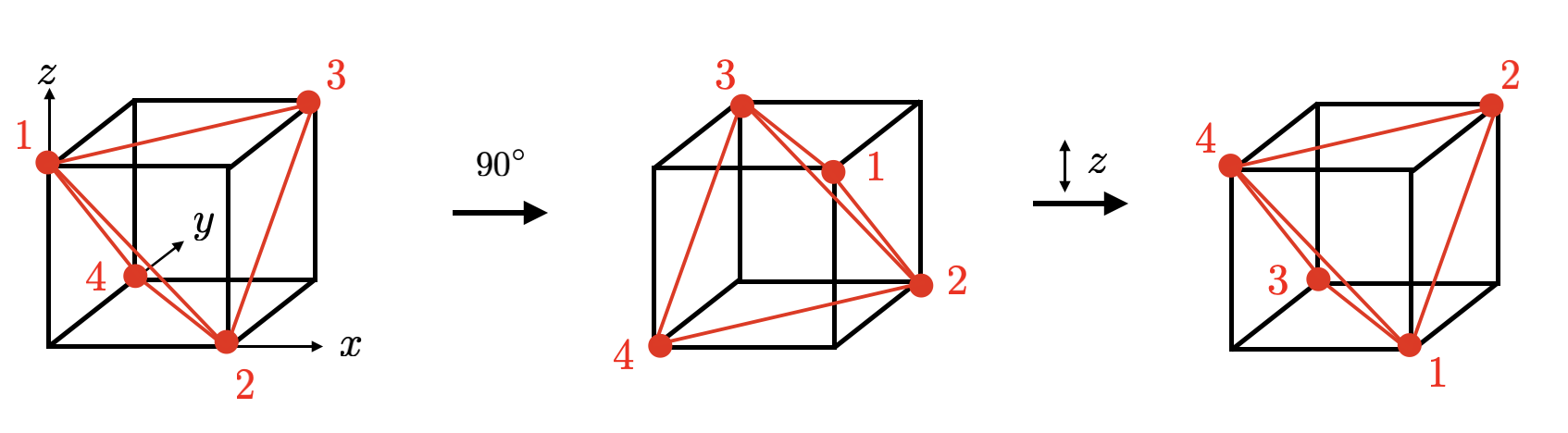}
\caption{The permutation $(1 \; 2 \; 3 \; 4)$ generating a $\mathbb{Z}_4$ as an element of $S_4$ acting on a tetrahedron (inside a cube).} 
\label{fig:Z4inS4}
\end{figure}

With Theorems \ref{Thm:CohoS4} and \ref{Thm:CohoZ4} in place, we can study the map on $\Z_2$ cohomology induced from the inclusion $\Z_4 \hookrightarrow S_4$. 
\begin{lem}
\label{lem:Z4inS4}    The inclusion $ \Z_4 \hookrightarrow S_4$ given by $1 \mapsto (1\;2\;3\;4)$ induces a map on $\Z_2$ cohomology 
    \begin{equation}\label{defnPhi}
    \begin{gathered}
        \Phi:H^*(BS_4;\Z_2) \rightarrow H^*(B\Z_4;\Z_2);\\
        a\mapsto x,\quad b\mapsto y,\quad c\mapsto xy
    \end{gathered}
    \end{equation}
\end{lem}
\begin{proof}
    The classes $a,b,c\in H^*(BS_4;\Z_2)$ are the first, second, and third Stiefel-Whitney classes of the representation of $S_4$ on $\R^3$ as the symmetries of the tetrahedron. Restricting to $\Z_4$, we obtain a $\Z_4$ \term{rotoreflection} representation $\rho$, the direct sum of a $\Z_4$ rotation symmetry on the $xy$-plane and a reflection symmetry $\Z_4\to\Z_2\to\GL_1(\R)$ on the $z$-axis. This representation is not contained in $\mathrm{SO}(3)$, so $w_1(\rho)\ne 0$, and one can also show $w_2(\rho)\ne 0$, e.g.\ because $w_2$ of the rotation representation is nonzero and then using the Whitney sum formula. Thus $w_1(\rho) = x$ and $w_2(\rho) = y$, since those are the only nonzero options, and therefore
    $a,b$ pullback to $x,y$, respectively. The fact that $c \mapsto xy$ then follows from the fact that $\Sq^1(b) = ab + c$ and $\Sq^1(y) = 0$. 
\end{proof}
\begin{lem}\label{image_Adams_lem}
Let
\begin{equation}
    \Phi_*\colon\Ext_{\cA(1)}(H^*(B\Z_4;\Z_2), \Z_2) \longrightarrow \Ext_{\cA(1)}(H^*(BS_4;\Z_2), \Z_2)
\end{equation}
be the map of Adams $E_2$-pages induced by the inclusion $\Z_4\hookrightarrow S_4$. Then
\begin{equation}\label{Ext_Phi}
\begin{alignedat}{3}
    \Phi_*(\theta) &= 0\qquad\qquad &\Phi_*(\theta\lambda_1) &= 0\qquad\qquad & \Phi_*(\theta\mu_1) &= h_0e\\
    \Phi_*(\kappa_1) &= a & \Phi_*(\kappa_2) &= f & \Phi_*(\xi_1) &= h_0c\\
    \Phi_*(\theta\kappa_1) &= d & \Phi_*(\mu_1) &= 0 & \Phi_*(\lambda_2)&= h_0i.\\
    \Phi_*(\lambda_1) &= b & \Phi_*(\theta\kappa_2) &= g
\end{alignedat}
\end{equation}
\end{lem}
The reader can check that, restricted to classes in degrees $8$ and below, the description of $\Phi_*$ in~\eqref{Ext_Phi} extends uniquely to an $\mathbb E$-module homomorphism.
\begin{proof}
In \cref{Thm:CohoZ4} and \cref{cor:A(1)decompPSL}, we gave direct-sum decompositions of $H^*(B\Z_4;\Z_2)$ and $H^*(B\SL(3,\F_2);\Z_2)$ as $\cA(1)$-modules; the map $\Phi^*\colon H^*(B\SL(3,\F_2);\Z_2)\to H^*(B\Z_4;\Z_2)$ mostly sends direct summands to direct summands. Specifically, using \cref{lem:Z4inS4} to evaluate $\Phi^*$ on an additive basis of $H^*(B\SL(3,\F_2);\Z_2)$ in degrees $*\le 8$, we obtain the following description of $\Phi^*$.
\begin{enumerate}
    \item There is a unique nonzero $\cA(1)$-module homomorphism $\phi_{\textcolor{BrickRed}{J}}\colon \textcolor{BrickRed}{\Sigma^2 J}\to \Sigma^2 C\eta$; it is surjective, with kernel $\Sigma^3 Q$.
    \item There is a unique nonzero $\cA(1)$-module homomorphism $\phi_{\textcolor{Green}{Q}}\colon \textcolor{Green}{\Sigma^3 Q}\to\Sigma^3 C\eta$; it is surjective, with kernel $\Sigma^6\Z_2$.
    \item There is a unique $\cA(1)$-module homomorphism $\phi_{\textcolor{Fuchsia}{\Z_2}}\colon \Sigma^6\cA(1)\oplus\textcolor{Fuchsia}{\Sigma^8 \Z_2}\to \Sigma^6 C\eta$ which is nonzero when restricted both to $\Sigma^6 \cA(1)$ and to $\textcolor{Fuchsia}{\Sigma^8 \Z_2}$; it is surjective, with kernel $\Sigma^7 R_2$.
    \item There is a unique nonzero $\cA(1)$-module homomorphism $\phi_{\textcolor{MidnightBlue}{R_2}}\colon \textcolor{MidnightBlue}{\Sigma^7 R_2}\to \Sigma^7 C\eta$; it is surjective, with kernel $\Sigma^8 J$.
\end{enumerate}
Then, there is some map $\phi_P\colon P\to \widetilde H^{\ge 10}(B\Z_4;\Z_2)$ such that, under the isomorphisms given in \cref{A1_mod_BZ4,cor:A(1)decompPSL},
\begin{equation}
    \Phi = \phi_{\textcolor{BrickRed}{J}}\oplus \phi_{\textcolor{Green}{Q}}\oplus \phi_{\textcolor{Fuchsia}{\Z_2}}\oplus \phi_{\textcolor{MidnightBlue}{R_2}}\oplus \phi_P\colon \widetilde H^*(B\SL(3,\F_2);\Z_2) \longrightarrow \widetilde H^*(B\Z_4;\Z_2).
\end{equation}
Thus it suffices to compute the effects of $\phi_{\textcolor{BrickRed}{J}}$, $\phi_{\textcolor{Green}{Q}}$, $\phi_{\textcolor{Fuchsia}{\Z_2}}$, and $\phi_{\textcolor{MidnightBlue}{R_2}}$ on Ext. For this we use the fact that each of these four maps is surjective, hence it and its kernel are a short exact sequence of $\cA(1)$-modules, which induces a long exact sequence in Ext. In practice, knowing the Ext of each of the three pieces of the short exact sequence usually dictates the maps on Ext by exactness, and that is true in all four cases under study here. See~\cite[\S 4.6]{Beaudry} for information on how to display and compute these long exact sequences; we run them in \cref{Z4inclS4}, and from those charts we can read off the lemma statement.
\end{proof}
\begin{rem}
The map $\Phi_*$ of Adams $E_2$-pages sends $\theta\mapsto 0$ but does not kill $\theta\kappa_i$. This is an algebraic manifestation of the fact that even though $\Omega_*^\Spin(B\Z_4)$ is a ring, and it has a map to $\Omega_*^\Spin(BS_4)$, the latter is not a module over the former; only over $\Omega_*^\Spin$. This is similar to the situation we observed with Spin-$\Mp(2, \Z)$ and Spin-$\GL^+(2, \Z)$ bordism in~\cite[\S 14.4]{Heckman}.
\end{rem}
\begin{prop}\label{lens_S4}
The map $\ko_*(B\Z_4)\to ko_*(BS_4)$ induced by the inclusion $\Z_4\hookrightarrow S_4$ from \cref{lem:Z4inS4} sends $\ell_3\mapsto \delta$ and $\ell_7\mapsto\zeta$, and therefore $L_4^3$ with either of its two spin structures and $S_4$-bundle induced from the $\Z_4$-bundle $S^4\to L_4^3$ represents $\delta$, and likewise for $L_4^7$ and $\zeta$.
\end{prop}
\begin{proof}
Recall from~\eqref{Ext_Phi} that
$\theta\kappa_1\mapsto d$ and $\theta\kappa_2\mapsto g$. All four of these classes survive to the $E_\infty$-page, so we can lift the map on $E_\infty$-pages to a map on $\ko$-homology using \cref{ringstr_Z4,algebra_koPSL} and conclude $\ell_3\mapsto \delta$ and $\ell_7\mapsto \zeta$. Finally, using the generators for $\ko_*(B\Z_4)$ we found in \cref{Z4_mult_gens}, we obtain the second part of the theorem statement. 
\end{proof}
\begin{rem}
Looking at~\eqref{Ext_Phi}, the classes $a\in\Ext^{0,2}$ and $f\in\Ext^{0,6}$ are also in the image of $\Phi_*$; specifically, $a = \Phi_*(\kappa_1)$ and $f = \Phi_*(\kappa_2)$. However, since $d_2(\kappa_i)\ne 0$ for $i = 1,2$, $a$ and $f$ are not in the image of $\Phi_*$ on the $E_\infty$-page, so the proof technique of \cref{lens_S4} does not furnish representing manifolds for $\alpha\in\ko_2(BS_4)$ or $f\in\ko_6(BS_4)$. We will find generators using a similar technique on a different subgroup of $S_4$ in \cref{f_manifold}.
\end{rem}

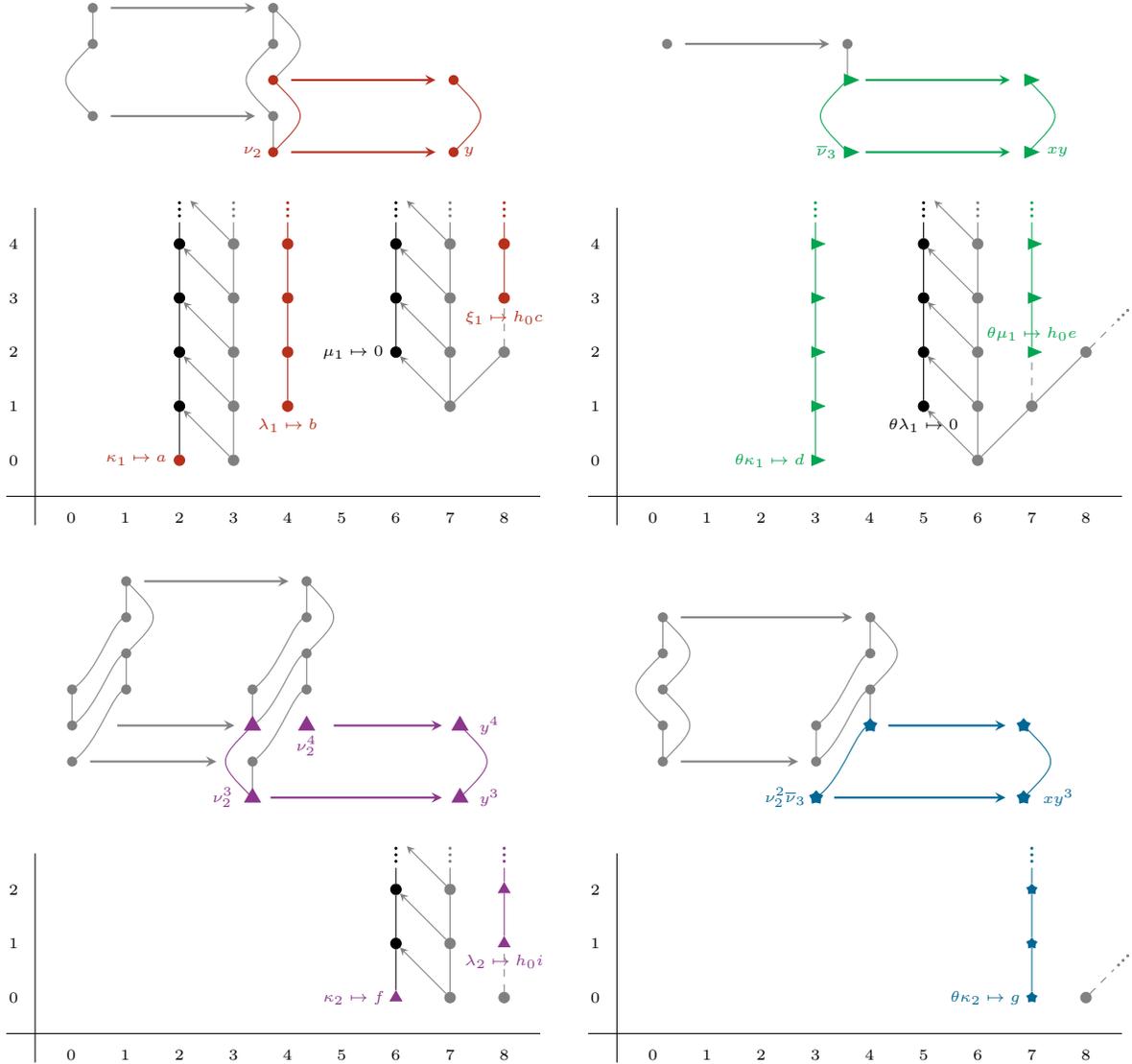
\begin{figure}[h!]
\centering
% Joker
\begin{subfigure}[b]{0.48\textwidth}
\centering
\begin{subfigure}[c]{\textwidth}
\centering
\begin{tikzpicture}[scale=0.5, every node/.style = {font=\tiny}, >=stealth]
\begin{scope}[gray]
	\SpanishQnMark{-5}{1}{}{}
	\SpanishQnMark{0}{1}{}{}
	\draw[->, thick] (-4.5, 1) -- (-0.5, 1);
	\draw[->, thick] (-4.5, 4) -- (-0.5, 4);
	\sqone(0, 0);
	\sqtwoR(0, 2);
\end{scope}
\begin{scope}[BrickRed]
	\foreach \x in {0, 5} {
		\tikzpt{\x}{0}{}{};
		\tikzpt{\x}{2}{}{};
		\sqtwoR(\x, 0);
	}
	\draw[thick, ->] (0.5, 0) -- (4.5, 0);
	\draw[thick, ->] (0.5, 2) -- (4.5, 2);
 \node[left] at (0, 0) {$\nu_2$};
 \node[right] at (5, 0) {$y$};
\end{scope}
\tikzpt{0}{-1}{}{white};
\end{tikzpicture}
\end{subfigure}
\begin{sseqdata}[name=jokerImage, classes=fill, xrange={0}{8}, yrange={0}{4}, scale=0.75, Adams grading, >=stealth,
tick style={font=\tiny}, x axis tail=0.4cm, y axis tail=0.4cm, x tick gap=0.3cm, y tick gap=0.3cm, class labels = {below = 0.05em, font=\tiny}]
\class(2, 1)\AdamsTower{}
\class(6, 2)\AdamsTower{} \classoptions[class labels = {left = 0.05em},"\mu_1\mapsto 0"](6, 2)
\class(10, 5)\AdamsTower{}
\begin{scope}[gray] % TODO maybe another color?
	\class(3, 0)\AdamsTower{}
	\class(7, 1)\AdamsTower{}
	\class(8, 2)
		\structline(7, 1)(8, 2)
	\class(9, 3)
\end{scope}
\begin{scope}[BrickRed]
	\class(2, 0)
 \classoptions[class labels = {left = 0.05em}, "\kappa_1\mapsto a"](2, 0)
	\class(4, 1)\AdamsTower{}
  \classoptions["\lambda_1\mapsto b"](4, 1)
	\class(8, 3)\AdamsTower{}
 \classoptions["\xi_1\mapsto h_0c"](8, 3)
	\class(10, 4)
\end{scope}
\structline(2, 0)(2, 1)
\structline[dashed, gray](8, 2)(8, 3)
\structline[dashed, gray](9, 3)(10, 4)
\structline(10, 4)(10, 5)
% now differentials
\foreach \y in {0, ..., 4} {
	\d[gray]1(3, \y)
}
\foreach \y in {1, ..., 4} {
	\d[gray]1(7, \y)
}
\end{sseqdata}
\printpage[name=jokerImage, page=1]
\end{subfigure}
% upside-down qn mark
\begin{subfigure}[b]{0.48\textwidth}
\centering
\begin{subfigure}[c]{\textwidth}
\centering
\begin{tikzpicture}[scale=0.5, every node/.style = {font=\tiny}, >=stealth]
\begin{scope}[gray]
	\tikzpt{-5}{3}{}{};
	\tikzpt{0}{3}{}{};
	\draw[->, thick] (-4.5, 3) -- (-0.5, 3);
	\sqone(0, 2);
\end{scope}
\begin{scope}[Green]
	\foreach \x in {0, 5} {
		\tikzpt{\x}{0}{}{isosceles triangle};
		\tikzpt{\x}{2}{}{isosceles triangle};
	}
	\sqtwoL(0, 0);
	\sqtwoR(5, 0);
	\draw[thick, ->] (0.5, 0) -- (4.5, 0);
	\draw[thick, ->] (0.5, 2) -- (4.5, 2);
  \node[left] at (0, 0) {$\overline\nu_3$};
 \node[right=0.125cm] at (5, 0) {$xy$};
\end{scope}
\tikzpt{0}{-1}{}{white};
\end{tikzpicture}
\end{subfigure}
\begin{sseqdata}[name=qnImage, classes=fill, xrange={0}{8}, yrange={0}{4}, scale=0.75, Adams grading, >=stealth,
tick style={font=\tiny}, x axis tail=0.4cm, y axis tail=0.4cm, x tick gap=0.3cm, y tick gap=0.3cm,
 class labels = {below = 0.05em, font=\tiny}]
\class(5, 1)\AdamsTower{}
 \classoptions["\theta\lambda_1\mapsto 0"](5, 1)
\class(9, 4)\AdamsTower{}
\begin{scope}[gray]
	\class(6, 0)\AdamsTower{}
	\class(7, 1)
		\structline(6, 0)(7, 1)
	\class(8, 2)\structline
	\class(10, 3)\AdamsTower{}
\end{scope}
\begin{scope}[Green, isosceles triangle]
	\class(3, 0)\AdamsTower{}
 \classoptions[class labels = {left = 0.05em},"\theta\kappa_1\mapsto d"](3, 0)
	\class(7, 2)\AdamsTower{}
  \classoptions[class labels = {above = 0.05em}, "\theta\mu_1\mapsto h_0e"](7, 2)
	\class(9, 3)
\end{scope}
\structline(9, 3)(9, 4)
\structline[dashed, gray](7, 1)(7, 2)
\structline[dashed, gray](8, 2)(9, 3)
% now differentials
\foreach \y in {0, ..., 4} {
	\d[gray]1(6, \y)
}
%\foreach \y in {3, ..., 4} {
%	\d[gray]1(10, \y)
%}
\end{sseqdata}
\printpage[name=qnImage, page=1]
\end{subfigure}
% F2
\begin{subfigure}[b]{0.48\textwidth}
\centering
\begin{subfigure}[c]{\textwidth}
\centering
\begin{tikzpicture}[scale=0.5, every node/.style = {font=\tiny}, >=stealth]
\begin{scope}[gray]
	\Rtwo{-5.75}{1}{}{}{};
	\tikzpt{-0.75}{1}{}{};
	\tikzpt{-0.75}{3}{}{};
	\tikzpt{0.75}{3}{}{};
	\tikzpt{0.75}{4}{}{};
	\tikzpt{0.75}{5}{}{};
	\tikzpt{0.75}{6}{}{};
	\sqone(-0.75, 0);
	\sqone(-0.75, 2);
	\sqone(0.75, 3);
	\sqone(0.75, 5);
	\sqtwoCR(-0.75, 1);
	\sqtwoCR(-0.75, 2);
	\sqtwoCR(-0.75, 3);
	\sqtwoR(0.75, 4);
	\draw[->, thick] (-5.25, 1) -- (-1.75, 1);
	\draw[->, thick] (-4.5, 2) -- (-1.25, 2);
	\draw[->, thick] (-3.75, 6) -- (0.25, 6);
\end{scope}
\begin{scope}[Fuchsia] % TODO: F2
	\foreach \x in {-0.75, 5} {
		\tikzpt{\x}{0}{}{regular polygon,regular polygon sides=3};
		\tikzpt{\x}{2}{}{regular polygon,regular polygon sides=3};
	}
	\tikzpt{0.75}{2}{}{regular polygon,regular polygon sides=3};
	\sqtwoL(-0.75, 0);
	\sqtwoR(5, 0);
	\draw[thick, ->] (-0.25, 0) -- (4.5, 0);
	\draw[thick, ->] (1.5, 2) -- (4.5, 2);
 \node[left] at (-1, 0) {$\nu_2^3$};
 \node[below] at (0.75, 2) {$\nu_2^4$};
 \node[right=0.125cm] at (5, 0) {$y^3$};
 \node[right=0.125cm] at (5, 2) {$y^4$};
\end{scope}
\tikzpt{0}{-1}{}{white};
\tikzpt{0}{7}{}{white};
\end{tikzpicture}
\end{subfigure}
\begin{sseqdata}[name=FtwoImage, classes=fill, xrange={0}{8}, yrange={0}{2}, scale=0.75, Adams grading, >=stealth,
tick style={font=\tiny}, x axis tail=0.4cm, y axis tail=0.4cm, x tick gap=0.3cm, y tick gap=0.3cm, class labels = {below = 0.05em, font=\tiny}]
\class(6, 1)\AdamsTower{}
\class(10, 3)\AdamsTower{}
\begin{scope}[gray] 
	\class(7, 0)\AdamsTower{}
	\class(8, 0)
	\class(9, 1)
\end{scope}
\begin{scope}[Fuchsia, regular polygon, regular polygon sides=3, minimum width=1ex]
	\class(6, 0)
  \classoptions[class labels = {left = 0.05em},"\kappa_2\mapsto f"](6, 0)
	\class(8, 1)\AdamsTower{}
  \classoptions["\lambda_2\mapsto h_0i"](8, 1)
	\class(10, 2)
\end{scope}
\structline(6, 0)(6, 1)
\structline(10, 2)(10, 3)
\structline[dashed, gray](8, 0)(8, 1)
\structline[dashed, gray](9, 1)(10, 2)
% now differentials
\foreach \y in {0, ..., 2} {
	\d[gray]1(7, \y)
}
\end{sseqdata}
\printpage[name=FtwoImage, page=1]
\end{subfigure}
\centering
% R2
\begin{subfigure}[b]{0.48\textwidth}
\centering
\begin{subfigure}[c]{\textwidth}
\centering
\begin{tikzpicture}[scale=0.5, every node/.style = {font=\tiny}, >=stealth]
\begin{scope}[gray]
	\Joker{-5}{1}{}{}
	\tikzpt{-0.75}{1}{}{};
	\tikzpt{-0.75}{2}{}{};
	\tikzpt{0.75}{3}{}{};
	\tikzpt{0.75}{4}{}{};
	\tikzpt{0.75}{5}{}{};
	\sqone(-0.75, 1);
	\sqone(0.75, 2);
	\sqone(0.75, 4);
	\sqtwoCR(-0.75, 1);
	\sqtwoCR(-0.75, 2);
	\sqtwoR(0.75, 3);
	\draw[->, thick] (-4.5, 1) -- (-1.25, 1);
	\draw[->, thick] (-4.5, 5) -- (0.25, 5);
%       \sqone(0, 0);
%       \sqtwoR(0, 2);
\end{scope}
\begin{scope}[MidnightBlue]
	\tikzpt{-0.75}{0}{}{star};
	\tikzpt{0.75}{2}{}{star};
	\sqtwoCR(-0.75, 0);
	\tikzpt{5}{0}{}{star};
	\tikzpt{5}{2}{}{star};
	\sqtwoR(5, 0);
	\draw[thick, ->] (-0.25, 0) -- (4.5, 0);
	\draw[thick, ->] (1.25, 2) -- (4.5, 2);
   \node[left] at (-0.75, 0) {$\nu_2^2\overline\nu_3$};
 \node[right=0.125cm] at (5, 0) {$xy^3$};
\end{scope}
\tikzpt{0}{-1}{}{white};
\end{tikzpicture}
\end{subfigure}
\begin{sseqdata}[name=R2Image, classes=fill, xrange={0}{8}, yrange={0}{2}, scale=0.75, Adams grading, >=stealth,
tick style={font=\tiny}, x axis tail=0.4cm, y axis tail=0.4cm, x tick gap=0.3cm, y tick gap=0.3cm, class labels = {left = 0.05em, font=\tiny}]
\class(9, 2)\AdamsTower{}
\begin{scope}[gray] 
	\class(8, 0)
	\class(10, 1)\AdamsTower{}
\end{scope}
\begin{scope}[MidnightBlue,star]
	\class(7, 0)\AdamsTower{}
        \classoptions["\theta\kappa_2\mapsto g"](7, 0)
	\class(9, 1)
\end{scope}
\structline(9, 1)(9, 2)
\structline[dashed, gray](8, 0)(9, 1)
% now differentials
%\foreach \y in {1, ..., 2} {%
%	\d[gray]1(10, \y)
%}
\end{sseqdata}
\printpage[name=R2Image, page=1]
\end{subfigure}
\caption{\label{Z4inclS4}The map $\Phi_*$ from the $E_2$-page of the Adams spectral sequence for $\ko_*(B\Z_4)$ to the $E_2$-page for $\ko_*(BS_4)$ can be calculated summand by summand. This is a picture of the proof of \cref{image_Adams_lem}; see also~\cite[Figure 23]{Heckman}. Each Adams chart is a depiction of a long exact sequence in Ext, as in~\cite[\S 4.6]{Beaudry}; color is the image of the quotient map in the short exact sequence, and black is the kernel. Upper left: the map $\phi_{\textcolor{BrickRed}{J}}$ and its kernel, and the resulting long exact sequence in Ext. Upper right: the same, for $\phi_{\textcolor{Green}{Q}}$. Lower left: the same, for $\phi_{\textcolor{Fuchsia}{\Z_2}}$. Lower right: the same, for $\phi_{\textcolor{MidnightBlue}{R_2}}$. Names of cohomology classes are as in \cref{Thm:CohoZ4} and \cref{Thm:SteenrodBPSL27}; names for Ext classes are as in \cref{Z4_Adams_ring,PSL_E2_gens} (drawn in \cref{Z4_E2,the_E2_page}, left), and the maps $\phi_{\bullet}$ are as defined in the proof of \cref{image_Adams_lem}.}
\label{image_from_spin_Z8}
\end{figure}

Similarly, there exists an inclusion $\Z_2\times \Z_2 \hookrightarrow S_4$ which helps determine the generator of $\widetilde{\ko}_6(B\SL(3,\F_2))$. 
\begin{lem}
\label{lem:Z2Z2S4inclu}
    The inclusion $j\colon \Z_2\times \Z_2 \hookrightarrow S_4$ defined by the transpositions $(1\;2)(3\;4)$ and $(1\;3)(2\;4)$ induces a map on $\Z_2$ cohomology 
    \begin{equation}
        \begin{gathered}
            H^*(BS_4;\Z_2)\rightarrow H^*(B(\Z_2\times \Z_2);\Z_2) \\
            a\mapsto 0,\quad b\mapsto x^2 + xy + y^2,\quad c\mapsto x^2y + xy^2
        \end{gathered}
    \end{equation}
    where $x,y\in H^1(B(\Z_2\times \Z_2);\Z_2)$ are the generators corresponding to the first, resp.\ second $\Z_2$ summands. 
\end{lem}
\begin{proof}
    The map $\Z_2 \times  \Z_2 \rightarrow  S_4$ factors through $A_4$, so we can first ask how $a$, $b$, and $c$ pull back to $A_4$, then restrict from $A_4$ to $\Z_2 \times  \Z_2$.
    
    For $S_4 \rightarrow  A_4$, $H^1(BA_4; \Z_2) = 0$, so $a$ pulls back to $0$. For the class $b$, note that it is the second Stiefel-Whitney class  of the representation of $S_4$ on $\R^3$ as the symmetries of a tetrahedron. Restricting to $A_4$, we obtain the orientation-preserving symmetries of the tetrahedron. This representation is not spin (its spin lift is the binary tetrahedral group, which is not $A_4 \times \Z_2$), so $w_2$ is nonzero. The only option in $H^2(BA_4; \Z_2) = \Z_2$ is the nonzero element, which we will denote as $u$ to match with the notation of \cite{Debray}. Finally, $ab + c = \Sq^1(b)$, so in $BA_4$, $c$ pulls back to $\Sq^1(u)$. 
    
    To get from $A_4$ to $\Z_2 \times  \Z_2$, we recall that $u$ pulls back to $x^2 + xy + y^2$; then $c$ pulls back to $\Sq^1(x^2 + xy + y^2)$; see \cite[Proposition 5.1]{Debray} for reference. 
\end{proof}
\begin{prop}\label{f_manifold}
Let $n = 1$ or $n = 3$.
Give $\RP^n\times\RP^n$ the ``tautological $(\Z_2\times\Z_2)$-bundle,'' i.e.\ the $(\Z_2\times\Z_2)$-bundle $\pi\colon P\to\RP^n\times\RP^n$ classified by the map
\begin{equation}
    \id\colon \pi_1(\RP^n\times\RP^n)\cong\Z_2\times\Z_2\longrightarrow \Z_2\times\Z_2.
\end{equation}
\begin{enumerate}
    \item For $n = 1$: $\int_{\RP^1\times\RP^1} j^*(\nu_2) = 1$, so $\RP^1\times\RP^1 \cong T^2$, with any of its four spin structures and with the principal $S_4$-bundle induced from $P$ by $j$, represents the class $\alpha\in\ko_2(BS_4)$.
    \item For $n = 3$: $\int_{\RP^3\times\RP^3} j^*(\nu_2^3) = 1$, so $\RP^3\times\RP^3$, with any of its spin structures and the principal $S_4$-bundle induced from $P$ by $j$, represents the class $\phi\in\ko_6(BS_4)$.
\end{enumerate}
\end{prop}
\begin{proof}
That $\int j^*(\nu_2)\ne 0$, resp.\ $\int j^*(\nu_2^3)\ne 0$ characterizes $\alpha$, resp.\ $\phi$, follows from the fact that the images of $\alpha$ and $\phi$ on the $E_\infty$-page, which are $a$ and $f$ respectively (\cref{algebra_koPSL}), are detected by the mod $2$ cohomology classes $\nu_2$, resp.\ $\nu_2^3$ (\cref{PSL_E2_gens}). 
By \cref{lem:Z2Z2S4inclu}, $j^*(\nu_2) = x^2 + xy + y^2$, so
\begin{equation}\label{jstar_monomials}
    j^*(\nu_2^3) = x^6+ x^5y + x^3y^3 + xy^5 + y^6.
\end{equation}
If we let $z$, resp.\ $w$ denote the generators of the left, resp.\ right-hand copies of $H^*(\RP^n;\Z_2)$ inside $H^*(\RP^n\times\RP^n;\Z_2)$, the K\"{u}nneth formula tells us $H^*(\RP^n\times\RP^n;\Z_2)\cong\Z_2[z,w]/(z^{n+1}, w^{n+1})$, and because the classifying map for $P$ is the identity, $x\mapsto z$ and $y\mapsto w$.
First consider $n = 1$: this means that, pulled back to $\RP^1\times\RP^1$, $x^2\mapsto 0$ and $y^2\mapsto 0$, but $xy\mapsto zw$, the non-zero top-degree cohomology class. Thus $\int_{\RP^1\times\RP^1} j^*(\nu_2) = 1$.

Now $n = 3$. The ring structure we just described for $H^*(\RP^3\times\RP^3;\Z_2)$ forces all of the monomials in~\eqref{jstar_monomials} to vanish except for $x^3y^3$, which is the nonzero class in $H^6(\RP^3\times\RP^3;\Z_2)$, so as needed, $\int_{\RP^3\times\RP^3} j^*(\nu_2^3) = 1$.
\end{proof}
\begin{rem}\label{Sq23}
How did we know to try finding a representative of $\phi$ for $\Z_2\times\Z_2$, rather than some other subgroup of $S_4$? One helpful fact is that, after a straightforward computation, one learns $\Sq^2(\Sq^2(\Sq^2(j^*(\nu_2^3)))\ne 0$. Thus $j^*(\nu_2^3)$ generates a free $\cA(1)$-submodule of $H^*(B(\Z_2\times\Z_2);\Z_2)$~\cite[\S D.4]{FH21}, which by Margolis' theorem~\cite{Margolis} means there is some closed spin $6$-manifold $M$ with $(\Z_2\times\Z_2)$-bundle $P\to M$ such that $\int_M j^*(\nu_2^3) = 1$. Thus this straightforward algebraic calculation guides us on the harder and less formulaic problem of finding generators.
\end{rem}

\subsubsection{\texorpdfstring{The generator $X_7$ in dimension $7$}{The generator X7 in dimension 7}}
\label{sec:X7gen}
The manifolds we have found so far in dimension $7$ do not generate $\widetilde\Omega_7^\Spin(B\SL(3, \F_2))$; we are missing one generator, whose image in the Adams spectral sequence is the green, triangle class $e\in E_\infty^{1,8}$, which generates a $\textcolor{Green}{\Z_2}$ summand. Under the inclusions $\Z_4\hookrightarrow S_4$ and $\Z_2\times \Z_2\hookrightarrow S_4$ given in \cref{lem:Z4inS4} and \cref{lem:Z2Z2S4inclu}, this class is not in the image of the induced maps on Adams $E_2$-pages. The following theorem of Mitchell and Priddy will therefore be of use to us.

\begin{thm}[{Mitchell-Priddy~\cite[Theorem A]{Priddy}}]
\label{D8_splitting}
The inclusion $D_8\hookrightarrow S_4$ is one of the maps in a 2-local stable equivalence 
\begin{equation}
    BD_8\simeq BS_4 \vee L(2) \vee B\Z_2
\end{equation}
\end{thm}

Hence, part of our strategy will be to construct a seven dimensional spin manifold $X_7$ with $D_8$ bundle and then use \cref{D8_splitting} to get from $D_8$ to $S_4$. Before we begin with the construction of $X_7$, we recall the $\Z_2$ action on a lens space %\textcolor{red}{MD: Same notation issue for lens spaces as above, maybe write $L^{2n-1}_k$?} 
$L^{2n-1}_k$ by `complex conjugation'.

\begin{defn}[{\!\!\cite[Definition 14.72]{Heckman}}]
\label{complx_conj}
Let $\zeta$ be a primitive $k$th root of unity and $L^{2n-1}_k$ denote the lens space which is the quotient of $S^{2n-1}\subset \C^n$ by the $\Z_k$-action on $\C^n$ which is multiplication by $\zeta$, which preserves the unit sphere. Complex conjugation exchanges $\zeta$ with another primitive $k$th root of unity, and therefore the image of a $\Z_k$-orbit of $S^{2n-1}$ under complex conjugation is another $\Z_k$-orbit. Hence, this involution descends to an involution on $L^{2n-1}_k$, which is also referred to as complex conjugation.
\end{defn}

Recall that the K3 surface is a closed, simply connected, spin $4$-manifold. Thus it has a unique spin structure; let $\mathcal B\to \text{K3}$ denote its principal $\Spin_4$-bundle of frames. K3 has an orientation-preserving free involution $\psi\colon \text{K3}\to \text{K3}$ which does \emph{not} lift to an involution on $\mathcal B$; this means that the quotient by the involution, the \term{Enriques surface} $E$, has a Spin-$\Z_4$ structure but is not spin~\cite[\S C.4]{GEM19}.
\begin{defn}\label{X7_defn}
    Let $\Z_2$ act diagonally on $L_4^3\times \text{K3}$ such that the action on the lens space is by complex conjugation and the action on the K3 surface is free. We define $X_7$ to be the quotient. 
\end{defn}
As usual for manifolds defined in this manner, quotienting by the lens space defines a fiber bundle $\pi\colon X_7\to E$ with fiber $L_4^3$.

Notice that the complex conjugation involution passes to inversion on $\pi_1(L_4^3) = \Z_4$. Hence, since $\text{K3}$ is simply connected, $\pi_1(X_7) = D_8$. 
\begin{thm}
$X_7$ has a Spin structure.
\end{thm}
\begin{proof}
First we need to show $X_7$ is orientable, i.e.\ that the involution in \cref{X7_defn} is orientation-preserving. This follows because complex conjugation on $L_4^3$ is orientation-preserving~\cite[\S 14.3.5]{Heckman} and the involution on K3 is orientation-preserving, as noted above.

In order to calculate $w_2$ and assess whether $X_7$ has a spin structure, we need more information on $TX_7$. We can stably split this vector bundle in the same manner as in~\cite[\S\S14.3.3, 14.3.5, 14.3.6]{Heckman}: specifically, following the same line of proof as in~\cite[Proposition 14.74]{Heckman}, we obtain an isomorphism
\begin{equation}\label{stable_splitting_X7}
    TX_7\oplus\underline\R \overset\cong\longrightarrow V\oplus V\oplus \pi^*(TE),
\end{equation}
where $V$ is the rank-two vector bundle associated to the principal $D_8$-bundle which is the universal cover $S^3\times \text{K3}\to X_7$ and the defining two-dimensional real representation of $D_8$.

Apply the Whitney sum formula to~\eqref{stable_splitting_X7} to deduce
\begin{equation}
    w_2(X_7) = w_1(V)^2 + \pi^*(w_2(E)).
\end{equation}
We want to show this vanishes. Since $E$ is Spin-$\Z_4$ but not Spin, there is a nontrivial principal $\Z_2$-bundle $P\to E$ such that $w_2(E) = w_1(P)^2$,\footnote{From~\cite[Footnote 13]{TY19}, we learn that a Spin-$\Z_4$ structure is equivalent to a $P$ and a Spin structure on $E\oplus \sigma_P^{\oplus 2}$, where $\sigma_P$ is the real line bundle associated to $P$. To get the condition we claim, use the Whitney sum formula.} and since $\pi_1(E)\cong\Z_2$, there is only one such $P$: the double cover $K3\to E$. Thus $w_2(X_7) = w_1(V)^2 + \pi^*(w_1(P))^2$, and so to prove the theorem it suffices to show $w_1(V) = \pi^*(w_1(P))$.

The classes $w_1(V)$ and $\pi^*(w_1(P))$ are equivalent to group homomorphisms $\pi_1(X_7)\cong D_8\to \Z_2$; $w_1(V)$ is identified with the map $D_8\to\Z_2$ given by quotienting by rotations, and $\pi^*(w_1(P))$ is the map $\pi_1(X_7)\overset{\pi_*}{\to}\pi_1(E)\overset\cong\to\Z_2$. Thus it suffices to know that the map $\pi_1(X_7)\to\pi_1(E)$ is exactly the map $D_8\to\Z_2$ that quotients by the rotation subgroup, which follows from the long exact sequence of homotopy groups of the fiber bundle $L_4^3\to X_7\to E$ and the fact that the map $\pi_1(L_4^3)\to\pi_1(X_7)$ is exactly the inclusion of the rotation subgroup $\Z_4\hookrightarrow D_8$.
\end{proof}

\begin{thm}
    $X_7$ with its $D_8$ bundle is nonzero in $\Omega^{\Spin}_7(BD_8)$ and linearly independent from the generator $L_4^7$ in $\Omega^{\Spin}_7(BD_8)$. 
\end{thm}
\begin{proof}
Suppose that $M$ is a spin manifold with $D_8$ bundle $P \rightarrow M$. The quotient of $P$ by $\Z_4 \subset D_8$ is a $\Z_2$-bundle $\widetilde{M}$ which is a double cover of $M$. The assignment $M \mapsto \widetilde{M}$ defines a homomorphism 
\begin{equation}
    T:\Omega^{\Spin}_k(BD_8)\rightarrow \Omega^{\Spin}_k(B\Z_4)
\end{equation}
The map $T$ is commonly referred to as the `transfer.'

Applying the map $T$ to $X_7$, we get $L_4^3 \times \text{K3}$. The manifold $L_4^3 \times \text{K3}$ is non-trivial in $\Omega^{\Spin}_7(B\Z_4)$. Indeed, the bordism invariant $\eta^D_1-\eta^D_0$ of the lens space $L_4^3$ is $-3/8$~\cite[Table 17]{Heckman}. With this, and the fact that 
\begin{equation}
    \eta^D(L_4^3 \times \text{K3}) = \text{Index}^{D}(\text{K3}) \, \eta^D(L_4^3),
\end{equation}
it follows that the bordism invariant $\eta^D_1-\eta^D_0$ evaluated on $L^3_4 \times \text{K3}$ is $3/4$. In particular, the manifold is non-trivial in bordism. Note here that we have used the fact that the $K3$ has Dirac index $(-2)$. 

The other generator we have for $\Omega^{\Spin}_7(BD_8)$ is $L_4^7$, see \cref{lens_S4} and \cref{D8_splitting} for reference. This generator vanishes when we apply the map $T$: its $D_8$-bundle is induced from a $\Z_4$-bundle (\cref{lens_S4}), and therefore $T$ produces the trivial double cover, which bounds $[0,1]\times L_4^7$. Hence, we conclude that $X_7$ is linearly independent from $L_4^7$ in $\Omega^{\Spin}_7(BD_8)$. 
\end{proof}

\subsubsection{\texorpdfstring{The generator $W_4$ of $\ko_4(L(2))$}{The generator W4 of ko4(L(2))}} 
\label{sec:W4}
The generator $W_4$ of $\ko_4(L(2))$ is detected by the $\Z_2$ cohomology class $a^2 b $. Consider the inclusion of $D_8$ into $S_4$. We have already considered the induced map on $\Z_2$ cohomology in the context of Theorem \ref{Thm:SteenrodBPSL27}. Indeed, the inclusion of $D_8$ into SL$(3,\F_2)$ induces a map on $\Z_2$ cohomology sending $\nu_2, \nu_3, \overline \nu_3$ to $x_1^2 + w$, $x_1 w$, and $x_2 w$. The following lemma then follows from Remark \ref{Rem:PSL(2,7)andS4}. 

\begin{lem}
\label{D8inS4}
    The inclusion $D_8\hookrightarrow S_4$ defined by $(1\;2\;3\;4)$ and $(1\;3)$ induces a map on $\Z_2$ cohomology defined by 
    \begin{equation}
    \begin{gathered}
        k:H^*(BS_4;\Z_2)\rightarrow H^*(BD_8;\Z_2), \\
        a\mapsto x_1 + x_2,\quad b\mapsto x_1^2 + w,\quad  c\mapsto x_2 w
    \end{gathered}
    \end{equation}
\end{lem}
The generator of $\widetilde{\ko}_4(L(2))$ is detected by $a^2b$ (\cref{complete_L2_decomp}). Hence, it suffices to determine a four dimensional manifold $W_4$ with $D_8$ bundle such that $\int_{W_4}k^*(ab^2)\neq 0$.  

\begin{defn}
    Let $\Z_2$ act on $L_4^3\times S^1$ by complex conjugation on the lens space, see \cref{complx_conj}, and the antipodal map on $S^1$. We define $W_4$ to be the quotient.
\end{defn}
We remark that $W_4$ is can be equivalently regarded as the quotient of $S^3 \times S^1$ by the $D_8$ action generated by two diffeomorphisms $r$ and $s$: $r$ is multiplication by $i$ on $S^3$ and the identity on $S^1$ and $s$ is reflection on $S^3$ and the antipodal map on $S^1$. 

\begin{lem}
    $W_4$ is orientable and spin.
\end{lem}
\begin{proof}
We first check that $W_4$ is orientable, or equivalently, whether the involution we used to define it is orientation preserving. The antipodal map is orientation preserving on $S^1$, and complex conjugation is orientation preserving on $L_4^3$. Therefore, the combination of these two involutions is orientation preserving.    

Next we check that $W_4$ is spin. We follow a similar approach as that used in the proof of \cite[Proposition 14.74]{Heckman}. First, we remark that there is an isomorphism of vector bundles 
\begin{equation}
    T(L_4^3 \times S^1) \oplus \underline{\R}^2 \xrightarrow{\cong} \mathcal{L}^2 \oplus \underline{\R}^2,
\end{equation}
where $\mathcal{L} \rightarrow L^{3}_4$ is the quotient of $\underline{\C} \rightarrow S^{3}$ by the $\Z_4$ action. 

Letting $\Z_2$ act on $L_4^3 \times S^1$ by complex conjugation on $L_4^3$ and the antipodal map on $S^1$, so that the quotient is $W_4$, there is an isomorphism of vector bundles
\begin{equation}
    TW_4 \oplus \underline{\R}^2\xrightarrow{\cong} V^2 \oplus \sigma^2 
\end{equation}
where $V$ and $\sigma$ are as follows: if $P \rightarrow W_4$ denotes the quotient $S^3 \times S^1 \rightarrow W_4 $, which is a principal $D_8$-bundle, then $\sigma$ is associated to $P$ and the sign representation $D_8 \rightarrow O(1)$ sending rotations to $1$ and reflections to $-1$, and $V$ is associated to the standard representation $D_8 \rightarrow O(2)$ as rotations and reflections on $\R^2$. Thus, $\sigma = \mathrm{Det}(V )$, so $w_1(V ) = w_1(\sigma)$. Using that $w_2( \sigma) = 0$, an application of the Whitney sum formula then reveals that $w_2(W_4) = 0$. 
\end{proof}

All that's left to do is verify that $\int_{W_4}k^*(ab^2)\neq 0$. To proceed, we need the $\Z_2$ cohomology of $W_4$. Notice that $W_4$ is obtained from $L_4^3\times S^1$ by a further $\Z_2$ quotient. Since $\Z_2$ acts as sign reversal in the base the resulting spaces are given by the fiber bundle $L_4^3 \hookrightarrow W_4 \rightarrow \mathbb{R}\mathbb{P}^1$. 

\begin{lem}
\label{CohoW4}
    The Serre spectral sequence for the fiber bundle $L_4^3 \rightarrow W_4\rightarrow \R\mathbb{P}^1$ collapses, providing an isomorphism 
    \begin{equation}
    \begin{gathered}
        H^*(W_4;\Z_2) \cong \Z_2[x,y,w]/(x^2, xy + y^2, w^2),\\
        |x| = |y| = 1,\quad |w| =2 
    \end{gathered}
    \end{equation}
\end{lem}
\begin{proof} The spectral sequence computing $H^*(W_4;\Z_2)$ is given in Figure \ref{W4Serre}. Clearly, the spectral sequence collapses; there is no room for differentials. 

\begin{figure}[h!]
\centering
\begin{sseqdata}[name=W4Serre, cohomological Serre grading, xrange={0}{1}, yrange={0}{3},
classes={draw=none}, >=stealth, xscale=2,
x label = {$\displaystyle{q\uparrow \atop p\rightarrow}$},
x label style = {font = \small, xshift = -11ex, yshift=5.4ex}
]
	\class["1"](0, 0)
	\class["x"](1, 0)

	\class["y"](0, 1)
	\class["xy"](1, 1)

	\class["w"](0, 2)
	\class["wx"](1, 2)

	\class["wy"](0, 3)
	\class["wxy"](1, 3)

\end{sseqdata}
\printpage[name=W4Serre, page=2]
\caption{The Serre spectral sequence computing $H^*(W_4; \Z_2)$.}
\label{W4Serre}
\end{figure}
The multiplicative structure is clear except for the relation $xy + y^2 = 0$. This follows from the fact that $x$ and $y$ pull back from $BD_8$, and $x_1x_2 + x_2^2 = 0$ in $H^*(BD_8; \Z_2)$.

\end{proof}
A straightforward application of \cref{D8inS4} reveals that 
\begin{equation}
    k^*(ab^2) = w(x_1 + x_2)^2 + (x_1+x_2)^4
\end{equation}
Lemma \ref{CohoW4} tells us that all terms vanish except $wx_2^2$. Hence, $\int_{W_4} a^2 b = 1$, and we conclude that $W_4$ represents the unique nonzero class in $\ko_4(L(2))$.

\section{\texorpdfstring{Calculations and generators for $B\Z_4\times BS_4$}{Calculations and generators for BZ4 x BS4}}
\label{2loc_smash}
The last piece of our computation is the most difficult: the low-degree reduced spin bordism groups of $B\Z_4\times BS_4$. There are a few immediate simplifications we can make:
\begin{enumerate}
    \item By \cref{distributivity}, $\widetilde \Omega_*^\Spin(B\Z_4\times BS_4)$ is a direct sum of $\widetilde\Omega_*^\Spin(B\Z_4)$, $\widetilde\Omega_*^\Spin(BS_4)$, and $\widetilde\Omega_*^\Spin(B\Z_4\wedge BS_4)$. We calculated the first two summands in low degrees in \cref{ringstr_Z4,algebra_koPSL,complete_L2_decomp}, so in this section we will only focus on $B\Z_4\wedge BS_4$.
    \item By \cref{S4_splitting}, $B\Z_4\wedge BS_4$ splits stably as a wedge sum of $B\Z_4\wedge B\Z_2$ (which we can ignore, by \cref{why_S4}), $B\Z_4\wedge B\SL_3(\F_2)$, and $B\Z_4\wedge L(2)$. We will tackle the latter two summands separately.
\end{enumerate}
The most powerful simplifying technique is to work $B\Z_4$-equivariantly. Generally for a group $G$, $BG$ is not a topological group, but for $A$ an abelian group there is a model for $BA$ which is a topological abelian group. Thus we can ask how $\Omega_*^\Spin(B\Z_4)$ acts on $\widetilde\Omega_*^\Spin(B\Z_4\wedge BS_4)$, similar to our $3$-primary analysis in \S\ref{ss:Z3D3}. By \cref{Thm:Kunneth}, this action is also present in the Adams spectral sequence, and differentials satisfy a Leibniz rule for it.

We take advantage of this extra symmetry in two ways: in \S\ref{Z4S4Adams}, we use it to cleanly describe and compute the differentials in the Adams spectral sequence for $\widetilde\Omega_*^\Spin(B\Z_4\wedge BS_4)$ in the range we need; then, in \S\ref{ss:Z4S4gens}, we describe all manifold representatives of the corresponding bordism classes in terms of just a few new generators.
\label{sec:Z4S4}

\subsection{\texorpdfstring{Adams spectral sequences for $\ko_*(B\Z_4 \wedge B\SL(3,\F_2))$ and $\ko_*(B\Z_4 \wedge L(2))$}{Adams spectral sequences for ko(BZ4 wedge BSL(3,F2)) and ko(BZ4 wedge L(2))}}
\label{Z4S4Adams}

In this subsection, we run the Adams spectral sequence computing $\ko_*(B\Z_4 \wedge B\SL(3,\F_2))$ in degrees $7$ and below. At all stages we describe the structure for $B\Z_4\wedge B\SL(3, \F_2)$ as a module over the corresponding structure for $B\Z_4$.

Our computation is front-loaded with algebraic calculations that will simplify the actual spectral sequence later. First, in \cref{joker_Ceta,question_Ceta,A1_Ceta}, we study the Ext structure on some tensor products of $\cA(1)$-modules. Using this, we describe the $E_2$-page as a module over $\Ext(H^*(B\Z_4;\Z_2))$ in \cref{the_smash_Ext_module}. Using this and the differentials we calculated in \cref{ko_Z4_diffs} and \S\ref{sec:BPSL}, we compute differentials in our spectral sequence in \cref{S4Z4_diffs}, then address extensions in \cref{the_Z4_module_str}.
\begin{rem}
It would be nice to apply the stable splitting techniques we have used for $B\Z_4\wedge B\SL(3,\F_2)$. However, Martino-Priddy-Douma show that $\Sigma^\infty(B\Z_4\wedge B\SL(3,\F_2))$ is indecomposable~\cite[Example 5.3]{MPD00}. Specifically, they show that the decomposition of $\Sigma^\infty(B\Z_4\wedge BD_8)$ into indecomposable summands can be found by taking the indecomposable summands of each factor and smashing them together --- the spectrum does not simplify further. Since $\Sigma^\infty B\Z_4$ is indecomposable and $B\SL(3,\F_2)$ is stably an indecomposable summand of $BD_8$~\cite[Theorem A]{Priddy}, $\Sigma^\infty(B\Z_4\wedge B\SL(3,\F_2))$ is indecomposable.
\end{rem}
Let $M$ and $N$ be $\cA(1)$-modules. Then there is a map
\begin{equation}
\label{Adams_Knn}
    \Ext_{\cA(1)}(M, \Z_2)\otimes\Ext_{\cA(1)}(N, \Z_2) \longrightarrow\Ext_{\cA(1)}(M\otimes N, \Z_2),
\end{equation}
which we call the \term{K\"{u}nneth map} -- often it arises by applying Ext to the K\"{u}nneth map in cohomology of a smash product of spaces. In this case, the K\"{u}nneth map on Ext converges in the Adams spectral sequence to the product map on $\ko$-homology groups. We will evaluate this K\"{u}nneth map in a few examples.

Recall from \cref{A1_mod_BZ4} that, as an $\cA(1)$-module, $\widetilde H^*(B\Z_4;\Z_2)$ is a sum of shifts of $\Z_2$ and $C\eta$. Tensoring with $\Z_2$ does not change the isomorphism type of an $\cA(1)$-module, so we focus on $C\eta$.
\begin{rem}
\label{E1_trick}
Determining the $\ExZ$-module structure on $\Ext(M\otimes C\eta, \Z_2)$ is easy, thanks to a trick: there is an isomorphism $C\eta\cong\cA(1)\otimes_{\cE(1)}\Z_2$, where $\cE(1)\coloneqq \ang{\Sq^1, \Sq^2\Sq^1 + \Sq^1\Sq^2}$, so by the change-of-rings theorem, $\Ext_{\cA(1)}(C\eta\otimes M)\cong\Ext_{\cE(1)}(M)$. The reference~\cite{Debray} works out Ext over $\cE(1)$ for several commonly occurring $\cE(1)$-modules. However, we will not use this shortcut much: it does not provide any insight on the K\"{u}nneth map, hence does not help much with the $\Omega_*^\Spin(B\Z_4)$-module structure on $\Omega_*^\Spin((B\Z_4)_+\wedge B\SL(3, \F_2))$. The reader may enjoy trying the $\cE(1)$ trick to compute the Ext groups we work out below more directly as a check of our calculations.
\end{rem}
\begin{lem}
\label{joker_Ceta}
\hfill
\begin{enumerate}
    \item With notation for $\Ext_{\cA(1)}(J)$ and $\Ext_{\cA(1)}(C\eta)$ as in \cref{ext_Ceta,PSL_E2_gens}, there is an isomorphism of $\ExZ$-modules
    \begin{equation}
     \Ext_{\cA(1)}(J\otimes C\eta)\cong (\ExZ/h_1)\set{j,k,\ell, \dotsc}/(h_0j, vj, \dotsc),
    \end{equation}
    with $j\in\Ext^{0,0}$, $k\in\Ext^{0,2}$, and $\ell\in\Ext^{1,5}$, and all generators and relations not listed are in topological degree $5$ and above.
    \item The K\"{u}nneth map sends $a\otimes\kappa\mapsto j$, $b\otimes\kappa\mapsto h_0k$, $b\otimes\lambda\mapsto h_0\ell$, and $a\otimes\lambda\mapsto 0$.
\end{enumerate}
\end{lem}
\begin{proof}
One can check by explicit computation that $J\otimes C\eta\cong \cA(1)\oplus \Sigma^2 C\eta$. $\Ext(\cA(1))\cong\Z_2$ in bidegree $(0, 0)$ with trivial $\ExZ$-action, and $\Ext(C\eta)$ appears in \cref{ext_Ceta} (drawn in \cref{Ext_eta_fig}); putting these together, we conclude the first part of the lemma.

For the K\"{u}nneth map in the second part of the lemma, two of the classes are easy: $\kappa$ and $a$ are both represented by the classes in $\Ext^0 = \Hom(\bl,\Z_2)$ which preserve the lowest-degree element in an $\cA(1)$-module and kill all others. One can explicitly compute the tensor product of those homomorphisms and see that it is non-trivial; since there is a unique nonzero class in $\Ext^{0,0}(J\otimes C\eta)$ by the first part of the lemma, we deduce $a\otimes\kappa$ must be this nonzero class, which is $j$.

Since $h_0a = 0$, $h_0(a\otimes\lambda) \mapsto 0$ under the K\"{u}nneth map. Multiplying by $h_0$, as a map $\Ext^{1,3}(J\otimes C\eta)\to\Ext^{2,4}(J\otimes C\eta)$, is injective by the first part of the lemma, so $h_0(a\otimes\lambda) \mapsto 0$ under the K\"{u}nneth map forces $a\otimes\lambda$ to also map to $0$.

The remaining two classes can be sorted by comparing with the K\"{u}nneth map over $\cA(0)\coloneqq\ang{\Sq^1}\subset\cA(1)$. The forgetful map $\Ext_{\cA(1)}(C\eta)\to\Ext_{\cA(0)}(C\eta)$ is injective in topological degree $2$ and below by \cref{A1_Ceta_A0}, so if we want to prove $b\otimes\kappa = h_0 k$ and $b\otimes\lambda = h_0\ell$, it suffices to do so for extensions of $\cA(0)$-modules.

As $\cA(0)$-modules, $C\eta\cong\Z_2\oplus \Sigma^2\Z_2$, and under the isomorphism $\Ext_{\cA(0)}(\Z_2, \Z_2)\cong\Z_2[h_0]$, $\kappa$ pulls back to $h_0$ for the degree-$0$ $\Z_2$ and $\lambda$ pulls back to $h_0^2$ for the degree-$2$ copy of $\Z_2$ by \cref{A1_Ceta_A0}. Thus, if we can show that $b\in\Ext_{\cA(1)}(J)$ is non-trivial when pulled back to $\Ext_{\cA(0)}(J)$, then $b\otimes\kappa$ and $b\otimes\lambda$ are nonzero; since they each live in a one-dimensional Ext group, they must be equal to $h_0 k$, resp.\ $h_0\ell$. So we finish the proof by showing an explicit $\cA(1)$-module extension representing $b$ is non-split as an $\cA(0)$-module extension.

$\Ext^{1,3}(J)\cong\Z_2$, so any non-split extension $0\to\Sigma^3\Z_2\to \widetilde J\to J\to 0$ represents $b$. For example, we could let $\widetilde J\coloneqq \cA(1)/(\Sq^2\Sq^1\Sq^2)$,\footnote{Baker~\cite[\S 5]{Bak18} calls this $\cA(1)$-module the \term{whiskered Joker}.} with the map to $J \cong \widetilde J/(\Sq^3)$ given by the quotient. We draw this extension in \cref{joker_ext_A1_A0}, left. The Adem relation $\Sq^1\Sq^2 = \Sq^3$ implies that as $\cA(0)$-modules, the extension is also not split, as we draw in \cref{joker_ext_A1_A0}, right: the image of $b\in\Ext_{\cA(0)}(J)$ is nonzero.
\end{proof}

\begin{figure}[h!]
\centering
\begin{subfigure}[c]{0.49\textwidth}
\centering
\begin{tikzpicture}[scale=0.6]
\sqone(0, 2);
\begin{scope}[BrickRed]
    \foreach \x in {-4, 0} {
        \tikzpt{\x}{3}{}{};
    }
    \draw[thick, ->] (-3.5, 3) -- (-0.5, 3);
\end{scope}
\begin{scope}[MidnightBlue]
    \foreach \x in {0, 4} {
        \foreach \y in {0, 1, 2} {
            \tikzpt{\x}{\y}{}{};
        }
        \tikzpt{\x+1.5}{3}{}{};
        \tikzpt{\x+1.5}{4}{}{};
        \sqone(\x, 0);
        \sqone(\x+1.5, 3);
        \sqtwoL(\x, 0);
        \sqtwoCR(\x, 1);
        \sqtwoCR(\x, 2);
    }
    \draw[thick, ->] (0.5, 0) -- (3.5, 0);
    \draw[thick, ->] (2, 4) -- (5, 4);
\end{scope}
\node at (-4, -1) {$\Sigma^3\Z_2$};
\node at (0, -1) {$\widetilde J$};
\node at (4, -1) {$J$};
\draw[->] (-3, -1) -- (-0.5, -1);
\draw[->] (0.5, -1) -- (3.5, -1);
\end{tikzpicture}
\end{subfigure}
\begin{subfigure}[c]{0.49\textwidth}
\centering
\begin{tikzpicture}[scale=0.6]
\sqone(0, 2);
\begin{scope}[BrickRed]
    \foreach \x in {-4, 0} {
        \tikzpt{\x}{3}{}{};
    }
    \draw[thick, ->] (-3.5, 3) -- (-0.5, 3);
\end{scope}
\begin{scope}[MidnightBlue]
    \foreach \x in {0, 4} {
        \foreach \y in {0, 1, 2} {
            \tikzpt{\x}{\y}{}{};
        }
        \tikzpt{\x+1.5}{3}{}{};
        \tikzpt{\x+1.5}{4}{}{};
        \sqone(\x, 0);
        \sqone(\x+1.5, 3);
 %       \sqtwoL(\x, 0);
 %       \sqtwoCR(\x, 1);
 %       \sqtwoCR(\x, 2);
    }
    \draw[thick, ->] (0.5, 2) -- (3.5, 2);
    %\draw[thick, ->] (2, 4) -- (5, 4);
\end{scope}
\node at (-4, -1) {$\Sigma^3\Z_2$};
\node at (0, -1) {$\widetilde J$};
\node at (4, -1) {$J$};
\draw[->] (-3, -1) -- (-0.5, -1);
\draw[->] (0.5, -1) -- (3.5, -1);
\end{tikzpicture}
\end{subfigure}
\caption{Left: an $\cA(1)$-module extension representing the class $b\in\Ext^{1,3}(J, \Z_2)$ defined in \cref{PSL_E2_gens}. Right: the same extension as $\cA(0)$-modules is not split. This is an ingredient in the proof of \cref{joker_Ceta}.}
\label{joker_ext_A1_A0}
\end{figure}
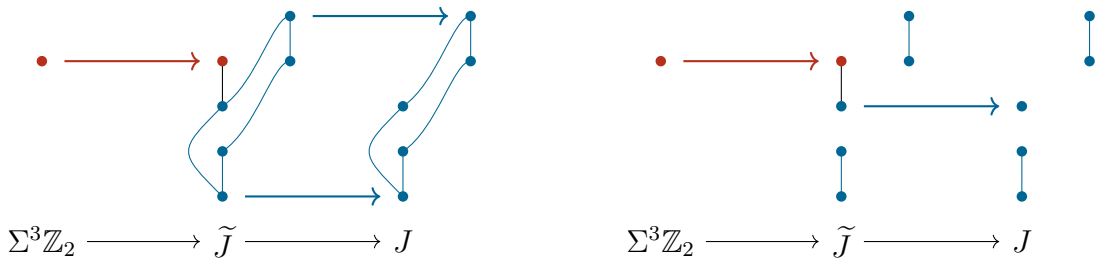
\begin{lem}
\label{question_Ceta}
\hfill
\begin{enumerate}
    \item With notation for $\Ext_{\cA(1)}(Q)$ and $\Ext_{\cA(1)}(C\eta)$ as in \cref{ext_Ceta,PSL_E2_gens}, there is an isomorphism of $\ExZ$-modules
    \begin{equation}
     \Ext_{\cA(1)}(Q\otimes C\eta)\cong (\ExZ/h_1)\set{n,p,\dotsc}/(\dotsc),
    \end{equation}
    with $n\in\Ext^{0,0}$ and $p\in\Ext^{0,2}$, and all generators and relations not listed are in topological degree $4$ and above.
    \item The K\"{u}nneth map sends $d\otimes\kappa\mapsto n$ and $d\otimes\lambda\mapsto h_0p$.
\end{enumerate}
\end{lem}
\begin{proof}
The first part of the lemma can be computed using the $\cE(1)$ trick (see \cref{E1_trick}); alternatively, by direct computation one finds an $\cA(1)$-module isomorphism
\begin{equation}
Q\otimes C\eta\cong \cA(1)\cdot x\oplus \Sigma^2\cA(1)\cdot y/(\Sq^1 x, \Sq^1 y, \Sq^2\Sq^3 x = \Sq^3 y),
\end{equation}
and the latter module is exactly the one studied by Davighi-Lohitsiri~\cite[Appendix A]{DL20} for applications to Spin-$\mathrm U(2)$ bordism. They compute enough about its Ext to imply the first part of the lemma.

The second part can be proven in exactly the same way as the second part of \cref{joker_Ceta}: $d\otimes k$ has to correspond to a nontrivial homomorphism $Q\otimes C\eta\to\Z_2$ and live in bidegree $(0, 0)$, so is $n$, and for $d\otimes\lambda$, one can check over $\cA(0)$, where the relation is easier.
\end{proof}
\begin{lem}
\label{A1_Ceta}
There is an $\cA(1)$-module isomorphism $\cA(1)\otimes C\eta\cong\cA(1)\oplus \Sigma^2\cA(1)$; thus there is an $\ExZ$-module isomorphism $\Ext_{\cA(1)}(\cA(1)\otimes C\eta)\cong\Z_2\cdot\{\sigma_1, \sigma_2\}$ with $\sigma_1\in\Ext^{0,0}$ and $\sigma_2\in\Ext^{0,2}$. If the $\cA(1)$ is generated by a mod $2$ cohomology class $g_1$ and $C\eta$ is generated by $g_2$, then $\sigma_1$ is detected by the mod $2$ cohomology class $g_1g_2$ and $\sigma_2$ is detected by $g_1\Sq^2(g_2)$. The image of the K\"{u}nneth map is the vector space generated by $\sigma_1$.
\end{lem}
Next, we want to understand $\widetilde H^*(B\Z_4\wedge B\SL(3, \F_2);\Z_2)$ as an $\cA(1)$-module, as a first step to understanding it as a module over $\Ext(H^*(B\Z_4;\Z_2))$. This amounts to plugging \cref{A1_mod_BZ4,cor:A(1)decompPSL} into the K\"{u}nneth formula:
\begin{prop}\label{Thm:CohoPSLwedgeZ4}
There is an isomorphism of $\cA(1)$-modules
\begin{equation}
\begin{gathered}
\widetilde H^*(B\SL(3,\F_2)\wedge B\Z_4;\Z_2) \\
\cong \Sigma \Z_2\otimes \widetilde H^*(B\SL(3,\F_2);\Z_2) \oplus \Sigma^2 \mathcal{C}\eta \otimes \widetilde H^*(B\SL(3,\F_2);\Z_2)\\
\oplus \Sigma^3\mathcal{C}\eta \otimes \widetilde H^*(B\SL(3,\F_2);\Z_2)\oplus \Sigma^6 \mathcal{C}\eta \otimes \widetilde H^*(B\SL(3,\F_2);\Z_2) \\
\oplus \Sigma^7 \mathcal{C}\eta \otimes \widetilde H^*(B\SL(3,\F_2);\Z_2) \oplus M,
\end{gathered}
\end{equation}
where $M$ is concentrated in degree $*\geq 10$.
\end{prop}
Recalling the $\cA(1)$-module structure on $\widetilde H^*(B\SL(3, \F_2);\Z_2)$ from \cref{cor:A(1)decompPSL}, we see that as an $\ExZ$-module, $\Ext(\widetilde H^*(B\Z_4\wedge B\SL(3, \F_2);\Z_2))$ is a direct sum of the following $\ExZ$-modules:
\begin{subequations}
\begin{itemize}
    \item Coming from $\Sigma\Z_2\otimes \widetilde H^*(B\SL(3, \F_2);\Z_2)$ we have
    \begin{equation}
        \Ext(\textcolor{BrickRed}{\Sigma^3 J}) \oplus \Ext(\textcolor{RedOrange}{\Sigma^4 Q}) \oplus \Ext({\Sigma^7 \cA(1)}) \oplus \Ext(\textcolor{Fuchsia}{\Sigma^8 R_2}) \oplus \Ext(M_1).
    \end{equation}
    \item Coming from $\Sigma^2 C\eta\otimes \widetilde H^*(B\SL(3, \F_2);\Z_2)$ we have
    \begin{equation}
    \label{where_gray_from}
    \Ext(\Sigma^4(J\otimes C\eta)) \oplus  \Ext(\textcolor{Goldenrod!67!black}{\Sigma^5(Q\otimes C\eta)})\oplus  \Ext(\Sigma^8 \cA(1)) \oplus \Ext(M_2).
    \end{equation}
    \item Coming from $\Sigma^3 C\eta\otimes\widetilde H^*(B\SL(3, \F_2);\Z_2)$ we have
    \begin{equation}
    \Ext(\Sigma^5(J\otimes C\eta)) \oplus  \Ext(\textcolor{PineGreen}{\Sigma^5(Q\otimes C\eta)})\oplus \Ext(M_3).
    \end{equation}
    \item Coming from $\Sigma^6 C\eta\otimes\widetilde H^*(B\SL(3, \F_2);\Z_2)$ we have
    \begin{equation}
    \Ext(\Sigma^8(J\otimes C\eta))\oplus \Ext(M_4).
    \end{equation}
\end{itemize}
\end{subequations}
Each $M_i$ is concentrated in degrees $9$ and above, so we can and do ignore it. Therefore we need to know $\Ext(\textcolor{BrickRed}{J})$, $\Ext(\textcolor{RedOrange}{Q})$ and $\Ext(\textcolor{Fuchsia}{R_2})$, which are given in \cref{PSL_E2_gens}, as well as $\Ext(J\otimes C\eta)$ and $\Ext(Q\otimes C\eta)$, which we computed in \cref{joker_Ceta,question_Ceta} respectively. Gathering these together, we obtain a complete description of the $E_2$-page of this Adams spectral sequence as an $\ExZ$-module. Next, we want to upgrade this to an $\Ext(H^*(B\Z_4;\Z_2))$-module:

\begin{thm}
\label{the_smash_Ext_module}
There is an isomorphism of $\Ext_{\cA(1)}(H^*(B\Z_4;\Z_2))$-modules
\begin{subequations}
\begin{equation}
    \Ext_{\cA(1)}(\widetilde H^*(B\Z_4\wedge (BS_4)_+; \Z_2)) \cong \ExZ\set{
        a, b, c, d, e, f, g, h, i, k, \ell, p, \dotsc
    }/\mathcal R_4,
\end{equation}
where $a\in\Ext^{0,2}$, $b\in\Ext^{1,5}$, $c\in\Ext^{2,10}$, $d\in\Ext^{0,3}$, $e\in\Ext^{1,8}$, $f,k\in\Ext^{0,6}$, $g,p\in\Ext^{0,7}$, $h,i\in\Ext^{0,8}$, and $\ell\in\Ext^{1,9}$, and
where the ideal $\mathcal R_4$ of relations is given by
\begin{equation}\label{R_four}
    \begin{aligned}
        \mathcal R_4 = &(h_0a, h_1a, va, \lambda_1a, \mu_1a,
        h_1d, vd + h_0^2e, \lambda_1d + h_0p,h_1b, vb + h_0^2c,\\
        &\phantom{(} \kappa_1 b + h_0k, \lambda_1 b + h_0\ell, h_0f, h_1f, h_1g, h_0h, h_1k, \kappa_1k, h_1p, \dotsc)
    \end{aligned}
\end{equation}
\end{subequations}
and all unlisted generators and relations are in topological degrees $9$ and above.
\end{thm}
\begin{proof}
We want to determine products of classes in $\Ext(H^*(B\Z_4;\Z_2))$ and $\Ext(\widetilde H^*((B\Z_4)_+\wedge B\SL(3, \F_2);\Z_2)$. If the class in $\Ext(H^*(B\Z_4;\Z_2))$ comes from the $\Sigma$ summand in degree $0$, that is the identity in $H^*(B\Z_4;\Z_2)$, so these Ext classes act the same way $\ExZ$ does. The Ext classes coming from $\Sigma \Z_2\subset H^*(B\Z_4;\Z_2)$ are a cyclic $\ExZ$-module generated by $\theta$; tensoring with $\Sigma \Z_2$ is the same as shifting upwards by $1$, so for $\Ext(\Sigma\Z_2\otimes \widetilde H^*(B\SL(3, \F_2));\Z_2)$, we get $\theta$ times the classes in $\Ext(H^*(B\SL(3, \F_2);\Z_2))$ that we computed in \cref{PSL_E2_gens}.

The remaining summands we need to understand are all of the form $\Sigma^kC\eta\otimes \widetilde H^*(B\SL(3, \F_2);\Z_2))$. The module structure follows from the description of the K\"{u}nneth map in \cref{joker_Ceta,question_Ceta,A1_Ceta} applied to whatever classes in $\Ext(H^*(B\Z_4;\Z_2))$ correspond to $\kappa$ and $\lambda$ of that particular $C\eta$ summand, namely\ $\kappa_1$ and $\lambda_1$ for $\Sigma^2 C\eta$; $\theta\kappa_1$ and $\theta\lambda_1$ for $\Sigma^3 C\eta$, and $\kappa_2$ and $\lambda_2$ for $\Sigma^6 C\eta$.
\end{proof}
\begin{figure}[h!]
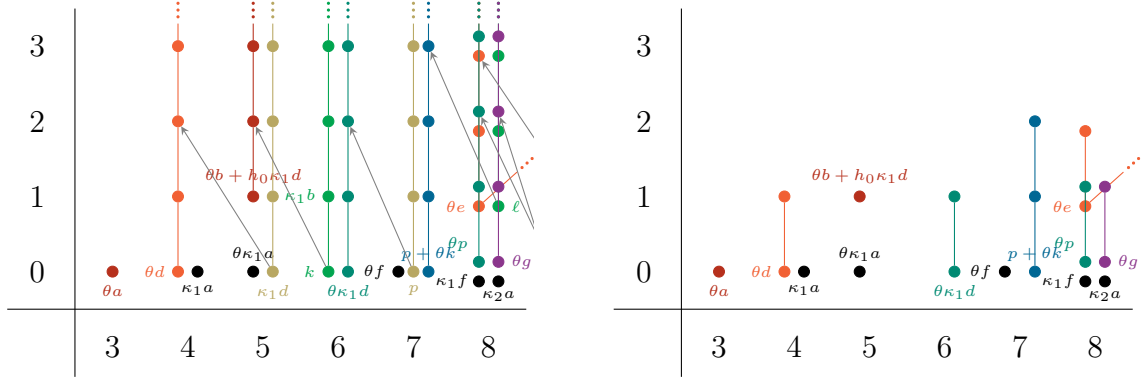

\centering
%\begin{subfigure}[c]{0.6\textwidth}
\begin{sseqdata}[name=wedge, classes=fill, xrange={3}{8}, yrange={0}{3}, scale=1, Adams grading, >=stealth,
class labels = {below = 0.05em, font=\tiny}]

\begin{scope}[BrickRed]
	\class(3, 0)
     \classoptions["\theta a"](3, 0)
	\class(5, 1)\AdamsTower{}
     \classoptions[class labels = {above = 0.05em}, "\theta b + h_0\kappa_1 d"](5, 1)
	\class(9, 2)\AdamsTower{}
      \classoptions[class labels = {left = 0.05em}, "\theta c"](9, 2)
	\class(10, 3)
		\structline(9, 2)(10, 3, -1)
        \class(11, 4)\structline
\end{scope}

\begin{scope}[RedOrange]
	\class(4, 0)\AdamsTower{}
     \classoptions[class labels = {left = 0.05em}, "\theta d"](4, 0)
	\class(8, 1)\AdamsTower{}
      \classoptions[class labels = {left = 0.05em}, "\theta e"](8, 1)
	\class(9, 2)
            \structline(8,1)(9,2,-1)
        \class(10,3)
            \structline(9,2,-1)(10,3,-1)
\end{scope}
\begin{scope}[draw=none, fill=none]
    \class(4, 1)\AdamsTower{}
    \class(7, 1)\AdamsTower{}
\end{scope}

\class(7,0)
\classoptions[class labels = {left = 0.05em}, "\theta f"](7, 0)

\begin{scope}[draw=none, fill=none]
	\class(9, 0)\AdamsTower{}
 \classoptions[class labels = {below left = 0.05em}, "\theta i"](9, 0, -1)
	\class(10, 1)
            \structline(10,1,-1)(9,0,-1)
        \class(11, 2)\structline
\end{scope}

%The second piece from Theorem 2.28
\class(4,0)
\classoptions["\kappa_1 a"](4, 0, -1)

%The third piece from Theorem 2.28
\class(5,0)
\classoptions[class labels = {above = -0.05em}, "\theta \kappa_1 a"](5, 0, -1)

\begin{scope}[Goldenrod!67!black]
	\class(5, 0)\AdamsTower{}
	\class(7, 0)\AdamsTower{}
	\class(9, 1)\AdamsTower{}
     \classoptions["\kappa_1 d"](5, 0, -1)
    \classoptions["p"](7, 0, -1)
    \classoptions["q"](9, 1, -1)

\end{scope}

\class(8,0)
\classoptions[class labels = {left = -0.05em}, "\kappa_1 f"](8, 0, -1)
%The fourth piece from Theorem 2.28
\class(8,0)
\classoptions[class labels = {below = -0.05em}, "\kappa_2 a"](8, 0, -1)

\begin{scope}[Green]
	\class(6, 0)\AdamsTower{}
	\class(8, 1)\AdamsTower{}
	\class(10, 2)\AdamsTower{}
      \classoptions[class labels = {left = 0.05em}, "k"](6, 0, -1)
    \classoptions[class labels = {left = 0.05em}, "\kappa_1 b"](6, 1, -1)
    \classoptions[class labels = {right = 0.05em}, "\ell"](8, 1, -1)
    \classoptions["m"](10, 2, -1)
\end{scope}

\begin{scope}[draw=none, fill=none]
%\begin{scope}[magenta]
	\class(9, 0)\AdamsTower{}
    \classoptions["\kappa_1 g"](9, 0, -1)
\end{scope}

\class(10,0)
\classoptions["\kappa_1 h"](10, 0, -1)

\begin{scope}[draw=none, fill=none]
%\begin{scope}[olive]
	\class(10, 0)\AdamsTower{}
\classoptions[class labels = {above left = 0.05em}, "\kappa_1 i"](10, 0, -1)
\end{scope}

\begin{scope}[PineGreen]
	\class(6, 0)\AdamsTower{}
	\class(8, 0)\AdamsTower{}
	\class(10, 1)\AdamsTower{}
 \classoptions["\theta \kappa_1 d"](6, 0, -1)
\classoptions[class labels = {above left = 0.05em}, "\theta p"](8, 0, -1)
\classoptions["\theta q"](10, 1, -1)

\end{scope}

\begin{scope}[MidnightBlue]
	\class(7, 0)\AdamsTower{}
	\class(9, 1)\AdamsTower{}
	\class(11, 2)\AdamsTower{}
 \classoptions[class labels = {above = -0.05em}, "p+\theta k"](7, 0, -1)
\classoptions["\theta \ell"](9, 1, -1)

\end{scope}
\begin{scope}[Fuchsia]
	\class(8, 0)\AdamsTower{}
    \classoptions[class labels = {right = 0.05em}, "\theta g"](8, 0, -1)
	\class(9, 0)
    \classoptions[class labels = {left = 0.05em}, "\theta h"](9, 0)
        \class(10,1)
           \structline(10,1)(9,0,-1)
        \class(10,0)
        \classoptions[class labels = {left = 0.05em}, "n"](10, 0, -1)
\end{scope}

\class(9,0)
\classoptions[class labels = {above left = -0.05em}, "\theta \kappa_1 f"](9, 0, -1)

\begin{scope}[draw=none, fill=none]
%\begin{scope}[purple]
	\class(10, 0)\AdamsTower{}
\classoptions[class labels = {above = -0.05em}, "\theta \kappa_1 g"](10, 0, -1)
\end{scope}

\begin{scope}[draw=none, fill=none]
%\begin{scope}[teal]
	\class(9, 0)\AdamsTower{}
 \classoptions[class labels = {below right = -0.05em}, "\kappa_2 d"](9, 0, -1)

\end{scope}

\begin{scope}[draw=none, fill=none]
%\begin{scope}[red]
	\class(10, 0)\AdamsTower{}
    \classoptions[class labels = {right = 0.2em}, "r"](10, 0, -1)
\end{scope}

%The fifth piece from Theorem 2.28
\class(9,0)
\classoptions[class labels = {above = -0.05em}, "\theta\kappa_2 a"](9, 0, -1)

\begin{scope}[draw=none, fill=none]
%\begin{scope}[violet]
	\class(10, 0)\AdamsTower{}
 \classoptions[class labels = {below = 0.05em}, "\theta\kappa_2 d"](10, 0, -1)
 \end{scope}

 % fake classes for differentials to hit
 \begin{scope}[draw=none, fill=none]
    \class(4, 5)
    \class(5, 5)
    \class(6, 5)
    \class(7, 5)\class(7, 5)\class(7, 5)
    \class(8, 5)\class(8, 5)
    \class(8, 5)\class(8, 5)
    \class(9, 3)
 \end{scope}

 % and now for the differentials
 \begin{scope}[gray]
    \d2(5, 0, -1)(4, 2, 1)
    \d2(6, 0, 1)(5, 2, 1)
    \d2(7, 0, 2)(6, 2, -1)
    \d2(8, 1, 2)(7, 3, 3)
    \d2(9, 1, 2)(8, 3, 1)
    \d2(9, 0, 1)(8, 2, 4)
    \d2(9, 0, 2)(8, 2, 3)
 \end{scope}
 \begin{scope}[draw=none]
    \foreach \y in {1, ..., 3} {
        \d2(5, \y, -1)(4, \y+2, 1)
        \d2(6, \y, 1)(5, \y+2, 1)
        \d2(7, \y, 2)(6, \y+2, -1)
        \d2(9, \y, 1)(8, \y+2, 4)
    }
    \foreach \y in {2, 3} {
        \d2(8, \y, 2)(7, \y+2, 3)
    }
    \d2(9, 1, 2)(8, 3, 3)
 \end{scope}

\end{sseqdata}
\begin{subfigure}[c]{0.48\textwidth}
\printpage[name=wedge, page=2]
\end{subfigure}
\begin{subfigure}[c]{0.48\textwidth}
\printpage[name=wedge, page=3]
\end{subfigure}
%\end{subfigure}
\caption{\label{AdmZ4BPSL} Left: the $E_2$ page of the Adams spectral sequence computing $\widetilde{\ko}_*(B\Z_4 \wedge B\SL(3,\F_2))$. The differentials shown here are calculated in \cref{S4Z4_diffs}. Right: the $E_3 = E_\infty$-page.}
\end{figure}

We draw the submodule coming from $B\Z_4\wedge BS_4$ in \cref{AdmZ4BPSL}, left. In particular, the relations in~\eqref{R_four} imply that in topological degree $8$ and below, every class is either part of an infinite $h_0$-tower, or is annihilated by $h_0$, and all of these classes except $\textcolor{RedOrange}{\theta e}$ are annihilated by $h_1$. We can therefore account for all classes in those degrees:
\begin{cor}\label{module_generators_list}
$\Ext(\widetilde H^*(B\Z_4\wedge BS_4; \Z_2)$ is generated as a $\Z_2[h_0]$-module in topological degrees $8$ and below by the following classes.
\begin{enumerate}
    \setcounter{enumi}{2}
    \item %3
    $\textcolor{BrickRed}{\theta a}\in\Ext^{0,3}$ (annihilated by $h_0$).
    \item %4
    $\textcolor{RedOrange}{\theta d}\in\Ext^{0,4}$ ($h_0$-tower) and $\kappa_1 a\in\Ext^{0,4}$ (annihilated).
    \item %5
    $\theta\kappa_1 a\in\Ext^{0,5}$ (annihilated), $\textcolor{Goldenrod!67!black}{\kappa_1 d}\in\Ext^{0,5}$ ($h_0$-tower), $\textcolor{BrickRed}{\theta b}\in\Ext^{1,6}$ ($h_0$-tower).
    \item %6
    $\textcolor{Green}{k}\in\Ext^{0,6}$ ($h_0$-tower), $\textcolor{PineGreen}{\theta\kappa_1 d}\in\Ext^{0,6}$ ($h_0$-tower).
    \item %7
    $\theta f\in\Ext^{0,7}$ (annihilated), $\textcolor{Goldenrod!67!black}{p}\in\Ext^{0,7}$ ($h_0$-tower), $\textcolor{MidnightBlue}{\theta k}$ ($h_0$-tower).
    \item %8
    $\textcolor{PineGreen}{\theta p}\in\Ext^{0,8}$ ($h_0$-tower), $\textcolor{Fuchsia}{\theta g}\in\Ext^{0,8}$ ($h_0$-tower), $\kappa_1f\in\Ext^{0,8}$ (annihilated), $\kappa_2 a\in\Ext^{0,8}$ (annihilated), $\textcolor{RedOrange}{\theta e}\in\Ext^{1,9}$ ($h_0$-tower), $\textcolor{Green}{\ell}\in\Ext^{1,9}$ ($h_0$-tower).
\end{enumerate}
\end{cor}
\begin{thm}\label{S4Z4_diffs}
\hfill
\begin{enumerate}
    \item In the Adams spectral sequence computing $\ko_*(B\Z_4\wedge B\SL(3, \F_2))$, whose $E_2$-page is as given in \cref{the_smash_Ext_module}, the value of $d_2$ on the classes listed in \cref{module_generators_list} is
    \begin{equation}
    \begin{alignedat}{2}
        d_2(\textcolor{BrickRed}{\theta b}) &= \textcolor{RedOrange}{h_0^3 \theta d} &\qquad\qquad d_2(\textcolor{Goldenrod!67!black}{p}) &= \textcolor{PineGreen}{h_0^2 \theta \kappa_1 d}\\
        d_2(\textcolor{Goldenrod!67!black}{\kappa_1 d}) &= \textcolor{RedOrange}{h_0^2 \theta d} &
        d_2(\textcolor{MidnightBlue}{\theta k}) &= \textcolor{PineGreen}{h_0^2\theta\kappa_1 d}\\
        d_2(\textcolor{Green}{k}) &= \textcolor{BrickRed}{h_0\theta b} + \textcolor{Goldenrod!67!black}{h_0^3\kappa_1 d} & d_2(\textcolor{Green}{\ell}) &= \textcolor{MidnightBlue}{h_0^3\theta k} + \textcolor{Goldenrod!67!black}{h_0^3 p}
    \end{alignedat}
    \end{equation}
    and is $0$ on the remaining classes in \cref{module_generators_list}.
    \item For $r > 2$, $d_r$ vanishes in topological degrees $8$ and below.
    \end{enumerate}
     Therefore this Adams spectral sequence collapses at $E_3$ in topological degree $7$ and below.
\end{thm}
\begin{proof}
We use the fact that Adams differentials satisfy a Leibniz rule for the module action: if $\alpha\in\Ext(H^*(B\Z_4;\Z_2))$ and $\beta\in\Ext(\widetilde H^*((B\Z_4)_+\wedge B\SL(3, \F_2);\Z_2))$, then
\begin{equation}
    d_2(\alpha\beta) = d_2(\alpha)\beta + \alpha d_2(\beta).
\end{equation}
Using the differentials we computed in \cref{ko_Z4_diffs} and \S\ref{sec:BPSL}, we can directly evaluate most of the differentials we need.

In many cases, $d_2(\alpha) = 0$ and $d_2(\beta) = 0$, so $d_2(\alpha\beta) = 0$ as well. This is the case for $\alpha\beta\in\set{\textcolor{BrickRed}{\theta a}, \textcolor{RedOrange}{\theta d}, \theta\kappa_1 a, \textcolor{PineGreen}{\theta\kappa_1 d}, \theta f, \textcolor{Fuchsia}{\theta g}, \textcolor{RedOrange}{\theta e}}$.

In some cases, $d_2(\alpha)\ne 0$, but the relations in~\eqref{R_four} kill the differential anyway. Namely, because $h_0a = 0$ and $h_0 f = 0$,
\begin{equation}
    \begin{aligned}
        d_2(\kappa_1 a) &= h_0^2\theta a + \kappa_1\cdot 0 = 0\\
        d_2(\kappa_1 f) &= h_0^2 \theta f + \kappa_1\cdot 0 = 0\\
        d_2(\kappa_2 a) &= h_0\theta\lambda_1 a + \kappa_2\cdot 0 = 0.
    \end{aligned}
\end{equation}
Directly calculating with the Leibniz rule also nets us a few of the nonvanishing differentials:
\begin{equation}
    \begin{aligned}
        d_2(\textcolor{BrickRed}{\theta b}) &= 0\cdot b + \textcolor{RedOrange}{h_0^3\theta d} = \textcolor{RedOrange}{h_0^3\theta d}\\
        d_2(\textcolor{Goldenrod!67!black}{\kappa_1 d}) &= \textcolor{RedOrange}{h_0^2\theta d} + \kappa_1\cdot 0 = \textcolor{RedOrange}{h_0^2\theta d}.
    \end{aligned}
\end{equation}
Now we tackle $\textcolor{Green}{k}$, $\textcolor{Green}{\ell}$, and $\textcolor{Goldenrod!67!black}{p}$ using $h_0$-injectivity in a manner reminiscent of the proof of \cref{more_products}. For example, the Leibniz rule tells us
\begin{subequations}
\begin{equation}
    d_2(\textcolor{Green}{\kappa_1 b}) = \textcolor{RedOrange}{h_0^2 \theta b} + \textcolor{Goldenrod!67!black}{h_0^3\kappa_1 d}.
\end{equation}
Since $\textcolor{Green}{h_0k} = \textcolor{Green}{\kappa_1 b}$ by~\eqref{R_four} and $d_2(h_0) = 0$,
\begin{equation}
    h_0 d_2(\textcolor{Green}{k}) = \textcolor{RedOrange}{h_0^2 \theta b} + \textcolor{Goldenrod!67!black}{h_0^3\kappa_1 d}.
\end{equation}
From \cref{module_generators_list}, we know that $h_0\colon\Ext^{2,7}\to\Ext^{3,8}$ is injective, so we can cancel a factor of $h_0$ and conclude
\begin{equation}
    d_2(\textcolor{Green}{k}) = \textcolor{RedOrange}{h_0 \theta b} + \textcolor{Goldenrod!67!black}{h_0^2\kappa_1 d}.
\end{equation}
\end{subequations}
In the same way, we have $d_2(\textcolor{Green}{\lambda_1 b}) = h_0^3 \theta\kappa_1b + \lambda_1(h_0^3 d) = \textcolor{MidnightBlue}{h_0^4\theta k} + \textcolor{Goldenrod!67!black}{h_0^4 p}$ from the Leibniz rule, together with $\textcolor{Green}{h_0\ell} = \textcolor{Green}{\lambda_1 b}$ from~\eqref{R_four}, and the action of $h_0$ is once again injective here, so we deduce $d_2(\textcolor{Green}{\ell}) = \textcolor{MidnightBlue}{h_0^3\theta k} + \textcolor{Goldenrod!67!black}{h_0^3 p}$. And lastly, $d_2(\textcolor{Goldenrod!67!black}{\lambda_1 d}) = \textcolor{PineGreen}{h_0^3\theta \kappa_1 d}$ and $\textcolor{Goldenrod!67!black}{h_0p} = \textcolor{Goldenrod!67!black}{\lambda_1 d}$, and $h_0$ acts injectively here, so $d_2(\textcolor{Goldenrod!67!black}{p}) = \textcolor{PineGreen}{h_0^2\theta\kappa_1d}$.

Only $\textcolor{MidnightBlue}{\theta k}$ and $\textcolor{PineGreen}{\theta p}$ remain, and they can be dispatched directly by the Leibniz rule, now that we know $d_2(\textcolor{Green}{k})$ and $d_2(\textcolor{Goldenrod!67!black}{p})$.

Thanks to these differentials $d_2$, on the $E_3$-page there are no classes in filtration $3$ or higher in degrees $7$ and below (see \cref{AdmZ4BPSL}, right, for a picture). Thus in this range, $d_r$ vanishes for $r\ge 3$.
\end{proof}
Now we compute extensions to determine the complete structure as a $\ko_*(B\Z_4)$-module.
\begin{thm}\hfill
\label{the_Z4_module_str}
\begin{enumerate}
    \item There is an isomorphism of $\ko_*(B\Z_4)$-modules
    \begin{equation}\label{big_exten}
        \widetilde \ko_*(B\Z_4\wedge (B\SL(3, \F_2))_+)\cong \ko_*\set{ \alpha, \delta, \epsilon, \phi, \zeta, A, B, C,\dotsc}/\mathcal R_5,
    \end{equation}
    where $\abs{\alpha} = 2$. $\abs\delta = 3$, $\abs\epsilon = 7$, $\abs\phi = 6$, $\abs\zeta = 7$, $\abs A = 4$, $\abs B = 5$, and $\abs C = 7$. The ideal $\mathcal R_5$ of relations is
    \begin{equation}\label{R_five}
    \begin{aligned}
        \mathcal R_5 = (&
            2\alpha, \eta a, v\alpha, q_5\alpha,
            8\delta, \eta\delta, 4\epsilon - v\delta,
            8\epsilon,
            2\phi, \eta\phi, 
            8\zeta - 4\epsilon,\\
            &
            2A, \eta A, \ell_1A - \ell_3\alpha, \ell_3 A, 2B, \eta B,
            \ell_1 B - 2\ell_3\delta,
            8C,\\
            &             \ell_3 A - c_1 \ell_1\phi + 4c_2 C,
            q_5\delta - 4c_3C,
            \dotsc
        )
    \end{aligned}
    \end{equation}
    and all generators and relations not listed are in degrees $8$ and above, where $c_1,c_2,c_3\in\Z_2$ are undetermined.
    \item The generators listed in~\eqref{big_exten} may be chosen so that their images in the $E_\infty$-page of the Adams spectral sequence are:
    \begin{equation}
        \begin{alignedat}{3}
            \alpha &\mapsto a, &\qquad\qquad \delta &\mapsto d &\qquad\qquad \phi &\mapsto f\\
            \epsilon &\mapsto e & \zeta &\mapsto g &&\\
            A &\mapsto \kappa_1 a & B &\mapsto \theta b + h_0\kappa_1 d & C &\mapsto p + \theta k.
        \end{alignedat}
    \end{equation}
\end{enumerate}
\end{thm}
\begin{proof}
The relations involving $\alpha$, $\delta$, $\phi$, $\epsilon$, and $\zeta$ and elements of $\ko_*(\pt)$ were addressed in \cref{algebra_koPSL}. Next, one observes that on the $E_\infty$-page for $\widetilde\ko(B\Z_4\wedge B\SL(3, \F_2))$ (\cref{AdmZ4BPSL}, right), the only classes that cannot be written as a product of a class in the $E_\infty$-page for $B\Z_4$ (\cref{Z4_E2}, right) and a class in the $E_\infty$-page for $B\SL(3, \F_2)$ (\cref{ThomZ4Einfty}, right) are $\kappa_1 a$, $\theta b + h_0\kappa_1 d$, and $p + \theta k$, together with multiples of those classes by powers of $h_0$. Thus there are classes $A$, $B$, and $C$ in $\widetilde\ko_*(B\Z_4\wedge B\SL(3, \F_2))$ lifting those classes on the $E_\infty$-page, and as claimed $\alpha$, $\delta$, $\phi$, $\epsilon$, $\zeta$, $A$, $B$, and $C$ generate $\widetilde\ko_*((B\Z_4)_+\wedge B\SL(3, \F_2))$ as a $\ko_*(B\Z_4)$-module. Now we must figure out the remaining relations.

For $\ko$-homology classes $x$, $y$, and $z$, we say $(x, y)$ is an \term{extension problem (for multiplication by $z$)} if it is not clear from the $E_\infty$-page of the Adams spectral sequence whether $xz$ is equal to $y$.

Our first stop is multiplication by $2$, i.e.\ $z = 2$. The $E_\infty$-page tells us some extensions via the action of $h_0$ (e.g.\ $h_0^2(p+\theta k)\ne 0$ lifts to $4C\ne 0$), and rules out some others because an extension by $2$ must preserve topological degree and raise the Adams filtration by at least $1$. Looking at \cref{AdmZ4BPSL}, right, this leaves four extension problems for multiplication by $2$: $(A, 2\ell_1\delta)$, $(\ell_3\alpha, B)$, $(\ell_1\phi, 2C)$, and $(\ell_1\phi, 4C)$. Only the second of these four extension problems can affect the $\ko_*$-module structure of the final answer, but all four affect the $\ko_*(B\Z_4)$-module structure. Fortunately, these four extensions all split, and for the same reason: Margolis' theorem~\cite{Margolis} shows that classes on the $E_2$-page arising from Ext of a free $\cA(1)$-module summand in the cohomology of a space or spectrum cannot support nonzero differentials or participate in nonzero extensions by elements in $\pi_*(\mathbb S)$, such as $2$ and $\eta$. In this specific case, to use Margolis' theorem we need that the classes $\kappa_1 a$, $\theta\kappa_1 a$, and $\theta f$ all come from free $\cA(1)$-module summands in $\widetilde H^*(B\Z_4\wedge B\SL(3, \F_2);\Z_2)$. Since $\Ext^0 = \Hom$, classes in Adams filtration $0$ can be represented by $\cA(1)$-module homomorphisms to $t\colon \Sigma^t\Z_2$ (see \cref{detector}); to check that this homomorphism generates a free summand, compute $\Sq^2\Sq^2\Sq^2$ on a class $x$ with $t(x)\ne 0$ (see Freed-Hopkins~\cite[\S D.4]{FH21}).
\begin{itemize}
    \item Since $\kappa_1$ came from $y\in H^2(B\Z_4;\Z_2)$ and $a$ came from $\nu_2\in H^2(B\SL(3, \F_2);\Z_2)$, the class $y\nu_2$ detects $\kappa_1 a$, so to invoke Margolis' theorem one can check that, indeed, $\Sq^2(\Sq^2(\Sq^2(y\nu_2)))\ne 0$.
    \item Likewise, $\theta$ is detected by $x\in H^1(B\Z_4;\Z_2)$, and $\Sq^2(\Sq^2(\Sq^2(xy\nu_2)))\ne 0$, preventing $\theta\kappa_1 a$ from participating in an extension.
    \item Finally, $f$ is detected by $\nu_2^3$ and $\theta$ by $x$, so $\theta f$ is detected by $x\nu_2^3$, and $\Sq^2(\Sq^2(\Sq^2(x\nu_2^3))) = 0$.
\end{itemize}
Next, extensions by $\eta$. A priori, the possible pairs for hidden extensions by $\eta$ are $(\ell_1\alpha, 2\ell_1\delta)$; $(c_1 \ell_1\delta + c_2 A, B)$ for some $c_1,c_2$ not both $0$; $(\ell_3\alpha, 2\ell_3\delta)$, and $(\ell_3\delta, c_3 C)$ for $C_3 = 2$, $4$, or $6$. However, from \cref{algebra_koPSL} we know $\eta\alpha = 0$ and $\eta\delta = 0$, which splits almost all of these hidden extensions, leaving only $(A, B)$, which splits by Margolis' theorem.

Hidden extensions by $v$ are ruled out by degree reasons: any such extension must increase Adams filtration by at least $3$, but all classes in topological degree at most $7$ on the $E_\infty$-page have Adams filtration at most $2$. And the Bott class is in too high of a degree to be relevant for us.

The relation $\ell_1^2 = 0$ leaves very little room for there to be hidden extensions by $\ell_1$, since on the $E_\infty$-page so many classes are products with $\theta$, or admit a nonzero product with $\theta$, lifting to provide most of the information on the $\ell_1$-action on $\widetilde\ko_*((B\Z_4)_+\wedge B\SL(3, \F_2))$. The only ambiguity is what $\ell_1 A$ is in terms of the generators in degree $5$, but this can be solved by redefining $B$ if necessary, and such a choice does not affect the rest of the generators and relations. So there are no hidden $\ell_1$-extensions. (The relation $\ell_1 B = 2\ell_3\delta$ is not hidden, as it is visible from the $\theta$-action on the $E_\infty$-page.)

For $\ell_3$, degree reasons suggest the following hidden extensions: $(\theta\alpha, c_1\ell_3\delta)$ for $c_1$ equal to $1$, $2$, or $3$; and $(c_2\theta\delta + c_3 A, c_4\ell_1\phi + c_5C)$ for $c_2,c_3\in\Z_2$ not both zero and $c_4\in\Z_2$ and $c_5\in\Z_8$ not both zero. The first class of hidden extensions are split by the fact $\theta \ell_3 = 0$, as are the extensions in the second class with $c_3 = 0$. We do not solve the second class of extensions in this theorem, though we note that since $2A = 0$, $c_5$ must be a multiple of $4$.

The only remaining relation in topological degree less than $8$ is $q_5\alpha$, which we have left unresolved.
\end{proof}

In a moment, when we search for manifold generators for the corresponding bordism groups, it will be helpful to unspool the module structure in \cref{the_Z4_module_str} into something more explicit.
\begin{cor}\label{ABC_and_more}
There are isomorphisms
\begin{equation}
\begin{aligned}
    \widetilde\ko_k(B\Z_4\wedge B\SL(3, \F_2)) &\cong 0, \qquad k\le 2,\\
    \widetilde\ko_3(B\Z_4\wedge B\SL(3, \F_2)) &\cong \Z_2\cdot \ell_1\alpha\\
    \widetilde\ko_4(B\Z_4\wedge B\SL(3, \F_2)) &\cong \Z_4\cdot \ell_1\delta \oplus \Z_2\cdot A\\
    \widetilde\ko_5(B\Z_4\wedge B\SL(3, \F_2)) &\cong \Z_2\cdot \ell_1 A \oplus \Z_2\cdot B\\
    \widetilde\ko_6(B\Z_4\wedge B\SL(3, \F_2)) &\cong \Z_4\cdot \ell_3\delta\\
    \widetilde\ko_7(B\Z_4\wedge B\SL(3, \F_2)) &\cong \Z_8\cdot C \oplus \Z_2\cdot \theta\phi.
\end{aligned}
\end{equation}
\end{cor}

Just as the Adams spectral sequence for $\ko(L(2))$ was simple in comparison to $\ko_*(B\SL(3,\F_2))$, the Adams spectral sequence for $\ko_*(B\Z_4 \wedge L(2))$ is much simpler than $\ko_*(B\Z_4 \wedge B\SL(3,\F_2))$. Indeed, an immediate consequence of \cref{Thm:CohoZ4} and \cref{Prop:cohoL2} is that all nonzero classes in the second page of the Adams spectral sequence computing $\ko_*(B\Z_4 \wedge L(2))$ lie in filtration zero, and for $n \leq 7$, all $E_2^{0,n}(B\Z_4 \wedge L(2))$ vanish except 
\begin{equation}
    E^{0,5}_2(B\Z_4 \wedge L(2)), \quad E^{0,6}_2(B\Z_4 \wedge L(2)),\quad E^{0,7}_2(B\Z_4 \wedge L(2))
\end{equation}
Furthermore, the spectral sequence collapses at $E_2$, giving us 
\begin{equation}
\label{L2_Z4}
\begin{gathered}
\ko_5(B\Z_4 \wedge L(2))\cong \Z_2,\quad ko_6(B\Z_4\wedge L(2)) \cong \Z_2,\quad \ko_7(B\Z_4\wedge L(2)) \cong \Z_2,
\end{gathered}
\end{equation}
detected by the mod $2$ cohomology classes $xa^2b$, $ya^2b$, and $xya^2b$ respectively.

\subsection{Generators}
\label{ss:Z4S4gens}
\subsubsection{Generators coming from the K\"{u}nneth map}
The Adams K\"{u}nneth map \eqref{Adams_Knn}, $X = B\Z_4$ and $Y = B\SL(3,\F_2)$, can be used to determine many of the generators for $\ko(B\Z_4\wedge B\SL(3,\F_2))$. Indeed, for those classes in $\Ext(\widetilde{H}^*(B\Z_4;\Z_2)\otimes \widetilde{H}^*(B\SL(3,\F_2);\Z_2))$ that are in the image of the Adams K\"{u}nneth map, the generators of the corresponding elements in $\widetilde{\ko}(B\Z_4\wedge B\SL(3,\F_2))$ are products of generators from $\widetilde{\ko}_*(B\Z_4)$ and $\widetilde{\ko}_*(B\SL(3,\F_2))$, which we determined in Sections \ref{sec:BZ4gens} and \ref{sec:BPSLgens}.
\begin{itemize}
    \item In dimension 3, the generator, detected by $\theta a$, is generated by $L_4^1 \times T^2 $.
    \item In dimension 4, the $\Z_4$ summand, detected by $\theta d$, is generated by $L^1_4 \times L^3_4$.
    \item In dimension 5, the $\Z_2$ summand detected by $\theta \kappa_1 a$ is generated by $L^3_4 \times T^2$.
    \item In dimension 6, the generator, detected by $\theta \kappa_1 d$, is generated by $L_4^3 \times L_4^3$. 
    \item In dimension 7, the $\Z_2$ summand detected by $\theta f$ is generated by $L_4^1 \times \mathbb{R}\mathbb{P}^3 \times \mathbb{R}\mathbb{P}^3$. 
\end{itemize}

\subsubsection{\texorpdfstring{Generators coming from $\widetilde{\ko}_*(B\Z_4\wedge B\Z_4)$ and $\widetilde{\ko}_*(B\Z_2\wedge B(\Z_2\times \Z_2))$}{Generators coming from reduced ko(BZ4 wedge BZ4) and reduced ko(BZ2wedge B(Z2 x Z2))}}
As in the case  of $\widetilde{ko}_*(B\SL(3,\F_2))$ in Section \ref{sec:BPSLgens}, we study the generators of $\widetilde{ko}(B\Z_4\wedge B\SL(3,\F_2))$ using \cref{S4_splitting}. That is, we instead consider the question of determining generators for $\widetilde{\ko}_*(B\Z_4\wedge BS_4)$. For this, subgroups of $\Z_4 \times S_4$ will be of particular use. 

\begin{lem}
\label{lem:Z2Z2Z2inclu}
    The inclusion $\Z_2 \times (\Z_2\times \Z_2)\hookrightarrow \Z_4 \times S_4$ given by $(1,0,0)\mapsto (2,\mathrm{id})$, $(0,1,0)\mapsto (0,(1\;2)(3\;4))$, and $(0,0,1) \mapsto(0,(1\;3)(2\;4))$ induces a map on $\Z_2$ cohomology 
    \begin{equation}
        \begin{gathered}
            H^*(B\Z_4 \times BS_4;\Z_2)\rightarrow H^*(B\Z_2 \times B( \Z_2 \times \Z_2);\Z_2)\\
            x\mapsto 0,\quad y\mapsto \alpha^2, \quad a\mapsto 0,\quad b \mapsto \beta^2 + \beta\gamma + \gamma^2, \quad c\mapsto \beta^2\gamma + \beta\gamma^2
        \end{gathered} 
    \end{equation}
where 
\begin{equation}
    \begin{gathered}
        H^*(B\Z_4 \times BS_4;\Z_2) \cong \Z_2[x,y,a,b,c]/(x^2, ac) \\
        H^*(B\Z_2 \times B(\Z_2 \times \Z_2);\Z_2)\cong \Z_2[\alpha,\beta,\gamma]
    \end{gathered}
\end{equation}
\end{lem}
\begin{proof}
    The claim follows from the K\"{u}nneth theorem and Lemma \ref{lem:Z2Z2S4inclu}.
\end{proof}
\begin{prop}\label{detected_A}
Consider $\RP^3\times S^1$ with any of its four spin structures and the principal $(\Z_4\times S_4)$-bundle $P\to\RP^3\times S^1$ specified by the homomorphism
\begin{equation}
 \pi_1(\RP^3\times S^1) \cong \Z_2\times \Z \overset{f_1}{\to} \Z_2\times\Z_2\times\Z_2\overset{f_2}{\to} \Z_4\times S_4,
 \end{equation}
 where $f_1(1, 0) = (1, 1, 0)$, $f_1(0, 1) = (0, 0, 1)$, $f_2(1, 0, 0) = (2, 0)$, $f_2(0, 1, 0) = (0, (1\ 2)(3\ 4))$, and $f_2(0, 0, 1) = (0, (1\ 3)(2\ 4))$. Then for some $\lambda\in\set{1,3}$, $[\RP^3\times S^1, P]\in\widetilde \Omega_4^\Spin(B\Z_4\wedge BS_4)$ is equal to $A + \lambda \ell_1\delta$.
\end{prop}
The exact value of $\lambda$ could depend on the choice of spin structure, and in any case does not matter to us; all we need from \cref{detected_A} is that this spin bordism class is linearly independent from $\ell_1 \delta$, so that $\RP^3\times S^1$ and $L_4^1\times L_4^3$ with their $(\Z_4\times S_4)$-bundles generate $\widetilde\Omega_4^\Spin(B\Z_4\wedge BS_4)$.
\begin{proof}
From \cref{ABC_and_more}, $\widetilde\Omega_4^\Spin(B\Z_4\wedge BS_4)\cong \Z_2\cdot A\oplus \Z_4\cdot \ell_1\delta$, and from the images of $A$ and $\ell_1\delta$ on the $E_\infty$-page (\cref{the_Z4_module_str}), we learn $A$ is detected by the mod $2$ cohomology class $yb$ and $\ell_1\delta$ is detected by $xc$. Thus to prove the lemma it suffices to show
\begin{equation}
    \int_{\RP^3\times S^1} y(P)b(P) = \int_{\RP^3\times S^1} x(P)c(P) = 1.
\end{equation}
Under the inclusion $\Z_2 \times (\Z_2 \times \Z_2) \hookrightarrow \Z_4 \times S_4$ described in \cref{lem:Z2Z2Z2inclu}, the class $yb$ pulls back to 
\begin{equation}
\label{expand_yb}
    \alpha^2\beta^2 +\alpha^2\beta\gamma +  \alpha^2\gamma^2
\end{equation}
A straightforward characteristic class computation reveals that for the bundle specified in the theorem statement, all terms in~\eqref{expand_yb} vanish except $\alpha^2 \beta\gamma$, which is nonzero, so $\int_{\mathbb{R}\mathbb{P}^3 \times S^1}yb =1$. In a similar way, \cref{lem:Z2Z2Z2inclu} implies $xc$ pulls back to $\alpha(\beta^2\gamma + \beta\gamma^2)$, and plugging in the bundle $P$ in the theorem statement, we have $\alpha\beta^2\gamma \ne 0$ but $\alpha\beta\gamma^2 = 0$, implying $\int_{\RP^3\times S^1} xc = 1$.
\end{proof}
\begin{prop}
\label{torus_gen}
Consider $T^4\times\RP^3$ with any of its $32$ spin structures and the principal $(\Z_4\times S^4)$-bundle $P\to T^4\times\RP^3$ specified by the homomorphism
\begin{equation}\label{C_torus_spec}
    \begin{aligned}
        \pi_1(T^4\times\RP^3)\cong \Z^4\times\Z_2 &\longrightarrow \Z_4\times S_4\\
        (\vec e_1, 0) &\longmapsto (1, 0)\\
        (\vec e_2, 0) &\longmapsto (0, (1\ 2)(3\ 4))\\
        (\vec e_3, 0) &\longmapsto (0, (1\ 3)(2\ 4))\\
        (\vec e_4, 0) &\longmapsto (0, (1\ 3)(2\ 4))\\
        (0, 1) &\longmapsto (2, (1\ 2)(3\ 4)).
    \end{aligned}
\end{equation}
Then there is an odd $\lambda$ such that $[T^4\times\RP^3, P] = \lambda C$ in $\widetilde\Omega_7^\Spin(B\Z_4\wedge BS_4)$, so that $\set{L_4^1\times\RP^3\times\RP^3, T^4\times\RP^3}$ generates $\widetilde\Omega_7^\Spin(B\Z_4\wedge BS_4)$.
\end{prop}
Again, we do not need to know the precise value of $\lambda$, and said value could depend on the choice of spin structure.
\begin{proof}
From \cref{ABC_and_more}, $\widetilde\Omega_7^\Spin(B\Z_4\wedge BS_4)\cong \Z_2\cdot \ell_1 \phi\oplus \Z_8\cdot C$, and from the images of $\ell_1\phi$ and $C$ on the $E_\infty$-page (\cref{the_Z4_module_str}), we learn $\ell_1\phi$ is detected by the mod $2$ cohomology class $xb^3$ and $C$ is detected by $y^2c + xyb^2$. Thus to prove the lemma it suffices to show
\begin{equation}
    \int_{T^4\times\RP^3} x(P)y(P)b(P)^2 = 1,\qquad \int_{T^4\times\RP^3} c(P)y(P)^2 = \int_{T^4\times\RP^3} x(P)b(P)^3 = 0.
\end{equation}
\begin{enumerate}
    \item There is a Fubini theorem for mod $2$ cohomology, so that
    \begin{equation}
        \int_{T^4\times\RP^3} x(P)y(P)b(P)^2 = \int_{S^1} x(P|_{S^1}) \int_{S^1\times S^1} b(P|_{S^1\times S^1}) \int_{S^1\times\RP^3} y(P|_{S^1\times\RP^3})b(P|_{S^1\times\RP^3}).
    \end{equation}
    Here the four $S^1$ factors in $T^4$ appear in the same order as they did in~\eqref{C_torus_spec}. The first factor of $S^1$, with $P|_{S^1}$, is exactly $L_4^1$, and $\int_{L_4^1}x\ne 0$ (\cref{Z4_mult_gens}); $P$ restricted to the next two factors of $S^1$ coincides with the bundle on $S^1\times S^1 = \RP^1\times\RP^1$ appearing in \cref{f_manifold}, so $\int_{S^1\times S^1} b\ne 0$; and $P$ on $S^1\times\RP^3$ coincides with the bundle appearing in \cref{detected_A}, so $\int_{S^1\times\RP^3} yb = 1$. Therefore the product of these three integrals is also nonzero,
    \item The homomorphism~\eqref{C_torus_spec} factors through the inclusion $j\colon \Z_4\times\Z_2\times\Z_2\hookrightarrow \Z_4\times S_4$. If we let $z$ and $w$ denote the generators of $H^1(B\Z_2\times B\Z_2;
    \Z_2)$ and $x$ and $y$ denote the usual generators of $H^*(B\Z_4;\Z_2)$, then $j^*(c) = xy + z^2w+ zw^2$ (\cref{lem:Z2Z2S4inclu}) and $j^*(y) = y$, so
    \begin{equation}\label{jcy2}
        j^*(cy^2) = y^2(xy + z^2w + zw^2).
    \end{equation}
    In the rest of this part of the proof, $y$ refers to the class in $H^*(B\Z_4\times B\Z_2\times B\Z_2;\Z_2)$.
    If $u$ denotes the generator of $H^*(\RP^3;\Z_2)$, then~\eqref{C_torus_spec} implies $y(P) = u^2$, so $y(P)^2 = 0$, and therefore the characteristic class~\eqref{jcy2} vanishes on $T^4\times \RP^3$.
    \item For $xb^3$, we once again pull back to $B(\Z_4\times\Z_2\times\Z_2)$; if $x$, $y$, $z$, and $w$ refer to the same classes as in the previous part of this proof, then $b$ pulls back to $z^2 + zw + w^2$. Thus $b^3$ pulls back to a product of terms $z^iw^j$ where $i > 3$ or $w\ge 3$. The K\"{u}nneth formula implies that for any class $q\in H^1(T^4\times\RP^3;\Z_2)$, $q^4 = 0$, and~\eqref{C_torus_spec} implies $w(P)$ pulls back across the projection $T^3\times\RP^3\to T^4$, and the cube of any class in $H^1(T^4;\Z_2)$ vanishes. Therefore $\int_{T^4\times\RP^3} xb^3 = 0$. \qedhere
\end{enumerate}
\end{proof}
\begin{prop}\label{Q5_gen}
Recall the spin $5$-manifold $Q_4^5$ and its $\Z_4$-bundle $S(V)\to Q_4^5$ from \cref{Q5_defn,Z4_mult_gens}, classified by the canonical identification $s\colon \pi_1(Q_4^5)\overset\cong\to\Z_4$. Let $P\to Q_4^5$ denote the principal $(\Z_4\times S_4)$-bundle classified by the homomorphism
\begin{equation}
    \pi_1(Q_4^5) \overset{s}{\longrightarrow} \Z_4\overset{g}{\longrightarrow} \Z_4\times S_4,
\end{equation}
where $g(1) \coloneqq (1, (1\ 2\ 3\ 4))$. Then $[Q_4^5, P] = B$ in $\widetilde\Omega_5^\Spin(B\Z_4\wedge BS_4)$.
\end{prop}
\begin{proof}
From \cref{ABC_and_more}, $\widetilde\Omega_5^\Spin(B\Z_4\wedge BS_4)\cong\Z_2\cdot \ell_1 A\oplus \Z_2\cdot B$, and from the image of $\ell_1 A$ on the $E_\infty$-page we know that $\ell_1 A$ is detected by the mod $2$ cohomology class $xyb$. It thus suffices to show $\int_{Q_4^5} x(P)y(P)b(P)= 0$ but that $Q_4^5$ is not null-bordant.

One can show that $H^*(Q_4^5;\Z_2)\cong\Z_2[x, y, t]/(x^2, y^2, t^2)$ with $\abs x = 1$ and $\abs y = \abs t = 2$ using a Serre spectral sequence argument similar to the one for $Q_4^{11}$ in~\cite[Proposition D.29]{Debray:2021vob}, and that $x(P) = x$, $y(P) = y$, and $b(P) = y(P) = y$, since the image of $b\in H^2(BS_4;\Z_2)$ under the pullback map to $H^*(B\Z_4;\Z_2)$ is $y$. Thus $xyb = xy^2 = 0$.

Let $\eta_k^D$ denote the $\eta$ invariant of the twisted Dirac operator associated to the charge-$k$ irreducible complex representation of $\Z_4$ (so that $k\in\Z_4$). Then $\eta_k^D$ is a bordism invariant $\Omega_5^\Spin(B\Z_4)\to\R/\Z$ and $\eta_1^D(Q_4^5) = -1/4$~\cite[Table 19]{Heckman}, so $\eta_k^D(Q_4^5, S(V)) = -k/4\bmod 1$. Because the permutation $(1\ 2\ 3\ 4)$ is an odd element of $S_4$, the charge $2$ irreducible representation of $\Z_4$ extends to a complex representation of $S_4$, namely the complexified sign representation. Therefore $\eta_2^D$ extends to a $\Z_2$-valued bordism invariant of $S_4$-bundles, and the product $(\eta_2^D)^{\mathrm{left}}(\eta_2^D)^{\mathrm{right}}$ (``left'' and ``right'' for the $\Z_4$ and $S_4$ factors) is a bordism invariant $\widetilde\Omega_5^\Spin(B\Z_4\wedge BS_4)\to\Z_2$ whose value on $(Q_4^5, P)$ is nonzero.
\end{proof}
This argument was effectively about the spin bordism of $B\Z_4\wedge B\Z_4$, rather than $B\Z_4\wedge BS_4$. The spin bordism of $B\Z_4\wedge B\Z_4$ is studied in detail by Bárcenas, García-Hernández, and Reinauer~\cite[\S 5, \S 6]{BGHR24}, who show that the class of $Q_4^5$ is nonzero in $\widetilde\Omega_5^\Spin(B\Z_4\wedge B\Z_4)$ using $\eta$-invariants, though they do not discuss extensions to $\Z_4\times S_4$.

\subsubsection{\texorpdfstring{Generators of $\ko_*(B\Z_4 \wedge L(2))$}{Generators of ko(BZ4 wedge L(2))}}
It is straightforward to show that the generators of $\widetilde{\ko}_5(B\Z_4 \wedge L(2))$, $\widetilde{\ko}_6(B\Z_4 \wedge L(2))$, and $\widetilde{\ko}_7(B\Z_4 \wedge L(2))$ are detected by the classes $xa^2b, y a^2b, xy a^2 \in H^*(B\Z_4 \times BS^4;\Z_2)$. It follows that $\widetilde{\ko}_5(B\Z_4 \wedge L(2))$ is generated by $L_4^1 \times W_4$, where $W_4$ was described in Section \ref{sec:W4}. Similarly, if we let $W_6$ denote the generator of $\widetilde{\ko}_6(B\Z_4 \wedge L(2))$, then $L_4^1 \times W_6$ generates $\widetilde{ko}_7(B\Z_4 \wedge L(2))$.

To find this last generator, we first shrink the search space by replacing $\Z_4\times S_4$ with a smaller group.
\begin{lem}
Let $\varphi\colon\Z_2^2\to \Z_4\times S_4$ be the homomorphism sending $(1, 0)\mapsto (2, (1\ 3)(2\ 4))$ and $(0, 1)\mapsto (0, (1\ 3))$.
%and $(0, 0, 1)\mapsto (0, (1\ 3)(2\ 4))$.
Then under the identifications $H^*(B\Z_4;\Z_2)\cong\Z_2[x,y]/(x^2)$, $H^*(BS_4;\Z_2)\cong\Z_2[a,b,c]/(ac)$ and $H^*(B\Z_2^2;\Z_2)\cong\Z_2[\alpha, \beta]$ given by Theorems~\ref{Thm:CohoZ4}, \ref{Thm:CohoS4}, and the K\"{u}nneth formula respectively,
\begin{equation}
    \begin{alignedat}{2}
        \varphi^*(x) &=0 \qquad\qquad& \varphi^*(a) &= \beta\\
        \varphi^*(y) &= \alpha^2 & \varphi^*(b) &= \alpha^2 + \alpha\beta\\
        && \varphi^*(c) &= 0.
    \end{alignedat}
\end{equation}
\end{lem}
\begin{proof}
Naturality of the K\"{u}nneth isomorphism implies that it suffices to understand the pullbacks by the maps $\phi_1\colon\Z_2\inj \Z_4$ and $\phi_2\colon \Z_2^2\inj S_4$ given by $\phi_1(1) = 2$, $\phi_2(1, 0) = (1\ 3)$, and $\phi_2(0, 1) = (1\ 3)(2\ 4)$. For $\phi_1$ the calculation on cohomology follows from~\cite[Lemma 14.38]{Heckman} (specifically $i_4$, not $\tilde\imath_4$). For $\phi_2$, we use the fact that $b$ and $c$ are $w_2$, resp.\ $w_3$ of the standard four-dimensional permutation representation $\boldsymbol 4$ of $S_4$~\cite[Proposition 5.1]{Priddy}: we can therefore restrict the representation to $\Z_2^2$ and calculate its Stiefel-Whitney classes there. If $\sigma_i$ denotes the real, one-dimensional representation which is nontrivial on the $i^{\mathrm{th}}$ copy of $\Z_2$ inside $\Z_2^2$, and which is trivial on the other copy, then the reader can directly check
\begin{equation}
    \phi_2^*(\boldsymbol 4)\cong\underline \R^2\oplus \sigma_2\oplus (\sigma_1\otimes\sigma_2).
\end{equation}
The Whitney sum formula quickly calculates $w_2$ and $w_3$ of this representation for us.
\end{proof}
The reader can then check that
\begin{equation}\label{ya2bpullback}
    \varphi^*(ya^2b) = \alpha^4\beta^2 + \alpha^3\beta^3.
\end{equation}
\begin{prop}
\label{found_W6}
Let $P_i\to\RP^3\times\RP^3$ be the principal $\Z_2$-bundle which has nontrivial monodromy around the $i^{\mathrm{th}}$ $\RP^3$ and is trivial around the other $\RP^3$. Then the bordism class of $(\RP^3\times\RP^3, \varphi(P_1, P_2))$ inside $\widetilde\Omega_6^\Spin(B\Z_4\wedge BS_4)$ is not in the span of $\set{C, \theta\phi}$, and therefore may be taken as our representative for the final remaining generator $W_6$.
\end{prop}
\begin{proof}
First, $\RP^3$ is spin, so $\RP^3\times\RP^3$ is also spin. As we discussed above, it therefore suffices to show that
\begin{equation}
\label{eval_integral}
\int_{\RP^3\times\RP^3} (ya^2b)(\varphi(P_1, P_2))
\end{equation}
is nonzero. Let $a\in H^2(B\Z_2;\Z_2)$ denote the generator; by~\eqref{ya2bpullback}, $\varphi^*(ya^2b) = \alpha^4\beta^2 + \alpha^3\beta^3$, and by definition we have $\alpha = a(P_1)$ and $\beta = a(P_2)$, so~\eqref{eval_integral} is equivalent to
\begin{equation}
\label{PE32int}
    \int_{\RP^3\times\RP^3} (a(P_1)^4 a(P_2)^2 + a(P_1)^3 a(P_2)^3),
\end{equation}
so we need to show this is nonzero. Now $H^*(\RP^3\times\RP^3;\Z_2)\cong\Z_2[a(P_1), a(P_2)]/(a(P_1)^4, a(P_2)^4)$, which immediately implies~\eqref{PE32int} equals $1$.
\end{proof}

\newpage

\appendix

\section{The shape of singular fibers}
\label{app:cryssingfiber}

The first step to deriving the singular fiber geometries for the cases
\begin{align}
T^3 / (\mathbb{Z}_3)_{\Gamma^{(2)}_3} \,, \quad T^3 / (\mathbb{Z}_4)_{\Gamma^{(2)}_4} \,,
\end{align}
is to find three-dimensional lattices, in whose basis the monodromies act like a rotation. For the 2-torus and $\gamma_3$ this meant choosing $\tau = e^{2 \pi i/3}$ which we will re-derive exemplifying the technique applicable to the higher-dimensional cases.

\subsubsection*{Example: The 2-torus}

Recall that the SL$(2,\mathbb{Z})$ transformation $\gamma_3$ is given by
\begin{equation}
\gamma_3 = \begin{pmatrix} -1 & -1 \\ 1 & 0 \end{pmatrix} \,.
\end{equation}
We want to find two vectors $\lambda_1$ and $\lambda_2$ that form a basis of the lattice $\Lambda_2$ defining the torus $T^2 = \mathbb{R}^2 / \Lambda_2$. And we want them to be such that $\gamma_3$ looks like a rotation in this necessarily non-orthogonal basis. For that we choose the ansatz
\begin{equation}
\lambda_1 = \begin{pmatrix} 1 \\ 0 \end{pmatrix} \,, \quad \lambda_2 = \begin{pmatrix} a \\ b \end{pmatrix} \,.
\label{eq:basisansatz}
\end{equation}
The action of $\gamma_3$ is given by
\begin{align}
\begin{split}
\lambda_1 & \rightarrow \widetilde{\lambda}_1 = \gamma_3 \begin{pmatrix} \lambda_1 \\ \lambda_2 \end{pmatrix} = - \lambda_1 - \lambda_2 = \begin{pmatrix} -1 - a \\ - b \end{pmatrix} \,, \\
\lambda_2 & \rightarrow \widetilde{\lambda}_2 = \gamma_3 \begin{pmatrix} \lambda_1 \\ \lambda_2 \end{pmatrix} = \lambda_1 = \begin{pmatrix} 1 \\ 0 \end{pmatrix} \,.
\end{split}
\end{align}
We are looking for a $(2 \times 2)$ matrix $A$ that has the property $A \widetilde{\lambda}_i = \lambda_i$, which is given by
\begin{equation}
A = \begin{pmatrix} a & - \frac{a(a+1)+1}{b} \\ b & - (a+1) \end{pmatrix} \,,
\end{equation}
where we have used that $b \neq 0$ for $\lambda_1$ and $\lambda_2$ to be linearly independent. Next, we demand that $A$ is a rotation matrix, i.e., an element in SO$(2)$, which can be phrased as $A^T A = \mathbf{1}$. One solution to this requirement, that also leads to a right-handed coordinate system is given by
\begin{equation}
a = - \tfrac{1}{2} \,, \quad b = \tfrac{\sqrt{3}}{2} \,.
\end{equation}
Indeed, expressed as a complex number by identifying $\mathbb{R}^2 \simeq \mathbb{C}$, we find
\begin{equation}
\begin{pmatrix} a \\ b \end{pmatrix} \simeq \tau = e^{2 \pi i/3} \,,
\end{equation}
precisely what we would expect from the action of modular transformations.

\subsubsection*{Fixing the shape of the 3-tori}

Now we apply the same technique to the three-dimensional tori. 

It is easy to see that $\Gamma_3^{(2)}$ is actually an element of SO$(3)$ which means that for the orbifold $T^3 / (\mathbb{Z}_3)_{\Gamma^{(2)}_3}$ we can choose the lattice $\Lambda_3 = \mathbb{Z}^3$. To determine the actual geometric realization of the singular fiber we will actually increase the fundamental domain of the torus, i.e., discuss a multi-cover. The basis of the enlarged fundamental domain is given by 
\begin{equation}
\widetilde{\lambda}_1 = \begin{pmatrix} 1 \\ 1 \\ 1 \end{pmatrix} \,, \quad \widetilde{\lambda}_2 = \begin{pmatrix*}[r] 1 \\ -1 \\ 0 \end{pmatrix*} \,, \quad \widetilde{\lambda}_3 = \begin{pmatrix*}[r] 0 \\ 1 \\ -1 \end{pmatrix*} \,,
\end{equation}
The order of the covering can be easily inferred from the determinant of the base change
\begin{equation}
\text{det} \begin{pmatrix*}[r] 1 & 1 & 0 \\ 1 & -1 & 1 \\ 1 & 0 & -1 \end{pmatrix*} = 3 \,.
\end{equation}
This new basis is chosen since $\widetilde{\lambda}_1$ coincides with the rotation axis, and $\widetilde{\lambda}_2$ and $\widetilde{\lambda}_3$ are perpendicular to it. To be more precise the action of $\Gamma_3^{(2)}$ is given by
\begin{equation}
\widetilde{\lambda}_1 \rightarrow \widetilde{\lambda}_1 \,, \quad \widetilde{\lambda}_2 \rightarrow - \widetilde{\lambda}_2 -  \widetilde{\lambda}_3 \,, \quad  \widetilde{\lambda}_3 \rightarrow  \widetilde{\lambda}_2 \,.
\end{equation}
This is exactly the transformation property imposed by $\gamma_3$ in the sub-torus spanned by $ \widetilde{\lambda}_2$ and $ \widetilde{\lambda}_3$. Therefore we find that the triple cover $\widetilde{T}^3$ allows for a straightforward quotient structure given by
\begin{equation}
\widetilde{T}^3 / (\mathbb{Z}_3)_{\Gamma_3^{(2)}} = (T^2/\mathbb{Z}_3) \times S^1 \,.
\end{equation}
To obtain the actual geometry we further have to reduce to the original fundamental domain of the torus. This can be done by additionally implementing shifts by internal lattice points, of which there are three since $\widetilde{T}^3$ is a triple cover. These internal lattice points are given by
\begin{equation}
p_1 = \begin{pmatrix} 0 \\ 0 \\ 0 \end{pmatrix} \,, \quad p_1 = \begin{pmatrix} 1 \\ 0 \\ 0 \end{pmatrix} = \tfrac{1}{3} \widetilde{\lambda}_1 + \tfrac{2}{3}  \widetilde{\lambda}_2 + \tfrac{1}{3}  \widetilde{\lambda}_3 \,, \quad p_3 = \begin{pmatrix} 1 \\ 1 \\ 0 \end{pmatrix} = \tfrac{2}{3}  \widetilde{\lambda}_1 + \tfrac{1}{3} \widetilde{\lambda}_2 + \tfrac{2}{3}  \widetilde{\lambda}_3 \,.
\end{equation}
Modding out by this $\mathbb{Z}_3^s$ translational action, we obtain the singular fiber
\begin{equation}
T^3 / (\mathbb{Z}_3)_{\Gamma_3^{(2)}} = \big( (T^2/\mathbb{Z}_3) \times S^1 \big) / \mathbb{Z}_3^s \,.
\end{equation}
This can be understood as a fibration of $T^2 / \mathbb{Z}_3$ over the circle with periodicity $\tfrac{1}{3} \widetilde{\lambda}_1$, see Figure~\ref{fig:T3Z3}.
\begin{figure}
\centering
\includegraphics[width = 0.6 \textwidth]{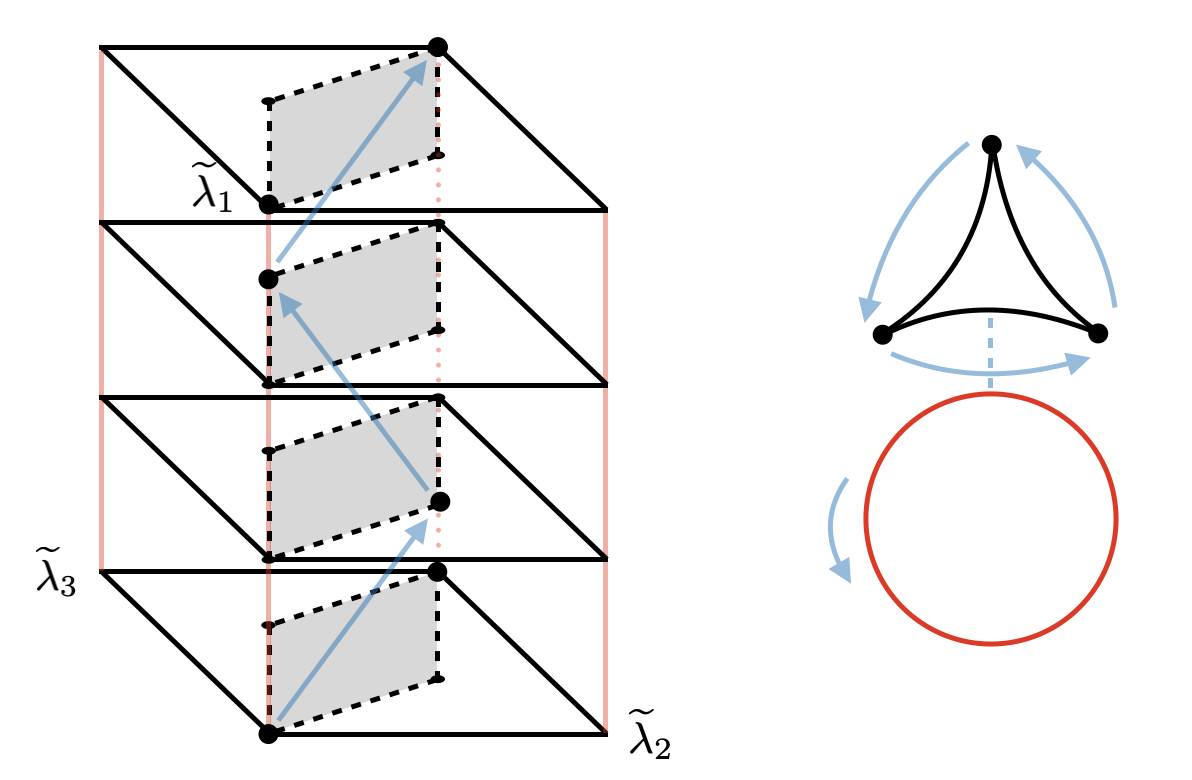}
\caption{The singular fiber $T^3 / (\mathbb{Z}_3)_{\Gamma_3^{(2)}}$ in terms of its triple cover on the left, where the translational symmetry on the fixed points of $T^2/\mathbb{Z}_3$ is indicated by shaded blue arrows, and its fibration structure on the right.}
\label{fig:T3Z3}
\end{figure}
In particular, we see that $\mathbb{Z}_3^s$ exchanges the fixed points of $T^2 / \mathbb{Z}_3$ when going around the circle, leaving a single line of local $\mathbb{C}/\mathbb{Z}_3$ singularities as expected.

We have to work a little harder for the singular fiber in the case of $\Gamma_4^{(2)}$, since it is not given by an element of SO$(3)$. As in the case of $T^2 /\mathbb{Z}_3$ discussed above, this means that we want to find a basis of the lattice $\Lambda_3$, which is necessary non orthogonal, in which $\Gamma_4^{(2)}$ acts as a rotation. Starting with the ansatz, in analogy to \eqref{eq:basisansatz}, 
\begin{equation}
\lambda_1 = \begin{pmatrix} 1 \\ 0 \\ 0 \end{pmatrix} \,, \quad \lambda_2 = \begin{pmatrix} a \\ b \\ 0 \end{pmatrix} \,, \quad \lambda_3 = \begin{pmatrix} c \\ d \\ e \end{pmatrix} \,,
\end{equation}
we can go through analogous steps as in the two-dimensional case and find an appropriate, right-handed basis given by
\begin{equation}
\lambda_1 = \begin{pmatrix} 1 \\ 0 \\ 0 \end{pmatrix} \,, \quad \lambda_2 = \begin{pmatrix} - \tfrac{1}{2} \\ \tfrac{\sqrt{3}}{2} \\ 0 \end{pmatrix} \,, \quad \lambda_3 = \begin{pmatrix} 0 \\ - \tfrac{1}{\sqrt{3}} \\ \sqrt{\tfrac{2}{3}} \end{pmatrix} \,.
\end{equation}
The explicit action of $\Gamma_4^{(2)}$ reads
\begin{equation}
\lambda_1 \mapsto \lambda_1 + \lambda_2 + \lambda_3 \,, \quad \lambda_2 \mapsto - \lambda_1 \,, \quad \lambda_3 \mapsto - \lambda_2 \,.
\end{equation}
As in the case of $\Gamma^{(2)}_3$ we want to obtain an easier understanding of this action by enlarging the fundamental domain. Guided by the fact that $\lambda_1 + \lambda_3$ maps to itself we choose
\begin{equation}
\widetilde{\lambda}_1 = \lambda_1 + \lambda_3 \,, \quad \widetilde{\lambda}_2 = \lambda_1 + \lambda_2 \,, \quad \widetilde{\lambda}_3 = \lambda_2 + \lambda_3 \,,
\end{equation}
with 
\begin{equation}
\text{det} \begin{pmatrix} 1 & 1 & 0 \\ 0 & 1 & 1 \\ 1 & 0 & 1 \end{pmatrix} = 2 \,,
\end{equation}
showing that it is a double cover of the original fundamental domain. The action of $\Gamma^{(2)}_4$ on this new basis is given by
\begin{equation}
\widetilde{\lambda}_1 \rightarrow \widetilde{\lambda}_1 \,, \quad \widetilde{\lambda}_2 \rightarrow \widetilde{\lambda}_3 \,, \quad \widetilde{\lambda}_3 \rightarrow - \widetilde{\lambda}_2 \,,
\end{equation}
which shows that it acts as $\gamma_4$, i.e., rotation by $\tfrac{\pi}{2}$, on $\widetilde{\lambda}_2$ and $\widetilde{\lambda}_3$. Thus, we find that taking the quotient by $(\mathbb{Z}_4)_{\Gamma^{(2)}_4}$ of the double cover $\widetilde{T}^3$ one has
\begin{equation}
\widetilde{T}^3 / (\mathbb{Z}_4)_{\Gamma^{(2)}_4} = (T^2 / \mathbb{Z}_4) \times S^1 \,.
\end{equation}
The internal lattice points are given by
\begin{equation}
p_1 = \begin{pmatrix} 0 \\ 0 \\ 0 \end{pmatrix} \,, \quad p_2 = \begin{pmatrix} \tfrac{1}{2} \\ \tfrac{1}{2 \sqrt{3}} \\ \sqrt{\tfrac{2}{3}}\end{pmatrix} = \lambda_1 + \lambda_2 + \lambda_3 = \tfrac{1}{2} \big(\widetilde{\lambda}_1 + \widetilde{\lambda}_2 + \widetilde{\lambda}_3 \big) \,.
\end{equation}
The singular fiber is therefore given by an additional quotient by a $\mathbb{Z}_2^s$ translational symmetry
\begin{equation}
T^3/ (\mathbb{Z}_4)_{\Gamma_4^{(2)}} = \big( (T^2/\mathbb{Z}_4) \times S^1 \big) / \mathbb{Z}^s_2
\end{equation}
as depicted in Figure~\ref{fig:T3Z4}.
\begin{figure}
\centering
\includegraphics[width = 0.6 \textwidth]{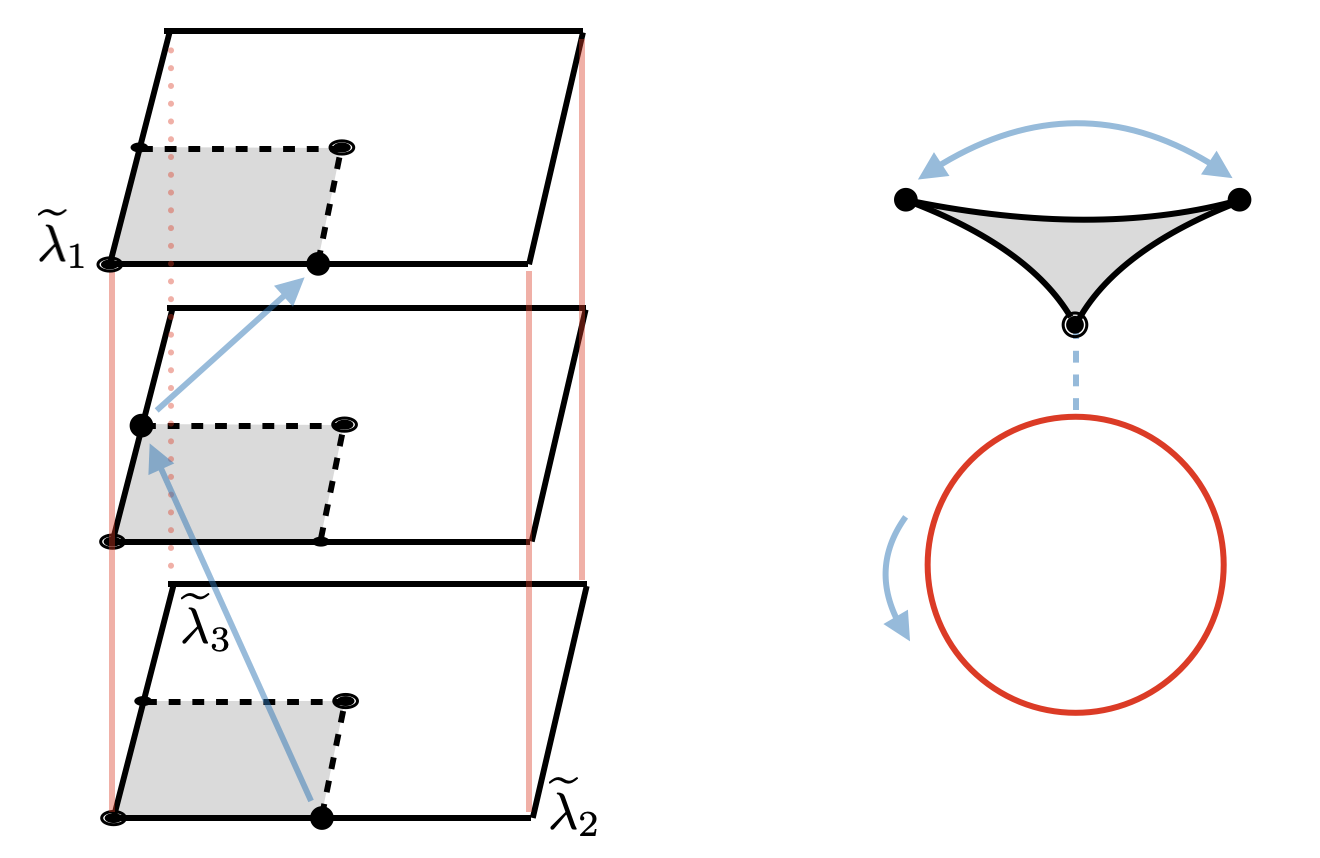}
\caption{The singular fiber $T^3 / (\mathbb{Z}_3)_{\Gamma_4^{(2)}}$ in terms of its double cover on the left, where the translational symmetry on the fixed points of $T^2/\mathbb{Z}_4$ is indicated by shaded blue arrows, and its fibration structure on the right.}
\label{fig:T3Z4}
\end{figure}
The action of $\mathbb{Z}_2^s$ can be understood as inducing a fibration of $T^2/\mathbb{Z}_4$ over a circle with periodicity $\widetilde{\lambda}_1$. Its action exchanges the two fixed points with local description of the form $\mathbb{C}/\mathbb{Z}_4$ and leaves the third fixed point of the form $\mathbb{C}/\mathbb{Z}_2$ invariant, giving two line of orbifold singularities.

Finally, we want to briefly discuss the singular fiber of the 3-torus modded and monodromy $M_2^{(2)}$ (an equivalent discussion applies for $M_1^{(2)}$), which defines the action
\begin{equation}
\lambda_1 \mapsto \lambda_3 \,, \quad \lambda_2 \mapsto - \lambda_1 - \lambda_2 - \lambda_3 \,, \quad \lambda_3 \mapsto \lambda_1 \,,
\end{equation}
from which we can define the invariant
\begin{equation}
    \widetilde{\lambda}_1 = \lambda_1 + \lambda_3 \mapsto \widetilde{\lambda}_1 \,,
\end{equation}
as well as the 
\begin{equation}
    \widetilde{\lambda}_2 = \lambda_1 + \lambda_2 \mapsto - \widetilde{\lambda}_2 \,, \quad \widetilde{\lambda}_3 = \lambda_2 + \lambda_3 \mapsto - \widetilde{\lambda}_3 \,.
\end{equation}
This shows that there is a double cover $\widetilde{T}^3$ of the original torus, such that
\begin{equation}
    \widetilde{T}^3 / (\mathbb{Z}_2)_{M_1^{(2)}} = (T^2/\mathbb{Z}_2) \times S^1 \,,
\end{equation}
with the $\mathbb{C}/\mathbb{Z}_2$ orbifold points at $\{ 0 \,, \tfrac{1}{2} \widetilde{\lambda}_2 \,, \tfrac{1}{2} \widetilde{\lambda}_3 \,, \tfrac{1}{2} \widetilde{\lambda}_2 + \tfrac{1}{2} \widetilde{\lambda}_3\}$
The only interior points are given by 
\begin{equation}
    p_1 = 0 \,, \quad p_2= 
    \tfrac{1}{2} \widetilde{\lambda}_1 + \tfrac{1}{2} \widetilde{\lambda}_2 + \tfrac{1}{2} \widetilde{\lambda}_3 = \lambda_1 + \lambda_2 + \lambda_3 \,,
\end{equation}
and one sees that shifts by it lead to shifts half-way around the $S^1$ and exchange the pair of orbifold points at $\{0 \,, \tfrac{1}{2} \widetilde{\lambda}_1 + \tfrac{1}{2} \widetilde{\lambda_2} \}$ and $\{ \tfrac{1}{2} \widetilde{\lambda}_1 \,, \tfrac{1}{2} \widetilde{\lambda}_2\}$. Thus, one has the singular geometry
\begin{equation}
    T^3 / (\mathbb{Z}_2)_{M_1^{(2)}} = \big( (T^2/\mathbb{Z}_2) \times S^1\big) / \mathbb{Z}_2^s
\end{equation}
sketched in Figure \ref{fig:T3Z2}. 
\begin{figure}
\centering
\includegraphics[width = 0.6 \textwidth]{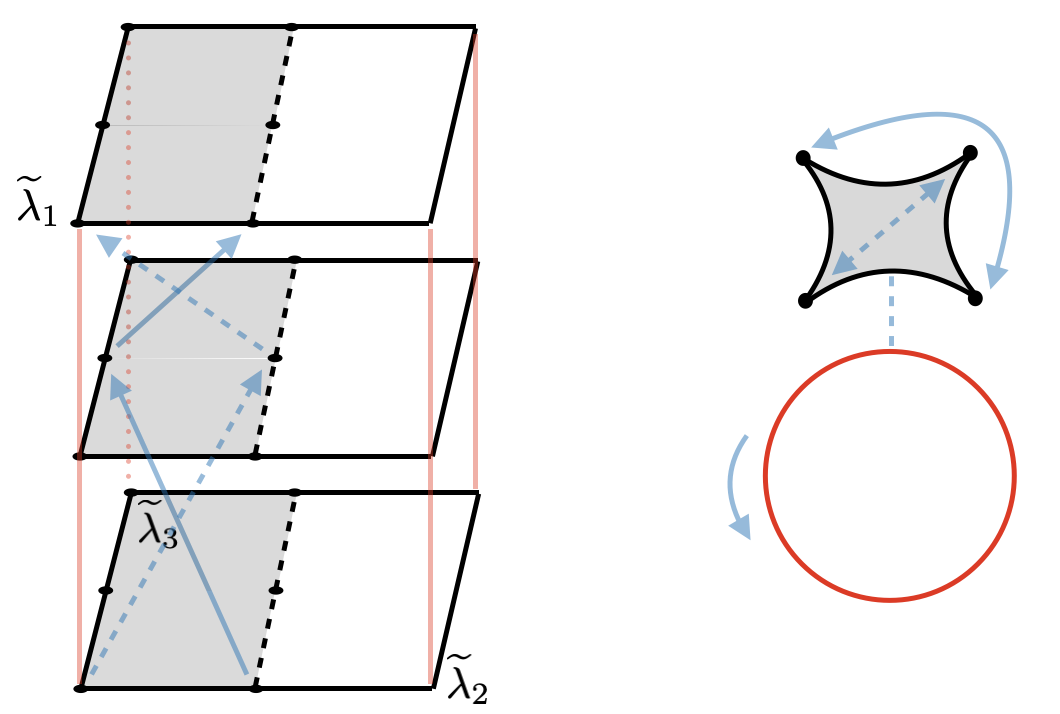}
\caption{The singular fiber $T^3 / (\mathbb{Z}_2)_{M_1^{(2)}}$ in terms of its double cover on the left, where the translational symmetry on the fixed points of $T^2/\mathbb{Z}_2$ is indicated by shaded blue arrows, and its fibration structure on the right.}
\label{fig:T3Z2}
\end{figure}

\subsection{Non-geometric singular fibers}
\label{subapp:singnongeo}

Finally, we want to apply the same method in the context of exceptional field theory, where we understand the simultaneous SL$(2,\mathbb{Z})$ and SL$(3,\mathbb{Z})$ action as the tensor product acting on a 6-torus (the auxiliary internal space).\footnote{Here, we will not explicitly work out a good basis for the lattice $\Lambda_6$ but determine the fibration structure form a carefully chosen multi-cover.} We will only consider the singular fiber for
\begin{equation}
    \gamma_3 \otimes \Gamma_3^{(2)} = \begin{pmatrix} 0 & -1 & 0 & 0 & -1 & 0 \\ 0 & 0 & -1 & 0 & 0 & -1 \\ -1 & 0 & 0 & -1 & 0 & 0 \\ 0 & 1 & 0 & 0 & 0 & 0 \\ 0 & 0 & 1 & 0 & 0 & 0 \\ 1 & 0 & 0 & 0 & 0 & 0 \end{pmatrix} \,,
\end{equation}
which acts as\footnote{In terms of the two index notation $\lambda_{\alpha a}$ one has: $\lambda_1 = \lambda_{1,1} \,, \lambda_2 = \lambda_{2,1} \,, \lambda_3 = \lambda_{3,1} \,, \lambda_4 = \lambda_{1,2} \,, \lambda_5 = \lambda_{2,2} \,, \lambda_6 = \lambda_{3,2}$ with SL$(3,\mathbb{Z})$ index $\alpha$ and SL$(2,\mathbb{Z})$ index $a$.}
\begin{equation}
    \begin{split}
        \lambda_1 &\mapsto - \lambda_2 - \lambda_5 \,, \quad \lambda_2 \mapsto - \lambda_3 - \lambda_6 \,, \quad \lambda_3 \mapsto - \lambda_1 - \lambda_4 \,, \\
        \lambda_4 &\mapsto \lambda_2 \,, \quad \lambda_5 \mapsto\lambda_3 \,, \quad \lambda_6 \mapsto \lambda_1 \,.
    \end{split}
\end{equation}
of course the same approach can be used for $\gamma_4 \otimes \Gamma_4^{(2)}$.

This action has a two-dimensional invariant subspace spanned by the lattice vectors
\begin{equation}
    \widetilde{\lambda}_1 = - \lambda_1 + \lambda_3 - \lambda_4 + \lambda_5 \,, \quad \widetilde{\lambda}_2 = - \lambda_2 + \lambda_3 -\lambda_4 + \lambda_6 \,,
\end{equation}
with $\widetilde{\lambda}_1 \mapsto \widetilde{\lambda}_1$ and $\widetilde{\lambda}_2 \mapsto \widetilde{\lambda}_2$. Next, we define the directions
\begin{equation}
\begin{split}
    \widetilde{\lambda}_3 &= \lambda_1 + \lambda_2 + \lambda_3 \mapsto - \lambda_1 - \lambda_2 - \lambda_3 - \lambda_4 - \lambda_5 - \lambda_6 = - \widetilde{\lambda}_3 - \widetilde{\lambda}_4  \,, \\
    \widetilde{\lambda}_4 &= \lambda_4 + \lambda_5 + \lambda_6 \mapsto \lambda_1 + \lambda_2 + \lambda_3 = \widetilde{\lambda}_3 \,,
\end{split}
\end{equation}
and
\begin{equation}
    \begin{split}
        \widetilde{\lambda}_5 &= \lambda_3 - \lambda_5 \mapsto - \lambda_3 - \lambda_1 - \lambda_4 = - \widetilde{\lambda}_5 - \widetilde{\lambda}_6 \,, \\
        \widetilde{\lambda}_6 &= \lambda_1 + \lambda_4 + \lambda_5 \mapsto \lambda_3-\lambda_5 = \widetilde{\lambda}_5 \,.
    \end{split}
\end{equation}
The $\widetilde{\lambda}_i$ define a $\widetilde{T}^6$ which is a 9-fold cover of the original $T^6$ over which the monodromy action simplifies and one has
\begin{equation}
    \widetilde{T}^6 / (\mathbb{Z}_3)_{\gamma_3 \otimes \Gamma_3^{(2)}} = T_{(12)}^2 \times (T_{(34)}^2/\mathbb{Z}_3) \times (T_{(56)}^2/\mathbb{Z}_3) \,,
\end{equation}
where we label the various sub-tori by their basis in terms of the $\widetilde{\lambda}_i$. This orbifold has nine orbifold singularities of the form $\mathbb{C}^2/\mathbb{Z}_3$ at the combinations of orbifold points of the two-dimensional sub-tori (at $\{0 \,, \tfrac{1}{3} \widetilde{\lambda}_3 + \tfrac{2}{3} \widetilde{\lambda}_4 \,, \tfrac{2}{3} \widetilde{\lambda}_3 + \tfrac{1}{3} \widetilde{\lambda}_4\}$ and $\{0 \,, \tfrac{1}{3} \widetilde{\lambda}_5 + \tfrac{2}{3} \widetilde{\lambda}_6 \,, \tfrac{2}{3} \widetilde{\lambda}_5 + \tfrac{1}{3} \widetilde{\lambda}_6\}$, respectively). This means that in the original basis the orbifold points are at
\begin{equation}
\begin{split}
\enspace \tfrac{1}{3} \widetilde{\lambda}_3 + \tfrac{2}{3} \widetilde{\lambda}_4 &= \tfrac{1}{3} \big( \lambda_1 + \lambda_2 + \lambda_3 \big) + \tfrac{2}{3} \big( \lambda_4 + \lambda_5 + \lambda_6 \big) \,, \\
\tfrac{2}{3} \widetilde{\lambda}_3 + \tfrac{1}{3} \widetilde{\lambda}_4 &= \tfrac{2}{3} \big( \lambda_1 + \lambda_2 + \lambda_3 \big) + \tfrac{1}{3} \big( \lambda_4 + \lambda_5 + \lambda_6 \big) \,, \enspace \tfrac{1}{3} \widetilde{\lambda}_5 + \tfrac{2}{3} \widetilde{\lambda}_6 = \tfrac{2}{3} \big( \lambda_1 + \lambda_4 \big) + \tfrac{1}{3} \big( \lambda_3 + \lambda_5 \big) \,, \\ 
\tfrac{2}{3} \widetilde{\lambda}_5 + \tfrac{1}{3} \widetilde{\lambda}_6 &= \tfrac{1}{3} \big( \lambda_1 + \lambda_4 - \lambda_5 \big) + \tfrac{2}{3} \lambda_3 \,.
\end{split}
\end{equation}
As above, to obtain the original fundamental domain one needs to mod our by shifts by interior points. These are given by:
\begin{equation}
\begin{split}
    p_1 &= 0  \,, \\ 
    p_2 &= \tfrac{1}{3} \widetilde{\lambda}_1 + \tfrac{2}{3} \widetilde{\lambda}_5 + \tfrac{1}{3} \widetilde{\lambda}_6 = \lambda_3 \,, \\ 
    p_3 &= \tfrac{2}{3} \widetilde{\lambda}_1 + \tfrac{1}{3} \widetilde{\lambda}_5 + \tfrac{2}{3} \widetilde{\lambda}_6 = \lambda_3 + \lambda_5 \,, \\
    p_4 &= \tfrac{1}{3} \big( \widetilde{\lambda}_1 + \widetilde{\lambda}_2 \big) + \tfrac{1}{3} \widetilde{\lambda}_3 + \tfrac{2}{3} \widetilde{\lambda}_4 = \lambda_3 + \lambda_5 + \lambda_6 \,, \\
    p_5 &= \tfrac{2}{3} \big( \widetilde{\lambda}_1 + \widetilde{\lambda}_2 \big) + \tfrac{2}{3} \widetilde{\lambda}_3 + \tfrac{1}{3} \widetilde{\lambda}_4 = 2 \lambda_3 - \lambda_4 + \lambda_5 + \lambda_6 \,, \\
    p_6 &= \tfrac{1}{3} \widetilde{\lambda}_2 + \tfrac{1}{3} \widetilde{\lambda}_3  + \tfrac{2}{3} \widetilde{\lambda}_4 + \tfrac{1}{3} \widetilde{\lambda}_5 + \tfrac{2}{3} \widetilde{\lambda}_6 = \lambda_1 + \lambda_3 + \lambda_4 + \lambda_5 + \lambda_6 \,, \\
    p_7 &= \tfrac{2}{3} \widetilde{\lambda}_2 + \tfrac{2}{3} \widetilde{\lambda}_3 + \tfrac{1}{3} \widetilde{\lambda}_4 + \tfrac{2}{3} \widetilde{\lambda}_5 + \tfrac{1}{3} \widetilde{\lambda}_6  = \lambda_1 + 2 \lambda_3 + \lambda_6 \,, \\ 
    p_8 &= \tfrac{1}{3} \widetilde{\lambda}_1 + \tfrac{2}{3} \widetilde{\lambda}_2 + \tfrac{2}{3} \widetilde{\lambda}_3 + \tfrac{1}{3} \widetilde{\lambda}_4 + \tfrac{1}{3} \widetilde{\lambda}_5 + \tfrac{2}{3} \widetilde{\lambda}_6 = \lambda_1 + 2 \lambda_3 + \lambda_5 + \lambda_6 \,, \\
    p_9 &= \tfrac{2}{3} \widetilde{\lambda}_1 + \tfrac{1}{3} \widetilde{\lambda}_2 + \tfrac{1}{3} \widetilde{\lambda}_3 + \tfrac{2}{3} \widetilde{\lambda}_4 + \tfrac{2}{3} \widetilde{\lambda}_5 + \tfrac{1}{3} \widetilde{\lambda}_6 = 2 \lambda_3 + \lambda_4 + \lambda_5 \,.
\end{split}
\end{equation}
As for the geometric cases above, the action exchanges the orbifold fixed point and can be written in terms of the two $\mathbb{Z}_3$ actions
\begin{equation}
\begin{split}
    \mathbb{Z}_3^s \text{ on } T^2_{(12)} \times (T^2_{(56)}/\mathbb{Z}_3):& \quad (z_{(12)} \,, z_{(34)} \,, z_{(56)}) \mapsto (z_{(12)} + \tfrac{1}{3} \widetilde{\lambda}_1 \,, z_{(34)} \,, z_{(56)} + \tfrac{2}{3} \widetilde{\lambda}_5 + \tfrac{1}{3} \widetilde{\lambda}_6) \,, \\
    \widetilde{\mathbb{Z}}_3^s  \text{ on } T^2_{(12)} \times (T^2_{(34)}/\mathbb{Z}_3):& \quad (z_{(12)} \,, z_{(34)} \,, z_{(56)}) \mapsto (z_{(12)} + \tfrac{1}{3} \widetilde{\lambda}_1 + \tfrac{1}{3} \widetilde{\lambda}_2 \,, z_{(34)} + \tfrac{1}{3} \widetilde{\lambda}_3 + \tfrac{2}{3} \widetilde{\lambda}_4 \,, z_{(56)}) \,,
\end{split}
\end{equation}
where we used the complex coordinate $z_{(ij)}$ for the 2-subtorus $T^2_{(ij)}$. From this action we see that we can understand the full singular $T^6$ fiber as a fibration of $(T^2_{(34)}/\mathbb{Z}_3) \times (T^2_{(56)}/ \mathbb{Z}_3)$ over a base $T^2$ spanned by $\tfrac{1}{3} \widetilde{\lambda}_1$ and $\tfrac{1}{3} \widetilde{\lambda}_1 + \tfrac{1}{3} \widetilde{\lambda}_2$, which leads to the volume reduction to that of the original required for the 9-fold cover, and the action on the fibers specified above and depicted in Figure \ref{fig:ngorb}.
\begin{figure}
\centering
\includegraphics[width = 0.45 \textwidth]{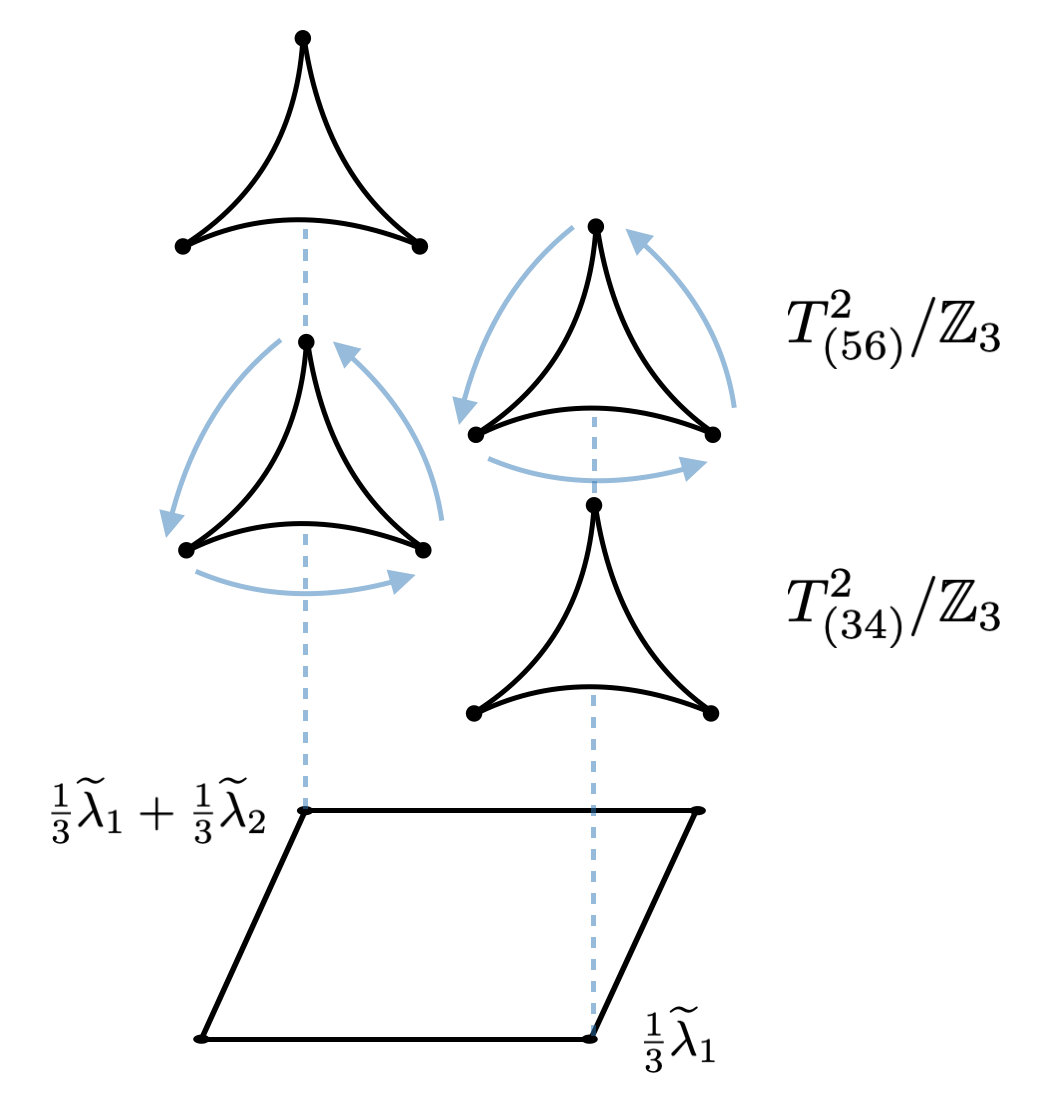}
\caption{Sketch of the singular fiber of the non-geometric quotient $T^6/(\mathbb{Z}_3)_{(\gamma_3, \Gamma_3^{(2)})}$, which can be understood as the fibration of two $T^2/\mathbb{Z}_3$ orbifolds over a base torus spanned by $\tfrac{1}{3} \widetilde{\lambda}_1$ and $\tfrac{1}{3} \widetilde{\lambda}_1 + \tfrac{1}{3} \widetilde{\lambda}_2$.}
\label{fig:ngorb}
\end{figure}

From the covering and the action one can already see that the singular fiber involves all of the 6-torus coordinates. It is therefore clear that the action cannot be made geometric in any duality frame and will involve the variation of other moduli fields. To illustrate this, let us solve the section constraint, by demanding that everything only depends on the internal coordinates $Y^{1,a}$, which translates to the periodicity conditions imposed by $\lambda_1$ and $\lambda_4$. Varying the coordinates over their full range
\begin{equation}
    Y^{1,1} \in [0,\lambda_1] \,, \quad Y^{1,2} \in [0,\lambda_4] \,,
\end{equation}
we find the following two-dimensional subspace in $\widetilde{T}^6$ (parameterized by coordinates $\widetilde{Y}^i \sim \widetilde{Y}^i + \widetilde{\lambda}_i$, with $i \in \{1\,,\dots, 6\}$):
\begin{equation}
    \widetilde{Y}^{i} = \begin{pmatrix} - Y^{1,1} - Y^{1,2} \\ -Y^{1,2} \\ Y^{1,1} \\ Y^{1,2} \\ 0 \\ Y^{1,1} + Y^{1,2}\end{pmatrix} \,,
\end{equation}
which is further subjects to the orbifold quotients and shifts, which in general has a very complicated singularity structure, whose discussion goes beyond the current investigation, but will be interesting to study in the future.

\section{Knit Product}
\label{App:Knit}
This appendix gives a brief overview of the knit product appearing in the description of the U-duality group given in \cref{lem:UdualityKnit}. We start with the definition of the knit product. As with the direct and semi-direct products, there is an internal knit product and an external knit product. We start with the internal definition.
\begin{defn}[Internal Knit Product]
    Let $G$ be a group and $H,K<G$. If $G = HK$ and $H\cap K = \{e\}$, then $G$ is said to be the internal knit product of $H$ and $K$, denoted $H\bowtie K$.
\end{defn}
The slightly more complicated definition is that of the external knit product. 
\begin{defn}[External Knit Product]
    Suppose that $H$ and $K$ are groups and suppose that there exist mappings $\alpha:K\times H \rightarrow H$ and $\beta:K\times H \rightarrow K$ satisfying the following properties:
    \begin{itemize}
        \item $\alpha(e,h) = h$ and $\beta(k,e) = k$ for all $h \in H$ and $k \in K$. 
        \item $\alpha(k_1k_2,h) = \alpha(k_1,\alpha(k_2,h))$
        \item $\beta(k,h_1h_2) = \beta(\beta(k,h_1),h_2)$
        \item $\alpha(k,h_1h_2) = \alpha(k,h_1)\alpha(\beta(k,h_1),h_2)$
        \item $\beta(k_1k_2,h) = \beta(k_1,\alpha(k_2,h))\beta(k_2,h)$
    \end{itemize}
    for all $h_1,h_2\in H$ and $k_1,k_2\in K$. The first three properties assert that the mapping $\alpha:K\times H \rightarrow H$ is a left action and $\beta:K\times H \rightarrow K$ is a right action. On the cartesian product $H \times K$, we then define a multiplication and an inversion mapping by 
    \begin{itemize}
        \item $(h_1,k_1)(h_2,k_2) = (h_1\alpha(k_1,h_2),\beta(k_1,h_2)k_2)$
        \item $(h,k)^{-1}=(\alpha(k^{-1},h^{-1}),\beta(k^{-1},h^{-1}))$
    \end{itemize}
    Then $H\times K$ is a group call the external knit product of $H$ and $K$, denoted $H\bowtie K$. Note that $H\times \{e\}$ and $\{e\}\times K$ are subgroups isomorphic to $H$ and $K$ and $H\times K$ is an internal knit product of $H\times \{e\}$ and $\{e\}\times K$. 
\end{defn}
The knit product is a natural generalization of the semi-direct product. For example, $G = X\rtimes Y$ requires that $X$ is a normal subgroup while $G = X\bowtie Y$ does not. Furthermore, the internal semi-direct product is a generalization of the internal direct product which requires both $X$ and $Y$ to be normal subgroups of $G$. We thus have
\begin{equation}
    \text{Knit Product}\supset \text{Semi-Direct Product}\supset \text{Direct Product}.
\end{equation}

\newpage
\bibliography{Utopia}
\bibliographystyle{utphys}

\end{document}